\documentclass[11pt,a4paper,twoside]{article}

\usepackage[ngerman,USenglish]{babel} 

\usepackage{amsmath}                
\usepackage{amsfonts}               
\usepackage{amssymb}                
\usepackage{amsopn}                 
\allowdisplaybreaks[3] 
\usepackage{amsthm}                
\usepackage{bbm}                    
\usepackage{mathrsfs}               
\usepackage{array}                  
\usepackage{IEEEtrantools}
\usepackage{calc}                   
\usepackage{graphicx}        
\usepackage{psfrag}                 
\usepackage{subfigure}              
\usepackage{wrapfig}                
\usepackage{color}           
\usepackage{fancyhdr}               
\usepackage{verbatim}               
\usepackage{moreverb}               
\usepackage{exscale}                
\usepackage{multirow}               
\usepackage{cite}
\usepackage{centernot}

\usepackage[multiple]{footmisc}

\usepackage{tikz}
\usetikzlibrary{circuits.logic.US,circuits.logic.IEC}
\usetikzlibrary{arrows,positioning} 
\usetikzlibrary{plotmarks,decorations.pathreplacing}

\usepackage{float}
\usepackage{wasysym}
\usepackage{amsmath}
\usepackage{amssymb}
\usepackage{amsfonts}
\usepackage{mathrsfs}
\usepackage{xspace}
\usepackage{bm}
\usepackage{fancyref}
\usepackage{textcomp}
\usepackage[Symbol]{upgreek}
\usepackage{centernot}
\usepackage{graphicx}

\newcommand{\safemath}[2]{\newcommand{#1}{\ensuremath{#2}\xspace}}
\renewcommand{\safemath}[2]{\newcommand{#1}{\ensuremath{#2}\xspace}}

\newcommand{\ssa}{\mathsf{a}}
\newcommand{\ssb}{\mathsf{b}}
\newcommand{\ssc}{\mathsf{c}}
\newcommand{\ssd}{\mathsf{d}}
\newcommand{\sse}{\mathsf{e}}
\newcommand{\ssf}{\mathsf{f}}
\newcommand{\ssg}{\mathsf{g}}
\newcommand{\ssh}{\mathsf{h}}
\newcommand{\ssi}{\mathsf{i}}
\newcommand{\ssj}{\mathsf{j}}
\newcommand{\ssk}{\mathsf{k}}
\newcommand{\ssl}{\mathsf{l}}
\newcommand{\ssm}{\mathsf{m}}
\newcommand{\ssn}{\mathsf{n}}
\newcommand{\sso}{\mathsf{o}}
\newcommand{\ssp}{\mathsf{p}}
\newcommand{\ssq}{\mathsf{q}}
\newcommand{\ssr}{\mathsf{r}}
\newcommand{\sss}{\mathsf{s}}
\newcommand{\sst}{\mathsf{t}}
\newcommand{\ssu}{\mathsf{u}}
\newcommand{\ssv}{\mathsf{v}}
\newcommand{\ssw}{\mathsf{w}}
\newcommand{\ssx}{\mathsf{x}}
\newcommand{\ssy}{\mathsf{y}}
\newcommand{\ssz}{\mathsf{z}}

\safemath{\bmsa}{\bm{\ssa}}
\safemath{\bmsb}{\bm{\ssb}}
\safemath{\bmsc}{\bm{\ssc}}
\safemath{\bmsd}{\bm{\ssd}}
\safemath{\bmse}{\bm{\sse}}
\safemath{\bmsf}{\bm{\ssf}}
\safemath{\bmsg}{\bm{\ssg}}
\safemath{\bmsh}{\bm{\ssh}}
\safemath{\bmsi}{\bm{\ssi}}
\safemath{\bmsj}{\bm{\ssj}}
\safemath{\bmsk}{\bm{\ssk}}
\safemath{\bmsl}{\bm{\ssl}}
\safemath{\bmsm}{\bm{\ssm}}
\safemath{\bmsn}{\bm{\ssn}}
\safemath{\bmso}{\bm{\sso}}
\safemath{\bmsp}{\bm{\ssp}}
\safemath{\bmsq}{\bm{\ssq}}
\safemath{\bmsr}{\bm{\ssr}}
\safemath{\bmss}{\bm{\sss}}
\safemath{\bmst}{\bm{\sst}}
\safemath{\bmsu}{\bm{\ssu}}
\safemath{\bmsv}{\bm{\ssv}}
\safemath{\bmsw}{\bm{\ssw}}
\safemath{\bmsx}{\bm{\ssx}}
\safemath{\bmsy}{\bm{\ssy}}
\safemath{\bmsz}{\bm{\ssz}}

\bmdefine{\bmualphad}{\upalpha}
\bmdefine{\bmubetad}{\upbeta}
\bmdefine{\bmuchid}{\upchi}
\bmdefine{\bmudeltad}{\updelta}
\bmdefine{\bmuepsilond}{\upepsilon}
\bmdefine{\bmuvarepsilond}{\upvarepsilon}
\bmdefine{\bmuetad}{\upeta}
\bmdefine{\bmugammad}{\upgamma}
\bmdefine{\bmuiotad}{\upiota}
\bmdefine{\bmukappad}{\upkappa}
\bmdefine{\bmulambdad}{\uplambda}
\bmdefine{\bmumud}{\upmu}
\bmdefine{\bmunud}{\upnu}
\bmdefine{\bmuomegad}{\upomega}
\bmdefine{\bmuphid}{\upphi}
\bmdefine{\bmuvarphid}{\upvarphi}
\bmdefine{\bmupid}{\uppi}
\bmdefine{\bmuvarpid}{\upvarpi}
\bmdefine{\bmupsid}{\uppsi}
\bmdefine{\bmurhod}{\uprho}
\bmdefine{\bmuvarrhod}{\upvarrho}
\bmdefine{\bmusigmad}{\upsigma}
\bmdefine{\bmuvarsigmad}{\upvarsigma}
\bmdefine{\bmutaud}{\uptau}
\bmdefine{\bmuthetad}{\uptheta}
\bmdefine{\bmuvarthetad}{\upvartheta}
\bmdefine{\bmuupsilond}{\upupsilon}
\bmdefine{\bmuxid}{\upxi}
\bmdefine{\bmuzetad}{\upzeta}

\safemath{\bmua}{\mathbf{a}}
\safemath{\bmub}{\mathbf{b}}
\safemath{\bmuc}{\mathbf{c}}
\safemath{\bmud}{\mathbf{d}}
\safemath{\bmue}{\mathbf{e}}
\safemath{\bmuf}{\mathbf{f}}
\safemath{\bmug}{\mathbf{g}}
\safemath{\bmuh}{\mathbf{h}}
\safemath{\bmui}{\mathbf{i}}
\safemath{\bmuj}{\mathbf{j}}
\safemath{\bmuk}{\mathbf{k}}
\safemath{\bmul}{\mathbf{l}}
\safemath{\bmum}{\mathbf{m}}
\safemath{\bmun}{\mathbf{n}}
\safemath{\bmuo}{\mathbf{o}}
\safemath{\bmup}{\mathbf{p}}
\safemath{\bmuq}{\mathbf{q}}
\safemath{\bmur}{\mathbf{r}}
\safemath{\bmus}{\mathbf{s}}
\safemath{\bmut}{\mathbf{t}}
\safemath{\bmuu}{\mathbf{u}}
\safemath{\bmuv}{\mathbf{v}}
\safemath{\bmuw}{\mathbf{w}}
\safemath{\bmux}{\mathbf{x}}
\safemath{\bmuy}{\mathbf{y}}
\safemath{\bmuz}{\mathbf{z}}

\safemath{\bmualpha}{\bmualphad}
\safemath{\bmubeta}{\bmubetad}
\safemath{\bmuchi}{\bumchid}
\safemath{\bmudelta}{\bmudeltad}
\safemath{\bmuepsilon}{\bmuepsilond}
\safemath{\bmuvarepsilon}{\bmuvarepsilond}
\safemath{\bmueta}{\bmuetad}
\safemath{\bmugamma}{\bmugammad}
\safemath{\bmuiota}{\bmuiotad}
\safemath{\bmukappa}{\bmukappad}
\safemath{\bmulambda}{\bmulambdad}
\safemath{\bmumu}{\bmumud}
\safemath{\bmunu}{\bmunud}
\safemath{\bmuomega}{\bmuomegad}
\safemath{\bmuphi}{\bmuphid}
\safemath{\bmuvarphi}{\bmuvarphid}
\safemath{\bmupi}{\bmupid}
\safemath{\bmuvarpi}{\bmuvarpid}
\safemath{\bmupsi}{\bmupsid}
\safemath{\bmurho}{\bmurhod}
\safemath{\bmuvarrho}{\bmuvarrhod}
\safemath{\bmusigma}{\bmusigmad}
\safemath{\bmuvarsigma}{\bmuvarsigmad}
\safemath{\bmutau}{\bmutaud}
\safemath{\bmutheta}{\bmuthetad}
\safemath{\bmuvartheta}{\bmuvarthetad}
\safemath{\bmuupsilon}{\bmuupsilond}
\safemath{\bmuxi}{\bmuxid}
\safemath{\bmuzeta}{\bmuzetad}

\bmdefine{\bmiad}{a}
\bmdefine{\bmibd}{b}
\bmdefine{\bmicd}{c}
\bmdefine{\bmidd}{d}
\bmdefine{\bmied}{e}
\bmdefine{\bmifd}{f}
\bmdefine{\bmigd}{g}
\bmdefine{\bmihd}{h}
\bmdefine{\bmiid}{i}
\bmdefine{\bmijd}{j}
\bmdefine{\bmikd}{k}
\bmdefine{\bmild}{l}
\bmdefine{\bmimd}{m}
\bmdefine{\bmind}{n}
\bmdefine{\bmiod}{o}
\bmdefine{\bmipd}{p}
\bmdefine{\bmiqd}{q}
\bmdefine{\bmird}{r}
\bmdefine{\bmisd}{s}
\bmdefine{\bmitd}{t}
\bmdefine{\bmiud}{u}
\bmdefine{\bmivd}{v}
\bmdefine{\bmiwd}{w}
\bmdefine{\bmixd}{x}
\bmdefine{\bmiyd}{y}
\bmdefine{\bmizd}{z}

\bmdefine{\bmialphad}{\alpha}
\bmdefine{\bmibetad}{\beta}
\bmdefine{\bmichid}{\chi}
\bmdefine{\bmideltad}{\delta}
\bmdefine{\bmiepsilond}{\epsilon}
\bmdefine{\bmivarepsilond}{\varepsilon}
\bmdefine{\bmietad}{\eta}
\bmdefine{\bmigammad}{\gamma}
\bmdefine{\bmiiotad}{\iota}
\bmdefine{\bmikappad}{\kappa}
\bmdefine{\bmivarkappad}{\varkappa}
\bmdefine{\bmilambdad}{\lambda}
\bmdefine{\bmimud}{\mu}
\bmdefine{\bminud}{\nu}
\bmdefine{\bmiomegad}{\omega}
\bmdefine{\bmiphid}{\phi}
\bmdefine{\bmivarphid}{\varphi}
\bmdefine{\bmipid}{\pi}
\bmdefine{\bmivarpid}{\varpi}
\bmdefine{\bmipsid}{\psi}
\bmdefine{\bmirhod}{\rho}
\bmdefine{\bmivarrhod}{\varrho}
\bmdefine{\bmisigmad}{\sigma}
\bmdefine{\bmivarsigmad}{\varsigma}
\bmdefine{\bmitaud}{\tau}
\bmdefine{\bmithetad}{\theta}
\bmdefine{\bmivarthetad}{\vartheta}
\bmdefine{\bmiupsilond}{\upsilon}
\bmdefine{\bmixid}{\xi}
\bmdefine{\bmizetad}{\zeta}

\safemath{\bmia}{\bmiad}
\safemath{\bmib}{\bmibd}
\safemath{\bmic}{\bmicd}
\safemath{\bmid}{\bmidd}
\safemath{\bmie}{\bmied}
\safemath{\bmif}{\bmifd}
\safemath{\bmig}{\bmigd}
\safemath{\bmih}{\bmihd}
\safemath{\bmii}{\bmiid}
\safemath{\bmij}{\bmijd}
\safemath{\bmik}{\bmikd}
\safemath{\bmil}{\bmild}
\safemath{\bmim}{\bmimd}
\safemath{\bmin}{\bmind}
\safemath{\bmio}{\bmiod}
\safemath{\bmip}{\bmipd}
\safemath{\bmiq}{\bmiqd}
\safemath{\bmir}{\bmird}
\safemath{\bmis}{\bmisd}
\safemath{\bmit}{\bmitd}
\safemath{\bmiu}{\bmiud}
\safemath{\bmiv}{\bmivd}
\safemath{\bmiw}{\bmiwd}
\safemath{\bmix}{\bmixd}
\safemath{\bmiy}{\bmiyd}
\safemath{\bmiz}{\bmizd}

\safemath{\bmialpha}{\bmialphad}
\safemath{\bmibeta}{\bmibetad}
\safemath{\bmichi}{\bmichid}
\safemath{\bmidelta}{\bmideltad}
\safemath{\bmiepsilon}{\bmiepsilond}
\safemath{\bmivarepsilon}{\bmivarepsilond}
\safemath{\bmieta}{\bmietad}
\safemath{\bmigamma}{\bmigammad}
\safemath{\bmiiota}{\bmiiotad}
\safemath{\bmikappa}{\bmikappad}
\safemath{\bmivarkappa}{\bmivarkappad}
\safemath{\bmilambda}{\bmilambdad}
\safemath{\bmimu}{\bmimud}
\safemath{\bminu}{\bminud}
\safemath{\bmiomega}{\bmiomegad}
\safemath{\bmiphi}{\bmiphid}
\safemath{\bmivarphi}{\bmivarphid}
\safemath{\bmipi}{\bmipid}
\safemath{\bmivarpi}{\bmivarpid}
\safemath{\bmipsi}{\bmipsid}
\safemath{\bmirho}{\bmirhod}
\safemath{\bmivarrho}{\bmivarrhod}
\safemath{\bmisigma}{\bmisigmad}
\safemath{\bmivarsigma}{\bmivarsigmad}
\safemath{\bmitau}{\bmitaud}
\safemath{\bmitheta}{\bmithetad}
\safemath{\bmivartheta}{\bmivarthetad}
\safemath{\bmiupsilon}{\bmiupsilond}
\safemath{\bmixi}{\bmixid}
\safemath{\bmizeta}{\bmizetad}

\bmdefine{\bmuDeltad}{\Updelta}
\bmdefine{\bmuGammad}{\Upgamma}
\bmdefine{\bmuLambdad}{\Uplambda}
\bmdefine{\bmuOmegad}{\Upomega}
\bmdefine{\bmuPhid}{\Upphi}
\bmdefine{\bmuPid}{\Uppi}
\bmdefine{\bmuPsid}{\Uppsi}
\bmdefine{\bmuSigmad}{\Upsigma}
\bmdefine{\bmuThetad}{\Uptheta}
\bmdefine{\bmuUpsilond}{\Upupsilon}
\bmdefine{\bmuXid}{\Upxi}

\safemath{\bmuA}{\mathbf{A}}
\safemath{\bmuB}{\mathbf{B}}
\safemath{\bmuC}{\mathbf{C}}
\safemath{\bmuD}{\mathbf{D}}
\safemath{\bmuE}{\mathbf{E}}
\safemath{\bmuF}{\mathbf{F}}
\safemath{\bmuG}{\mathbf{G}}
\safemath{\bmuH}{\mathbf{H}}
\safemath{\bmuI}{\mathbf{I}}
\safemath{\bmuJ}{\mathbf{J}}
\safemath{\bmuK}{\mathbf{K}}
\safemath{\bmuL}{\mathbf{L}}
\safemath{\bmuM}{\mathbf{M}}
\safemath{\bmuN}{\mathbf{N}}
\safemath{\bmuO}{\mathbf{O}}
\safemath{\bmuP}{\mathbf{P}}
\safemath{\bmuQ}{\mathbf{Q}}
\safemath{\bmuR}{\mathbf{R}}
\safemath{\bmuS}{\mathbf{S}}
\safemath{\bmuT}{\mathbf{T}}
\safemath{\bmuU}{\mathbf{U}}
\safemath{\bmuV}{\mathbf{V}}
\safemath{\bmuW}{\mathbf{W}}
\safemath{\bmuX}{\mathbf{X}}
\safemath{\bmuY}{\mathbf{Y}}
\safemath{\bmuZ}{\mathbf{Z}}

\safemath{\bmuZero}{\mathbf{0}}
\safemath{\bmuOne}{\mathbf{1}}

\safemath{\bmuDelta}{\bmuDeltad}
\safemath{\bmuGamma}{\bmuGammad}
\safemath{\bmuLambda}{\bmuLambdad}
\safemath{\bmuOmega}{\bmuOmegad}
\safemath{\bmuPhi}{\bmuPhid}
\safemath{\bmuPi}{\bmuPid}
\safemath{\bmuPsi}{\bmuPsid}
\safemath{\bmuSigma}{\bmuSigmad}
\safemath{\bmuTheta}{\bmuThetad}
\safemath{\bmuUpsilon}{\bmuUpsilond}
\safemath{\bmuXi}{\bmuXid}

\bmdefine{\bmiAd}{A}
\bmdefine{\bmiBd}{B}
\bmdefine{\bmiCd}{C}
\bmdefine{\bmiDd}{D}
\bmdefine{\bmiEd}{E}
\bmdefine{\bmiFd}{F}
\bmdefine{\bmiGd}{G}
\bmdefine{\bmiHd}{H}
\bmdefine{\bmiId}{I}
\bmdefine{\bmiJd}{J}
\bmdefine{\bmiKd}{K}
\bmdefine{\bmiLd}{L}
\bmdefine{\bmiMd}{M}
\bmdefine{\bmiOd}{N}
\bmdefine{\bmiPd}{O}
\bmdefine{\bmiQd}{P}
\bmdefine{\bmiRd}{R}
\bmdefine{\bmiSd}{S}
\bmdefine{\bmiTd}{T}
\bmdefine{\bmiUd}{U}
\bmdefine{\bmiVd}{V}
\bmdefine{\bmiWd}{W}
\bmdefine{\bmiXd}{X}
\bmdefine{\bmiYd}{Y}
\bmdefine{\bmiZd}{Z}

\bmdefine{\bmiDeltad}{\Delta}
\bmdefine{\bmiGammad}{\Gamma}
\bmdefine{\bmiLambdad}{\Lambda}
\bmdefine{\bmiOmegad}{\Omega}
\bmdefine{\bmiPhid}{\Phi}
\bmdefine{\bmiPid}{\Pi}
\bmdefine{\bmiPsid}{\Psi}
\bmdefine{\bmiSigmad}{\Sigma}
\bmdefine{\bmiThetad}{\Theta}
\bmdefine{\bmiUpsilond}{\Upsilon}
\bmdefine{\bmiXid}{\Xi}

\safemath{\bmiA}{\bmiAd}
\safemath{\bmiB}{\bmiBd}
\safemath{\bmiC}{\bmiCd}
\safemath{\bmiD}{\bmiDd}
\safemath{\bmiE}{\bmiEd}
\safemath{\bmiF}{\bmiFd}
\safemath{\bmiG}{\bmiGd}
\safemath{\bmiH}{\bmiHd}
\safemath{\bmiI}{\bmiId}
\safemath{\bmiJ}{\bmiJd}
\safemath{\bmiK}{\bmiKd}
\safemath{\bmiL}{\bmiLd}
\safemath{\bmiM}{\bmiMd}
\safemath{\bmiN}{\bmiNd}
\safemath{\bmiO}{\bmiOd}
\safemath{\bmiP}{\bmiPd}
\safemath{\bmiQ}{\bmiQd}
\safemath{\bmiR}{\bmiRd}
\safemath{\bmiS}{\bmiSd}
\safemath{\bmiT}{\bmiTd}
\safemath{\bmiU}{\bmiUd}
\safemath{\bmiV}{\bmiVd}
\safemath{\bmiW}{\bmiWd}
\safemath{\bmiX}{\bmiXd}
\safemath{\bmiY}{\bmiYd}
\safemath{\bmiZ}{\bmiZd}

\safemath{\bmiDelta}{\bmiDeltad}
\safemath{\bmiGamma}{\bmiGammad}
\safemath{\bmiLambda}{\bmiLambdad}
\safemath{\bmiOmega}{\bmiOmegad}
\safemath{\bmiPhi}{\bmiPhid}
\safemath{\bmiPi}{\bmiPid}
\safemath{\bmiPsi}{\bmiPsid}
\safemath{\bmiSigma}{\bmiSigmad}
\safemath{\bmiTheta}{\bmiThetad}
\safemath{\bmiUpsilon}{\bmiUpsilond}
\safemath{\bmiXi}{\bmiXid}

\safemath{\setA}{\mathcal{A}}
\safemath{\setB}{\mathcal{B}}
\safemath{\setC}{\mathcal{C}}
\safemath{\setD}{\mathcal{D}}
\safemath{\setE}{\mathcal{E}}
\safemath{\setF}{\mathcal{F}}
\safemath{\setG}{\mathcal{G}}
\safemath{\setH}{\mathcal{H}}
\safemath{\setI}{\mathcal{I}}
\safemath{\setJ}{\mathcal{J}}
\safemath{\setK}{\mathcal{K}}
\safemath{\setL}{\mathcal{L}}
\safemath{\setM}{\mathcal{M}}
\safemath{\setN}{\mathcal{N}}
\safemath{\setO}{\mathcal{O}}
\safemath{\setP}{\mathcal{P}}
\safemath{\setQ}{\mathcal{Q}}
\safemath{\setR}{\mathcal{R}}
\safemath{\setS}{\mathcal{S}}
\safemath{\setT}{\mathcal{T}}
\safemath{\setU}{\mathcal{U}}
\safemath{\setV}{\mathcal{V}}
\safemath{\setW}{\mathcal{W}}
\safemath{\setX}{\mathcal{X}}
\safemath{\setY}{\mathcal{Y}}
\safemath{\setZ}{\mathcal{Z}}
\safemath{\emptySet}{\varnothing}

\safemath{\colA}{\mathscr{A}}
\safemath{\colB}{\mathscr{B}}
\safemath{\colC}{\mathscr{C}}
\safemath{\colD}{\mathscr{D}}
\safemath{\colE}{\mathscr{E}}
\safemath{\colF}{\mathscr{F}}
\safemath{\colG}{\mathscr{G}}
\safemath{\colH}{\mathscr{H}}
\safemath{\colI}{\mathscr{I}}
\safemath{\colJ}{\mathscr{J}}
\safemath{\colK}{\mathscr{K}}
\safemath{\colL}{\mathscr{L}}
\safemath{\colM}{\mathscr{M}}
\safemath{\colN}{\mathscr{N}}
\safemath{\colO}{\mathscr{O}}
\safemath{\colP}{\mathscr{P}}
\safemath{\colQ}{\mathscr{Q}}
\safemath{\colR}{\mathscr{R}}
\safemath{\colS}{\mathscr{S}}
\safemath{\colT}{\mathscr{T}}
\safemath{\colU}{\mathscr{U}}
\safemath{\colV}{\mathscr{V}}
\safemath{\colW}{\mathscr{W}}
\safemath{\colX}{\mathscr{X}}
\safemath{\colY}{\mathscr{Y}}
\safemath{\colZ}{\mathscr{Z}}

\safemath{\opA}{\mathbb{A}}
\safemath{\opB}{\mathbb{B}}
\safemath{\opC}{\mathbb{C}}
\safemath{\opD}{\mathbb{D}}
\safemath{\opE}{\mathbb{E}}
\safemath{\opF}{\mathbb{F}}
\safemath{\opG}{\mathbb{G}}
\safemath{\opH}{\mathbb{H}}
\safemath{\opI}{\mathbb{I}}
\safemath{\opJ}{\mathbb{J}}
\safemath{\opK}{\mathbb{K}}
\safemath{\opL}{\mathbb{L}}
\safemath{\opM}{\mathbb{M}}
\safemath{\opN}{\mathbb{N}}
\safemath{\opO}{\mathbb{O}}
\safemath{\opP}{\mathbb{P}}
\safemath{\opQ}{\mathbb{Q}}
\safemath{\opR}{\mathbb{R}}
\safemath{\opS}{\mathbb{S}}
\safemath{\opT}{\mathbb{T}}
\safemath{\opU}{\mathbb{U}}
\safemath{\opV}{\mathbb{V}}
\safemath{\opW}{\mathbb{W}}
\safemath{\opX}{\mathbb{X}}
\safemath{\opY}{\mathbb{Y}}
\safemath{\opZ}{\mathbb{Z}}
\safemath{\opZero}{\mathbb{O}}
\safemath{\identityop}{\opI}

\safemath{\sca}{a}
\safemath{\scb}{b}
\safemath{\scc}{c}
\safemath{\scd}{d}
\safemath{\sce}{e}
\safemath{\scf}{f}
\safemath{\scg}{g}
\safemath{\sch}{h}
\safemath{\sci}{i}
\safemath{\scj}{j}
\safemath{\sck}{k}
\safemath{\scl}{l}
\safemath{\scm}{m}
\safemath{\scn}{n}
\safemath{\sco}{o}
\safemath{\scp}{p}
\safemath{\scq}{q}
\safemath{\scr}{r}
\safemath{\scs}{s}
\safemath{\sct}{t}
\safemath{\scu}{u}
\safemath{\scv}{v}
\safemath{\scw}{w}
\safemath{\scx}{x}
\safemath{\scy}{y}
\safemath{\scz}{z}

\safemath{\scA}{A}
\safemath{\scB}{B}
\safemath{\scC}{C}
\safemath{\scD}{D}
\safemath{\scE}{E}
\safemath{\scF}{F}
\safemath{\scG}{G}
\safemath{\scH}{H}
\safemath{\scI}{I}
\safemath{\scJ}{J}
\safemath{\scK}{K}
\safemath{\scL}{L}
\safemath{\scM}{M}
\safemath{\scN}{N}
\safemath{\scO}{O}
\safemath{\scP}{P}
\safemath{\scQ}{Q}
\safemath{\scR}{R}
\safemath{\scS}{S}
\safemath{\scT}{T}
\safemath{\scU}{U}
\safemath{\scV}{V}
\safemath{\scW}{W}
\safemath{\scX}{X}
\safemath{\scY}{Y}
\safemath{\scZ}{Z}

\safemath{\scalpha}{\alpha}
\safemath{\scbeta}{\beta}
\safemath{\scchi}{\chi}
\safemath{\scdelta}{\delta}
\safemath{\scepsilon}{\epsilon}
\safemath{\scvarepsilon}{\varepsilon}
\safemath{\sceta}{\eta}
\safemath{\scgamma}{\gamma}
\safemath{\sciota}{\iota}
\safemath{\sckappa}{\kappa}
\safemath{\scvarkappa}{\varkappa}
\safemath{\sclambda}{\lambda}
\safemath{\scmu}{\mu}
\safemath{\scnu}{\nu}
\safemath{\scomega}{\omega}
\safemath{\scphi}{\phi}
\safemath{\scvarphi}{\varphi}
\safemath{\scpi}{\pi}
\safemath{\scvarpi}{\varpi}
\safemath{\scpsi}{\psi}
\safemath{\scrho}{\rho}
\safemath{\scvarrho}{\varrho}
\safemath{\scsigma}{\sigma}
\safemath{\scvarsigma}{\varsigma}
\safemath{\sctau}{\tau}
\safemath{\sctheta}{\theta}
\safemath{\scvartheta}{\vartheta}
\safemath{\scupsilon}{\upsilon}
\safemath{\scxi}{\xi}
\safemath{\sczeta}{\zeta}

\safemath{\veca}{\mathbf{a}}
\safemath{\vecb}{\mathbf{b}}
\safemath{\vecc}{\mathbf{c}}
\safemath{\vecd}{\mathbf{d}}
\safemath{\vece}{\mathbf{e}}
\safemath{\vecf}{\mathbf{f}}
\safemath{\vecg}{\mathbf{g}}
\safemath{\vech}{\mathbf{h}}
\safemath{\veci}{\mathbf{i}}
\safemath{\vecj}{\mathbf{j}}
\safemath{\veck}{\mathbf{k}}
\safemath{\vecl}{\mathbf{l}}
\safemath{\vecm}{\mathbf{m}}
\safemath{\vecn}{\mathbf{n}}
\safemath{\veco}{\mathbf{o}}
\safemath{\vecp}{\mathbf{p}}
\safemath{\vecq}{\mathbf{q}}
\safemath{\vecr}{\mathbf{r}}
\safemath{\vecs}{\mathbf{s}}
\safemath{\vect}{\mathbf{t}}
\safemath{\vecu}{\mathbf{u}}
\safemath{\vectU}{\mathbf{U}}
\safemath{\vecv}{\mathbf{v}}
\safemath{\vecw}{\mathbf{w}}
\safemath{\vecx}{\mathbf{x}}
\safemath{\vectX}{\mathbf{X}}
\safemath{\vecy}{\mathbf{y}}
\safemath{\vecz}{\mathbf{z}}
\safemath{\veczero}{\mathbf{0}}
\safemath{\vecone}{\mathbf{1}}

\safemath{\vecalpha}{\upalpha}
\safemath{\vecbeta}{\upbeta}
\safemath{\vecchi}{\upchi}
\safemath{\vecdelta}{\updelta}
\safemath{\vecepsilon}{\upepsilon}
\safemath{\vecvarepsilon}{\upvarepsilon}
\safemath{\veceta}{\upeta}
\safemath{\vecgamma}{\upgamma}
\safemath{\veciota}{\upiota}
\safemath{\veckappa}{\upkappa}
\safemath{\veclambda}{\uplambda}
\safemath{\vecmu}{\text{\textmu}}
\safemath{\vecnu}{\upnu}
\safemath{\vecomega}{\upomega}
\safemath{\vecphi}{\upphi}
\safemath{\vecvarphi}{\upvarphi}
\safemath{\vecpi}{\uppi}
\safemath{\vecvarpi}{\upvarpi}
\safemath{\vecpsi}{\uppsi}
\safemath{\vecrho}{\uprho}
\safemath{\vecvarrho}{\upvarrho}
\safemath{\vecsigma}{\upsigma}
\safemath{\vecvarsigma}{\upvarsigma}
\safemath{\vectau}{\uptau}
\safemath{\vectheta}{\uptheta}
\safemath{\vecvartheta}{\upvartheta}
\safemath{\vecupsilon}{\upupsilon}
\safemath{\vecxi}{\upxi}
\safemath{\veczeta}{\upzeta}

\safemath{\matA}{\mathrm{A}}
\safemath{\matB}{\mathrm{B}}
\safemath{\matC}{\mathrm{C}}
\safemath{\matD}{\mathrm{D}}
\safemath{\matE}{\mathrm{E}}
\safemath{\matF}{\mathrm{F}}
\safemath{\matG}{\mathrm{G}}
\safemath{\matH}{\mathrm{H}}
\safemath{\matI}{\mathrm{I}}
\safemath{\matJ}{\mathrm{J}}
\safemath{\matK}{\mathrm{K}}
\safemath{\matL}{\mathrm{L}}
\safemath{\matM}{\mathrm{M}}
\safemath{\matN}{\mathrm{N}}
\safemath{\matO}{\mathrm{O}}
\safemath{\matP}{\mathrm{P}}
\safemath{\matQ}{\mathrm{Q}}
\safemath{\matR}{\mathrm{R}}
\safemath{\matS}{\mathrm{S}}
\safemath{\matT}{\mathrm{T}}
\safemath{\matU}{\mathrm{U}}
\safemath{\matV}{\mathrm{V}}
\safemath{\matW}{\mathrm{W}}
\safemath{\matX}{\mathrm{X}}
\safemath{\matY}{\mathrm{Y}}
\safemath{\matZ}{\mathrm{Z}}
\safemath{\matzero}{\mathrm{0}}

\safemath{\matDelta}{\Updelta}
\safemath{\matGamma}{\Upgammma}
\safemath{\matLambda}{\Uplambda}
\safemath{\matOmega}{\Upomega}
\safemath{\matPhi}{\Upphi}
\safemath{\matPi}{\Uppi}
\safemath{\matPsi}{\Uppsi}
\safemath{\matSigma}{\Upsigma}
\safemath{\matTheta}{\Uptheta}
\safemath{\matUpsilon}{\Upupsilon}
\safemath{\matXi}{\Upxi}

\safemath{\matidentity}{\matI}
\safemath{\matone}{\matO}

\safemath{\rnda}{\bmia}
\safemath{\rndb}{\bmib}
\safemath{\rndc}{\bmic}
\safemath{\rndd}{\bmid}
\safemath{\rnde}{\bmie}
\safemath{\rndf}{\bmif}
\safemath{\rndg}{\bmig}
\safemath{\rndh}{\bmih}
\safemath{\rndi}{\bmii}
\safemath{\rndj}{\bmij}
\safemath{\rndk}{\bmik}
\safemath{\rndl}{\bmil}
\safemath{\rndm}{\bmim}
\safemath{\rndn}{\bmin}
\safemath{\rndo}{\bmio}
\safemath{\rndp}{\bmip}
\safemath{\rndq}{\bmiq}
\safemath{\rndr}{\bmir}
\safemath{\rnds}{\bmis}
\safemath{\rndt}{\bmit}
\safemath{\rndu}{\bmiu}
\safemath{\rndv}{\bmiv}
\safemath{\rndw}{\bmiw}
\safemath{\rndx}{\bmix}
\safemath{\rndy}{\bmiy}
\safemath{\rndz}{\bmiz}

\safemath{\rndA}{\scA}
\safemath{\rndB}{\scB}
\safemath{\rndC}{\scC}
\safemath{\rndD}{\scD}
\safemath{\rndE}{\scE}
\safemath{\rndF}{\scF}
\safemath{\rndG}{\scG}
\safemath{\rndH}{\scH}
\safemath{\rndI}{\scI}
\safemath{\rndJ}{\scJ}
\safemath{\rndK}{\scK}
\safemath{\rndL}{\scL}
\safemath{\rndM}{\scM}
\safemath{\rndN}{\scN}
\safemath{\rndO}{\scO}
\safemath{\rndP}{\scP}
\safemath{\rndQ}{\scQ}
\safemath{\rndR}{\scR}
\safemath{\rndS}{\scS}
\safemath{\rndT}{\scT}
\safemath{\rndU}{\scU}
\safemath{\rndV}{\scV}
\safemath{\rndW}{\scW}
\safemath{\rndX}{\scX}
\safemath{\rndY}{\scY}
\safemath{\rndZ}{\scZ}

\safemath{\rndalpha}{\bmialpha}
\safemath{\rndbeta}{\bmibeta}
\safemath{\rndchi}{\bmichi}
\safemath{\rnddelta}{\bmidelta}
\safemath{\rndepsilon}{\bmiepsilon}
\safemath{\rndvarepsilon}{\bmivarepsilon}
\safemath{\rndeta}{\bmieta}
\safemath{\rndgamma}{\bmigamma}
\safemath{\rndiota}{\bmiiota}
\safemath{\rndkappa}{\bmikappa}
\safemath{\rndlambda}{\bmilambda}
\safemath{\rndmu}{\bmimu}
\safemath{\rndnu}{\bminu}
\safemath{\rndomega}{\bmiomega}
\safemath{\rndphi}{\bmiphi}
\safemath{\rndvarphi}{\bmivarphi}
\safemath{\rndpi}{\bmipi}
\safemath{\rndvarpi}{\bmivarpi}
\safemath{\rndpsi}{\bmipsi}
\safemath{\rndrho}{\bmirho}
\safemath{\rndvarrho}{\bmivarrho}
\safemath{\rndsigma}{\bmisigma}
\safemath{\rndvarsigma}{\bmivarsigma}
\safemath{\rndtau}{\bmitau}
\safemath{\rndtheta}{\bmitheta}
\safemath{\rndvartheta}{\bmivartheta}
\safemath{\rndupsilon}{\bmiupsilon}
\safemath{\rndxi}{\bmixi}
\safemath{\rndzeta}{\bmizeta}

\safemath{\rveca}{\bmua}
\safemath{\rvecb}{\bmub}
\safemath{\rvecc}{\bmuc}
\safemath{\rvecd}{\bmud}
\safemath{\rvece}{\bmue}
\safemath{\rvecf}{\bmuf}
\safemath{\rvecg}{\bmug}
\safemath{\rvech}{\bmuh}
\safemath{\rveci}{\bmui}
\safemath{\rvecj}{\bmuj}
\safemath{\rveck}{\bmuk}
\safemath{\rvecl}{\bmul}
\safemath{\rvecm}{\bmum}
\safemath{\rvecn}{\bmun}
\safemath{\rveco}{\bmuo}
\safemath{\rvecp}{\bmup}
\safemath{\rvecq}{\bmuq}
\safemath{\rvecr}{\bmur}
\safemath{\rvecs}{\bmus}
\safemath{\rvect}{\bmut}
\safemath{\rvecu}{\bmuu}
\safemath{\rvecv}{\bmuv}
\safemath{\rvecw}{\bmuw}
\safemath{\rvecx}{\bmux}
\safemath{\rvecy}{\bmuy}
\safemath{\rvecz}{\bmuz}

\safemath{\rvecalpha}{\bmualpha}
\safemath{\rvecbeta}{\bmubeta}
\safemath{\rvecchi}{\bmuchi}
\safemath{\rvecdelta}{\bmudelta}
\safemath{\rvecepsilon}{\bmuepsilon}
\safemath{\rvecvarepsilon}{\bmuvarepsilon}
\safemath{\rveceta}{\bmueta}
\safemath{\rvecgamma}{\bmugamma}
\safemath{\rveciota}{\bmuiota}
\safemath{\rveckappa}{\bmukappa}
\safemath{\rveclambda}{\bmulambda}
\safemath{\rvecmu}{\bmumu}
\safemath{\rvecnu}{\bmunu}
\safemath{\rvecomega}{\bmuomega}
\safemath{\rvecphi}{\bmuphi}
\safemath{\rvecvarphi}{\bmuvarphi}
\safemath{\rvecpi}{\bmupi}
\safemath{\rvecvarpi}{\bmuvarpi}
\safemath{\rvecpsi}{\bmupsi}
\safemath{\rvecrho}{\bmurho}
\safemath{\rvecvarrho}{\bmuvarrho}
\safemath{\rvecsigma}{\bmusigma}
\safemath{\rvecvarsigma}{\bmuvarsigma}
\safemath{\rvectau}{\bmutau}
\safemath{\rvectheta}{\bmutheta}
\safemath{\rvecvartheta}{\bmuvartheta}
\safemath{\rvecupsilon}{\bmuupsilon}
\safemath{\rvecxi}{\bmuxi}
\safemath{\rveczeta}{\bmuzeta}

\safemath{\rmatA}{\bmuA}
\safemath{\rmatB}{\bmuB}
\safemath{\rmatC}{\bmuC}
\safemath{\rmatD}{\bmuD}
\safemath{\rmatE}{\bmuE}
\safemath{\rmatF}{\bmuF}
\safemath{\rmatG}{\bmuG}
\safemath{\rmatH}{\bmuH}
\safemath{\rmatI}{\bmuI}
\safemath{\rmatJ}{\bmuJ}
\safemath{\rmatK}{\bmuK}
\safemath{\rmatL}{\bmuL}
\safemath{\rmatM}{\bmuM}
\safemath{\rmatN}{\bmuN}
\safemath{\rmatO}{\bmuO}
\safemath{\rmatP}{\bmuP}
\safemath{\rmatQ}{\bmuQ}
\safemath{\rmatR}{\bmuR}
\safemath{\rmatS}{\bmuS}
\safemath{\rmatT}{\bmuT}
\safemath{\rmatU}{\bmuU}
\safemath{\rmatV}{\bmuV}
\safemath{\rmatW}{\bmuW}
\safemath{\rmatX}{\bmuX}
\safemath{\rmatY}{\bmuY}
\safemath{\rmatZ}{\bmuZ}

\safemath{\rmatDelta}{\bmuDelta}
\safemath{\rmatGamma}{\bmuGamma}
\safemath{\rmatLambda}{\bmuLambda}
\safemath{\rmatOmega}{\bmuOmega}
\safemath{\rmatPhi}{\bmuPhi}
\safemath{\rmatPi}{\bmuPi}
\safemath{\rmatPsi}{\bmuPsi}
\safemath{\rmatSigma}{\bmuSigma}
\safemath{\rmatTheta}{\bmuTheta}
\safemath{\rmatUpsilon}{\bmuUpsilon}
\safemath{\rmatXi}{\bmuXi}

\safemath{\rndvecV}{\bmiV} 
\safemath{\rndvecU}{\bmiU} 
\safemath{\rndvecW}{\bmiW} 
\safemath{\rndvecX}{\bmiX} 
\safemath{\rndvecY}{\bmiY} 
\safemath{\rndvecZ}{\bmiZ} 
\safemath{\rndvecC}{\bmiC} 
\safemath{\rndvecS}{\bmiS} 

\newenvironment{textbmatrix}{	\setlength{\arraycolsep}{2.5pt}%
								\big[\begin{matrix}}{\end{matrix}\big]%
								\raisebox{0.08ex}{\vphantom{M}}}

\def\be{\begin{equation}}
\def\ee{\end{equation}}
\def\een{\nonumber \end{equation}}
\def\mat{\begin{bmatrix}}
\def\emat{\end{bmatrix}}
\def\btm{\begin{textbmatrix}}
\def\etm{\end{textbmatrix}}

\def\ba#1\ea{\begin{align}#1\end{align}}
\def\bas#1\eas{\begin{align*}#1\end{align*}}
\def\bs#1\es{\begin{split}#1\end{split}} 
\def\bg#1\eg{\begin{gather}#1\end{gather}}
\def\bml#1\eml{\begin{multline}#1\end{multline}}
\def\bi#1\ei{\begin{itemize}#1\end{itemize}} 
\def\bipi#1\eipi{\begin{inparaitem}#1\end{inparaitem}}

\newcommand{\lefto}{\mathopen{}\left}

\DeclareMathOperator{\adj}{adj}				
\DeclareMathOperator{\Exop}{\opE}			

\safemath{\fun}{\scf}						
\safemath{\altfun}{\scg}
\safemath{\aaltfun}{\sch}
\safemath{\bel}{\sce}					
\safemath{\altbel}{\sce}					
\safemath{\frel}{g}					
\safemath{\altfrel}{g}					
\safemath{\dfrel}{\tilde{g}}					
\safemath{\altdfrel}{\tilde{g}}					

\newcommand{\nullspace}{\setN}	 			
\newcommand{\Ex}[2]{\ensuremath{\Exop_{#1}\lefto[#2\right]}} 	
\newcommand{\bigEx}[2]{\ensuremath{\Exop_{#1} \bigl[#2\bigr]}} 	

\newcommand{\ind}[1]{\mathbbm{1}_{#1}}				
\newcommand{\conj}[1]{\ensuremath{\overline{#1}}} 	
\newcommand{\inv}[1]{\ensuremath{#1^{-1}}} 	

\safemath{\dirac}{\delta}					
\safemath{\diracp}{\dirac(\time)}			
\safemath{\krond}{\dirac}					

\safemath{\upto}{\uparrow}
\safemath{\downto}{\downarrow}
\safemath{\iu}{i}							
\safemath{\maj}{\succ}
\newcommand{\dftmat}[1]{\matF_{#1}}			
\safemath{\mdft}{\dftmat{}}					
\safemath{\runity}{\beta}					
\safemath{\eval}{\lambda}					
\safemath{\veval}{\veclambda}				
\safemath{\littleo}{\sco}					

\let\im\undefined
\safemath{\re}{\mathfrak{Re}}				
\safemath{\im}{\mathfrak{Im}}				

\safemath{\euclidspace}{\complexset}			
\safemath{\confunspace}{\setC}				
\newcommand{\banachseqspace}[1]{l^{#1}}		
\safemath{\hilseqspace}{\banachseqspace{2}}	
\newcommand{\banachfunspace}[1]{\setL^{#1}}	
\safemath{\hilfunspace}{\banachfunspace{2}}	
\safemath{\schwarzspace}{\setS}				

\newcommand{\hadj}[1]{#1^{\star}}			

\safemath{\SNR}{\text{\sc snr}} 				
\safemath{\No}{N_0}							
\safemath{\Es}{E_s}							
\safemath{\Eb}{E_b}							
\safemath{\EbNo}{\frac{\Eb}{\No}}
\safemath{\EsNo}{\frac{\Es}{\No}}

\let\time\undefined
\safemath{\time}{\sct}						
\safemath{\dtime}{\sck}						
\safemath{\delay}{\sctau}					
\safemath{\ddelay}{\scl}						
\safemath{\doppler}{\scnu}					
\safemath{\ddoppler}{\scm}					
\safemath{\freq}{\scf}						
\safemath{\dfreq}{\scn}						
\safemath{\Dtime}{\Delta\time}
\safemath{\Dfreq}{\Delta\freq}
\safemath{\Ddtime}{\Delta\dtime}
\safemath{\Ddfreq}{\Delta\dfreq}
\safemath{\bandwidth}{\scB}
\safemath{\maxdoppler}{\doppler_{0}}			
\safemath{\maxdelay}{\delay_{0}}				
\safemath{\spread}{\Delta_{\CHop}}			
\DeclareMathOperator{\CHop}{\ensuremath{\opH}} 
\safemath{\kernel}{\rndk_{\CHop}}			
\safemath{\kernelp}{\kernel(\time,\time')}	
\safemath{\tvir}{\rndh_{\CHop}}				
\safemath{\tvirp}{\tvir(\time,\delay)}		
\safemath{\tvirc}{\conj{\rndh}_{\CHop}}		
\safemath{\tvtf}{\rndl_{\CHop}}				
\safemath{\tvtfp}{\tvtf(\time,\freq)}			
\safemath{\tvtfc}{\conj{\rndl}_{\CHop}}		
\safemath{\spf}{\rnds_{\CHop}}				
\safemath{\spfp}{\spf(\doppler,\delay)}		
\safemath{\spfc}{\conj{\rnds}_{\CHop}}		
\safemath{\bff}{\rndb_{\CHop}}				
\safemath{\bffp}{\bff(\doppler,\freq)}		

\safemath{\irc}{\scr_{\rndh}}				
\safemath{\tfc}{\scr_{\rndl}}				
\safemath{\spc}{\scr_{\rnds}}				
\safemath{\bfc}{\scr_{\rndb}}				

\safemath{\scaf}{\scc_{\rnds}}				
\safemath{\scafp}{\scaf(\doppler,\delay)}		
\safemath{\ccf}{\scc_{\rndl}}				
\safemath{\ccfp}{\ccf(\Dtime,\Dfreq)}			
\safemath{\cic}{\scc_{\rndh}}				
\safemath{\cicp}{\cic(\Dtime,\delay)}			

\safemath{\mi}{\scI}							
\safemath{\capacity}{\scC}					

\DeclareMathOperator{\Prob}{\opP}		
\safemath{\normal}{\mathcal{N}}			
\safemath{\jpg}{\mathcal{CN}}			
\safemath{\mchain}{\leftrightarrow}		

\safemath{\dB}{\,\mathrm{dB}}
\safemath{\dBm}{\,\mathrm{dBm}}
\safemath{\Hz}{\,\mathrm{Hz}}
\safemath{\kHz}{\,\mathrm{kHz}}
\safemath{\MHz}{\,\mathrm{MHz}}
\safemath{\GHz}{\,\mathrm{GHz}}
\safemath{\s}{\,\mathrm{s}}
\safemath{\ms}{\,\mathrm{ms}}
\safemath{\mus}{\,\mathrm{\text{\textmu}s}}
\safemath{\ns}{\,\mathrm{ns}}
\safemath{\ps}{\,\mathrm{ps}}
\safemath{\meter}{\,\mathrm{m}}
\safemath{\mm}{\,\mathrm{mm}}
\safemath{\cm}{\,\mathrm{cm}}
\safemath{\m}{\,\mathrm{m}}
\safemath{\W}{\,\mathrm{W}}
\safemath{\mW}{\, \mathrm{mW}}
\safemath{\J}{\,\mathrm{J}}
\safemath{\K}{\,\mathrm{K}}
\safemath{\bit}{\,\mathrm{bit}}
\safemath{\nat}{\,\mathrm{nat}}

\safemath{\define}{\triangleq}					

\safemath{\equivalent}{\sim}
\safemath{\distas}{\sim}					
\safemath{\sdiff}{\Delta}				

\safemath{\reals}{\mathbb R}
\safemath{\positivereals}{\reals_{+}}
\safemath{\integers}{\mathbb Z}
\safemath{\posint}{\integers_{+}}
\safemath{\naturals}{\mathbb N}
\safemath{\posnaturals}{\naturals_{+}}
\safemath{\complexset}{\mathbb C}
\safemath{\rationals}{\mathbb Q}

\newcommand*{\fancyrefparlabelprefix}{par}		
\newcommand*{\fancyrefchalabelprefix}{cha}		
\newcommand*{\fancyrefapplabelprefix}{app}		
\newcommand*{\fancyrefthmlabelprefix}{thm}		
\newcommand*{\fancyreflemlabelprefix}{lem}		
\newcommand*{\fancyrefcorlabelprefix}{cor}		
\newcommand*{\fancyrefdeflabelprefix}{def}		

\frefformat{vario}{\fancyrefparlabelprefix}{Part~#1}
\frefformat{vario}{\fancyrefchalabelprefix}{Chapter~#1}
\frefformat{vario}{\fancyrefseclabelprefix}{Section~#1}
\frefformat{vario}{\fancyrefthmlabelprefix}{Theorem~#1}
\frefformat{vario}{\fancyreflemlabelprefix}{Lemma~#1}
\frefformat{vario}{\fancyrefcorlabelprefix}{Corollary~#1}
\frefformat{vario}{\fancyrefdeflabelprefix}{Definition~#1}
\frefformat{vario}{\fancyreffiglabelprefix}{Figure~#1}
\frefformat{vario}{\fancyrefapplabelprefix}{Appendix~#1}
\frefformat{vario}{\fancyrefeqlabelprefix}{(#1)}

\safemath{\iSet}{\setI}
\safemath{\rel}{\bowtie}					
\safemath{\eqrel}{\sim}					
\safemath{\rlord}{\leq}					
\safemath{\slord}{<}						
\safemath{\rpord}{\preceq}				
\safemath{\rrpord}{\succeq}				
\safemath{\spord}{\prec}					
\safemath{\sig}{\sigma}					
\safemath{\metric}{d}					

\safemath{\setfun}{\Phi}					
\safemath{\measure}{\mu}					
\newcommand{\outerm}[1]{#1^{\star}}		
\newcommand{\innerm}[1]{#1_{\star}}		
\safemath{\omeasure}{\outerm{\measure}}		
\safemath{\imeasure}{\innerm{\measure}}		
\safemath{\aecol}{\colS^{\star}_{\measure}} 
\safemath{\emeasure}{\bar{\measure}_{0}}	
\safemath{\rmeasure}{\tilde{\measure}}	
\safemath{\bmeasure}{\measure_{0}}		
\safemath{\glength}{\measure_{\altfun}}	
\safemath{\lebmea}{\lambda}				
\safemath{\blebmea}{\lebmea_{0}}			
\safemath{\sfun}{s}						
\safemath{\absintspace}{\colL^{1}}		
\safemath{\sqintspace}{\colL^{2}}		
\safemath{\abssumspace}{l^{1}}		
\safemath{\sqsumspace}{l^{2}}		
\safemath{\boundspace}{l^{\infty}}	

\safemath{\sfield}{\setF}				
\safemath{\vectors}{\setV}				
\safemath{\vecspace}{(\vectors,\sfield)}	
\safemath{\linspace}{\setV}				
\safemath{\clinspace}{(\linspace,\sfield)} 
\safemath{\nspace}{\setU}				
\safemath{\metspace}{\setM}				
\safemath{\bspace}{\setB}				
\safemath{\ipspace}{\setG}				
\safemath{\hilspace}{\setH}				
\safemath{\blospace}{\setG}				
\safemath{\lop}{\opT}					
\safemath{\altlop}{\opS}					
\safemath{\nullsp}{\nullspace(\lop)}		
\safemath{\lfun}{l}						
\safemath{\altlfun}{g}					
\newcommand{\dual}[1]{#1^{'}}			
\safemath{\dsum}{\oplus}					
\safemath{\funspace}{\colL}				
\renewcommand{\adj}[1]{#1^{\times}}		
\safemath{\adjlop}{\adj{\lop}}			
\safemath{\hadjlop}{\hadj{\lop}}			
\safemath{\tow}{\xrightarrow{w}}			
\safemath{\tows}{\xrightarrow{w^{*}}}		
\safemath{\cparam}{\lambda}
\safemath{\lopl}{\lop_{\cparam}}		
\safemath{\iop}{\opI}					
\safemath{\resolop}{\opR}				
\safemath{\resolvent}{\resolop_{\cparam}(\lop)}	
\safemath{\altresolvent}{\resolop_{\cparam}(\altlop)} 
\safemath{\reset}{\setQ}
\safemath{\spectrum}{\setS}
\safemath{\resolset}{\reset(\lop)}		
\safemath{\lopspec}{\spectrum(\lop)}		
\safemath{\altlopspec}{\spectrum(\altlop)} 
\safemath{\pspec}{\spectrum_{p}(\lop)}	
\safemath{\cspec}{\spectrum_{c}(\lop)}	
\safemath{\rspec}{\spectrum_{r}(\lop)}	
\safemath{\ev}{\cparam}					
\newcommand{\specrad}[1]{r_{#1}}			
\safemath{\lopsrad}{\specrad{\lop}}		
\safemath{\pop}{\opP}					

\safemath{\specfam}{\colE}				
\safemath{\specop}{\opE_{\cparam}}		
\safemath{\altspecop}{\opE_{\mu}}		
\safemath{\vmulti}{\vecone}				
\safemath{\unitaryop}{\opU}				
\safemath{\sval}{\sigma}					
\safemath{\corrcoef}{\rho}				
\safemath{\sangle}{\theta}				

\let\time\undefined
\safemath{\iset}{\setI}				
\safemath{\shift}{\nu}
\safemath{\scale}{\alpha}
\safemath{\time}{t}
\safemath{\specfreq}{\theta}	
\newcommand{\transopgen}[1]{\opT_{#1}} 
\safemath{\transop}{\transopgen{\delay}}
\newcommand{\modopgen}[1]{\opM_{#1}}	
\safemath{\modop}{\modopgen{\shift}}
\newcommand{\dilopgen}[1]{\opD_{#1}}	
\safemath{\dilop}{\dilopgen{\scale}}
\safemath{\fram}{\setF}				
\safemath{\dfram}{\dual{\fram}}		
\safemath{\ufb}{\scB}					
\safemath{\lfb}{\scA}					
\safemath{\sop}{\hadj{\aop}}				
\safemath{\aop}{\opT}			
\safemath{\fop}{\opS}				
\safemath{\daop}{\tilde\opT}			
\safemath{\dsop}{\hadj{\tilde\opT}}				

\safemath{\ifop}{\inv{\fop}}			
\safemath{\rifop}{\fop^{-1/2}}			
\newcommand{\ft}[1]{\widehat{#1}}	
\safemath{\transeq}{\setT}			
\safemath{\nfun}{\Phi}				
\safemath{\funvec}{\vecf}			
\safemath{\altfunvec}{\vecg}

\safemath{\samplespace}{\Omega}
\safemath{\probspace}{(\samplespace,\sfield,\Prob)}	
\safemath{\ccoef}{\rho}			

\safemath{\infstate}{\vecpi}				
\safemath{\typset}{\setA_{\epsilon}^{(N)}}	
\safemath{\expequal}{\doteq}				
\safemath{\code}{C}						
\safemath{\dstringset}{\setD^{\star}}		
\safemath{\cwlength}{l}					
\safemath{\elength}{L}					
\safemath{\extension}{C^{\star}}			
\safemath{\approaches}{\rightarrow}		
\safemath{\evnt}{\setA}					
\safemath{\altevnt}{\setB}					
\safemath{\rv}{\rndx}					
\safemath{\altrv}{\rndy}					
\safemath{\complexrv}{\rndu}					
\safemath{\altcrv}{\rndv}				
\safemath{\rvec}{\rvecx}					
\safemath{\altrvec}{\rvecy}				
\safemath{\crvec}{\rvecu}				
\safemath{\altcrvec}{\rvecv}				
\safemath{\variance}{\sigma^{2}}			
\safemath{\map}{T}						
\safemath{\jacobian}{J}					
\safemath{\wvec}{\rvecw}					
\safemath{\wrv}{\rndw}					
\safemath{\orthmat}{\matQ}				
\safemath{\evmat}{\matLambda}			
\safemath{\identity}{\matidentity}		
\safemath{\innovec}{\vecv}				
\safemath{\convas}{\xrightarrow{\text{a.s.}}}	
\safemath{\convr}{\xrightarrow{\text{r}}}	
\safemath{\convp}{\xrightarrow{\text{P}}}	
\safemath{\convd}{\xrightarrow{\text{D}}}	
\safemath{\ltis}{\opL}				
\safemath{\ir}{h}					
\safemath{\tf}{\MakeUppercase{\ir}}	

{\begin{list}{}%
         {\setlength{\leftmargin}{#1}}%
         \item[]%
}
{\end{list}}

\safemath{\signal}{\scx} 
\safemath{\signalFt}{\widehat{\signal}}	
\safemath{\sigTime}{\signal \! \left( \time \right)}	
\safemath{\orFt}{\widehat{\signal} \! \left(\freq\right)}	
\safemath{\sampFt}{\widehat{\signal}_{\scd} \! \left(\freq\right)}	
\safemath{\fZero}{\freq_{0}}	
\safemath{\fSamp}{\freq_{\scs}}	
\safemath{\specOc}{\setI} 
\safemath{\sampSet}{\setP} 
\safemath{\sampVal}{\signal \! \left(\time_\scn\right)}	
\safemath{\beuDen}{\setD^{-} \! \left(\sampSet\right)}	
\safemath{\samPer}{\scT_\scs}	
\safemath{\hSp}{\setH}	
\safemath{\multSig}{\setB \! \left( \specOc \right)}	
\safemath{\Ltwo}{\sqintspace \! \left(\reals\right)}	
\safemath{\cellNo}{\scL} 
\safemath{\noIn}{s} 
\safemath{\cosetNo}{\scK} 
\safemath{\sampSig}{\scX \! \left(\freq\right)} 
\newcommand{\vandEnt}[1]{\scx_{#1}} 
\safemath{\vandM}{\matV \! \left(\vandEnt{0}, \vandEnt{1}, \ldots, \vandEnt{\cosetNo - 1} \right)} 
\safemath{\sampMat}{\matA} 
\safemath{\suppXhat}{\gamma} 
\safemath{\orFtSupp}{\widehat{\signal}_{\suppXhat} \! \left(\freq\right)} 
\safemath{\maxCard}{\scC} 

\safemath{\dict}{\matD} 
\safemath{\measVec}{\vecy} 
\safemath{\sigVec}{\vecx} 
\safemath{\meas}{\scm} 
\safemath{\sigDim}{\scn} 
\safemath{\spars}{\scs} 
\safemath{\FM}{\matF}	
\safemath{\IM}{\matI}	
\safemath{\fONB}{\matA}	
\safemath{\sONB}{\matB}	
\safemath{\dictCol}{d} 
\safemath{\recSigVec}{\hat{\sigVec}} 
\safemath{\PZ}{\left( \text{P0} \right)} 
\safemath{\BP}{\left( \text{BP} \right)} 
\safemath{\PO}{\left( \text{P1} \right)} 
\safemath{\coher}{\mu} 
\safemath{\supp}{\setS}	
\safemath{\corMeas}{\vecz} 
\safemath{\erVec}{\vece} 
\safemath{\concSig}{\check{\sigVec}} 
\safemath{\fVec}{\vecp} 
\safemath{\sVec}{\vecq} 
\safemath{\conVec}{\vecv} 
\safemath{\sgnl}{\vecs} 
\safemath{\fONBCol}{\veca} 
\safemath{\sONBCol}{\vecb} 
\safemath{\fSupp}{\setP} 
\safemath{\sSupp}{\setQ} 
\safemath{\fDim}{\sigDim_{\fONBCol}} 
\safemath{\sDim}{\sigDim_{\sONBCol}} 
\safemath{\fCoher}{\coher_{\fONBCol}} 
\safemath{\sCoher}{\coher_{\sONBCol}} 
\safemath{\sigSupp}{\setX} 
\safemath{\erSupp}{\setE} 
\safemath{\PZErSup}{\left( \text{P0, } \erSupp \right)} 
\safemath{\falseSigVec}{\tilde{\sigVec}} 
\safemath{\falseErVec}{\tilde{\erVec}} 
\safemath{\sigSpars}{\sigDim_{\sigVec}} 
\safemath{\erSpars}{\sigDim_{\erVec}} 
\safemath{\BPErSup}{\left( \text{BP, } \erSupp \right)} 
\safemath{\gram}{\matG} 

\safemath{\WHT}{\scT}	
\safemath{\WHF}{\scF}	
\safemath{\prot}{\scg}	
\safemath{\vprot}{\vecg}	
\safemath{\proti}{\prot_{\scm, \scn}}	
\safemath{\prott}{\prot \! \left( \time \right) }	
\safemath{\protit}{\proti \! \left(\time\right)}	
\safemath{\vproti}{\vprot_{\scm, \scn}}	
\safemath{\dprot}{\tilde{\scg}}	
\safemath{\dvprot}{\tilde{\vecg}}	
\safemath{\dproti}{\dprot_{\scm, \scn}}	
\safemath{\dprott}{\dprot \! \left( \time \right) }	
\safemath{\dprotit}{\dproti \! \left(\time\right)}	
\safemath{\dvproti}{\dvprot_{\scm, \scn}}	
\safemath{\WO}{\scW}		
\safemath{\WOpam}{\WO^{\left( \WHT, \WHF \right)}_{\scm,\scn}}	
\newcommand{\WOind}[2]{\WO_{#1, #2}} 
\safemath{\WOsind}{\WOind \scm \scn }
\safemath{\zak}{\opZ} 
\safemath{\zakpar}{\zak^{\WHT,\WHF}} 
\safemath{\zaksig}{\zak_{\signal} \! \left( \time, \freq \right)} 
\newcommand{\zakprot}[2]{\zak_{\prot} \! \left( #1, #2 \right)} 
\safemath{\zakprots}{\zakprot{\time}{\freq}}
\safemath{\zakprotis}{\zak_{\prot_{\scm,\scn}} \! \left( \time, \freq \right)} 
\safemath{\tfr}{\scR} 
\safemath{\funmin}{\scm \! \left( \prott; \WHT \right)}	
\safemath{\funmax}{\scM \! \left( \prott; \WHT \right)}	

\safemath{\sclf}{\phi}	
\safemath{\vscf}{\vecphi}	
\newcommand{\scfa}[1]{\sclf \! \left( #1 \right)}	
\safemath{\scft}{\scfa \time}	
\newcommand{\vscfi}[2]{\vscf_{#1,#2}}	
\safemath{\vscfs}{\vscfi \scj \scn}
\safemath{\fscf}{\ft \sclf}	
\safemath{\vfscf}{\ft \vscf}	
\newcommand{\fscfa}[1]{\fscf \! \left( #1 \right)}	
\safemath{\fscff}{\fscfa \freq}	

\safemath{\wav}{\psi}	
\safemath{\vwav}{\vecpsi}	
\newcommand{\wava}[1]{\wav \! \left( #1 \right)}	
\safemath{\wavt}{\wava \time}	
\newcommand{\vwavi}[2]{\vwav_{#1,#2}}	
\safemath{\vwavs}{\vwavi \scj \scn} 
\safemath{\fwav}{\ft \wav}	
\safemath{\vfwav}{\ft \vwav}	
\newcommand{\fwava}[1]{\fwav \! \left( #1 \right)}	
\safemath{\fwavf}{\fwava \freq}	

\DeclareMathOperator{\unif}{Unif}			
\newcommand{\channel}[2]{W ( #1 | #2 )} 
\newcommand{\bigchannel}[2]{W \bigl( #1 \bigl| #2 \bigr)} 

\safemath{\oprobf}{\hat{p}}	
\safemath{\probf}{p} 

\safemath{\binind}{\setI}			
\safemath{\pool}{\bm{ \mathcal P}}			
\safemath{\poolre}{\setP}			
\safemath{\eptyp}{\setT^{(n)}_\epsilon}	
\safemath{\dist}{\mathbb{P}}			
\safemath{\tdist}{\tilde \dist}			
\newcommand{\distof}[1]{\dist [ #1 ]}	
\newcommand{\bigdistof}[1]{\dist \bigl[ #1 \bigr]}	
\newcommand{\Bigdistof}[1]{\dist \Bigl[ #1 \Bigr]}	
\safemath{\ry}{\setY}			
\safemath{\rz}{\setZ}			
\safemath{\trho}{\frac{1}{1+\rho}}
\safemath{\tirho}{\tilde \rho}
\DeclareRobustCommand{\impliedby}{\reflectbox {$\implies$}}

\DeclareRobustCommand{\notimpliedby}{\centernot {\reflectbox {$\implies$}}}

\newtheorem{theorem}{Theorem}
\numberwithin{theorem}{section}
\newtheorem{lemma}[theorem]{Lemma}
\newtheorem{corollary}[theorem]{Corollary}

\newtheorem{definition}[theorem]{Definition}
\newtheorem{remark}[theorem]{Remark}
\newtheorem{example}[theorem]{Example}

\fussy
\large \normalsize
\begin{document}
%
%
\selectlanguage{USenglish}
\pagenumbering{arabic}
%
\title{The Zero-Error Feedback Capacity of State-Dependent Channels}
  
\author{Annina Bracher and Amos Lapidoth}

\maketitle

\huge
\begin{abstract}
  \normalsize
  \vspace{0.5cm}
  
  \let\thefootnote\relax\footnotetext{The results in this paper were presented in part at the IEEE International Symposium on Information Theory (ISIT), Barcelona, Spain, Jul.\ 2016.}
  \let\thefootnote\relax\footnotetext{A.\ Bracher and A.\ Lapidoth are with the Department of Information Technology and Electrical Engineering, ETH Zurich, Switzerland (e-mail: bracher@isi.ee.ethz.ch; lapidoth@isi.ee.ethz.ch).}

The zero-error feedback capacity of the Gelfand-Pinsker channel is
established. It can be positive even if the channel's zero-error
capacity is zero in the absence of feedback. Moreover, the error-free
transmission of a single bit may require more than one channel
use. These phenomena do not occur when the state is revealed to the
transmitter causally, a case that is solved here using Shannon
strategies. Cost constraints on the channel inputs or channel states
are also discussed, as is the scenario where---in addition to the
message---also the state sequence must be recovered.
\end{abstract}
\normalsize

\section{Introduction}\label{sec:introduction}

Motivated by Shannon's characterization of the zero-error capacity of
the discrete memoryless channel (DMC) with a feedback link from the channel output to the encoder \cite{shannon56}, we compute the corresponding capacity
for the state-dependent DMC (SD-DMC) whose state is revealed acausally to
the transmitter. This ``Gelfand-Pinsker
channel,'' which was introduced by Gelfand and Pinsker in
\cite{gelfandpinsker80,merhavweissman05}, is more general than the
channel studied by Shannon, and, indeed, when there is only one state we recover Shannon's result. But, more interestingly, this channel's zero-error feedback capacity exhibits phenomena that are not observed on the state-less channel: it
can be positive even if the zero-error capacity is zero in the absence of feedback; the error-free
transmission of a single bit may require more than one channel use;
and Shannon's sequential coding technique cannot be applied naively.

Like Shannon's, our coding scheme is a two-phase scheme where the
first phase reduces the receiver's ambiguity to a manageable size, and
the second removes it entirely. But our first phase differs from Shannon's
sequential approach and draws instead on Dueck's scheme
for zero-error communication over the multiple-access channel with
feedback \cite{dueck85}, which in turn draws on Ahlswede's work
\cite{ahlswede73,ooiwornell98,merhavweissman05}. The second phase is tricky, because
sending a single bit reliably may require more than one channel use, so
``uncoded'' transmission need not work.

We also compute the zero-error feedback capacity of the SD-DMC $W(y|x,s)$ when the state is revealed to the transmitter causally. As we show, causal state information (SI) is utilized optimally using Shannon strategies. Consequently, when the SI is causal, the zero-error capacity is positive with feedback if, and only if, (iff) it is positive without it, and one channel use suffices to transmit a single bit error-free.

Several extensions are also discussed: we compute the zero-error feedback capacity of the Gelfand-Pinkser channel for the case where---in addition to the message---the encoder wishes to convey error-free also the state sequence; and we present capacity results for the Gelfand-Pinsker channel with cost constraints on the channel inputs or channel states. Under channel-input constraints a naive application of Shannon's sequential coding technique turns out to be suboptimal even on the state-less channel.\\

The rest of this paper is structured as follows. We conclude this section by introducing some notation; by recalling the zero-error feedback capacity of the state-less DMC; and by exploring connections with the m-capacity of an arbitrarily-varying channel (AVC). Section~\ref{sec:main} contains the problem formulation and the results. The main results for the Gelfand-Pinkser channel are proved in Section~\ref{sec:proofs}, and the paper concludes with a brief summary.

\subsection{Notation and Terminology}

We consider a SD-DMC of transition law $\channel y {x,s}$, which is governed by an IID $\sim Q$ state process. The channel-input alphabet $\setX$, the channel-state alphabet $\setS$, and the channel-output alphabet $\setY$ are all finite. By possibly redefining $\setS$, we can assume without loss of generality that
\begin{equation}
Q (s) > 0, \quad s \in \setS. \label{eq:assQAssignsPosProbs}
\end{equation}
Subject to \eqref{eq:assQAssignsPosProbs}, the exact nature of the PMF $Q$ is immaterial.

By default $\log (\cdot)$ denotes base-$2$ logarithm, and $\ln (\cdot)$ denotes natural logarithm. We denote by $h_{\textnormal b} (\cdot)$ the binary entropy function. If $\xi$ is a real number, then $[\xi]^+$ denotes the maximum of $\xi$ and zero. Chance variables are denoted by upper-case letters and their realizations or the elements of their support sets by lower-case letters, e.g., $Y$ denotes the random channel output and $y \in \setY$ a value it may take. Sets are denoted by calligraphic letters and in boldface if they are random, so the set of all messages is denoted $\setM$, and $\bm {\mathcal M}_1$ could be the set of messages of positive posterior probability given a first block of (random) channel outputs. Sequences are in bold lower- or upper-case letters depending on whether they are deterministic or random, e.g., $\rndvecY$ is the length-$n$ channel-output sequence, and $\vecy$ is an $n$-tuple from $\setY^n$. The positive integer $n \in \naturals$ stands for the blocklength, and unless otherwise specified sequences are of length~$n$.

Variables pertaining to Time~$i$ have the subscript $i$, so $S_i$ denotes the Time-$i$ channel state. Sequences of variables that occur in the time-range $j$ to $i$ bear a subscript $j$ and a superscript $i$, where the subscript $j = 1$ may be dropped, e.g., $S_4^5$ denotes the fourth and fifth state, and $S^n$ denotes all the states through Time~$n$. We also use a similar notation for sequences whose indices need not coincide with time, e.g., if $\vecs$ is a 5-tuple from $\setS^5$, then $s_3$ denotes its third component, $s^5_4$ its fourth and fifth component, and $s^5$ the entire 5-tuple.

If the input $X$ to the channel $\channel y x$ is of PMF $P$, then $P \times W$ denotes the joint distribution of $X$ and the channel output $Y$ $$( P \times W ) ( x,y ) = P ( x ) \, \channel y x, \quad (x,y) \in \setX \times \setY,$$ and $P W$ denotes the corresponding $Y$-marginal $$( P W ) ( y ) = \sum_{x \in \setX} ( P \times W ) ( x,y  ) = \sum_{x \in \setX} P ( x ) \, \channel y x, \quad y \in \setY.$$

Given two PMFs $P_1$ and $P_2$ on some finite set $\setZ$, we say that $P_2$ is absolutely continuous w.r.t.\ $P_1$ and write $$P_2 \ll P_1,$$ if $P_2 (z)$ is zero whenever $P_1 (z)$ is. If $P_2$ is absolutely continuous w.r.t.\ $P_1$, then the events that have probability zero w.r.t.\ $P_1$ must also have probability zero w.r.t.\ $P_2$. Likewise for events of probability one.

For an SD-DMC $\channel y {x,s}$ we denote by $\mathscr P ( W )$ the set of transition laws $P_{Y|X,S}$ from $\setX \times \setS$ to $\setY$ for which for every pair $(x,s) \in \setX \times \setS$ $$P_{Y|X,S} (\cdot|x,s) \ll \channel \cdot {x,s}.$$ For a state-less DMC $\channel y x$ we drop $s$, and $\mathscr P ( W )$ denotes the set of transition laws $P_{Y|X}$ from $\setX$ to $\setY$ for which for every $x \in \setX$ $$P_{Y|X} (\cdot|x) \ll \channel \cdot x.$$

The empirical type of an $n$-tuple $\vecx \in \setX^n$ is denoted $P_\vecx$, i.e., $$P_\vecx (x) = \frac{N (x|\vecx)}{n}, \quad x \in \setX,$$ where $N (x|\vecx)$ is the number of components of the $n$-tuple $\vecx$ that equal $x$. For a PMF $P$ on $\setX$ the type class $\setT^{(n)}_P$ comprises the elements of $\setX^n$ whose empirical type is $P$. If $\setT^{(n)}_P$ is nonempty, then we say that $P$ is an $n$-type. For an $n$-type $P$ on $\setX$, a transition law~$W$ from $\setX$ to $\setY$, and an element $\vecx$ of $\setT^{(n)}_P$ the $W$-shell $\setT^{(n)}_{W} (\vecx)$ comprises the $n$-tuples $\vecy \in \setT^{(n)}_{P W}$ that satisfy $(\vecx,\vecy) \in \setT^{(n)}_{P \times W}$.

\subsection{State-Less Channels} \label{sec:stateLess}

Shannon showed in \cite{shannon56} that the zero-error capacity of the state-less DMC $\channel y x$ (with or without feedback) is positive iff
\begin{IEEEeqnarray}{l}
\exists \, x, \, x^\prime \in \setX \textnormal{ s.t.\ } \Bigl( \channel y x \, \channel y {x^\prime} = 0, \,\, \forall \, y \in \setY \Bigr). \label{eq:positiveShannon}
\end{IEEEeqnarray}
When \eqref{eq:positiveShannon} holds, the error-free transmission of a single bit requires one channel use. He also showed that, when it is positive, the zero-error feedback capacity of $\channel y x$ is
\begin{IEEEeqnarray}{l}
\max_{P_X} \min_{y \in \setY} - \log \sum_{x \in \setX \colon \channel y x > 0} P_X (x). \label{eq:capacityShannon}
\end{IEEEeqnarray}
Ahlswede \cite{ahlswede73} proved that \eqref{eq:capacityShannon} can be alternatively expressed as
\begin{IEEEeqnarray}{l}
\max_{P_X} \min_{P_{Y|X} \in \mathscr P (W)} I (X;Y), \label{eq:capacityAhlswede}
\end{IEEEeqnarray}
where the mutual information is computed w.r.t.\ the joint PMF $P_X \times P_{Y|X}$. He also provided an alternative coding scheme. Unlike \eqref{eq:positiveShannon}, the formulas \eqref{eq:capacityShannon} and \eqref{eq:capacityAhlswede} are only for channels with feedback. Indeed, feedback can increase the zero-error capacity of a DMC \cite{shannon56}.

\subsection{Connection to the AVC} \label{sec:connAVC}

There are interesting connections between the problem of computing the
zero-error capacity of a DMC and that of computing the m-capacity (the
capacity under the maximal-probability-of-error criterion) of an AVC \cite{ahlswede70}. Indeed,
given a DMC $W(y|x)$ with input alphabet~$\setX$ and output
alphabet~$\setY$, the following construction produces an AVC
$\widetilde{W}(y|x,\sigma)$ whose m-capacity is equal to the zero-error capacity of
the channel $W(y|x)$ \cite[Section~2]{ahlswede70}, \cite[Problem~12.3]{csiszarkoerner11}. To construct the AVC we consider the functions
$\sigma \colon \setX \to \setY$ that satisfy that
$W(\sigma(x)|x)$ is positive for all $x \in \setX$. With each such
function $\sigma(\cdot)$ we associate a state $\sigma$ and the transition
law
\begin{equation}
  \label{eq:constructA}
  \widetilde{W}(y|x, \sigma) = \begin{cases}
    1 & \text{if $y = \sigma(x)$}, \\
    0 & \text{otherwise}.
    \end{cases}
\end{equation}
The constructed AVC has two important properties. The first is that
to every pair of input and output sequences $x_{1}, \ldots, x_{n}$
and $y_{1},\ldots, y_{n}$ for which $\prod_k W(y_{k}|x_{k})$ is
positive, there corresponds a sequence of states $\sigma_{1}, \ldots,
\sigma_{n}$ such that $y_{k} = \sigma_{k}(x_{k})$ for $k=1,\ldots,
n$. The second is that $\widetilde{W}(y|x,\sigma)$ is $\{0,1\}$-valued in
the sense that
\begin{equation*}
\widetilde{W}(y|x,\sigma) \in \{0,1\}, \,\, \forall \, y, \, x, \, \sigma.
\end{equation*}
This latter property guarantees that the conditional probability of error over
the AVC (conditional on the transmitted message and the state
sequence) is $\{0,1\}$-valued and thus small (say, smaller than $1/2$) only if
it is zero. 

This relationship between the zero-error capacity and the m-capacity
fails when the original channel whose zero-error capacity we seek is
state-dependent and the state is revealed to the encoder. To see why, let us denote by $W(y|x,s)$ the transition law of the
state-dependent channel whose zero-error capacity we seek when the state is
revealed to the encoder, and suppose we want to construct an AVC
$\widetilde{W}(y|x,\sigma)$ whose m-capacity when the state $\sigma$ is
revealed to the encoder is equal to the zero-error capacity we seek.
We have intentionally used different letters $s$ and $\sigma$ for the
state of the original channel and of the AVC because the two need not
\emph{prima facie} be the same. For example, if there is only one
state $s^{\star}$, then we are back to the state-less case and the
construction we described above in \eqref{eq:constructA} results in
the number of AVC states being equal to the number of functions
$\sigma \colon \setX \to \setY$ that satisfy that
$W(\sigma(x)|x,s^{\star})$ is positive for all $x \in \setX$. However,
in this case the $m$-capacity of the AVC $\widetilde{W}(y|x,\sigma)$ is
equal to the zero-error capacity we seek only if the state $\sigma$ is
\emph{not} revealed to the encoder. In attempting to construct the
AVC we are faced with two conflicting requirements. For the
state information (SI) that is revealed to the encoder in the two
scenarios to be identical, the states $s$ and $\sigma$ should be identical.
But for the AVC to have a $\{0,1\}$-law, the number of AVC states $\sigma$
should typically be larger than the number of states $s$. 

The construction does go through in the special case where the original
state-dependent transition law $W(y|x,s)$ happens to be $\{0,1\}$-valued.
In this special case we can choose $\sigma$ to equal~$s$, and the
m-capacity equals the zero-error capacity. In this case feedback is
superfluous, because from the state (which is revealed to the encoder)
and from the input (that it produces) the encoder can compute
the output. We thus see that when $W(y|x,s)$ is
$\{0,1\}$-valued the zero-error feedback capacity with acausal SI can be inferred from Ahlswede's results on the feedback-less AVC with SI at the encoder
\cite{ahlswede86}; but in general it cannot.

\section{Problem Formulation and Results}\label{sec:main}

We consider an SD-DMC $\channel y {x,s}$ with feedback whose encoder is furnished with the state sequence either acausally (Figure~\ref{fig:model}), or causally (Figure~\ref{fig:modelCausal}), or strictly-causally (Figure~\ref{fig:modelStrCaus}). Using $n$ channel uses, the encoder wants to convey to the receiver error-free a message~$m$ from some finite set of messages $\setM$. To this end it uses an $(n, \setM)$ zero-error code:

\begin{definition} \label{def:zeroErrorFBCode}
Given a finite set $\setM$ and a positive integer $n \in \naturals$, an $(n, \setM)$ zero-error feedback code for the SD-DMC $\channel {y}{x,s}$ with acausal SI to the encoder consists of $n$ encoding mappings
\begin{IEEEeqnarray}{l}
f_i \colon \setM \times \setS^n \times \setY^{i-1} \rightarrow \setX, \quad i \in [1:n] \label{eq:encAcaus}
\end{IEEEeqnarray}
and $|\setM|$ disjoint decoding sets $$\setD_m \subseteq \setY^n, \quad m \in \setM$$ such that, for every $m \in \setM$ and every realization $\vecs \in \setS^n$ of the state sequence, the probability of a decoding error is zero, i.e., $$\distof {Y^n \notin \setD_m | M = m, S^n = \vecs } = 0, \,\, \forall \, m \in \setM, \, \vecs \in \setS^n,$$ where
\begin{IEEEeqnarray}{l}
\distof {Y^n \notin \setD_m | M = m, S^n = \vecs } = \sum_{\vecy \in \setY^n \setminus \setD_m } \prod^n_{i = 1} W \bigl( y_i \bigl| f_i (m, \vecs, y^{i-1}), s_i \bigr). 
\end{IEEEeqnarray}
A rate~$R$ is achievable if for every sufficiently-large blocklength~$n$ there exists an $(n, \setM)$ zero-error feedback code with $$\log |\setM| \geq n R.$$ The zero-error feedback capacity with acausal SI is the supremum of all achievable rates and is denoted $C_{\textnormal{f},0}$.

The zero-error feedback capacities with causal and strictly-causal SI are denoted $C^{\textnormal{caus}}_{\textnormal{f},0}$ and $C^{\textnormal{s-caus}}_{\textnormal{f},0}$, respectively. They are defined like $C_{\textnormal{f},0}$ except that the encoding mappings \eqref{eq:encAcaus} are replaced by
\begin{IEEEeqnarray}{l}
f_i \colon \setM \times \setS^i \times \setY^{i-1} \rightarrow \setX, \quad i \in [1:n] \label{eq:encCaus}
\end{IEEEeqnarray}
in the causal case and by
\begin{IEEEeqnarray}{l}
f_i \colon \setM \times \setS^{i-1} \times \setY^{i-1} \rightarrow \setX, \quad i \in [1:n] \label{eq:endStrCaus}
\end{IEEEeqnarray}
in the strictly-causal case.
\end{definition}

Note that the PMF $Q$ governing the state does not appear in Definition~\ref{def:zeroErrorFBCode} and therefore does not affect the zero-error feedback capacities with acausal, causal, and strictly-causal SI. Also note that our definition assumes deterministic encoders. This assumption is not restrictive:

\begin{remark}\label{re:detEncWlg}
Allowing stochastic encoders does not increase the zero-error feedback capacities with acausal, causal, and strictly-causal SI.
\end{remark}

\begin{proof}
A proof for the case where the encoder observes the SI acausally is provided in Appendix~\ref{sec:pfDetEncWlg}. The proof goes through also when the SI is causal or strictly-causal.
\end{proof}

\subsection{Acausal SI} \label{sec:acaus}

In this section we assume that the encoder observes the SI acausally (see Figure~\ref{fig:model}). Our main result is presented in the following two theorems, which together provide a single-letter characterization of $C_{\textnormal{f},0}$. The first characterizes the channels for which it is positive, and the second provides a formula for $C_{\textnormal{f},0}$ when it is positive.

\begin{figure}[ht]
\vspace{-2mm}

\begin{center}
\def\pgfsysdriver{pgfsys-dvipdfm.def}
\begin{tikzpicture}[circuit logic US]
	\tikzstyle{sensor}=[draw, minimum width=1em, text centered, minimum height=2em]
	\tikzstyle{stategen}=[draw, minimum width=1em, text centered, minimum height=1em]
	\tikzstyle{delay}=[draw, minimum width=0.5em, text centered, minimum height=0.5em]
	\def\blockdist{2.4}
	\def\edgedist{2.5}
    \node (naveq) [sensor] {$\channel y {x,s}$};
    \path (naveq.west)+(-0.8*\blockdist,0) node (enc) [sensor] {Encoder};
    \path (naveq.east)+(0.8*\blockdist,0) node (dec) [sensor] {Decoder};
    \path (dec.east)+(0.5*\blockdist,0) node (guess) [text centered] {};
    \path (enc.west)+(-0.5*\blockdist,0) node (sour) [text centered] {};
    \path (naveq.east)+(0.2*\blockdist,-0.4*\blockdist) node (del) [delay] {\tiny $D$};
    \path (naveq.center)+(0,0.65*\blockdist) node (state) [sensor] {$Q (s)$};

    \fill [black] (del.north |- dec.west) circle (2pt);
    \path [draw, ->] (sour) -- node [above] {$M$}
    		(enc.west |- sour);  		
    \path [draw, ->] (enc) -- node [above] {$X_i$} 
        (naveq.west |- enc);
    \path [draw] (state.west) -- node [above] {$S^n$} (enc.north |- state.west);
    \path [draw, ->] (enc.north |- state.west) -- (enc.north);
    \path [draw, ->] (state) -- node [right] {$S_i$} 
        (naveq.north);    
    	\path [draw, <-] (dec) -- node [above] {$Y_i$}
    		(naveq.east |- dec);		
    	\path [draw, ->] (dec) -- node [above] {$\widehat{M}$} 
        (guess);
    \path [draw] (del) -- node [above] {} (del.north |- dec.west);
	\path [draw] (del.west) -- node [below] {$Y^{i-1}$} (enc.south |- del.west);
    \path [draw, <-] (enc.south) -- node [left] {} (enc.south |- del.west);
    
\end{tikzpicture}

\caption[SD-DMC with acausal SI and feedback]{SD-DMC with acausal SI and feedback.}
\label{fig:model}
\end{center}
\vspace{-2mm}

\end{figure}
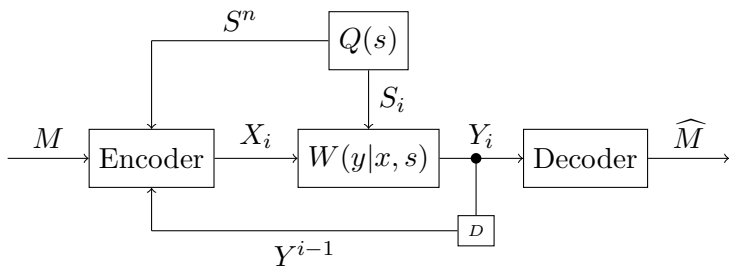

\begin{theorem}\label{th:positive}
A necessary and sufficient condition for $C_{\textnormal{f},0}$ to be positive is
\begin{IEEEeqnarray}{l} 
\forall \, s, \, s^\prime \in \setS \quad \exists \, x, \, x^\prime \in \setX \textnormal{ s.t.\ } \Bigl(  \channel y {x,s} \, \channel y {x^\prime,s^\prime} = 0, \,\, \forall \, y \in \setY \Bigr). \label{eq:positive}
\end{IEEEeqnarray}
\end{theorem}

\begin{proof}
See Section~\ref{sec:pfThPositive}.
\end{proof}

\begin{theorem}\label{th:capacity}
If $C_{\textnormal{f},0}$ is positive, then
\begin{IEEEeqnarray}{l} 
C_{\textnormal{f},0} = \min_{P_S} \max_{P_{U,X|S}} \min_{ \substack{ P_{Y|U,X,S} \colon \\ P_{Y|U=u,X,S} \in \mathscr P (W), \,\, \forall \, u \in \setU }} I (U;Y) - I (U;S), \label{eq:capacity}
\end{IEEEeqnarray}
where $U$ is an auxiliary chance variable taking values in a finite set $\setU$, and the mutual informations are computed w.r.t.\ the joint PMF $P_S \times P_{U,X|S} \times P_{Y|U,X,S}$. Restricting $X$ to be a function of $U$ and $S$, i.e., $P_{U,X|S}$ to have the form
\begin{equation}
P_{U,X|S} (u,x|s) = P_{U|S} (u|s) \, \ind {x = g (u,s)}, \label{eq:thCapCardUXFuncUS}
\end{equation}
does not change the RHS of \eqref{eq:capacity}, nor does restricting the cardinality of $\setU$ to
\begin{equation}
|\setU| \leq |\setX|^{|\setS|}. \label{eq:cardU}
\end{equation}
\end{theorem}

\begin{proof}
See Section~\ref{sec:pfThCapacity}.
\end{proof}

\begin{remark}\label{re:posSuffButNotNec1}
The hypothesis in Theorem~\ref{th:capacity} that  $C_{\textnormal{f},0}$ be positive is essential: the RHS of \eqref{eq:capacity} may be positive even when  $C_{\textnormal{f},0}$ is zero.
\end{remark}

In fact, as we prove in Appendix~\ref{sec:pfPosSuffButNotNec}:

\begin{remark}\label{re:posSuffButNotNec}
The RHS of \eqref{eq:capacity} is positive iff
\begin{IEEEeqnarray}{l} 
\forall \, ( s,y ) \in \setS \times \setY \quad \exists \, x \in \setX \textnormal{ s.t.\ } \channel y {x,s} = 0. \label{eq:condCapacityPos}
\end{IEEEeqnarray}
\end{remark}

Theorems~\ref{th:positive} and \ref{th:capacity} generalize to the SD-DMC with feedback and acausal SI Shannon's characterization \cite[Theorem~7]{shannon56} of the zero-error feedback capacity of the (state-less) DMC $\channel y x$ (see \eqref{eq:positiveShannon} and \eqref{eq:capacityShannon} in Section~\ref{sec:stateLess}). That \eqref{eq:positive} reduces to \eqref{eq:positiveShannon} when $|\setS| = 1$ is evident. That \eqref{eq:capacity} reduces to \eqref{eq:capacityShannon} when $|\setS| = 1$ becomes evident when we recall from \cite{ahlswede73} Ahlswede's alternative form \eqref{eq:capacityAhlswede} for \eqref{eq:capacityShannon}: clearly, \eqref{eq:capacity} specializes to \eqref{eq:capacityAhlswede} and thus to \eqref{eq:capacityShannon} when $|\setS| = 1$. The way in which \eqref{eq:capacity} generalizes \eqref{eq:capacityAhlswede} is reminiscent of the way the Gelfand-Pinsker capacity generalizes the ordinary capacity of the state-less DMC (cf.\ \cite{shannon48,gelfandpinsker80}).\\

In the remainder of this section we discuss how feedback affects the zero-error capacity with acausal SI. By considering the case of a single state, i.e., $|\setS| = 1$, and invoking Shannon's result \cite{shannon56} that feedback can increase the zero-error capacity of a DMC, we readily obtain that feedback can also increase the zero-error capacity of an SD-DMC with acausal SI. But, in the presence of acausal SI, more is true. Unlike the stateless channel, here feedback can increase the capacity from zero:

\begin{theorem}\label{th:zeroWithoutFBPosWithFB}
The zero-error capacity of an SD-DMC with acausal SI can be positive with feedback yet zero without it.
\end{theorem}

\begin{proof}
See Section~\ref{sec:pfThZeroWithoutFBPosWithFB}.
\end{proof}

Condition~\eqref{eq:positive} is thus only for channels with feedback: the no-feedback zero-error capacity of the SD-DMC $\channel y {x,s}$ with acausal SI can be zero also when the channel satisfies \eqref{eq:positive}. Because feedback can help only if the encoder uses the channel more than once, we obtain the following corollary, which marks another difference to the state-less case:

\begin{corollary}\label{co:moreThanOneChUse}
On the SD-DMC with acausal SI and feedback, the error-free transmission of a single bit may require more than one channel use.
\end{corollary}

This result will be strengthened in Section~\ref{sec:caus}, where we show that also in the absence of feedback the error-free transmission of a single bit may require more than one channel use (Corollary~\ref{co:moreThanOneChUseWithoutFB}).\\

As we have seen in Section~\ref{sec:connAVC}, if the transition law $\channel y {x,s}$ of the SD-DMC happens to be $\{0,1\}$-valued, then $C_{\textnormal{f},0}$ is related to Ahlswede's AVC with acausal SI. As we show in Appendix~\ref{sec:anExYFunXS}, in this case Theorems~\ref{th:positive} and \ref{th:capacity} can be greatly simplified:

\begin{example}\label{ex:YFunXS}
If the transition law $\channel y {x,s}$ of an SD-DMC is $\{0,1\}$-valued, then
\begin{IEEEeqnarray}{l}
C_{\textnormal{f},0} = \min_{s \in \setS} \log \bigl| \bigl\{ y \in \setY \colon \exists \, x \in \setX \textnormal{ s.t.\ } \channel {y}{x,s} > 0 \bigr\} \bigr|. \label{eq:capYFunXS}
\end{IEEEeqnarray}
\end{example}

Remark~\ref{re:posSuffButNotNec1} not withstanding, if $\channel y {x,s}$ is $\{0,1\}$-valued, then the RHS of \eqref{eq:capacity}---which in this case is equal to the RHS of \eqref{eq:capYFunXS}---is positive iff $C_{\textnormal{f},0}$ is positive. This agrees with Ahlswede's observation \cite{ahlswede86} that the formula for the (a- and m-) capacity of the general AVC $\channel y {x,s}$ whose state sequence is revealed acausally to the encoder not only applies when the capacity is positive but also determines whether it is positive.

\subsection{Causal SI} \label{sec:caus}

In this section we assume that the encoder observes the SI causally (see Figure~\ref{fig:modelCausal}). The following two theorems together provide a single-letter characterization of $C^{\textnormal{caus}}_{\textnormal{f},0}$. The first characterizes the channels for which it is positive, and the second provides a formula for the capacity when it is positive.

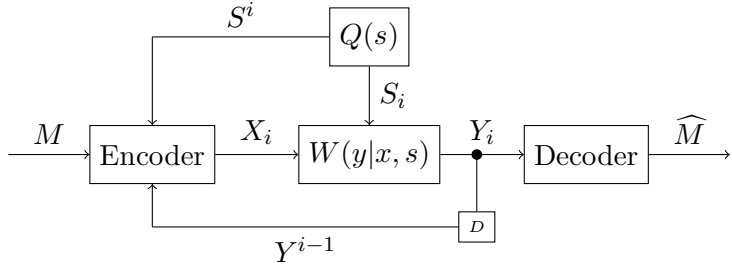
\begin{figure}[ht]
\vspace{-2mm}

\begin{center}
\def\pgfsysdriver{pgfsys-dvipdfm.def}
\begin{tikzpicture}[circuit logic US]
	\tikzstyle{sensor}=[draw, minimum width=1em, text centered, minimum height=2em]
	\tikzstyle{stategen}=[draw, minimum width=1em, text centered, minimum height=1em]
	\tikzstyle{delay}=[draw, minimum width=0.5em, text centered, minimum height=0.5em]
	\def\blockdist{2.4}
	\def\edgedist{2.5}
    \node (naveq) [sensor] {$\channel y {x,s}$};
    \path (naveq.west)+(-0.8*\blockdist,0) node (enc) [sensor] {Encoder};
    \path (naveq.east)+(0.8*\blockdist,0) node (dec) [sensor] {Decoder};
    \path (dec.east)+(0.5*\blockdist,0) node (guess) [text centered] {};
    \path (enc.west)+(-0.5*\blockdist,0) node (sour) [text centered] {};
    \path (naveq.east)+(0.2*\blockdist,-0.4*\blockdist) node (del) [delay] {\tiny $D$};
    \path (naveq.center)+(0,0.65*\blockdist) node (state) [sensor] {$Q (s)$};

    \fill [black] (del.north |- dec.west) circle (2pt);
    \path [draw, ->] (sour) -- node [above] {$M$}
    		(enc.west |- sour);  		
    \path [draw, ->] (enc) -- node [above] {$X_i$} 
        (naveq.west |- enc);
    \path [draw] (state.west) -- node [above] {$S^i$} (enc.north |- state.west);
    \path [draw, ->] (enc.north |- state.west) -- (enc.north);
    \path [draw, ->] (state) -- node [right] {$S_i$} 
        (naveq.north);    
    	\path [draw, <-] (dec) -- node [above] {$Y_i$}
    		(naveq.east |- dec);		
    	\path [draw, ->] (dec) -- node [above] {$\widehat{M}$} 
        (guess);
    \path [draw] (del) -- node [above] {} (del.north |- dec.west);
	\path [draw] (del.west) -- node [below] {$Y^{i-1}$} (enc.south |- del.west);
    \path [draw, <-] (enc.south) -- node [left] {} (enc.south |- del.west);
    
\end{tikzpicture}

\caption[SD-DMC with causal SI and feedback]{SD-DMC with causal SI and feedback.}
\label{fig:modelCausal}
\end{center}
\vspace{-2mm}

\end{figure}

\begin{theorem}\label{th:causPositive}
A necessary and sufficient condition for $C^{\textnormal{caus}}_{\textnormal{f},0}$ to be positive is that there exist a partition $\setY_0, \, \setY_1$ of $\setY$ for which
\begin{IEEEeqnarray}{l}
\forall \, s \in \setS \quad \exists \, x, \, x^\prime \in \setX \textnormal{ s.t.\ } \channel {\setY_0}{x,s} = \channel {\setY_1}{x^\prime,s} = 1. \label{eq:causPositive}
\end{IEEEeqnarray}
If $C^{\textnormal{caus}}_{\textnormal{f},0}$ is positive, then one channel use suffices to transmit a single bit error-free, and therefore the zero-error capacity with causal SI is positive with feedback iff it is positive without it.
\end{theorem}

\begin{proof}
See Appendix~\ref{sec:pfCausPositive}.
\end{proof}

\begin{theorem}\label{th:causCapacity}
If $C^{\textnormal{caus}}_{\textnormal{f},0}$ is positive, then
\begin{IEEEeqnarray}{rCl}
C^{\textnormal{caus}}_{\textnormal{f},0} & = & \max_{P_U} \min_{P_{Y|U} \in \mathscr P (W^\prime)} I (U;Y) \label{eq:causCapacityAhlswedeFrom} \\
& = & \max_{P_U} \min_y - \log \sum_{u \colon W^\prime (y|u) > 0} P_U (u), \label{eq:causCapacity}
\end{IEEEeqnarray}
where $U$ is an auxiliary chance variable taking values in a finite set $\setU$ of cardinality $|\setU| = |\setX|^{|\setS|}$; the mutual information is computed w.r.t.\ the joint PMF $P_U  \times P_{Y|U}$; and
\begin{IEEEeqnarray}{l}
W^\prime (y|u) = \sum_{s \in \setS} Q_S (s) \, \bigchannel {y}{g (u,s),s}, \quad (u,y) \in \setU \times \setY,
\end{IEEEeqnarray}
where $\bigl\{ g (u,\cdot) \colon u \in \setU \bigr\}$ is the set of functions from $\setS$ to $\setX$, i.e., $\setX^\setS$. Because
\begin{IEEEeqnarray}{l}
\Bigl( W^\prime (y|u) > 0 \Bigr) \iff \Bigl( \exists \, s \in \setS \textnormal{ s.t.\ } \bigchannel {y}{g(u,s),s} > 0 \Bigr), \label{eq:causCapWPrimeExistsSW}
\end{IEEEeqnarray}
$P_{Y|U} \in \mathscr P (W^\prime)$ holds iff
\begin{IEEEeqnarray}{l}
\Bigl( \bigchannel {y}{g(u,s),s} = 0, \,\, \forall \, s \in \setS \Bigr) \implies \Bigl( P_{Y|U} (y|u) = 0 \Bigr).
\end{IEEEeqnarray}
\end{theorem}

\begin{proof}
The proof draws on Shannon's results \cite{shannon56, shannon58} (see Appendix~\ref{sec:pfCausCapacity}).
\end{proof}

\begin{remark}\label{re:causPosSuffButNotNec1}
The hypothesis in Theorem~\ref{th:capacity} that $C^{\textnormal{caus}}_{\textnormal{f},0}$ be positive is essential: the RHS of \eqref{eq:causCapacityAhlswedeFrom} may be positive even when $C^{\textnormal{caus}}_{\textnormal{f},0}$ is zero.
\end{remark}

In fact, as we prove in Appendix~\ref{sec:pfCausPosSuffButNotNec}:\footnote{Remarks~\ref{re:posSuffButNotNec} and \ref{re:causPosSuffButNotNec} imply that, like the ordinary capacities with causal and acausal SI \cite{shannon58,gelfandpinsker80}, the RHS of \eqref{eq:capacity} is positive iff that of \eqref{eq:causCapacityAhlswedeFrom} is positive. As we shall see, however, this does not hold for the capacities: the zero-error capacity can be positive with acausal SI yet zero with causal SI (see Theorem~\ref{re:zeroCausPosAcaus} ahead).}

\begin{remark}\label{re:causPosSuffButNotNec}
The RHS of \eqref{eq:causCapacityAhlswedeFrom} is positive iff
\begin{IEEEeqnarray}{l}
\forall \, (s,y) \in \setS \times \setY \quad \exists \, x \in \setX \textnormal{ s.t.\ } \channel y {x,s} = 0. \label{eq:condCausCapFormulaPos}
\end{IEEEeqnarray}
\end{remark}

Theorems~\ref{th:causPositive} and \ref{th:causCapacity} generalize to the SD-DMC with feedback and causal SI Shannon's characterization \cite[Theorem~7]{shannon56} of the zero-error feedback capacity of the (state-less) DMC $\channel y x$ (see \eqref{eq:positiveShannon} and \eqref{eq:capacityShannon} in Section~\ref{sec:stateLess}). The way in which \eqref{eq:causPositive} and \eqref{eq:causCapacity} generalize \eqref{eq:positiveShannon} and \eqref{eq:capacityShannon} is reminiscent of the way the ordinary capacity with causal SI generalizes the ordinary capacity of the state-less DMC (cf.\ \cite{shannon48,shannon58}): in both cases causal SI is utilized optimally by using \emph{Shannon strategies}. To see this, recall that by using Shannon strategies the encoder transforms the SD-DMC $\channel y {x,s}$ with causal SI into the state-less DMC $$W^\prime (y|u) = \sum_{s \in \setS} Q_S (s) \, \bigchannel {y}{g (u,s),s}$$ with input alphabet $\setU$ of cardinality $|\setU| = |\setX|^{|\setS|}$, where $\bigl\{ g (u,\cdot) \colon u \in \setU \bigr\}$ equals $\setX^\setS$: an encoder with causal SI is said to use Shannon strategies if it performs the encoding over the set $\setU$ and obtains the Time-$i$ channel-input by evaluating the function $g (\cdot,\cdot) \colon \setU \times \setS \rightarrow \setX$ for the $i$-th codeword-symbol $u_i \in \setU$ and the Time-$i$ channel-state $S_i$ (see Figure~\ref{fig:shannonStrategy} and \cite[Remark~7.6]{gamalkim11}). By comparing \eqref{eq:causPositive} and \eqref{eq:causCapacity} to \eqref{eq:positiveShannon} and \eqref{eq:capacityShannon}, respectively, we see that, indeed, the zero-error feedback capacity of the SD-DMC $\channel y {x,s}$ with causal SI equals the zero-error feedback capacity of the state-less DMC $W^\prime (y|u)$, and hence causal SI is utilized optimally by using Shannon strategies.\\

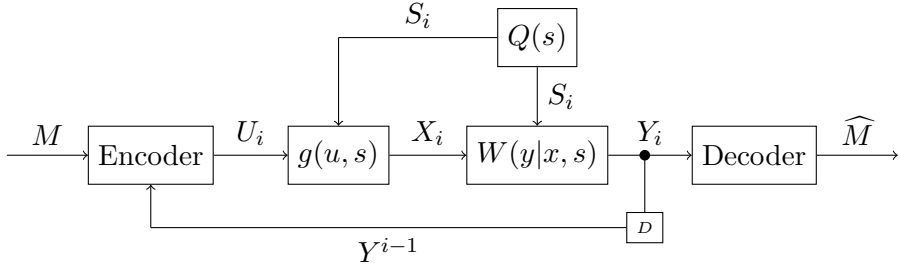
\begin{figure}[ht]
\vspace{-2mm}

\begin{center}
\def\pgfsysdriver{pgfsys-dvipdfm.def}
\begin{tikzpicture}[circuit logic US]
	\tikzstyle{sensor}=[draw, minimum width=1em, text centered, minimum height=2em]
	\tikzstyle{stategen}=[draw, minimum width=1em, text centered, minimum height=1em]
	\tikzstyle{delay}=[draw, minimum width=0.5em, text centered, minimum height=0.5em]
	\def\blockdist{2.4}
	\def\edgedist{2.5}
    \node (naveq) [sensor] {$\channel y {x,s}$};
    \path (naveq.west)+(-0.7*\blockdist,0) node (shStrat) [sensor] {$g (u,s)$};
    \path (shStrat.west)+(-0.75*\blockdist,0) node (enc) [sensor] {Encoder};
    \path (naveq.east)+(0.8*\blockdist,0) node (dec) [sensor] {Decoder};
    \path (dec.east)+(0.5*\blockdist,0) node (guess) [text centered] {};
    \path (enc.west)+(-0.5*\blockdist,0) node (sour) [text centered] {};
    \path (naveq.east)+(0.2*\blockdist,-0.4*\blockdist) node (del) [delay] {\tiny $D$};
    \path (naveq.center)+(0,0.65*\blockdist) node (state) [sensor] {$Q (s)$};

    \fill [black] (del.north |- dec.west) circle (2pt);
    \path [draw, ->] (sour) -- node [above] {$M$}
    		(enc.west |- sour);
    \path [draw, ->] (enc) -- node [above] {$U_i$} 
        (shStrat.west |- enc);		
    \path [draw, ->] (shStrat) -- node [above] {$X_i$} 
        (naveq.west |- shStrat);
    \path [draw] (state.west) -- node [above] {$S_i$} (shStrat.north |- state.west);
    \path [draw, ->] (shStrat.north |- state.west) -- (shStrat.north);
    \path [draw, ->] (state) -- node [right] {$S_i$} 
        (naveq.north);    
    	\path [draw, <-] (dec) -- node [above] {$Y_i$}
    		(naveq.east |- dec);		
    	\path [draw, ->] (dec) -- node [above] {$\widehat{M}$} 
        (guess);
    \path [draw] (del) -- node [above] {} (del.north |- dec.west);
	\path [draw] (del.west) -- node [below] {$Y^{i-1}$} (enc.south |- del.west);
    \path [draw, <-] (enc.south) -- node [left] {} (enc.south |- del.west);
    
\end{tikzpicture}

\caption[Shannon strategies]{Shannon strategies.}
\label{fig:shannonStrategy}
\end{center}
\vspace{-2mm}

\end{figure}

In the remainder of this section we briefly contrast how feedback affects the zero-error capacities with acausal and causal SI. As in the acausal case, by considering the case of a single state, i.e., $|\setS| = 1$, and invoking Shannon's result \cite{shannon56} that feedback can increase the zero-error capacity of a DMC, we readily obtain that feedback can also increase the zero-error capacity in the causal case. However, unlike the acausal case, the zero-error capacity with causal SI is positive with feedback iff it is positive without it (see Theorem~\ref{th:causPositive}).

Since acausal SI is better than causal SI, and since the zero-error capacity with causal SI is positive with feedback iff it is positive without it, the condition in Theorem~\ref{th:causPositive} is sufficient for the no-feedback zero-error capacity of the SD-DMC $\channel y {x,s}$ with acausal SI to be positive. (Alternatively, this is obtained by noting that \eqref{eq:causPositive} of Theorem~\ref{th:causPositive} implies \eqref{bl:suffCondCapWithoutFBPositive} of Lemma~\ref{le:suffZeroWithoutFB} ahead, which is a sufficient condition for the no-feedback zero-error capacity of the SD-DMC with acausal SI to be positive. As we shall see in Example~\ref{ex:zeroCausPosNoncausWithoutFB} ahead, the reverse implication need not hold.) By Theorem~\ref{th:zeroWithoutFBPosWithFB} the zero-error capacity with acausal SI can be positive with feedback yet zero without it. Consequently, unlike the ordinary capacities with causal and acausal SI (see \cite{shannon58,gelfandpinsker80}) or the RHSs of \eqref{eq:capacity} and \eqref{eq:causCapacityAhlswedeFrom} (see Remarks~\ref{re:posSuffButNotNec} and \ref{re:causPosSuffButNotNec}), the zero-error feedback capacity can be positive with acausal SI yet zero with causal SI. In fact, more is true:

\begin{theorem}\label{re:zeroCausPosAcaus}
The zero-error capacity can be positive with acausal SI yet zero with causal SI even when feedback is available in the latter setting and absent in the former.
\end{theorem}

Because acausal SI can be better than causal SI only if the encoder uses the channel more than once, we obtain the following corollary, which strengthens Corollary~\ref{co:moreThanOneChUse}:

\begin{corollary}\label{co:moreThanOneChUseWithoutFB}
On the SD-DMC with acausal SI, the error-free transmission of a single bit may require more than one channel use also in the absence of feedback.
\end{corollary}

To prove Theorem~\ref{re:zeroCausPosAcaus}, we provide an example for which the zero-error capacity (with and without feedback) is positive with acausal SI yet zero with causal SI:

\begin{example}\label{ex:zeroCausPosNoncausWithoutFB}
Consider a deterministic SD-DMC $\channel y {x,s}$ over the alphabets $\setX = \{ 0, 1 \}$ and $\setS = \setY = \{ 1, 2, 3 \}$. Let the output corresponding to the input $x$ and the state $s$ be the single element of the set $\setY_{x,s}$ that is given in Table~\ref{tb:exZeroCausPosNoncausWithoutFB}
\begin{IEEEeqnarray}{l}
\bigl\{ y \in \setY \colon \channel {y}{x,s} > 0 \bigr\} = \setY_{x,s}, \,\, \forall \, (x,s) \in \setX \times \setS.
\end{IEEEeqnarray}
Since this channel violates \eqref{eq:causPositive} but satisfies \eqref{bl:suffCondCapWithoutFBPositive} of Lemma~\ref{le:suffZeroWithoutFB} ahead for $\kappa = \lambda = 3$,
\begin{IEEEeqnarray}{l}
x (s,k) = \begin{cases} 0 &\textnormal{if } k = 1 \textnormal{ or } (s,k) = (3,2), \\ 1 &\textnormal{otherwise}, \end{cases} \quad (s,k) \in \{ 1,2,3 \} \times \{ 1,2,3 \}, \nonumber
\end{IEEEeqnarray}
and $\setY_\ell = \{ \ell \}, \,\, \ell \in \{ 1,2,3\}$ (cf.\ Remark~\ref{re:lemmaSuffZeroWithoutFB} ahead), its zero-error capacity (both with and without feedback) is positive with acausal SI yet zero with causal SI.
\end{example}

\begin{table}[h]
\begin{center}
  \renewcommand{\arraystretch}{1.2}
  \begin{tabular}{ c c | c c c }
    \multicolumn{2}{c|}{\multirow{2}{*}{$\setY_{x,s}$}} & \multicolumn{3}{c}{$s$} \\
  	&& 1 & 2 & 3 \\
    \hline    
    \multirow{2}{*}{$x$} & 0 & \{ 2 \} & \{ 1 \} & \{ 1 \} \\
    & 1 & \{ 3 \} & \{ 3 \} & \{ 2 \}
  \end{tabular}
  \caption{Nonzero transitions of the SD-DMC in Example~\ref{ex:zeroCausPosNoncausWithoutFB}.}
  \label{tb:exZeroCausPosNoncausWithoutFB}
\end{center}
\end{table}

\subsection{Strictly-Causal SI}

In this section we assume that the encoder observes the SI strictly-causally (see Figure~\ref{fig:modelStrCaus}).

\begin{figure}[ht]
\vspace{-2mm}

\begin{center}
\def\pgfsysdriver{pgfsys-dvipdfm.def}
\begin{tikzpicture}[circuit logic US]
	\tikzstyle{sensor}=[draw, minimum width=1em, text centered, minimum height=2em]
	\tikzstyle{stategen}=[draw, minimum width=1em, text centered, minimum height=1em]
	\tikzstyle{delay}=[draw, minimum width=0.5em, text centered, minimum height=0.5em]
	\def\blockdist{2.4}
	\def\edgedist{2.5}
    \node (naveq) [sensor] {$\channel y {x,s}$};
    \path (naveq.west)+(-0.8*\blockdist,0) node (enc) [sensor] {Encoder};
    \path (naveq.east)+(0.8*\blockdist,0) node (dec) [sensor] {Decoder};
    \path (dec.east)+(0.5*\blockdist,0) node (guess) [text centered] {};
    \path (enc.west)+(-0.5*\blockdist,0) node (sour) [text centered] {};
    \path (naveq.east)+(0.2*\blockdist,-0.4*\blockdist) node (del) [delay] {\tiny $D$};
    \path (naveq.center)+(0,0.65*\blockdist) node (state) [sensor] {$Q (s)$};
    \path (state.west)+(-0.3*\blockdist,0) node (delState) [delay] {\tiny $D$};

    \fill [black] (del.north |- dec.west) circle (2pt);
    \path [draw, ->] (sour) -- node [above] {$M$}
    		(enc.west |- sour);  		
    \path [draw, ->] (enc) -- node [above] {$X_i$} 
        (naveq.west |- enc);
    \path [draw] (state) -- (delState.east);
    \path [draw] (delState.west) -- node [above] {$S^{i-1}$} (enc.north |- state.west);
    \path [draw, ->] (enc.north |- state.west) -- (enc.north);
    \path [draw, ->] (state) -- node [right] {$S_i$} 
        (naveq.north);    
    	\path [draw, <-] (dec) -- node [above] {$Y_i$}
    		(naveq.east |- dec);		
    	\path [draw, ->] (dec) -- node [above] {$\widehat{M}$} 
        (guess);
    \path [draw] (del) -- node [above] {} (del.north |- dec.west);
	\path [draw] (del.west) -- node [below] {$Y^{i-1}$} (enc.south |- del.west);
    \path [draw, <-] (enc.south) -- node [left] {} (enc.south |- del.west);
    
\end{tikzpicture}

\caption[SD-DMC with strictly-causal SI and feedback]{SD-DMC with strictly-causal SI and feedback.}
\label{fig:modelStrCaus}
\end{center}
\vspace{-2mm}

\end{figure}
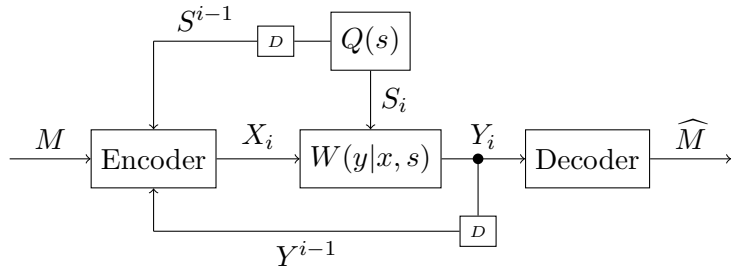

The results \eqref{eq:positiveShannon}--\eqref{eq:capacityAhlswede} for the state-less DMC also provide the zero-error feedback capacity $C_{\textnormal{f},0}^{\textnormal{s-caus}}$ of the state-dependent channel whose state is revealed strictly-causally to the encoder:

\begin{remark} \label{re:strCaus}
Shannon's proof of \eqref{eq:positiveShannon} and \eqref{eq:capacityShannon} in \cite{shannon56} goes through also when the channel is state-dependent and the SI is revealed strictly-causally to the encoder. Consequently, such SI cannot increase the zero-error feedback capacity. That is, if we define
\begin{equation}
\widetilde W ( y | x ) = \sum_{s \in \setS} Q (s) \, \channel y {x,s}, \quad (x,y) \in \setX \times \setY,
\end{equation}
then a necessary and sufficient condition for $C_{\textnormal{f},0}^{\textnormal{s-caus}}$ to be positive is that \eqref{eq:positiveShannon} hold for the channel $\widetilde W ( y | x )$, and if $C_{\textnormal{f},0}^{\textnormal{s-caus}}$ is positive, then it can be computed by substituting $\widetilde W ( y | x )$ for $\channel y x$ in \eqref{eq:capacityShannon} or \eqref{eq:capacityAhlswede}.\footnote{Note that by \eqref{eq:assQAssignsPosProbs} $\widetilde W ( y | x )$ is positive iff there exists some state for which $\channel y {x,s}$ is positive.}
\end{remark}

\subsection{Zero-Error Rate-and-State} \label{sec:alsoStates}

In this section we consider a scenario where---in addition to the message~$m$---the encoder wishes to convey to the receiver (error-free) also the state sequence $S^n$, which it observes acausally. For the standard setting where the probability of a message error need not be zero but can be arbitrarily small, Kim, Sutivong, and Cover \cite{kimsutivongcover08} introduced and solved a related problem with list decoding of state sequences. Choudhuri, Kim, and Mitra \cite{choughurikimmitra13} studied the causal and strictly-causal settings subject to a constraint on the distortion between the state sequence and its receiver-side estimate. Analogous results in the presence of feedback were recently reported by Bross and Lapidoth \cite{brosslapidoth16}.

We begin with the basic definitions of an $(n, \setM)$ zero-error code:

\begin{definition}
Given a finite set $\setM$ and a positive integer $n \in \naturals$, an $(n, \setM)$ zero-error state-conveying feedback code for the SD-DMC $\channel y {x,s}$ with acausal SI to the encoder consists of $n$ encoding mappings $$f_i \colon \setM \times \setS^n \times \setY^{i-1} \rightarrow \setX, \quad i \in [1:n]$$ and $|\setM| \, |\setS|^n$ disjoint decoding sets $$\setD_{m,\vecs} \subseteq \setY^n, \quad (m,\vecs) \in \setM \times \setS^n$$ such that for every $(m,\vecs) \in \setM \times \setS^n$ the probability of a decoding error is zero, i.e., $$\distof {Y^n \notin \setD_{m,\vecs} | M = m, S^n = \vecs } = 0, \,\, \forall \, (m,\vecs) \in \setM \times \setS^n,$$ where
\begin{IEEEeqnarray}{l}
\distof {Y^n \notin \setD_{m,\vecs} | M = m, S^n = \vecs } = \sum_{\vecy \in \setY^n \setminus \setD_{m,\vecs} } \prod^n_{i = 1} W \bigl( y_i \bigl| f_i (m, \vecs, y^{i-1}), s_i \bigr). 
\end{IEEEeqnarray}
A positive rate~$R$ is called achievable if for every sufficiently-large blocklength~$n$ there exists an $(n, \setM)$ zero-error state-conveying feedback code satisfying $$\frac{1}{n} \log |\setM| \geq R.$$ The zero-error state-conveying feedback capacity is the supremum of all achievable rates and is denoted $C^{\textnormal{m} + \textnormal{s}}_{\textnormal{f},0}$. If no positive rate is achievable, then we say that $C^{\textnormal{m} + \textnormal{s}}_{\textnormal{f},0} = 0$, regardless of whether or not it is possible to convey the state sequence error-free.
\end{definition}

Our definition of an $(n, \setM)$ zero-error state-conveying code does not depend on the PMF $Q$ of the state and assumes a deterministic encoder. Like the scenario where the encoder need not convey the state, $C^{\textnormal{m} + \textnormal{s}}_{\textnormal{f},0}$ does not depend on the PMF $Q$, and allowing stochastic encoders does not increase it (cf.\ the proof of Remark~\ref{re:detEncWlg}).\\

The following theorem provides a single-letter characterization of $C^{\textnormal{m} + \textnormal{s}}_{\textnormal{f},0}$:

\begin{theorem}\label{th:stateAmplification}
A necessary condition for $C^{\textnormal{m} + \textnormal{s}}_{\textnormal{f},0}$ to be positive is \eqref{eq:positive}, and if \eqref{eq:positive} holds, then
\begin{IEEEeqnarray}{l}
C^{\textnormal{m} + \textnormal{s}}_{\textnormal{f},0} = \biggl[ \min_{P_S} \max_{P_{X|S}} \min_{P_{Y|X,S} \in \mathscr P (W)} I (X,S;Y) - H (S) \biggr]^+, \label{eq:capacitySA}
\end{IEEEeqnarray}
where the mutual information and the entropy are computed w.r.t.\ the joint PMF $P_S \times P_{X|S} \times P_{Y|X,S}$.
\end{theorem}

\begin{proof}
The result is proved in Appendix~\ref{sec:pfThStateAmplification} by adapting the proofs of Theorems~\ref{th:positive} and \ref{th:capacity} so as to guarantee that the receiver can decode also the state sequence error-free.
\end{proof}

\subsection{Constrained Inputs} \label{sec:constrainedInputs}

In this section we establish the zero-error feedback capacity of the SD-DMC $\channel y {x,s}$ with acausal SI subject to a cost constraint on the channel inputs. Consider some nonnegative ``cost-function'' $\gamma \colon \setX \rightarrow \reals^+_0$, and define $$\gamma_{\textnormal {min}} = \min_{x \in \setX} \gamma (x) \quad \textnormal{and} \quad \gamma_{\textnormal {max}} = \max_{x \in \setX} \gamma (x).$$ Let the set $\setX^\prime \subseteq \setX$ comprise all the minimizers of $\gamma ( \cdot )$ $$\setX^\prime = \bigl\{ x \in \setX \colon \gamma (x) = \gamma_{\textnormal {min}} \bigr\}.$$ The cost constraint we study is that, at every blocklength~$n$ and for every transmitted message $m \in \setM$, the channel inputs' average cost $$\gamma^{(n)} (X^n) = \frac{1}{n} \sum^n_{i = 1} \gamma (X_i)$$ satisfy the cost constraint
\begin{equation}
\gamma^{(n)} (X^n) \leq \Gamma \label{eq:ccInputs}
\end{equation}
for some given $\Gamma$ satisfying
\begin{equation}
\gamma_{\textnormal {min}} < \Gamma < \gamma_{\textnormal {max}}. \label{eq:ccGamma}
\end{equation}
The zero-error feedback capacity with acausal SI subject to \eqref{eq:ccInputs} is denoted $C_{\textnormal{f},0} (\Gamma)$. We restrict $\Gamma$ to \eqref{eq:ccGamma}, because all other values of $\Gamma$ are uninteresting: if $\Gamma < \gamma_{\textnormal {min}}$, then \eqref{eq:ccInputs} cannot hold; if $\Gamma = \gamma_{\textnormal {min}}$, then the encoder can only use inputs in $\setX^\prime$, and the zero-error feedback capacity is thus that of the channel with input alphabet $\setX^\prime$ and without a cost constraint; and if $\Gamma \geq \gamma_{\textnormal {max}}$, then \eqref{eq:ccInputs} always holds, and the cost constraint can be ignored.

As we argue next,
\begin{equation}
C_{\textnormal{f},0} (\Gamma) \geq \frac{\Gamma - \gamma_{\textnormal {min}}}{\gamma_{\textnormal {max}} - \gamma_{\textnormal {min}}} \, C_{\textnormal{f},0}. \label{eq:ccInpustConstFact}
\end{equation}
In fact $C_{\textnormal{f},0} (\cdot)$ is nondecreasing and concave on $[\gamma_{\textnormal {min}}, \gamma_{\textnormal {max}}]$. Indeed, we can divide the blocklength-$n$ transmission into two frames, Frame~1 and Frame~2, with the former of $\alpha n$ channel uses and the latter of $(1 - \alpha) n$ channel uses, where $$\alpha = \frac{\gamma_{\textnormal {max}} - \Gamma}{\gamma_{\textnormal {max}} - \gamma_{\textnormal {min}}}.$$ If in Frame~1 the encoder repeatedly transmits an element of $\setX^\prime$, then the cost constraint will be satisfied irrespective of the inputs in Frame~2. Those can thus be chosen to achieve the unconstrained capacity $C_{\textnormal{f},0}$, with the resulting rate being the RHS of \eqref{eq:ccInpustConstFact}. This proves \eqref{eq:ccInpustConstFact}.

It follows from \eqref{eq:ccInpustConstFact} that $C_{\textnormal{f},0} (\Gamma)$ is positive iff $C_{\textnormal{f},0}$ is positive. By adapting the proof of Theorem~\ref{th:capacity} to account for the cost constraint \eqref{eq:ccInputs} (see Appendix~\ref{sec:pfThCcInputs}), we obtain the following generalization of Theorems~\ref{th:positive} and \ref{th:capacity}:

\begin{theorem}\label{th:ccInputs}
Given any $\Gamma$ satisfying \eqref{eq:ccGamma}, a necessary and sufficient condition for $C_{\textnormal{f},0} (\Gamma)$ to be positive is \eqref{eq:positive}. If $C_{\textnormal{f},0} (\Gamma)$ is positive, then
\begin{IEEEeqnarray}{l}
C_{\textnormal{f},0} (\Gamma) = \min_{P_S} \max_{ \substack{ P_{U,X|S} \colon \\ \Ex {}{\gamma (X)} \leq \Gamma } } \min_{ \substack{ P_{Y|U,X,S} \colon \\ P_{Y|U = u,X,S} \in \mathscr P (W), \,\, \forall \, u \in \setU }} I (U;Y) - I (U;S), \label{eq:ccCapacity}
\end{IEEEeqnarray}
where $U$ is an auxiliary chance variable taking values in a finite set $\setU$, the expectation is computed w.r.t\ the joint PMF $P_S \times P_{U,X|S}$, and the mutual informations are computed w.r.t.\ the joint PMF $P_S \times P_{U,X|S} \times P_{Y|U,X,S}$. Restricting $X$ to be a function of $U$ and $S$, i.e., $P_{U,X|S}$ to have the form \eqref{eq:thCapCardUXFuncUS}, does not change the RHS of \eqref{eq:ccCapacity}, nor does restricting the cardinality of $\setU$ to \eqref{eq:cardU}.
\end{theorem}

Specializing Theorem~\ref{th:ccInputs} to the state-less case, we obtain:

\begin{corollary}\label{co:ccInputsStateless}
For a state-less DMC $\channel y x$ and any $\Gamma$ satisfying \eqref{eq:ccGamma}, $C_{\textnormal{f},0} (\Gamma)$ is positive iff \eqref{eq:positiveShannon} holds. If $C_{\textnormal{f},0} (\Gamma)$ is positive, then it is given by
\begin{IEEEeqnarray}{l}
C_{\textnormal{f},0} (\Gamma) = \max_{ \substack{P_{X} \colon \\ \Ex {}{\gamma (X)} \leq \Gamma } } \min_{ P_{Y|X} \in \mathscr P (W) } I (X;Y), \label{eq:ccCapacityStateless}
\end{IEEEeqnarray}
where the expectation is computed w.r.t.\ the PMF $P_X$ and the mutual information w.r.t.\ the joint PMF $P_X \times P_{Y|X}$.
\end{corollary}

\begin{proof}[Proof of Corollary~\ref{co:ccInputsStateless}]
This follows from Theorem~\ref{th:ccInputs} when we consider an SD-DMC $\channel y {x,s}$ with a single state, i.e., $|\setS| = 1$, whose transition law is
\begin{IEEEeqnarray}{l}
\channel y {x,s} = \channel y x, \,\, \forall \, ( x,s,y ) \in \setX \times \setS \times \setY,
\end{IEEEeqnarray}
because on this channel SI is useless, \eqref{eq:positive} is equivalent to \eqref{eq:positiveShannon}, and the RHS of \eqref{eq:ccCapacity} equals that of \eqref{eq:ccCapacityStateless}.
\end{proof}

The RHS of \eqref{eq:ccCapacityStateless} is a natural generalization of Ahlswede's capacity formula \eqref{eq:capacityAhlswede} to the setting with the cost constraint \eqref{eq:ccInputs}. Since Ahlswede's capacity formula \eqref{eq:capacityAhlswede} is an alternative form for Shannon's capacity formula \eqref{eq:capacityShannon}, one might wonder whether the RHS of \eqref{eq:ccCapacityStateless} can also be expressed as the ``natural'' generalization of Shannon's formula \eqref{eq:capacityShannon}, namely as
\begin{IEEEeqnarray}{l}
\max_{ \substack{ P_X \colon \\ \Ex {}{\gamma (X)} \leq \Gamma } } \min_{ y \in \setY } - \log \sum_{x \in \setX \colon \channel y x > 0 } P_X (x), \label{eq:ccCapacityStatelessShannon}
\end{IEEEeqnarray}
where the expectation is computed w.r.t.\ the PMF $P_X$. The answer is no:

\begin{remark}\label{re:ccShannonSmallerAhlswede}
For any $\Gamma \geq \gamma_{\textnormal{min}}$ and every state-less DMC $\channel y x$
\begin{IEEEeqnarray}{l}
\max_{ \substack{ P_X \colon \\ \Ex {}{\gamma (X)} \leq \Gamma } } \min_{ y \in \setY } - \log \sum_{x \in \setX \colon \channel y x > 0 } P_X (x) \leq \max_{ \substack{P_{X} \colon \\ \Ex {}{\gamma (X)} \leq \Gamma } } \min_{ P_{Y|X} \in \mathscr P (W) } I (X;Y), \label{eq:ccShannonSmallerAhlswede}
\end{IEEEeqnarray}
where the expectations are computed w.r.t.\ the PMF $P_X$ and the mutual information w.r.t.\ the joint PMF $P_X \times P_{Y|X}$. The inequality can be strict.
\end{remark}

The inequality \eqref{eq:ccShannonSmallerAhlswede} is proved in Appendix~\ref{sec:pfReCcShannonSmallerAhlswede}. That it can be strict follows from the example below:

\begin{example}\label{ex:ccShannonStrictlySmallerAhlswede}
Suppose $$\setX = \setY = \{ 0,1 \};$$ that for every $(x,y) \in \setX \times \setY$
\begin{IEEEeqnarray}{rCl}
\channel y x & = & \ind { \{ y = x \} }, \\
\gamma (x) & = & x;
\end{IEEEeqnarray}
and that $0 < \Gamma < 1/2$. The RHS of \eqref{eq:ccShannonSmallerAhlswede} evaluates to
\begin{IEEEeqnarray}{l}
\max_{ \substack{P_{X} \colon \\ \Ex {}{\gamma (X)} \leq \Gamma } } \min_{ P_{Y|X} \in \mathscr P (W) } I (X;Y) = h_{\textnormal b} (\Gamma);
\end{IEEEeqnarray}
the LHS of \eqref{eq:ccShannonSmallerAhlswede} evaluates to
\begin{IEEEeqnarray}{l}
\max_{ \substack{ P_X \colon \\ \Ex {}{\gamma (X)} \leq \Gamma } } \min_{ y \in \setY } - \log \sum_{x \in \setX \colon \channel y x > 0 } P_X (x) = - \log (1 - \Gamma);
\end{IEEEeqnarray}
and
\begin{equation}
- \log (1 - \Gamma) < h_{\textnormal b} (\Gamma), \quad 0 < \Gamma < 1/2.
\end{equation}
\end{example}

The following may explain why the inequality in \eqref{eq:ccShannonSmallerAhlswede} can be strict. Recall Shannon's sequential coding scheme \cite{shannon56}, which achieves the zero-error feedback capacity \eqref{eq:capacityShannon} of the state-less DMC: The encoder selects some PMF $P_X$, and, before every channel use, it maps a fraction of approximately $P_X (x)$ of the survivor set to the input symbol~$x$. If the channel output is $y \in \setY$, then the survivor set is reduced by a factor of nearly
\begin{equation}
\Biggl( \sum_{x \in \setX \colon \channel y x > 0} P_X (x) \Biggr)^{-1}.
\end{equation}
The generalization \eqref{eq:ccCapacityStatelessShannon} of Shannon's capacity formula \eqref{eq:capacityShannon} is obtained when the PMF $P_X$ is restricted to satisfy $\bigEx {}{\gamma (X)} \leq \Gamma$. As the following argument suggests, a more adaptive coding scheme may be required in the presence of the cost constraint \eqref{eq:ccInputs}. To see why, fix some PMF $P_X$ w.r.t.\ which $\bigEx {}{\gamma (X)} \leq \Gamma$, and let $y^\star \in \setY$ maximize
\begin{equation}
\sum_{x \in \setX \colon \channel {y^\star} x > 0} P_X (x).
\end{equation}
If the cost of every input symbol $x \in \setX$ for which $\channel {y^\star} x > 0$ is smaller than $\Gamma$, then the cost constraint loosens for the remaining channel uses, and the encoder should take advantage of this.

\subsection{Constrained States} \label{sec:constrainedStates}

This section provides some insight into how cost constraints on the channel states affect the zero-error feedback capacity of the SD-DMC $\channel y {x,s}$ with acausal SI. Consider some nonnegative ``cost-function'' $\lambda \colon \setS \rightarrow \reals^+_0$, define $$\lambda_{\textnormal {min}} = \min_{s \in \setS} \lambda (s) \quad \textnormal{and} \quad \lambda_{\textnormal {max}} = \max_{s \in \setS} \lambda (s),$$ and let
\begin{equation}
\lambda_{\textnormal {min}} < \Lambda < \lambda_{\textnormal {max}}. \label{eq:ccLambda}
\end{equation}
Like the cost constraint \eqref{eq:ccInputs} on the channel inputs, where we restrict $\Gamma$ to \eqref{eq:ccGamma}, we restrict $\Lambda$ to \eqref{eq:ccLambda}, because all other values of $\Lambda$ are uninteresting. In the following, we shall consider two different cost constraints on the channel states.\\

The first is that, at every blocklength~$n$, the channel states' average cost $$\lambda^{(n)} (S^n) = \frac{1}{n} \sum^n_{i = 1} \lambda (S_i)$$ satisfy the cost constraint
\begin{equation}
\lambda^{(n)} (S^n) \leq \Lambda \label{eq:ccStates1}
\end{equation}
for some given $\Gamma$ satisfying \eqref{eq:ccLambda}. Let $C_{\textnormal{f},0}^{(1)} (\Lambda)$ denote the zero-error feedback capacity with acausal SI subject to \eqref{eq:ccStates1}. Unlike the cost constraint on the channel inputs \eqref{eq:ccInputs}, the cost constraint on the channel states \eqref{eq:ccStates1} affects not only the formula for $C_{\textnormal{f},0}$ when it is positive but also whether $C_{\textnormal{f},0}$ is positive. The reason for this is that the time-sharing argument of Section~\ref{sec:constrainedInputs} does not work for the adversarial state selector: since the state is revealed acausally to the encoder, if the state selector chooses only ``benign'' states of low cost during Frame~1 and only ``hurtful'' states of high cost during Frame~2, then the encoder can concentrate its transmission in the first frame, where the state assumes only ``benign'' realizations of low cost.

Indeed, the cost constraint \eqref{eq:ccStates1} can increase the zero-error feedback capacity with acausal SI from zero:

\begin{remark}\label{re:ccStates1Pos}
Even when $\Lambda$ satisfies \eqref{eq:ccLambda}, the zero-error feedback capacity of an SD-DMC with acausal SI can be zero in the absence of a state cost-constraint yet be positive in its presence.
\end{remark}

We prove Remark~\ref{re:ccStates1Pos} by means of the following example:

\begin{example}\label{ex:stuckAtOne}
Consider a deterministic SD-DMC $\channel y {x,s}$ over the binary alphabets $\setX = \setS = \setY = \{ 0,1 \}$ with the state cost-function
\begin{equation}
\lambda (s) = s, \quad s \in \setS.
\end{equation}
Let the output corresponding to the input $x$ and the state $s$ be the single element of the set $\setY_{x,s}$ that is given in Table~\ref{tb:exStuckAtOne}
\begin{IEEEeqnarray}{l}
\bigl\{ y \in \setY \colon \channel {y}{x,s} > 0 \bigr\} = \setY_{x,s}, \,\, \forall \, (x,s) \in \setX \times \setS.
\end{IEEEeqnarray}
Since \eqref{eq:positive} does not hold for this channel, Theorem~\ref{th:positive} implies that $C_{\textnormal{f},0}$ is zero. However, as shown in Appendix~\ref{sec:exStuckAtOne}, $C_{\textnormal{f},0}^{(1)} (\Lambda)$ is positive when $\Lambda > 0$ is sufficiently small so that
\begin{equation}
\Lambda + h_{\textnormal b} (\Lambda) < 1. \label{eq:exStuckAtOneLambda}
\end{equation}
This holds also in the absence of feedback: because $\channel y {x,s}$ is $\{ 0,1 \}$-valued, the encoder can compute the output from the state (which is revealed to it acausally) and from the input (that it produces), and feedback does not, therefore, increase capacity. 
\end{example}

\begin{table}[h]
\begin{center}
  \begin{tabular}{ c c | c c }
    \multicolumn{2}{c|}{\multirow{2}{*}{$\setY_{x,s}$}} & \multicolumn{2}{c}{$s$} \\
  	&& 0 & 1 \\
    \hline    
    \multirow{2}{*}{$x$} & 0 & \{ 0 \} & \{ 1 \} \\
    & 1 & \{ 1 \} & \{ 1 \}
  \end{tabular}
  \caption{Nonzero transitions of the SD-DMC in Example~\ref{ex:stuckAtOne}.}
  \label{tb:exStuckAtOne}
\end{center}
\end{table}

If $\channel y {x,s}$ satisfies \eqref{eq:positive}, i.e., if $C_{\textnormal{f},0}$ is positive in the absence of a state cost-constraint, then we can adapt the proof of Theorem~\ref{th:capacity} to account for the cost constraint \eqref{eq:ccStates1} and to thus express $C_{\textnormal{f},0}^{(1)} (\Lambda)$ as the ``natural'' generalization of \eqref{eq:capacity}, i.e., as the RHS of \eqref{eq:capacityCCStates} ahead. However, by Remark~\ref{re:ccStates1Pos} the capacity can be positive also when \eqref{eq:positive} does not hold; and for this case we do not have a generalization of Theorems~\ref{th:positive} and \ref{th:capacity}. The difficulty in extending Theorems~\ref{th:positive} and \ref{th:capacity} to this case is that the cost constraint \eqref{eq:ccStates1} allows the adversarial state selector to choose whichever states it likes in $\beta n$ epochs, where
\begin{equation}
\beta = \frac{\Lambda - \lambda_{\textnormal {min}}}{\lambda_{\textnormal {max}} - \lambda_{\textnormal {min}}}, \label{eq:fracChUsesStatesUnconst}
\end{equation}
and these epochs are not revealed to the receiver. This is problematic, because the coding schemes by which we prove the direct parts of Theorems~\ref{th:positive} and~\ref{th:capacity} comprise multiple short transmission phases. For example, the last block of the coding scheme by which we prove the direct part of Theorem~\ref{th:capacity} is of negligible length compared to $n$ and consequently also compared to $\beta n$, and hence the adversarial state selector is free to choose whichever states it likes during the last block.\\

The second type of cost constraint we consider is that, for some fixed $l \in \naturals$ and at every blocklength~$n$, the channel states satisfy the cost constraint
\begin{IEEEeqnarray}{l}
\frac{1}{l} \sum^{j l}_{i = ( j - 1 ) l + 1} \lambda (S_i) \leq \Lambda, \,\, \Bigl( \forall \, j \in \naturals \textnormal{ s.t.\ } j l \leq n \Bigr). \label{eq:ccStates2}
\end{IEEEeqnarray}
Note that \eqref{eq:ccStates2} is more stringent than \eqref{eq:ccStates1}, because it constrains the average cost of prespecified $l$-blocks of consecutive channel states and consequently also the channel states' average over the entire blocklength. The zero-error capacity subject to \eqref{eq:ccStates2}, $C^{(2)}_{\textnormal{f},0} (\Lambda,l)$, depends on $l$. We define the zero-error feedback capacity of the SD-DMC $\channel y {x,s}$ with acausal SI under this type of constraint as $$\liminf_{l \rightarrow \infty} C^{(2)}_{\textnormal{f},0} (\Lambda,l),$$ and we denote it $C_{\textnormal{f},0}^{(2)} (\Lambda)$. By adapting the proofs of Theorems~\ref{th:positive} and \ref{th:capacity} to account for the cost constraint \eqref{eq:ccStates2} (see Appendix~\ref{sec:pfThCcStates2Capacity}), we obtain the following single-letter characterization of $C_{\textnormal{f},0}^{(2)} (\Lambda)$:

\begin{theorem}\label{th:ccStates2Capacity}
Given any $\Lambda$ satisfying \eqref{eq:ccLambda}, a necessary condition for $C^{(2)}_{\textnormal{f},0} (\Lambda)$ to be positive is that
\begin{IEEEeqnarray}{l}
\biggl( \forall \, s, \, s^\prime \in \setS \textnormal{ s.t.\ } \frac{\lambda (s) + \lambda (s^\prime)}{2} \leq \Lambda \biggr) \quad \exists \, x, \, x^\prime \in \setX \textnormal{ s.t.\ } \nonumber \\*[-0.5\normalbaselineskip]
  \label{eq:positiveCCStates}
\\*[-0.5\normalbaselineskip]
\qquad \Bigl( \channel y {x,s} \, \channel y {x^\prime,s^\prime} = 0, \,\, \forall \, y \in \setY \Bigr). \nonumber
\end{IEEEeqnarray}
If this condition holds, then
\begin{IEEEeqnarray}{l}
C^{(2)}_{\textnormal{f},0} (\Lambda) = \min_{ \substack{ P_S \colon \\ \Ex {}{\lambda (S)} \leq \Lambda} } \max_{ P_{U,X|S} } \min_{ \substack{ P_{Y|U,X,S} \colon \\ P_{Y|U = u,X,S} \in \mathscr P (W), \,\, \forall \, u \in \setU }} I (U;Y) - I (U;S), \label{eq:capacityCCStates}
\end{IEEEeqnarray}
where $U$ is an auxiliary chance variable taking values in a finite set $\setU$, the expectation is computed w.r.t\ the PMF $P_S$, and the mutual informations are computed w.r.t.\ the joint PMF $P_S \times P_{U,X|S} \times P_{Y|U,X,S}$. Restricting $X$ to be a function of $U$ and $S$, i.e., $P_{U,X|S}$ to have the form \eqref{eq:thCapCardUXFuncUS}, does not change the RHS of \eqref{eq:capacityCCStates}, nor does restricting the cardinality of $\setU$ to \eqref{eq:cardU}.
\end{theorem}

We do not know whether \eqref{eq:positiveCCStates} guarantees that the RHS of \eqref{eq:capacityCCStates} be positive, and hence we do not know whether \eqref{eq:positiveCCStates} is also sufficient for $C_{\textnormal{f},0} (\Lambda)$ to be positive.

For the deterministic SD-DMC of Example~\ref{ex:stuckAtOne}, Theorem~\ref{th:ccStates2Capacity} yields the following result:

\begin{example}
For the channel and cost-function of Example~\ref{ex:stuckAtOne}
\begin{equation}
C_{\textnormal{f},0}^{(2)} (\Lambda) = 1 - \Lambda, \quad 0 \leq \Lambda \leq 1.
\end{equation}
\end{example}

\begin{proof}
Here \eqref{eq:positiveCCStates} holds iff $\Lambda < 1$, so the capacity is zero if $\Lambda = 1$. (This could have also been established by noting that the all-one state-sequence results in the output being one irrespective of the input.) If $0 \leq \Lambda < 1$, then the capacity is
\begin{IEEEeqnarray}{rCl}
C^{(2)}_{\textnormal{f},0} (\Lambda) & = & \min_{ \substack{ P_S \colon \\ P_S (1) \leq \Lambda} } \max_{ P_{U,X|S} } I (U;Y) - I (U;S) \\
& = & \min_{ \substack{ P_S \colon \\ P_S (1) \leq \Lambda} } \sum_{s \in \setS} P_S (s) \log \bigl| \bigl\{ y \in \setY \colon \exists \, x \in \setX \textnormal{ s.t.\ } \channel y {x,s} > 0 \bigr\} \bigr| \\
& = & 1 - \Lambda,
\end{IEEEeqnarray}
where the mutual informations are computed w.r.t.\ the joint PMF $P_S \times P_{U,X|S} \times W$, and the first two equalities can be proved similarly as in Appendix~\ref{sec:anExYFunXS}.
\end{proof}

\section{Selected Proofs}\label{sec:proofs}

This section contains the proofs of the results in Section~\ref{sec:acaus}: Theorem~\ref{th:positive} is proved in Section~\ref{sec:pfThPositive}; Theorem~\ref{th:capacity} in Section~\ref{sec:pfThCapacity}; and Theorem~\ref{th:zeroWithoutFBPosWithFB} in Section~\ref{sec:pfThZeroWithoutFBPosWithFB}.

\subsection{A Proof of Theorem~\ref{th:positive}}\label{sec:pfThPositive}

The proof consists of a direct and a converse part. We first establish the direct part. In fact, we prove the following stronger result:

\begin{remark}\label{re:chUsesPerBit}
Consider an SD-DMC $\channel {y}{x,s}$ with feedback whose encoder is furnished with acausal SI. If \eqref{eq:positive} holds, then $n_{\textnormal{bit}}$ channel uses suffice for the error-free transmission of a bit, where $n_{\textnormal{bit}}$ is 1 if $|\setS| = 1$, and is otherwise upper-bounded by\footnote{Note that all logarithms in \eqref{eq:nbitState} are nonnegative, because \eqref{eq:positive} implies that $|\setY| \geq 2$.}
\begin{IEEEeqnarray}{l}
\frac{ 2 \, |\setY| \log |\setS| - \log |\setY|}{\log |\setY| - \log \bigl( |\setY| - 1 \bigr)} + 1 + 2 \, |\setY|. \label{eq:nbitState}
\end{IEEEeqnarray}
\end{remark}

The direct part of Theorem~\ref{th:positive} follows from Remark~\ref{re:chUsesPerBit}, because if \eqref{eq:positive} is satisfied, then, by Remark~\ref{re:chUsesPerBit},
\begin{IEEEeqnarray}{l}
C_{\textnormal{f},0} \geq 1 / n_{\textnormal{bit}} > 0.
\end{IEEEeqnarray}

In proving Remark~\ref{re:chUsesPerBit} we focus on the case $|\setS| \geq 2$, because the case $|\setS| = 1$ follows directly from Shannon \cite{shannon56}. (In this case \eqref{eq:positive} is equivalent to \eqref{eq:positiveShannon}.) 

Before we prove Remark~\ref{re:chUsesPerBit}, we briefly describe the coding scheme that we propose. Because the zero-error capacity of the SD-DMC $\channel y {x,s}$ with acausal SI can be zero without feedback but positive with feedback (Theorem~\ref{th:zeroWithoutFBPosWithFB}), it is not always possible to transmit a single bit error-free in only one channel use (Corollary~\ref{co:moreThanOneChUse}). Our scheme thus requires more than one channel use, and it utilizes the feedback link.

The scheme has two phases. Phase~1 is not used to convey the bit but rather to reduce the decoder's ambiguity about the Phase-2 state-sequence. This is attained with an adaptive feedback code reminiscent of the one used in the first phase of Shannon's coding scheme for the stateless DMC \cite{shannon56}. But in our Phase~1, the encoder utilizes the Phase-1 state-sequence (albeit only causally). After Phase~1 the decoder computes the set of Phase-2 state-sequences of positive posterior probability given the Phase-1 outputs. This set can also be computed by the encoder thanks to the Phase-1 feedback. This enables the encoder to transmit the bit error-free in Phase~2. The feedback link is not used in Phase~2.

The condition in Theorem~\ref{th:positive} ensures that Phase~1 and 2 are feasible. As we shall see, Phase~1 is feasible iff \eqref{eq:condCapacityPos} holds, whereas Phase~2 is feasible iff \eqref{eq:positive} holds, where by Remarks~\ref{re:posSuffButNotNec1} and \ref{re:posSuffButNotNec} $$\eqref{eq:positive} \implies \eqref{eq:condCapacityPos} \quad \text{and} \quad \eqref{eq:positive} \notimpliedby \eqref{eq:condCapacityPos},$$ so feasibility is easier to attain in Phase~1 than in Phase~2.

\begin{proof}[Proof of Remark~\ref{re:chUsesPerBit}]
The case $|\setS| = 1$ follows from Shannon \cite{shannon56}, and we hence assume that $|\setS| \geq 2$. To transmit a single bit $m \in \{ 0,1 \}$, we divide the blocklength-$n_{\textnormal{bit}}$ transmission into Phase~1 and Phase~2 of $n_1$ and $n_2$ channel uses, where
\begin{equation}
n_{\textnormal{bit}} = n_1 + n_2. \label{eq:nBitSumN1N2}
\end{equation}
For now, $( n_{\textnormal{bit}}, n_1, n_2 )$ could be any triple of positive integers satisfying \eqref{eq:nBitSumN1N2}. At the end of the proof, we shall exhibit a choice of the triple for which the transmission is error-free and $n_{\textnormal{bit}}$ is upper-bounded by \eqref{eq:nbitState}. Before we do that, we describe Phase~1 and Phase~2, beginning with Phase~1.


Let $\setS^{n_1 + n_2}$ denote the set of possible length-$(n_1 + n_2)$ state-sequences, and let $\setS^{n_2}$ denote the set of possible state sequences occurring during Phase~2. Before the transmission begins, the encoder observes the entire state sequence $S^{n_1 + n_2}$. The goal of Phase~1 is to produce a random subset $\bm { \mathcal S}_{n_1} \subseteq \setS^{n_2}$ with the following three properties: 1)~$\bm { \mathcal S}_{n_1}$ is determined by the Phase-1 outputs $Y_1, \ldots, Y_{n_1}$, so both encoder and decoder know $\bm { \mathcal S}_{n_1}$ before Phase~2 begins; 2)~with probability one $\bm { \mathcal S}_{n_1}$ contains the Phase-2 state-sequence $S^{n_1 + n_2}_{n_1 + 1}$; and 3)~the cardinality of $\bm { \mathcal S}_{n_1}$ is upper-bound by
\begin{IEEEeqnarray}{l}
|\bm { \mathcal S}_{n_1}| \leq \biggl( \frac{|\setY| - 1}{|\setY|} \biggr)^{\!\! n_1} |\setS|^{n_2} + |\setY|. \label{eq:noPhase2SLeft}
\end{IEEEeqnarray}
To that end we partition the set $\bm { \mathcal S}_0 = \setS^{n_2}$ into $|\setY|$ different subsets whose size is between $\bigl\lfloor |\bm { \mathcal S}_0| / |\setY| \bigr\rfloor$ and $\bigl\lceil |\bm { \mathcal S}_0| / |\setY| \bigr\rceil$. We index the $|\setY|$ subsets by the output alphabet $\setY$ and reveal the result to the encoder and decoder. To every pair $( s, y) \in \setS \times \setY$ we assign an input $x (s,y) \in \setX$ for which
\begin{equation}
W \bigl ( y \bigl| x (s,y), s \bigr) = 0. \label{eq:avoidY}
\end{equation}
Such an $x (s,y)$ exists, because substituting $s$ for both $s$ and $s^\prime$ in \eqref{eq:positive} demonstrates that \eqref{eq:positive} implies that there exists a pair of inputs $x^\prime, \, x^{\prime\prime} \in \setX$ for which
\begin{IEEEeqnarray}{l}
\channel {y}{x^\prime,s} \, \channel {y}{x^{\prime\prime},s} = 0, \,\, \forall \, y \in \setY,
\end{IEEEeqnarray}
i.e., for which for every $y \in \setY$ either $\channel y {x^\prime,s}$ or $\channel {y}{x^{\prime\prime},s}$ is zero. We can thus choose $x (s,y)$ to be $x^\prime$ when $\channel y {x^\prime,s}$ is zero and to be $x^{\prime\prime}$ when it is not.\footnote{This is nothing else but $\eqref{eq:positive} \implies \eqref{eq:condCapacityPos}$, which follows from Remarks~\ref{re:posSuffButNotNec1} and \ref{re:posSuffButNotNec}.} If, thanks to its acausal SI, the encoder knows that the Time-1 state $S_1$ is $s$ and that $S^{n_1 + n_2}_{n_1 + 1}$ is in the subset of $\bm { \mathcal S}_0$ indexed by $y$, then at Time~1 it transmits $x (s,y)$. This choice guarantees by \eqref{eq:avoidY} that, upon observing the Time-1 output $Y_1$, the decoder will know that the Phase-2 state-sequence is not an element of the subset of $\bm { \mathcal S}_0$ indexed by $Y_1$, and that it is thus in the $\bm { \mathcal S}_0$-complement of this subset, which we denote $\bm { \mathcal S}_1$. Note that: 1)~both encoder and decoder know $\bm { \mathcal S}_1$ after Channe-Use~1;~2) $\bm { \mathcal S}_1$ contains $S^{n_1 + n_2}_{n_1 + 1}$; and 3)~the cardinality of $\bm { \mathcal S}_1$ is upper-bounded by
\begin{IEEEeqnarray}{l}
|\bm { \mathcal S}_1| \leq |\bm { \mathcal S}_0| - \biggl\lfloor \frac{|\bm { \mathcal S}_0|}{|\setY|} \biggr\rfloor = \biggl\lceil \frac{|\setY| - 1}{|\setY|} \, |\bm { \mathcal S}_0| \biggr\rceil \leq \frac{|\setY| - 1}{|\setY|} \, |\bm { \mathcal S}_0| + 1.
\end{IEEEeqnarray}

Phase~1 continues in the same fashion: Let $i \in [2:n_1]$, and assume that the first $i-1$ channel uses have produced a random subset $\bm { \mathcal S}_{i-1}$ of $\setS^{n_2}$ with the following three properties: 1)~both encoder and decoder know $\bm { \mathcal S}_{i-1}$ after Channel-Use~$(i-1)$; 2)~$\bm { \mathcal S}_{i-1}$ contains $S^{n_1 + n_2}_{n_1 + 1}$; and 3)~the cardinality of $\bm { \mathcal S}_{i-1}$ is upper-bounded by
\begin{IEEEeqnarray}{l}
|\bm {\mathcal S}_{i-1}| \leq \frac{|\setY| - 1}{|\setY|} \, |\bm { \mathcal S}_{i-2}| + 1.
\end{IEEEeqnarray}
After Channel-Use~$(i-1)$, we partition $\bm { \mathcal S}_{i-1}$ into $|\setY|$ different subsets whose size is between $\bigl\lfloor |\bm { \mathcal S}_{i-1}| / |\setY| \bigr\rfloor$ and $\bigl\lceil |\bm { \mathcal S}_{i-1}| / |\setY| \bigr\rceil$. We index the subsets by the elements of the output alphabet $\setY$ and reveal the result to the encoder and decoder. If, thanks to its acausal SI, the encoder knows that the Time-$i$ state $S_i$ is $s$ and that $S^{n_1 + n_2}_{n_1 + 1}$ is an element of the subset of $\bm {\mathcal S}_{i-1}$ indexed by $y$, then it transmits $x (s,y)$ at Time~$i$. This choice guarantees by \eqref{eq:avoidY} that, upon observing the Time-$i$ output $Y_i$, the decoder will know that the Phase-2 state-sequence is not an element of the subset indexed by $Y_i$, and that it is thus in the $\bm { \mathcal S}_{i-1}$-complement of this subset, which we denote $\bm { \mathcal S}_i$. Note that: 1)~both encoder and decoder know $\bm { \mathcal S}_i$ after Channel-Use~$i$; 2)~$\bm { \mathcal S}_i$ contains $S^{n_1 + n_2}_{n_1 + 1}$; and 3)~the cardinality of $\bm { \mathcal S}_i$ is upper-bounded by
\begin{IEEEeqnarray}{l}
|\bm { \mathcal S}_i| \leq |\bm { \mathcal S}_{i-1}| - \biggl\lfloor \frac{|\bm { \mathcal S}_{i-1}|}{|\setY|} \biggr\rfloor = \biggl\lceil \frac{|\setY| - 1}{|\setY|} \, |\bm { \mathcal S}_{i-1}| \biggr\rceil \leq \frac{|\setY| - 1}{|\setY|} \, |\bm { \mathcal S}_{i-1}| + 1. \label{eq:sizeSeti}
\end{IEEEeqnarray}

Since this holds for every $i \in [1 : n_1]$, the goal of Phase~1 is attained, and the first $n_1$ channel uses produce a random subset $\bm { \mathcal S}_{n_1}$ of $\setS^{n_2}$ with the following three properties: 1)~both encoder and decoder know $\bm { \mathcal S}_{n_1}$ before Phase~2 begins; 2)~$\bm { \mathcal S}_{n_1}$ contains the Phase-2 state-sequence $S^{n_1 + n_2}_{n_1 + 1}$; and~3) the cardinality of $\bm { \mathcal S}_{n_1}$ is upper-bound by
\begin{IEEEeqnarray}{rCl}
|\bm { \mathcal S}_{n_1}| & \leq & \biggl( \frac{|\setY| - 1}{|\setY|} \biggr)^{\!\! n_1} |\bm { \mathcal S}_0| + \sum^{n_1 - 1}_{i = 0} \biggl( \frac{|\setY| - 1}{|\setY|} \biggr)^{\!\! i} \\
& = & \biggl( \frac{|\setY| - 1}{|\setY|} \biggr)^{\!\! n_1} |\setS|^{n_2} + \frac{|\setY|^{n_1} - \bigl( |\setY| - 1 \bigr)^{n_1}}{|\setY|^{n_1} - \bigl( |\setY| - 1 \bigr) |\setY|^{n_1 - 1}} \\
& = & \biggl( \frac{|\setY| - 1}{|\setY|} \biggr)^{\!\! n_1} |\setS|^{n_2} + \frac{|\setY|^{n_1} - \bigl( |\setY| - 1 \bigr)^{n_1}}{|\setY|^{n_1-1}} \\
& \leq & \biggl( \frac{|\setY| - 1}{|\setY|} \biggr)^{\!\! n_1} |\setS|^{n_2} + |\setY|.
\end{IEEEeqnarray}

We next turn to Phase~2 whose goal is to transmit the bit error-free. To that end the encoder allocates to every bit value $m \in \{ 0, 1 \}$ and every state sequence $\vecs$ in $\bm { \mathcal S}_{n_1}$ a length-$n_2$ codeword $\vecx (m, \vecs)$, where the codewords are chosen so that
\begin{IEEEeqnarray}{l}
\!\!\!\!\!\!\!\! \forall \, \vecs, \, \vecs^\prime \in \bm { \mathcal S}_{n_1} \quad \exists \, i \in [1 : n_2] \textnormal{ s.t.\ } \Bigl( W \bigl( y \bigl| x_i (0, \vecs), s_i \bigr) \, W \bigl( y \bigl| x_i (1, \vecs^\prime), s^\prime_i \bigr) = 0, \,\, \forall \, y \in \setY \Bigr). \label{eq:condPhase2}
\end{IEEEeqnarray}
(We will shortly show how this can be done.) If the value of the bit to be sent is $m \in \{ 0,1 \}$ and if the Phase-2 state-sequence is $\vecs$, then the encoder transmits in Phase~2 the codeword $\vecx (m, \vecs)$. Condition \eqref{eq:condPhase2} implies that, upon observing the realization $\vecy \in \setY^{n_2}$ of the Phase-2 output-sequence $Y^{n_1 + n_2}_{n_1 + 1}$, the decoder, who knows $\bm { \mathcal S}_{n_1}$ and the codewords $\bigl\{ \vecx (\tilde m, \tilde \vecs) \bigr\}$, can determine the value of $m$ error-free, because for the true realization $\vecs \in \bm { \mathcal S}_{n_1}$ of the Phase-2 state-sequence
\begin{IEEEeqnarray}{l}
\prod^{n_2}_{i = 1} W \bigl( y_i \bigl| x_i (m, \vecs), s_i \bigr) > 0,
\end{IEEEeqnarray}
whereas \eqref{eq:condPhase2} implies for $m^\prime \neq m$
\begin{IEEEeqnarray}{l}
\prod^{n_2}_{i = 1} W \bigl( y_i \bigl| x_i (m^\prime, \tilde \vecs), \tilde s_i \bigr) = 0, \,\, \forall \, \tilde \vecs \in \bm { \mathcal S}_{n_1}.
\end{IEEEeqnarray}
The decoder can thus calculate $\prod_i W \bigl( y_i \bigl| x_i (\tilde m, \tilde \vecs), \tilde s_i \bigr)$ for each $\tilde m \in \{ 0,1 \}$ and $\tilde \vecs \in \bm { \mathcal S}_{n_1}$ and produce the $\tilde m$ for which this product is positive for some $\tilde \vecs \in \bm {\mathcal S}_{n_1}$.

One (inefficient) way to achieve \eqref{eq:condPhase2} is the following. Let $x^\star$ be an arbitrary fixed element of $\setX$, and for every pair $s, \, s^\prime \in \setS$ choose a pair $x (s, s^\prime), \, x^\prime (s, s^\prime) \in \setX$ for which
\begin{IEEEeqnarray}{l}
W \bigl( y \bigl| x (s, s^\prime), s \bigr) \, W \bigl( y \bigl| x^\prime (s, s^\prime), s^\prime \bigr) = 0, \,\, \forall \, y \in \setY. \label{eq:choiceXXPrimeSSPrime}
\end{IEEEeqnarray}
By \eqref{eq:positive} such a pair $x (s, s^\prime), \, x^\prime (s, s^\prime)$ exists. Now choose
\begin{equation}
n_2 \geq |\bm { \mathcal S}_{n_1}|^2; \label{eq:ineffChoiceN2}
\end{equation}
allocate to every ordered pair $(\vecs, \vecs^\prime) \in \bm { \mathcal S}_{n_1} \times \bm { \mathcal S}_{n_1}$ a different index $i \in \bigl[ 1 : |\bm { \mathcal S}_{n_1}|^2 \bigr]$; and for the allocated index $i$ choose $x_i (0, \vecs) = x (s_i, s^\prime_i)$ and $x_i (1, \vecs^\prime) = x^\prime (s_i, s^\prime_i)$, and thus guarantee, by \eqref{eq:choiceXXPrimeSSPrime}, that
\begin{IEEEeqnarray}{l}
\Bigl( W \bigl( y \bigl| x_i (0, \vecs), s_i \bigr) \, W \bigl( y \bigl| x_i (1, \vecs^\prime), s^\prime_i \bigr) = 0, \,\, \forall \, y \in \setY \Bigr).
\end{IEEEeqnarray}
The above specifies $|\bm { \mathcal S}_{n_1}|$ out of $n_2 \geq |\bm { \mathcal S}_{n_1}|^2$ symbols of each codeword $\vecx (m, \vecs)$. How we choose the other $n_2 - |\bm { \mathcal S}_{n_1}|$ symbols is immaterial. To be explicit, we choose each of them to be $x^\star$. The described choice of the codewords $\bigl\{ \vecx (m, \vecs) \bigr\}$ clearly satisfies \eqref{eq:condPhase2}. Hence, it would only remain to exhibit some choice of the triple $( n_{\textnormal{bit}}, n_1, \, n_2 )$ satisfying \eqref{eq:nBitSumN1N2} and \eqref{eq:ineffChoiceN2}. This can be done using \eqref{eq:noPhase2SLeft}, but the resulting value of $n_{\textnormal{bit}}$ need not be upper-bounded by \eqref{eq:nbitState}. To fix this, we allocate the indices more efficiently. Note that for every $i \in \bigl[ 1 : |\bm { \mathcal S}_{n_1}|^2 \bigr]$ the above choice of the codewords $\bigl\{ \vecx (m, \vecs) \bigr\}$ allocates meaningful values to the $i$-th symbols of only two codewords, namely $\vecx (0, \vecs)$ and $\vecx (1, \vecs^\prime)$, where $( \vecs, \vecs^\prime )$ is the ordered pair to which Index~$i$ has been allocated. More efficiently, we can allocate the same index~$i$ to several distinct pairs $(\vecs, \vecs^\prime)$. (Still, we let $x_i (0, \vecs) = x (s_i, s^\prime_i)$ and $x_i (1, \vecs^\prime) = x^\prime (s_i, s^\prime_i)$ when Index~$i$ has been allocated to the ordered pair $(\vecs, \vecs^\prime)$, and we choose each codeword symbol that has not been assigned a value to be $x^\star$.) This works whenever any two distinct pairs $( \vecs, \vecs^\prime ), \, ( \tilde \vecs, \tilde \vecs^\prime )$ that are allocated the same index~$i$ satisfy $\vecs \neq \tilde \vecs$ and $\vecs^\prime \neq \tilde \vecs^\prime$, because then every codeword symbol $x_i (m, \vecs)$ is assigned exactly one value. An efficient way to allocate the indices and guarantee that this requirement is met is the following. Instead of \eqref{eq:ineffChoiceN2}, choose any integer $n_2$ that satisfies
\begin{equation}
n_2 \geq |\bm { \mathcal S}_{n_1}|. \label{eq:pfPositiveRequirementN2}
\end{equation}
(An explicit choice for which $n_{\textnormal{bit}}$ is upper-bounded by \eqref{eq:nbitState} will be given in \eqref{bl:explChoiceN1N2}.) Index the elements of $\bm { \mathcal S}_{n_1}$ by $\bigl[ 1 : |\bm { \mathcal S}_{n_1}| \bigr]$, where $\vecs (j)$ denotes the element of $\bm { \mathcal S}_{n_1}$ indexed by $j$. Allocate to every ordered pair $\bigl( \vecs (k), \vecs (\ell) \bigr)$, where $k, \, \ell \in \bigl[ 1 : |\bm { \mathcal S}_{n_1}| \bigr]$, the index
\begin{IEEEeqnarray}{l}
i (k,\ell) = \bigl( \ell - k \mod |\bm { \mathcal S}_{n_1}| \bigr) + 1, \label{eq:pfPositiveIEllJ}
\end{IEEEeqnarray}
which clearly satisfies
\begin{equation}
i \in \bigl[ 1 : |\bm { \mathcal S}_{n_1}| \bigr] \subseteq [ 1 : n_2 ].
\end{equation}
By \eqref{eq:pfPositiveIEllJ} any two distinct pairs $\bigl( \vecs (k), \vecs (\ell) \bigr), \, \bigl( \vecs (k^\prime), \vecs (\ell^\prime) \bigr)$ that are allocated the same index $i$ satisfy $k \neq k^\prime$ and $\ell \neq \ell^\prime$, so $\vecs (k) \neq \vecs (k^\prime)$ and $\vecs (\ell) \neq \vecs (\ell^\prime)$.

To conclude the direct part, it remains to exhibit some choice of the triple $( n_{ \textnormal{bit} }, n_1, \, n_2 )$ satisfying \eqref{eq:nBitSumN1N2} and \eqref{eq:pfPositiveRequirementN2}. By \eqref{eq:noPhase2SLeft} these are satisfied if
\begin{subequations}\label{bl:explChoiceN1N2}
\begin{IEEEeqnarray}{rCl}
n_1 & = & \biggl\lceil \frac{ 2 \, |\setY| \log |\setS| - \log |\setY|}{\log |\setY| - \log \bigl( |\setY| - 1 \bigr)} \biggr\rceil, \\
n_2 & = & 2 \, |\setY|, \\
n_{\textnormal{bit}} & = & \biggl\lceil \frac{ 2 \, |\setY| \log |\setS| - \log |\setY|}{\log |\setY| - \log \bigl( |\setY| - 1 \bigr)} \biggr\rceil + 2 \, |\setY|,
\end{IEEEeqnarray}
\end{subequations}
and for this choice $n_{\textnormal{bit}}$ is upper-bounded by \eqref{eq:nbitState}.
\end{proof}

We next prove the converse part of Theorem~\ref{th:positive}.

\begin{proof}[Converse Part]
To show that \eqref{eq:positive} is necessary for $C_{\textnormal {f}, 0}$ to be positive, we need to prove that if \eqref{eq:positive} does not hold, i.e., if there exists a pair of states $s, \, s^\prime \in \setS$ such that
\begin{IEEEeqnarray}{l}
\nexists \, x, \, x^\prime \in \setX \textnormal{ s.t.\ } \Bigl( \channel y {x,s} \, \channel y {x^\prime,s^\prime} = 0, \,\, \forall \, y \in \setY \Bigr), \label{eq:oppOfPositive}
\end{IEEEeqnarray}
then it is impossible to transmit a single bit error-free. Condition~\eqref{eq:oppOfPositive} can be alternatively expressed as
\begin{IEEEeqnarray}{l}
\forall \, x, \, x^\prime \in \setX \quad \exists \, y \in \setY \colon \, \channel y {x,s} \, \channel y {x^\prime,s^\prime} > 0, \label{eq:oppOfPositive2}
\end{IEEEeqnarray}
which makes the claim almost obvious. Indeed, \eqref{eq:oppOfPositive2} implies that, if the state sequence is all $s$ or all $s^\prime$, then---during every channel use and irrespective of the inputs $x, \, x^\prime$ that we choose---the pairs $(x,s)$ and $(x^\prime,s^\prime)$ can produce the same output. This implies that for every pair of messages $m, \, m^\prime \in \setM$ and every encoding mappings there exists an output sequence of positive probability conditional on each of the following two events: 1) the message is $m$, and the state sequence is all $s$; or 2) the message is $m^\prime$, and the state sequence is all $s^\prime$.

To prove this formally, let the bit take values in the set $\setM = \{ 0,1 \}$, and fix a blocklength~$n$ and $n$ encoding mappings $$f_i \colon \setM \times \setS^n \times \setY^{i-1} \rightarrow \setX, \quad i \in [1:n].$$ Denote by $\vecs \in \setS^n$ the all-$s$ and by $\vecs^\prime \in \setS^n$ the all-$s^\prime$ state-sequence, so
\begin{equation}
s_i = s \quad \textnormal{and} \quad s_i^\prime = s^\prime, \,\, \forall \, i \in [1:n]. \label{eq:badOutputsSiSiPrime}
\end{equation}
To show that the mappings do not achieve error-free transmission, we will exhibit an output sequence $\vecy \in \setY^n$ that for every $i \in [1:n]$ satisfies
\begin{IEEEeqnarray}{l} 
W \bigl( y_i \bigl| f_i (0,\vecs,y^{i-1}), s_i \bigr) \, W \bigl( y_i \bigl| f_i (1,\vecs^\prime,y^{i-1}), s^\prime_i \bigr) > 0. \label{eq:badOutputs}
\end{IEEEeqnarray}
This will rule out error-free transmission, because if the state sequence is either $\vecs$ or $\vecs^\prime$, then the decoder, not knowing which, cannot recover the bit.

Our construction of $\vecy \in \setY^n$ is inductive, i.e., we first exhibit a Time-$1$ output $y_1 \in \setY$ that satisfies \eqref{eq:badOutputs} for $i = 1$, and we then repeatedly increment $i$ by one (until it reaches $n$) and exhibit a Time-$i$ output $y_i \in \setY$ that---together with the previously constructed $\{ y_j \}_{j \in [1:i-1]}$---satisfies \eqref{eq:badOutputs}.

We start by exhibiting a Time-1 output $y_1 \in \setY$ that satisfies \eqref{eq:badOutputs} for $i = 1$. To this end we observe from \eqref{eq:oppOfPositive2} and \eqref{eq:badOutputsSiSiPrime} that
\begin{IEEEeqnarray}{l}
\exists \, y \in \setY \textnormal{ s.t.\ } W \bigl( y \bigl| f_1 (0,\vecs), s_1\bigr) \, W \bigl( y \bigl| f_1 (1,\vecs^\prime), s^\prime_1 \bigr) > 0. \label{eq:badOutputsTime1}
\end{IEEEeqnarray}
If $y$ is as promised in \eqref{eq:badOutputsTime1}, then we choose $y_1 = y$ with the result that \eqref{eq:badOutputs} holds for $i = 1$.

For the inductive step, suppose $\ell \in [2:n]$, and that we have already constructed $\{ y_i \}_{i \in [1:\ell - 1]}$ for which \eqref{eq:badOutputs} holds for every $i \in [1:\ell - 1]$. We construct a Time-$\ell$ output $y_\ell \in \setY$ that---together with the previously constructed $\{ y_i \}_{i \in [1:\ell - 1]}$---satisfies \eqref{eq:badOutputs} when we substitute $\ell$ for $i$ in \eqref{eq:badOutputs}, i.e., we show that
\begin{IEEEeqnarray}{l}
\exists \, y_\ell \in \setY \textnormal{ s.t.\ } W \bigl( y_\ell \bigl| f_\ell (0,\vecs,y^{\ell-1}), s_\ell \bigr) \, W \bigl( y_\ell \bigl| f_\ell (1,\vecs^\prime,y^{\ell-1}), s^\prime_\ell \bigr) > 0. \label{eq:badOutputsTimeEll}
\end{IEEEeqnarray}
In fact, \eqref{eq:badOutputsTimeEll} follows from \eqref{eq:oppOfPositive2} and \eqref{eq:badOutputsSiSiPrime}.

Since the construction goes through for every $\ell \in [1:n]$, when $\ell$ reaches $n$ we have constructed an output sequence $\vecy \in \setY^n$ that for every $i \in [1:n]$ satisfies \eqref{eq:badOutputs}.
\end{proof}

\subsection{A Proof of Theorem~\ref{th:capacity}}\label{sec:pfThCapacity}

As we prove in Appendix~\ref{sec:leCardU}, restricting $X$ to be a function of $U$ and $S$, i.e., $P_{U,X|S}$ to have the form \eqref{eq:thCapCardUXFuncUS}, does not change the RHS of \eqref{eq:capacity}, nor does restricting the cardinality of $\setU$ to \eqref{eq:cardU} (Lemma~\ref{le:cardU}). To prove Theorem~\ref{th:capacity} it thus suffices to establish a direct part for the case where the cardinality of $\setU$ is restricted to \eqref{eq:cardU} and a converse part for the case where $\setU$ is any finite set. We first establish the direct part.

\begin{proof}[Direct Part]
Our coding scheme can be roughly described as follows. We partition the blocklength-$n$ transmission into $B + 1$ blocks, with each of the first $B$ blocks being of length~$k$. Each of these blocks is guaranteed to reduce the ``survivor set''---i.e., the set of messages of positive posterior probability given the channel outputs---by at least a factor of nearly $$\min_{P_S} \max_{P_{U,X|S}} \min_{ \substack{ P_{Y|U,X,S} \colon \\ P_{Y|U = u, X,S} \in \mathscr P (W) , \,\, \forall \, u \in \setU}} 2^{ k ( I (U;Y) - I (U;S) ) },$$ where $U$ is an auxiliary chance variable taking values in a finite set $\setU$, and where the mutual informations are computed w.r.t.\ the joint PMF $P_{S} \times P_{U,X|S} \times P_{Y|U,X,S}$. The parameter $B$ is chosen so that the post-Block-$B$ survivor-set be ``small.'' The last block further reduces the survivor-set from a small set to a singleton containing the transmitted message. The coding scheme asymptotically achieves the rate on the RHS of \eqref{eq:capacity}, because, when $B$ and $k$ are large, the last block is of negligible length compared to $B k$ and therefore does not affect the code's asymptotic rate.\\

In the first $B$ blocks our scheme draws on Dueck's scheme for zero-error communication over the multiple-access channel with feedback \cite{dueck85}. Dueck's scheme in turn draws on Ahlswede's work \cite{ahlswede73, ooiwornell98, merhavweissman05}, which was originally motivated by the AVC with feedback, and which on the (state-less) DMC $\channel y x$ achieves the zero-error feedback capacity \eqref{eq:capacityAhlswede} \cite{ahlswede73}. We next describe Blocks~1 through $B$ of Ahlswede's scheme and then show how to adapt them to the present setting.

Fix positive integers $B, \, k$ and a $k$-type $P_X$ on $\setX$. Let $\bm { \mathcal M}_0 \triangleq \setM$ be the set of possible messages, and for every $b \in [1:B]$ let $\bm { \mathcal M}_b$ be the post-Block-$b$ survivor-set, i.e., the (random) set of messages of positive posterior probability given the channel outputs $Y^{b k}$ during the first $b$ blocks. Thus, $\bm { \mathcal M}_b$ is the (random) subset of $\bm { \mathcal M}_{b-1}$ comprising the messages in $\bm { \mathcal M}_{b-1}$ of positive posterior probability given the Block-$b$ outputs $\vecy^{(b)} \triangleq Y^{b k}_{(b-1)k + 1}$. Ahlswede's scheme is designed so as to guarantee that
\begin{IEEEeqnarray}{l}
| \bm { \mathcal M}_b | \lesssim \biggl( \max_{\substack{ P_{Y|X} \in \mathscr P (W)}} 2^{ - k I (X;Y) } \biggr) | \bm { \mathcal M}_{b-1} |, \label{eq:ahlswedeSchemeExpRed}
\end{IEEEeqnarray}
where the mutual information is computed w.r.t.\ the joint PMF $P_X \times P_{Y|X}$.

For every $b \in [1:B]$ Ahlswede's Block-$b$ transmission can be described as follows. Thanks to the feedback link, the set $\bm { \mathcal M}_{b-1}$ can be computed by both transmitter and receiver after Block~$(b-1)$. They can thus agree on a partition of $\bm { \mathcal M}_{b-1}$ into $\bigl| \setT^{(k)}_{P_X} \bigr|$ message sets whose size is between $\bigl\lfloor | \bm { \mathcal M}_{b-1} | / \bigl| \setT^{(k)}_{P_X} \bigr| \bigr\rfloor$ and $\bigl\lceil | \bm { \mathcal M}_{b-1} | / \bigl| \setT^{(k)}_{P_X} \bigr| \bigr\rceil$, and they can agree on a way to associate with each message set a different $k$-tuple from $\setT^{(k)}_{P_X}$. To transmit Message~$m \in \bm { \mathcal M}_{b-1}$, the encoder transmits the $k$-tuple $\vecx^{(b)} \in \setT^{(k)}_{P_X}$ associated with the message set containing $m$. Based on the Block-$b$ outputs $\vecy^{(b)}$, the encoder and decoder compute $\bm { \mathcal M}_b$ as follows: they identify all the $k$-tuples in $\setT^{(k)}_{P_X}$ that could have produced the Block-$b$ outputs $\vecy^{(b)}$, and they compute $\bm { \mathcal M}_b$ as the union of the message sets with which these $k$-tuples are associated.

We next establish \eqref{eq:ahlswedeSchemeExpRed}, or more precisely that
\begin{subequations}\label{bl:ahlswedeCardMbUB}
\begin{IEEEeqnarray}{l}
|\bm { \mathcal M}_b| \leq \biggl( \max_{P_{Y|X} \in \mathscr P (W)} 2^{ - k (I (X;Y) - \alpha_k)} \biggr) |\bm { \mathcal M}_{b-1}|, \label{eq:ahlswedeCardMbUB}
\end{IEEEeqnarray}
whenever
\begin{IEEEeqnarray}{l}
|\bm { \mathcal M}_{b-1}| \geq \bigl| \setT^{(k)}_{P_X} \bigr|, \label{eq:ahlswedeCardCondMbMin1}
\end{IEEEeqnarray}
\end{subequations}
where the mutual information is computed w.r.t.\ the joint PMF $P_X \times P_{Y|X}$, and where $\alpha_k$ is given by
\begin{IEEEeqnarray}{l}
\alpha_k = \frac{\log (1 + k) |\setX| \bigl( 1 + |\setY| \bigr) + 1}{k}
\end{IEEEeqnarray}
and hence converges to zero as $k$ tends to infinity. To this end assume that \eqref{eq:ahlswedeCardCondMbMin1} holds and note that, with probability one, the empirical type of the pair of Block-$b$ inputs and outputs $\bigl( \vecx^{(b)}, \vecy^{(b)} \bigr)$ satisfies
\begin{subequations}\label{bl:ahlswedeCodeConstEmpType}
\begin{IEEEeqnarray}{l}
P_{\vecx^{(b)}} = P_X, \\
\Bigl( \channel y x = 0 \Bigr) \implies \Bigl( P_{\vecx^{(b)},\vecy^{(b)}} (x,y) = 0 \Bigr).
\end{IEEEeqnarray}
\end{subequations}
This allows us to upper-bound the number of $k$-tuples in $\setT^{(k)}_{P_X}$ that could have produced the observed Block-$b$ outputs $\vecy^{(b)}$: For every fixed $k$-type $P_{X,Y}$ on $\setX \times \setY$, the number of $k$-tuples $\vecx$ that satisfy $\bigl( \vecx, \vecy^{(b)} \bigr) \in \setT^{(k)}_{P_{X,Y}}$ cannot exceed $2^{k H (X|Y)}$, where the conditional entropy is computed w.r.t.\ the joint PMF $P_{X,Y}$ \cite[Lemma~2.5]{csiszarkoerner11}. This, combined with \eqref{bl:ahlswedeCodeConstEmpType} and the fact that the number of $k$-types on $\setX \times \setY$ cannot exceed $(1 + k)^{|\setX| \, |\setY|}$, implies that the number of $k$-tuples in $\setT^{(k)}_{P_X}$ that could have produced the observed Block-$b$ outputs $\vecy^{(b)}$ is upper-bounded by
\begin{IEEEeqnarray}{l}
2^{\log (1 + k) |\setX| \, |\setY|} \max_{P_{Y|X} \in \mathscr P (W)} 2^{k H (X|Y)}, \label{eq:ahlswedeCodeUBSeqX}
\end{IEEEeqnarray}
where the conditional entropy is computed w.r.t.\ the joint PMF $P_X \times P_{Y|X}$. Every $k$-tuple from $\setT^{(k)}_{P_X}$ is associated with a message set whose size is at most $\bigl\lceil | \bm { \mathcal M}_{b-1} | / \bigl| \setT^{(k)}_{P_X} \bigr| \bigr\rceil$; and, by the assumption that \eqref{eq:ahlswedeCardCondMbMin1} holds,
\begin{IEEEeqnarray}{rCl}
\Bigl\lceil | \bm { \mathcal M}_{b-1} | / \bigl| \setT^{(k)}_{P_X} \bigr| \Bigr\rceil & \stackrel{(a)}\leq & 2 \, | \bm { \mathcal M}_{b-1} | / \bigl| \setT^{(k)}_{P_X} \bigr| \\
& \stackrel{(b)}\leq & | \bm { \mathcal M}_{b-1} | \, 2^{ - k H (X) + \log (1 + k) |\setX| + 1 }, \label{eq:ahlswedeCodeUBMsgSet}
\end{IEEEeqnarray}
where $(a)$ follows from \eqref{eq:ahlswedeCardCondMbMin1}; and $(b)$ follows from the inequality $\bigl| \setT^{(k)}_{P_X} \bigr| \geq (1 + k)^{- |\setX|} \, 2^{k H (X)}$, where the entropy is computed w.r.t.\ $P_X$ \cite[Lemma~2.3]{csiszarkoerner11}. From \eqref{eq:ahlswedeCodeUBSeqX} and \eqref{eq:ahlswedeCodeUBMsgSet} we obtain \eqref{bl:ahlswedeCardMbUB}.\\

We next sketch our adaption of Ahlswede's scheme to the present setting. For every $b \in [ 1 : B ]$ the Block-$b$ transmission can be described as follows. Before the transmission begins, the encoder is revealed the realization $\vecs^{(b)} \triangleq S^{b k}_{(b-1) k + 1}$ of the Block-$b$ state-sequence. Assume for now that the decoder---while incognizant of $\vecs^{(b)}$---knows its empirical type $P_{\vecs^{(b)}}$: the latter will be conveyed to the decoder error-free in Block~$B + 1$. Let $\bm { \mathcal M}_0 \triangleq \setM$ be the set of possible messages, and let $\bm { \mathcal M}_b$ be the post-Block-$b$ survivor-set, i.e., the (random) subset of $\bm { \mathcal M}_{b-1}$ comprising the messages in $\bm { \mathcal M}_{b-1}$ of positive posterior probability given the Block-$b$ outputs $\vecy^{(b)}$ and the empirical type $P_{\vecs^{(b)}}$. Choose some $k$-type $P_{U,X,S}^{(b)}$ whose $\setS$-marginal $P_S^{(b)}$ equals $P_{\vecs^{(b)}}$. In the following, unless otherwise specified, all entropies and mutual informations are computed w.r.t.\ the joint PMF $P_{U,X,S}^{(b)}$. Unlike Ahlswede's Block~$b$, which partitions $\bm { \mathcal M}_{b-1}$ into $\bigl| \setT^{(k)}_{P_X^{(b)}} \bigr|$ message sets and associates with each a different $k$-tuple from $\setT^{(k)}_{P_X^{(b)}}$, we fix some $\epsilon > 0$ and partition $\bm { \mathcal M}_{b-1}$ into
\begin{equation}
\Theta \triangleq \Bigl\lceil 2^{k (H (U|S) - \epsilon)} \Bigr\rceil \label{eq:defTheta}
\end{equation}
message sets whose size is between $\bigl\lfloor |\bm { \mathcal M}_{b-1}| / \Theta \bigr\rfloor$ and $\bigl\lceil |\bm { \mathcal M}_{b-1}| / \Theta \bigr\rceil$; and we associate with each message set a different bin from the bins $$\setB_\ell \subseteq \setT^{(k)}_{P_U^{(b)}}, \quad \ell \in [ 1 : \Theta ],$$ where the bins $\{ \setB_\ell \}_{\ell \in [ 1 : \Theta ]}$ are pairwise disjoint subsets of $\setT^{ (k) }_{P_U^{(b)}}$
\begin{subequations}\label{bl:binsBellConds}
\begin{IEEEeqnarray}{l}
\setB_\ell \cap \setB_{\ell^\prime} = \emptyset, \,\, \Bigl( \forall \, \ell, \, \ell^\prime \in [ 1 : \Theta ] \textnormal{ s.t.\ } \ell^\prime \neq \ell \Bigr), \label{eq:binsBellDisjoint}
\end{IEEEeqnarray}
and where each bin ``covers'' $\setT^{ (k) }_{P_S^{(b)}}$ in the sense that
\begin{IEEEeqnarray}{l}
\forall \, ( \vecs, \ell ) \in \setT^{ (k) }_{P_S^{(b)}} \times [ 1 : \Theta ] \quad \exists \, \vecu \in \setB_\ell \textnormal{ s.t.\ } ( \vecu, \vecs ) \in \setT^{ (k) }_{ P^{(b)}_{U,S} }. \label{eq:coveringCond}
\end{IEEEeqnarray}
\end{subequations}
(Lemma~\ref{le:covering} ahead guarantees the existence of such bins whenever $k$ is sufficiently large.) To transmit Message~$m \in \bm { \mathcal M}_{b-1}$, the encoder picks from the bin that is associated with the message set containing $m$ a $k$-tuple $\vecu^{(b)}$ satisfying $\bigl( \vecu^{(b)}, \vecs^{(b)} \bigr) \in \setT^{ (k) }_{ P^{(b)}_{U,S} }$. (By \eqref{eq:coveringCond} such a $k$-tuple $\vecu^{(b)}$ exists.) It then chooses as the Block-$b$ channel-inputs some $k$-tuple $\vecx^{(b)}$ satisfying $\bigl( \vecu^{(b)}, \vecx^{(b)}, \vecs^{(b)} \bigr) \in \setT^{ (k) }_{ P^{(b)}_{U,X,S} }$. (This is possible, because $\setT^{ (k) }_{ P^{(b)}_{U,X,S} }$ is not empty since $P_{U,X,S}^{(b)}$ is a $k$-type, and because, by \eqref{eq:coveringCond}, $\bigl( \vecu^{(b)}, \vecs^{(b)} \bigr) \in \setT^{ (k) }_{ P^{(b)}_{U,S} }$.) Based on the Block-$b$ outputs $\vecy^{(b)}$ and the empirical type $P_{\vecs^{(b)}}$, the encoder and decoder compute $\bm { \mathcal M}_b$ as follows. First, they identify all the $k$-tuples in $\setT^{ (k) }_{ P^{(b)}_{U} }$ that could have produced the observed Block-$b$ outputs $\vecy^{(b)}$. Then, they determine all the bins that contain at least one of the identified $k$-tuples. Finally, they compute $\bm { \mathcal M}_b$ as the union of the message sets with which these bins are associated.\footnote{Our Blocks~1 through $B$ are reminiscent of Merhav and Weissman's $\epsilon$-error scheme for the state-dependent DMC with acausal SI and feedback to the encoder \cite[Section~III]{merhavweissman05}, which also draws on \cite{ahlswede73, ooiwornell98}. Unlike the $\epsilon$-error scheme, our Block~$b$ must, however, reduce $\bm {\mathcal M}_{b-1}$ with probability one and hence differs from Block~$b$ of the $\epsilon$-error scheme in the following three aspects: 1) it can deal with every possible Block-$b$ state-sequence, regardless of whether or not its empirical type is close to the PMF $Q$ of the state; 2) for every fixed $k$-type $P^{(b)}_{U,S}$ on $\setU \times \setS$, every Block-$b$ state-sequence $\vecs^{(b)}$ of empirical type $P^{(b)}_S$, and every message $m$ in $\bm {\mathcal M}_{b-1}$, the bin allocated to the message set containing $m$ contains some $k$-tuple $\vecu^{(b)}$ that satisfies $\bigl( \vecu^{(b)},\vecs^{(b)} \bigr) \in \setT^{(k)}_{P_{U,S}}$; and 3) our Block~$b$ can deal with every possible Block-$b$ output-sequence, regardless of whether or not the sequence is typical according to $\channel y {x,s}$.}

Using arguments similar to those for the state-less DMC, we next show that
\begin{subequations}\label{bl:noMessMbGivenOutputs}
\begin{IEEEeqnarray}{l}
| \bm { \mathcal M}_b | \leq \left( \max_{\substack{P_{Y|U,X,S} \colon \\ P_{Y|U = u, X,S} \in \mathscr P (W), \,\, \forall \, u \in \setU}} 2^{ - k ( I (U;Y) - I (U;S) - ( \epsilon + \beta_k ) ) } \right) \! |\bm { \mathcal M}_{b-1}|, \label{eq:noMessMbGivenOutputs}
\end{IEEEeqnarray}
whenever
\begin{IEEEeqnarray}{l}
|\bm { \mathcal M}_{b-1}| \geq 2^{k (H (U|S) - \epsilon)}, \label{eq:noMessMbCondMbMin1}
\end{IEEEeqnarray}
\end{subequations}
where the mutual informations are computed w.r.t.\ the joint PMF $P^{(b)}_{U,X,S} \times P_{Y|U,X,S}$, and where $\beta_k$ is given by
\begin{IEEEeqnarray}{l}
\beta_{k} = \frac{\log (1 + k) |\setU| \, |\setY| + 1}{k} \label{eq:gammaKPfCapacity}
\end{IEEEeqnarray}
and hence converges to zero as $k$ tends to infinity. To this end assume that \eqref{eq:noMessMbCondMbMin1} holds and note that, with probability one, the empirical type of the tuple $(\vecu^{(b)}, \vecx^{(b)}, \vecs^{(b)}, \vecy^{(b)})$ satisfies
\begin{subequations}\label{bl:adaptAhlswedeCodeConstEmpType}
\begin{IEEEeqnarray}{l}
P_{\vecu^{(b)},\vecx^{(b)},\vecs^{(b)}} = P_{U,X,S}^{(b)}, \\
\Bigl( \channel y {x,s} = 0 \Bigr) \implies \Bigl( P_{\vecu^{(b)},\vecx^{(b)},\vecs^{(b)},\vecy^{(b)}} (u,x,s,y) = 0, \,\, \forall \, u \in \setU \Bigr).
\end{IEEEeqnarray}
\end{subequations}
This allows us to upper-bound the number of $k$-tuples in $\setT^{ (k) }_{ P^{(b)}_{U} }$ that could have produced the observed Block-$b$ outputs $\vecy^{(b)}$: For every fixed $k$-type $P_{U,Y}$ on $\setU \times \setY$, the number of $k$-tuples $\vecu$ that satisfy $(\vecu,\vecy^{(b)}) \in \setT^{(k)}_{P_{U,Y}}$ cannot exceed $2^{k H (U|Y)}$, where the conditional entropy is computed w.r.t.\ the joint PMF $P_{U,Y}$ \cite[Lemma~2.5]{csiszarkoerner11}. This, combined with \eqref{bl:adaptAhlswedeCodeConstEmpType} and the fact that the number of $k$-types on $\setU \times \setY$ cannot exceed $(1 + k)^{|\setU| \, |\setY|}$, implies that the number of $k$-tuples in $\setT^{(k)}_{P_U^{(b)}}$ that could have produced the observed Block-$b$ outputs $\vecy^{(b)}$ is upper-bounded by
\begin{IEEEeqnarray}{l}
2^{\log (1 + k) |\setU| \, |\setY|} \max_{\substack{P_{Y|U,X,S} \colon \\ P_{Y|U = u,X,S} \in \mathscr P (W), \,\, \forall \, u \in \setU}} 2^{k H (U|Y)}, \label{eq:noSeqU}
\end{IEEEeqnarray}
where the conditional entropy is computed w.r.t.\ the joint PMF $P_{U,X,S}^{(b)} \times P_{Y|U,X,S}$. Since the bins are pairwise disjoint \eqref{eq:binsBellDisjoint}, no $k$-tuple is contained in more than one bin, and \eqref{eq:noSeqU} is thus also an upper bound on the number of bins that contain at least one $k$-tuple that could have produced the observed Block-$b$ outputs. Every bin is associated with a message set whose size is at most $\bigl\lceil |\bm { \mathcal M}_{b-1}| / \Theta \bigr\rceil$; and, by \eqref{eq:defTheta} and the assumption that \eqref{eq:noMessMbCondMbMin1} holds,
\begin{IEEEeqnarray}{l}
\bigl\lceil |\bm { \mathcal M}_{b-1}| / \Theta \bigr\rceil \leq \Bigl\lceil 2^{-k (H (U|S) - \epsilon)} \, |\bm { \mathcal M}_{b-1}| \Bigr\rceil \leq 2^{- k (H (U|S) - \epsilon) + 1} \, |\bm { \mathcal M}_{b-1}|. \label{eq:noMsgPerBl}
\end{IEEEeqnarray}
From \eqref{eq:noSeqU}, \eqref{eq:noMsgPerBl}, and the fact that
\begin{IEEEeqnarray}{l}
H (U|S) - H (U|Y) = I (U;Y) - I (U;S)
\end{IEEEeqnarray}
we obtain \eqref{bl:noMessMbGivenOutputs}.

Since $H (U|S) \leq \log |\setU|$ and $\epsilon > 0$, it follows from \eqref{bl:noMessMbGivenOutputs} that
\begin{subequations}\label{bl:noMessMbGivenOutputs2}
\begin{IEEEeqnarray}{l}
| \bm { \mathcal M}_b | \leq \left( \max_{\substack{P_{Y|U,X,S} \colon \\ P_{Y|U = u, X,S} \in \mathscr P (W), \,\, \forall \, u \in \setU}} 2^{ - k ( I (U;Y) - I (U;S) - ( \epsilon + \beta_k ) ) } \right) \! |\bm { \mathcal M}_{b-1}|, \label{eq:noMessMbGivenOutputs2}
\end{IEEEeqnarray}
whenever
\begin{IEEEeqnarray}{l}
|\bm { \mathcal M}_{b-1}| \geq 2^{k \log |\setU|}, \label{eq:noMessMbCondMbMin12}
\end{IEEEeqnarray}
\end{subequations}
where the mutual informations are computed w.r.t.\ the joint PMF $P^{(b)}_{U,X,S} \times P_{Y|U,X,S}$, and where $\beta_k$ is defined in \eqref{eq:gammaKPfCapacity}.

From \eqref{bl:noMessMbGivenOutputs2}, which holds for every $b \in [1:B]$, we infer that we can choose $B$ to be the smallest integer for which
\begin{equation}
| \bm { \mathcal M}_B | \leq 2^{k \log |\setU|}. \label{eq:MBSm2PowkLogU}
\end{equation}

In Block~$(B + 1)$ we resolve the post-Block-$B$ survivor-set $\bm { \mathcal M}_B$, and we transmit the empirical types $P_{\vecs^{(1)}}, \ldots, P_{\vecs^{(B)}}$ of the state sequences pertaining to Blocks~1 through $B$. It follows from \eqref{eq:MBSm2PowkLogU} that, when $B$ is large, the number of bits that are needed to resolve $\bm { \mathcal M}_B$ is negligible compared to $B k$. Moreover, when $k$ is large, $B \log (1 + k) |\setS|$, which upper-bounds the number of bits needed to represent $P_{\vecs^{(1)}}, \ldots, P_{\vecs^{(B)}}$, is small compared to $B k$. If we thus choose $B$ and $k$ sufficiently large, then---compared to $B k$---the encoder will only need to transmit few bits error-free in Block~$(B + 1)$, and by Remark~\ref{re:chUsesPerBit} this can be achieved with the length of the last block negligible compared to $B k$.\\

We next describe and analyze our coding scheme in detail, beginning with Blocks~1 through $B$ and ending with the last block. Throughout, we assume that $C_{\textnormal{f},0}$ is positive, which (by Theorem~\ref{th:positive}) is equivalent to the assumption that \eqref{eq:positive} holds.

For Blocks~1 through $B$ we only provide the missing details. Fix positive integers $B, \, k$, some finite set $\setU$ of cardinality
\begin{equation}
|\setU| \leq |\setX|^{|\setS|}, \label{eq:pfCapacityBoundCardU}
\end{equation}
and some $\epsilon > 0$. Assume for now that the decoder knows the empirical types $\bigl\{ P_{\vecs^{(b)}} \bigr\}_{b \in [ 1 : B ]}$ of the state sequences $\bigl\{ \vecs^{(b)} \bigr\}_{b \in [ 1 : B ]}$: those will be conveyed to the decoder error-free in Block~$B + 1$. Let $\bm { \mathcal M}_0 \triangleq \setM$ be the set of possible messages, and for every $b \in [1:B]$ let $\bm { \mathcal M}_b$ be the post-Block-$b$ survivor-set, i.e., the (random) set of messages of positive posterior probability given the channel outputs $Y^{b k}$ and the empirical types $\bigl\{ P_{\vecs^{(b^\prime)}} \bigr\}_{b^\prime \in [ 1 : b ]}$. Thus, $\bm { \mathcal M}_b$ is the subset of $\bm { \mathcal M}_{b-1}$ comprising the messages in $\bm { \mathcal M}_{b-1}$ of positive posterior probability given the Block-$b$ outputs $\vecy^{(b)}$ and the empirical type $P_{\vecs^{(b)}}$ of the Block-$b$ state-sequence $\vecs^{(b)}$. We already described the Block-$b$ transmission for every $b \in [ 1 : B ]$; it only remains to show that we can find bins $$\setB_\ell \subseteq \setT^{(k)}_{P_U^{(b)}}, \quad \ell \in [ 1 : \Theta ]$$ such that \eqref{bl:binsBellConds} holds. This follows from the following lemma:

\begin{lemma}\label{le:covering}
Let $\setU$ and $\setS$ be finite sets. For every $\epsilon > 0$ we can find a positive integer $\eta_0 = \eta_0 \bigl( |\setU|, |\setS|, \epsilon \bigr)$ that will guarantee that, for every $k \geq \eta_0$ and every $k$-type $P_{U,S}$, there exist a partition $\{ \setB_\ell \}_{\ell \in [ 1 : \Theta ]}$ of the type class $\setT^{ (k) }_{P_U}$ with the property that
\begin{IEEEeqnarray}{l}
\forall \, ( \vecs, \ell ) \in \setT^{ (k) }_{P_S} \times [ 1 : \Theta ] \quad \exists \, \vecu \in \setB_\ell \textnormal{ s.t.\ } ( \vecu, \vecs ) \in \setT^{ (k) }_{ P_{U,S} }, \label{eq:leCoveringCond}
\end{IEEEeqnarray}
where $\Theta = \bigl\lceil 2^{k (H (U|S) - \epsilon)} \bigr\rceil$ with $H (U|S)$ being computed w.r.t.\ the joint PMF $P_{U,S}$.
\end{lemma}

\begin{proof}
See Appendix~\ref{sec:leCovering}.
\end{proof}

By Lemma~\ref{le:covering} and \eqref{eq:pfCapacityBoundCardU} we can find a positive integer $\eta_0 = \eta_0 \bigl( |\setX|, |\setS|, \epsilon \bigr)$ that guarantees that, for every $k \geq \eta_0$ and $k$-type $P_U^{(b)}$, there exist bins $$\setB_\ell \subseteq \setT^{(k)}_{P_U^{(b)}}, \quad \ell \in [ 1 : \Theta ]$$ satisfying \eqref{bl:binsBellConds}.

Henceforth, assume that $k \geq \eta_0$ and that the bins are as above. We next conclude the analysis of Blocks~1 through $B$ by showing that each of these blocks can reduce the survivor set by at least a factor of nearly $$\min_{P_S} \max_{P_{U,X|S}} \min_{ \substack{ P_{Y|U,X,S} \colon \\ P_{Y|U = u, X,S} \in \mathscr P (W) , \,\, \forall \, u \in \setU}} 2^{ k ( I (U;Y) - I (U;S) ) },$$ where the mutual informations are computed w.r.t.\ the joint PMF $P_{S} \times P_{U,X|S} \times P_{Y|U,X,S}$. To that end recall that if \eqref{eq:noMessMbCondMbMin12} holds, then $|\bm { \mathcal M}_b|$ can be upper-bounded in terms of $|\bm { \mathcal M}_{b-1}|$ using \eqref{eq:noMessMbGivenOutputs2}, where the mutual informations are computed w.r.t.\ the joint PMF $P^{(b)}_{U,X,S} \times P_{Y|U,X,S}$, and where $\beta_k$ is defined in \eqref{eq:gammaKPfCapacity}. Since we can choose any $k$-type $P_{U,X,S}^{(b)}$ whose $\setS$-marginal $P_S^{(b)}$ is $P_{\vecs^{(b)}}$, we can choose $P_{U,X,S}^{(b)} = P_{\vecs^{(b)}} \times P^{(b)}_{U,X|S}$, where $P^{(b)}_{U,X|S}$ is the conditional $k$-type that---among all conditional $k$-types---maximizes
\begin{IEEEeqnarray}{l}
\min_{\substack{P_{Y|U,X,S} \colon \\ P_{Y|U = u, X,S} \in \mathscr P (W), \,\, \forall \, u \in \setU}} I (U;Y) - I (U;S),
\end{IEEEeqnarray}
where the mutual informations are computed w.r.t.\ the joint PMF $P_{\vecs}^{(b)} \times P^{(b)}_{U,X|S} \times P_{Y|U,X,S}$. Every conditional PMF can be approximated in the total variation distance by a conditional $k$-type when $k$ is sufficiently large; and, because entropy and mutual information are continuous in this distance \cite[Lemma~2.7]{csiszarkoerner11}, it follows that---for the above choice of the conditional $k$-type and some $\gamma_k = \gamma_k \bigl( |\setU|, |\setX|, |\setS|, |\setY| \bigr)$, which converges to zero as $k$ tends to infinity---\eqref{bl:noMessMbGivenOutputs2} implies that when $| \bm { \mathcal M}_{b-1} | \geq 2^{ k \log |\setU| }$
\begin{IEEEeqnarray}{l}
| \bm { \mathcal M}_b | \leq \left( \max_{P_S} \min_{P_{U,X|S}} \max_{ \substack{P_{Y|U,X,S} \colon \\ P_{Y|U = u, X,S} \in \mathscr P (W), \,\, \forall \, u \in \setU}} 2^{ - k ( I (U;Y) - I (U;S) - \epsilon - \gamma_k ) } \right) \! | \bm { \mathcal M}_{b-1} |, \label{eq:cardBSetMSuffLarge}
\end{IEEEeqnarray}
where the mutual informations are computed w.r.t.\ the joint PMF $P_S \times P_{U,X|S} \times P_{Y|U,X,S}$. Because our scheme works for any $\epsilon > 0$, it follows that for every $\epsilon > 0$ and positive integer $k \geq \eta_0 \bigl( |\setX|, |\setS|, \epsilon \bigr)$ each of Blocks~1 through $B$ is guaranteed to reduce the survivor set by a factor of at least
\begin{IEEEeqnarray}{l}
\min_{P_S} \max_{P_{U,X|S}} \min_{ \substack{ P_{Y|U,X,S} \colon \\ P_{Y|U = u, X,S} \in \mathscr P (W) , \,\, \forall \, u \in \setU}} 2^{ k ( I (U;Y) - I (U;S) - \delta (\epsilon, k) ) }, \label{eq:cardBSetMSuffLarge2}
\end{IEEEeqnarray}
until $|\bm { \mathcal M}_B|$ is smaller than $2^{k \log |\setU|}$. Here the mutual informations are computed w.r.t.\ the joint PMF $P_{S} \times P_{U,X|S} \times P_{Y|U,X,S}$, and
\begin{equation}
\delta (\epsilon, k) = \epsilon + \gamma_k \label{eq:defDeltaOfEpsK}
\end{equation}
and hence converges to zero as $\epsilon$ tends to zero and $k$ to infinity.

Since $C_{\textnormal{f},0}$ is positive, so is the RHS of \eqref{eq:capacity}; and, because $\delta (\epsilon, k)$ converges to zero as $\epsilon \downarrow 0$ and $k \rightarrow \infty$, it follows that we can choose $\epsilon$ sufficiently small and $B$ and $k$ sufficiently large so that
\begin{subequations}\label{bl:choiceBKEpsForLastBlock}
\begin{IEEEeqnarray}{l}
k \geq \eta_0 \bigl( |\setX|, |\setS|, \epsilon \bigr) \label{eq:choiceBKEpsForLastBlockK}
\end{IEEEeqnarray}
and
\begin{IEEEeqnarray}{l}
\left( \max_{P_S} \min_{P_{U,X|S}} \max_{ \substack{ P_{Y|U,X,S} \colon \\ P_{Y|U = u, X,S} \in \mathscr P (W) , \,\, \forall \, u \in \setU}} 2^{ - B k ( I (U;Y) - I (U;S) - \delta (\epsilon, k) ) } \right) \! |\setM| \leq 2^{k \log |\setU|}. \label{eq:choiceBKEpsForLastBlock}
\end{IEEEeqnarray}
\end{subequations}
This guarantees that
\begin{equation}
|\bm { \mathcal M}_B| \leq 2^{k \log |\setU|}, \label{eq:setMBSmUPowK}
\end{equation}
because each block reduces the survivor set by the factor in \eqref{eq:cardBSetMSuffLarge2} until $|\bm { \mathcal M}_B|$ is smaller than $2^{k \log |\setU|}$.

We now deal with Block~$B + 1$. In Block~$(B + 1)$ we resolve the post-Block-$B$ survivor-set $\bm { \mathcal M}_B$, and we transmit the empirical types $P_{\vecs^{(1)}}, \ldots, P_{\vecs^{(B)}}$ of the state sequences pertaining to Blocks~1 through $B$. By \eqref{eq:setMBSmUPowK} the resolution of $\bm { \mathcal M}_B$ requires at most $k \log |\setU|$ bits. And since the empirical type of each $\vecs^{(b)}$ can take on at most $(1 + k)^{|\setS|}$ values, we need at most $B \log (1 + k) \, |\setS|$ bits to describe $P_{\vecs^{(1)}}, \ldots, P_{\vecs^{(B)}}$. In the last block we thus need to transmit at most
\begin{equation}
\bigl\lceil k \log |\setU| + B \log (1 + k) \, |\setS| \bigr\rceil
\end{equation}
bits error-free. Remark~\ref{re:chUsesPerBit} and the assumption that $C_{\textnormal{f},0}$ is positive guarantee that this can be achieved by choosing the length of the last block to be
\begin{equation}
\bigl\lceil k \log |\setU| + B \log (1 + k) \, |\setS| \bigr\rceil n_{\textnormal{bit}}, \label{eq:chUsesLastBlock}
\end{equation}
where $n_{\textnormal{bit}} = n_{\textnormal{bit}} \bigl( |\setS|, |\setY| \bigr)$.\\

We are now ready to join the dots and conclude that the coding scheme asymptotically achieves any rate smaller than the RHS of \eqref{eq:capacity}. More precisely, we will show that, for every rate $R$ smaller than the RHS of \eqref{eq:capacity} and every sufficiently-large blocklength~$n$, our coding scheme can convey $n R$ bits error-free in $n$ channel uses. It follows from \eqref{bl:choiceBKEpsForLastBlock} and \eqref{eq:chUsesLastBlock} that if the positive integers $n, \, B, \, k$ and $\epsilon > 0$ are such that
\begin{subequations}\label{bl:choiceBKEpsForLastBlock2}
\begin{IEEEeqnarray}{l}
k \geq \eta_0 \bigl( |\setX|, |\setS|, \epsilon \bigr) \label{eq:choiceBKEpsForLastBlock2K}
\end{IEEEeqnarray}
and
\begin{IEEEeqnarray}{l}
n R \leq B k \left( \min_{P_S} \max_{P_{U,X|S}} \min_{ \substack{P_{Y|U,X,S} \colon \\ P_{Y|U = u, X,S} \in \mathscr P (W), \,\, \forall \, u \in \setU}} I (U;Y) - I (S;Y) - \delta (\epsilon, k ) \right) \!\!, \label{eq:choiceBKEpsForLastBlock2}
\end{IEEEeqnarray}
\end{subequations}
then our coding scheme can convey $n R$ bits error-free in
\begin{IEEEeqnarray}{l}
B k + \bigl\lceil k \log |\setU| + B \log (1 + k) \, |\setS| \bigr\rceil n_{\textnormal{bit}} \label{eq:chUsesBlocks1ThroughBPl1}
\end{IEEEeqnarray}
channel uses. It thus remains to exhibit positive integers $B, \, k$ and some $\epsilon > 0$ such that, for every sufficiently-large blocklength~$n$, \eqref{bl:choiceBKEpsForLastBlock2} holds and
\begin{IEEEeqnarray}{l}
B k + \bigl\lceil k \log |\setU| + B \log (1 + k) \, |\setS| \bigr\rceil n_{\textnormal{bit}} \leq n. \label{eq:chUsesBlocks1ThroughBPl1SmN}
\end{IEEEeqnarray}
As we argue next, when $n$ is sufficiently large we can choose
\begin{subequations}\label{bl:defBk}
\begin{IEEEeqnarray}{rCl}
B & = & \lfloor \sqrt n \rfloor - \bigl\lceil \log |\setU| + \log (1 + \sqrt n) \, |\setS| \bigr\rceil n_{\textnormal{bit}}, \\
k & = & \lfloor \sqrt n \rfloor, \label{eq:defBkk}
\end{IEEEeqnarray}
\end{subequations}
and we can choose any $\epsilon > 0$ for which
\begin{IEEEeqnarray}{l}
R + \epsilon < \min_{P_S} \max_{P_{U,X|S}} \min_{ \substack{P_{Y|U,X,S} \colon \\ P_{Y|U = u, X,S} \in \mathscr P (W), \,\, \forall \, u \in \setU}} I (U;Y) - I (S;Y). \label{eq:defEps}
\end{IEEEeqnarray}
Note that, whenever $n$ is sufficiently large, $B$ is positive and \eqref{eq:chUsesBlocks1ThroughBPl1SmN} is satisfied. To see that also \eqref{bl:choiceBKEpsForLastBlock2} holds whenever $n$ is sufficiently large, we first observe from \eqref{eq:defBkk} that $k$ tends to infinity as $n$ tends to infinity. This implies that \eqref{eq:choiceBKEpsForLastBlock2K} holds whenever $n$ is sufficiently large, and that $\delta (\epsilon, k)$ (which is defined in \eqref{eq:defDeltaOfEpsK}, where $\gamma_k = \gamma_k \bigl( |\setU|, |\setX|, |\setS|, |\setY| \bigr)$ converges to zero as $k$ tends to infinity) converges to $\epsilon$ as $n$ tends to infinity. We next observe that \eqref{bl:defBk} implies that $B k / n$ converges to one as $n$ tends to infinity. This, combined with the fact that $\delta (\epsilon, k)$ converges to $\epsilon$ as $n$ tends to infinity and with \eqref{eq:defEps}, implies that \eqref{eq:choiceBKEpsForLastBlock2} holds whenever $n$ is sufficiently large.

\end{proof}

We next prove the converse part of Theorem~\ref{th:capacity}.

\begin{proof}[Converse Part]
Fix a finite set $\setM$, a blocklength~$n$, and an $(n, \setM)$ zero-error code with $n$ encoding mappings
\begin{equation}
f_i \colon \setM \times \setS^n \times \setY^{i-1} \rightarrow \setX, \quad i \in [1:n] \label{eq:convCapEnc}
\end{equation}
and $|\setM|$ disjoint decoding sets $\setD_m \subseteq \setY^n, \,\, m \in \setM$. We will show that, for some chance variable $U$ of finite support $\setU$, the rate $\frac{1}{n} \log |\setM|$ of the code is upper-bounded by the RHS of \eqref{eq:capacity}.

Draw $M$ uniformly over $\setM$, and denote its distribution $P_M$. Since the code is a zero-error code,
\begin{equation}
\distof {Y^n \in \setD_M} = 1, \label{eq:convZeroDecErr}
\end{equation}
where $\dist$ is the distribution of $(M,S^n,X^n,Y^n)$ induced by $P_M$, the state distribution $Q$, the encoding mappings \eqref{eq:convCapEnc}, and the channel law $\channel y {x,s}$, so for every $(m,\vecs,\vecx,\vecy) \in \setM \times \setS^n \times \setX^n \times \setY^n$
\begin{IEEEeqnarray}{l}
\bigdistof {(M,S^n,X^n,Y^n)=(m,\vecs,\vecx,\vecy)} \nonumber \\
\quad = P_M (m) \, Q^n (\vecs) \prod^n_{i = 1} \Bigl( P_{X_i|M,S^n,Y^{i-1}} (x_i|m,\vecs,y^{i-1}) \, \channel {y_i}{x_i,s_i} \Bigr), \label{eq:capacityConvDist}
\end{IEEEeqnarray}
where
\begin{IEEEeqnarray}{l}
P_{X_i|M,S^n,Y^{i-1}} (x_i|m,\vecs,y^{i-1}) = \begin{cases} 1 & \textnormal{if } x_i = f_i (m,\vecs,y^{i-1}), \\ 0 & \textnormal{otherwise}. \end{cases} \label{eq:convCondPMFXDetEncMapp}
\end{IEEEeqnarray}
Fix any PMF $\tilde P_S$ on $\setS$ and any collection of $n$ conditional PMFs $\bigl\{ \tilde P_{Y_i|M,Y^{i-1},S^n_{i+1},X_i,S_i} \bigr\}_{i \in [1:n]}$ that satisfy
\begin{IEEEeqnarray}{l}
\tilde P_{Y_i|M,Y^{i-1},S^n_{i+1},X_i,S_i} (\cdot|m,y^{i-1},s^n_{i+1},x_i,s_i) \ll \channel {\cdot}{x_i,s_i}, \nonumber \\
\quad \forall \, ( m, y^{i-1}, s^n_{i+1}, x_i, s_i ) \in \setM \times \setY^{i-1} \times \setS^{n - i} \times \setX \times \setS. \label{eq:convAbsContTrProb}
\end{IEEEeqnarray}
These PMFs induce the PMF on $\setM \times \setS^n \times \setX^n \times \setY^n$
\begin{IEEEeqnarray}{l}
\tilde P_{M, S^n, X^n, Y^n} = P_M \times \tilde P_S^n \times \prod^n_{i = 1} \bigl( P_{X_i|M,S^n,Y^{i-1}} \times \tilde P_{Y_i|M,Y^{i-1},S^n_{i + 1},X_i,S_i} \bigr). \label{eq:convAbsContPMF}
\end{IEEEeqnarray}
It follows from \eqref{eq:assQAssignsPosProbs} and \eqref{eq:convAbsContTrProb} that $\tilde P_{M, S^n, X^n, Y^n} \ll \dist$ and consequently that \eqref{eq:convZeroDecErr} implies
\begin{equation}
\tilde P_{M, S^n, X^n, Y^n} [Y^n \in \setD_M] = 1. \label{eq:convZeroDecErrImp}
\end{equation}

We upper-bound $\frac{1}{n} \log |\setM|$ by carrying out the following calculation as in \cite[Section~7.6]{gamalkim11} but under $\tilde P_{M, S^n, X^n, Y^n}$ of \eqref{eq:convAbsContPMF}:
\begin{IEEEeqnarray}{l}
\frac{1}{n} \log |\setM| \nonumber \\
\quad \stackrel{(a)}= \frac{1}{n} H (M) \\
\quad \stackrel{(b)}= \frac{1}{n} \Bigl[ I ( M ; Y^n ) - I ( M ; S^n ) \Bigr] \\
\quad \stackrel{(c)}= \frac{1}{n} \sum^n_{i = 1} \Bigl[ I ( M ; Y_i | Y^{i-1} ) - I ( M ; S_i | S^n_{i + 1} ) \Bigr] \\
\quad \stackrel{(d)}= \frac{1}{n} \sum^n_{i = 1} \Bigl[ I ( M, Y^{i-1}, S^n_{i + 1} ; Y_i ) - I ( Y^{i-1} ; Y_i ) - I ( S^n_{i + 1} ; Y_i | M, Y^{i-1} ) \nonumber \\
\qquad - I ( M, Y^{i-1}, S^n_{i + 1} ; S_i ) + I ( S^n_{i + 1} ; S_i ) + I ( Y^{i-1} ; S_i | M, S^n_{i + 1} ) \Bigr] \\
\quad \stackrel{(e)}\leq \frac{1}{n} \sum^n_{i = 1} \Bigl[ I ( M, Y^{i-1}, S^n_{i + 1} ; Y_i ) - I ( M, Y^{i-1}, S^n_{i + 1} ; S_i ) \Bigr], \label{eq:convFixedDists}
\end{IEEEeqnarray}
where $(a)$ holds because under $\tilde P_{M, S^n, X^n, Y^n}$ $M$ is uniform over $\setM$; $(b)$ holds by \eqref{eq:convZeroDecErrImp} and because under $\tilde P_{M,S^n,X^n,Y^n}$ $M$ is independent of $S^n$; $(c)$ and $(d)$ follow from the chain rule; and $(e)$ follows from Csisz\'{a}r's sum-identity, the nonnegativity of mutual information, and the independence of $S_i$ and $S^n_{i + 1}$ under $\tilde P_{M, S^n, X^n, Y^n}$.

For every $i \in [1:n]$ define the chance variable
\begin{IEEEeqnarray}{l}
U_i = ( M, Y^{i-1}, S^n_{i + 1} ). \label{eq:convDefRVUi}
\end{IEEEeqnarray}
From \eqref{eq:convFixedDists} it then follows that every choice of $\tilde P_S$ and $\bigl\{ \tilde P_{Y_i|M,Y^{i-1},S^n_{i+1},X_i,S_i} \bigr\}_{i \in [1:n]}$ satisfying \eqref{eq:convAbsContTrProb} gives rise to an upper bound
\begin{IEEEeqnarray}{l}
\frac{1}{n} \log |\setM| \leq \frac{1}{n} \sum^n_{i = 1} \Bigl[ I ( U_i ; Y_i ) - I ( U_i ; S_i ) \Bigr], \label{eq:convFixedDists2}
\end{IEEEeqnarray}
where the mutual informations in the $i$-th summand are computed w.r.t.\ the joint PMF $\tilde P_{U_i,X_i,S_i,Y_i}$ induced by $\tilde P_{M, S^n, X^n, Y^n}$.

We will conclude the proof by exhibiting a PMF $\tilde P_S$ and a collection of conditional PMFs $\bigl\{ \tilde P_{Y_i|M,Y^{i-1},S^n_{i+1},X_i,S_i} \bigr\}_{i \in [1:n]}$ satisfying \eqref{eq:convAbsContTrProb} for which each summand on the RHS of \eqref{eq:convFixedDists2} is upper-bounded by the RHS of \eqref{eq:capacity}.

We begin with the choice of $\bigl\{ \tilde P_{Y_i|M,Y^{i-1},S^n_{i+1},X_i,S_i} \bigr\}_{i \in [1:n]}$. To this end note from \eqref{eq:convDefRVUi} the one-to-one correspondence between $\tilde P_{Y_i|M,Y^{i-1},S^n_{i+1},X_i,S_i}$ and $\tilde P_{Y_i|U_i,X_i,S_i}$:
\begin{IEEEeqnarray}{l}
\tilde P_{Y_i|M,Y^{i-1},S^n_{i + 1},X_i,S_i} ( y_i | m, y^{i-1}, s^n_{i + 1}, x_i, s_i ) = \tilde P_{Y_i|U_i,X_i,S_i} \bigl( y_i \bigl| (m, y^{i-1}, s^n_{i + 1}), x_i, s_i \bigr), \nonumber \\
\quad \forall \, ( m, y^{i-1}, s^n_{i+1}, x_i, s_i ) \in \setM \times \setY^{i-1} \times \setS^{n - i} \times \setX \times \setS. \label{eq:convDistYCondUXS}
\end{IEEEeqnarray}
This implies that choosing a conditional PMF $\tilde P_{Y_i|M,Y^{i-1},S^n_{i+1},X_i,S_i}$ that satisfies \eqref{eq:convAbsContTrProb} is tantamount to choosing a conditional PMF $\tilde P_{Y_i|U_i,X_i,S_i}$ that satisfies
\begin{IEEEeqnarray}{l}
\tilde P_{Y_i|U_i = u_i,X_i,S_i} \in \mathscr P (W), \,\, \forall \, u_i \in \setU_i, \label{eq:convDistYCondUXSAbsContTrProb}
\end{IEEEeqnarray}
and consequently choosing a collection of conditional PMFs $\bigl\{ \tilde P_{Y_i|M,Y^{i-1},S^n_{i+1},X_i,S_i} \bigr\}_{i \in [1:n]}$ that satisfy \eqref{eq:convAbsContTrProb} is tantamount to choosing a collection of conditional PMFs $\bigl\{ \tilde P_{Y_i|U_i,X_i,S_i} \bigr\}_{i \in [1:n]}$ that satisfy \eqref{eq:convDistYCondUXSAbsContTrProb}. We shall choose the latter collection, and we shall do so as follows.

We first choose $\tilde P_{Y_i|U_i,X_i,S_i}$ for $i = 1$, and we then repeatedly increment $i$ by one until it reaches $n$. Key to our choice is the observation, which will be justified shortly, that $\tilde P_{U_i,X_i,S_i}$ is determined by $\tilde P_S$ and $\bigl\{ \tilde P_{Y_j|U_j,X_j,S_j} \bigr\}_{j \in [1:i-1]}$. Our choice of $\tilde P_{Y_i|U_i,X_i,S_i}$ can thus depend not only on our choice of $\tilde P_S$ and our previous choices of $\bigl\{ \tilde P_{Y_j|U_j,X_j,S_j} \bigr\}_{j \in [1:i-1]}$ but also on $\tilde P_{U_i,X_i,S_i}$. This will allow us to choose $\tilde P_{Y_i|U_i,X_i,S_i}$ as one that---among all conditional PMFs satisfying \eqref{eq:convDistYCondUXSAbsContTrProb}---minimizes
\begin{equation}
I ( U_i ; Y_i ) - I ( U_i ; S_i ),
\end{equation}
where the mutual informations are computed w.r.t.\ the joint PMF $\tilde P_{U_i,X_i,S_i} \times \tilde P_{Y_i|U_i,X_i,S_i}$. Since \eqref{eq:convAbsContPMF} implies that
\begin{IEEEeqnarray}{l}
\tilde P_{S_i} = \tilde P_{S}, \quad i \in [1:n], \label{eq:convDistSi}
\end{IEEEeqnarray}
we will then find that, for our choice of $\bigl\{ \tilde P_{Y_i|U_i,X_i,S_i} \bigr\}_{i \in [1:n]}$,
\begin{IEEEeqnarray}{l}
I ( U_i ; Y_i ) - I ( U_i ; S_i ) \nonumber \\
\quad \leq \max_{\tilde P_{U_i,X_i|S_i}} \min_{\substack{\tilde P_{Y_i|U_i,X_i,S_i} \colon \\ \tilde P_{Y_i|U_i = u_i,X_i,S_i} \in \mathscr P (W), \,\, \forall \, u_i \in \setU_i }} I ( U_i ; Y_i ) - I ( U_i ; S_i ), \quad i \in [1:n], \label{eq:convChoiceCondPMFsGuarantee}
\end{IEEEeqnarray}
where the mutual informations are computed w.r.t.\ the joint PMF $\tilde P_{S_i} \times \tilde P_{U_i,X_i|S_i} \times \tilde P_{Y_i|U_i,X_i,S_i}$. The chosen conditional PMFs $\bigl\{ \tilde P_{Y_i|U_i,X_i,S_i} \bigr\}_{i \in [1:n]}$ satisfy \eqref{eq:convDistYCondUXSAbsContTrProb}, and hence \eqref{eq:convFixedDists2} and \eqref{eq:convChoiceCondPMFsGuarantee} will imply that
\begin{IEEEeqnarray}{l}
\frac{1}{n} \log |\setM| \nonumber \\
\quad \leq \frac{1}{n} \sum^n_{i = 1} \max_{\tilde P_{U_i,X_i|S_i}} \min_{\substack{\tilde P_{Y_i|U_i,X_i,S_i} \colon \\ \tilde P_{Y_i|U_i = u_i,X_i,S_i} \in \mathscr P (W), \,\, \forall \, u_i \in \setU_i }} I ( U_i ; Y_i ) - I ( U_i ; S_i ), \label{eq:convFixedDists3}
\end{IEEEeqnarray}
where the mutual informations in the $i$-th summand are computed w.r.t.\ $\tilde P_{S_i} \times \tilde P_{U_i,X_i|S_i} \times \tilde P_{Y_i|U_i,X_i,S_i}$.

We now prove that indeed $\tilde P_{U_i,X_i,S_i}$ is determined by $\tilde P_S$ and $\bigl\{ \tilde P_{Y_j|U_j,X_j,S_j} \bigr\}_{j \in [1:i-1]}$. In fact, we will show that the latter two determine $\tilde P_{M,S^n,X^i,Y^{i-1}}$. The latter determines $\tilde P_{U_i,X_i,S_i}$, because, by \eqref{eq:convDefRVUi}, the tuple $(U_i,X_i,S_i)$ is determined by $(M,S^n,X^i,Y^{i-1})$ and consequently its PMF $\tilde P_{U_i,X_i,S_i}$ is determined by $\tilde P_{M,S^n,X^i,Y^{i-1}}$.

We use mathematical induction, but first we note that the PMF $\tilde P_{M, S^n, X^n, Y^n}$ is constructed inductively: by \eqref{eq:convAbsContPMF}
\begin{equation}
\tilde P_{M, S^n, X_1} = P_M \times \tilde P_S^n \times P_{X_1|M,S^n} \label{eq:convDistMSnX1}
\end{equation}
and, for every $\ell \in [2:n]$, $\tilde P_{M, S^n, X^\ell, Y^{\ell-1}}$ is constructed from $\tilde P_{M, S^n, X^{\ell-1}, Y^{\ell-2}}$ by
\begin{IEEEeqnarray}{l}
\tilde P_{M, S^n, X^\ell, Y^{\ell-1}} = \tilde P_{M, S^n, X^{\ell - 1}, Y^{\ell - 2}} \times \tilde P_{Y_{\ell-1}|M,Y^{\ell-2},S^n_\ell,X_{\ell-1},S_{\ell-1}} \times P_{X_\ell|M,S^n,Y^{\ell-1}}. \label{eq:convAbsContPMFXYFromiToiPl1}
\end{IEEEeqnarray}
In describing the proof we shall make the dependence on $P_M$, our choice of $\tilde P_S$, and $\bigl\{ P_{X_j|M,S^n,Y^{j-1}} \bigr\}_{j \in [1:n]}$, whose components are determined by the encoding mappings \eqref{eq:convCapEnc} via \eqref{eq:convCondPMFXDetEncMapp}, implicit.
\begin{enumerate}
\item Basis $\ell = 1$: It follows from \eqref{eq:convDistMSnX1} that $\tilde P_{M,S^n,X_1}$ is determined.
\item Inductive Step: Fix $\ell \in [2:i]$, and suppose that $\tilde P_{M,S^n,X^{\ell-1},Y^{\ell-2}}$ is determined by $\bigl\{ \tilde P_{Y_j|U_j,X_j,S_j} \bigr\}_{j \in [1:\ell - 2]}$. Since $\tilde P_{Y_{\ell-1}|M,Y^{\ell-2},S^n_\ell,X_{\ell-1},S_{\ell-1}}$ is by \eqref{eq:convDefRVUi} in a one-to-one correspondence with $\tilde P_{Y_{\ell-1}|U_{\ell-1},X_{\ell-1},S_{\ell-1}}$, this implies that $\tilde P_{M,S^n,X^{\ell-1},Y^{\ell-2}}$ and $\tilde P_{Y_{\ell-1}|M,Y^{\ell-2},S^n_\ell,X_{\ell-1},S_{\ell-1}}$ are determined by $\bigl\{ \tilde P_{Y_j|U_j,X_j,S_j} \bigr\}_{j \in [1:\ell - 1]}$. Consequently, it follows from \eqref{eq:convAbsContPMFXYFromiToiPl1} that $\tilde P_{M,S^n,X^\ell,Y^{\ell-1}}$ is determined by $\bigl\{ \tilde P_{Y_j|U_j,X_j,S_j} \bigr\}_{j \in [1:\ell - 1]}$.
\end{enumerate}
This proves that, for every $i \in [1:n]$, $\tilde P_{M,S^n,X^i,Y^{i-1}}$ and consequently also $\tilde P_{U_i,X_i,S_i}$ are determined by $\tilde P_S$ and $\bigl\{ \tilde P_{Y_j|U_j,X_j,S_j} \bigr\}_{j \in [1:i-1]}$, and hence \eqref{eq:convFixedDists3} holds.

Having established \eqref{eq:convFixedDists3}, we are now ready to conclude the proof. By the definition of $U_i$ \eqref{eq:convDefRVUi} the cardinality of the support $\setU_i$ of $U_i$ satisfies
\begin{equation}
|\setU_i| \leq |\setM| \max \bigl\{ |\setY|, |\setS| \bigr\}^n, \quad i \in [1:n]. \label{eq:convSuppRVUi}
\end{equation}
Consequently, \eqref{eq:convDistSi} and \eqref{eq:convFixedDists3} imply that
\begin{IEEEeqnarray}{l}
\frac{1}{n} \log |\setM| \leq \max_{\tilde P_{U,X|S}} \min_{\substack{\tilde P_{Y|U,X,S} \colon \\ \tilde P_{Y|U = u,X,S} \in \mathscr P (W), \,\, \forall \, u \in \setU}} I ( U ; Y ) - I ( U ; S ), \label{eq:convFixedDists4}
\end{IEEEeqnarray}
where $U$ is an auxiliary chance variable taking values in a finite set $\setU$, and the mutual informations are computed w.r.t.\ the joint PMF $\tilde P_S \times \tilde P_{U,X|S} \times \tilde P_{Y|U,X,S}$. Since we can choose any PMF $\tilde P_S$ on $\setS$, we can choose one that---among all PMFs on $\setS$---yields the tightest bound, i.e., minimizes
\begin{IEEEeqnarray}{l}
\max_{\tilde P_{U,X|S}} \min_{\substack{\tilde P_{Y|U,X,S} \colon \\ \tilde P_{Y|U = u,X,S} \in \mathscr P (W), \,\, \forall \, u \in \setU}} I ( U ; Y ) - I ( U ; S ),
\end{IEEEeqnarray}
where $U$ is an auxiliary chance variable taking values in a finite set $\setU$, and the mutual informations are computed w.r.t.\ the joint PMF $\tilde P_S \times \tilde P_{U,X|S} \times \tilde P_{Y|U,X,S}$. For this choice of $\tilde P_S$ \eqref{eq:convFixedDists4} implies that
\begin{IEEEeqnarray}{l}
\frac{1}{n} \log |\setM| \leq \min_{\tilde P_S} \max_{\tilde P_{U,X|S}} \min_{\substack{\tilde P_{Y|U,X,S} \colon \\ \tilde P_{Y|U = u,X,S} \in \mathscr P (W), \,\, \forall \, u \in \setU }} I ( U ; Y ) - I ( U ; S ), \label{eq:convRateUB}
\end{IEEEeqnarray}
where $U$ is an auxiliary chance variable taking values in a finite set $\setU$, and the mutual informations are computed w.r.t.\ the joint PMF $\tilde P_S \times \tilde P_{U,X|S} \times \tilde P_{Y|U,X,S}$.
\end{proof}

\subsection{A Proof of Theorem~\ref{th:zeroWithoutFBPosWithFB}}\label{sec:pfThZeroWithoutFBPosWithFB}

We use the following lemma to establish Theorem~\ref{th:zeroWithoutFBPosWithFB}:

\begin{lemma}[No Feedback]\label{le:suffZeroWithoutFB}
In the absence of feedback, a sufficient condition for the zero-error capacity of the SD-DMC $\channel y {x,s}$ with acausal SI to be zero is
\begin{IEEEeqnarray}{l}
\!\!\!\!\!\!\!\! \exists \,  s \in \setS \quad \forall \, x \in \setX \quad \exists \, s^\prime \in \setS \quad \forall \, x^\prime \in \setX \quad \exists \, y \in \setY \textnormal{ s.t.\ } \channel y {x,s} \, \channel y {x^\prime,s^\prime} > 0. \label{eq:suffCondCapWithoutFBZero}
\end{IEEEeqnarray}
A sufficient condition for the capacity in the absence of feedback to be positive is that for some $\kappa \in \bigl[ 2 : |\setY| \bigr]$ and $\lambda \in \bigl[ 2 : \kappa \, |\setX| \bigr]$ there exist channel inputs $$x (s,k), \quad (s,k) \in \setS \times [ 1 : \kappa ]$$ and $\lambda$ pairwise-disjoint nonempty subsets $\setY_1, \ldots, \setY_\ell \subset \setY$ such that the following two conditions hold:
\begin{subequations}\label{bl:suffCondCapWithoutFBPositive}
\begin{IEEEeqnarray}{l}
\forall \, (s,k) \in \setS \times [ 1 : \kappa ] \quad \exists \, \ell \in [ 1 : \lambda ] \textnormal{ s.t.\ } \bigchannel {\setY_\ell}{x (s,k), s} = 1 \label{eq:suffCondCapWithoutFBPositive1}
\end{IEEEeqnarray}
and
\begin{IEEEeqnarray}{l}
\forall \, \ell \in [ 1 : \lambda ] \quad \exists \, k^\prime \in [ 1 : \kappa ] \textnormal{ s.t.\ } \Bigl( \bigchannel {\setY_\ell}{x (s^\prime,k^\prime), s^\prime} = 0, \,\, \forall \, s^\prime \in \setS  \Bigr). \label{eq:suffCondCapWithoutFBPositive2}
\end{IEEEeqnarray}
\end{subequations}
\end{lemma}

\begin{proof}
We first prove that if \eqref{eq:suffCondCapWithoutFBZero} holds, then without feedback it is impossible to transmit a single bit error-free. Let the bit take values in the set $\setM = \{ 0,1 \}$, and fix a blocklength~$n$, an encoding mapping $f \colon \setM \times \setS^n \rightarrow \setX^n$, and two disjoint decoding sets $\setD_m \subseteq \setY^n, \,\, m \in \setM$. By \eqref{eq:suffCondCapWithoutFBZero} there exists some state $s^\star \in \setS$ for which
\begin{IEEEeqnarray}{l}
\forall \, x \in \setX \quad \exists \, s^\prime \in \setS \quad \forall \, x^\prime \in \setX \quad \exists \, y \in \setY \textnormal{ s.t.\ } \channel y {x,s^\star} \, \channel y {x^\prime,s^\prime} > 0. \label{eq:suffCondCapWithoutFBZero1}
\end{IEEEeqnarray}
Let $\vecs^\star \in \setS^n$ be the all-$s^\star$ state-sequence, so $s_i^\star = s^\star, \,\, i \in [1:n]$, and let $\vecx = f (0,\vecs^\star)$. Choosing $x$ in \eqref{eq:suffCondCapWithoutFBZero1} to be the $i$-th component $x_i$ of $f (0,\vecs^\star)$, it follows from \eqref{eq:suffCondCapWithoutFBZero1} that for every $i \in [1:n]$ there exists some $s^\prime \in \setS$, say $s^\prime (i)$, for which
\begin{IEEEeqnarray}{l}
\forall \, x^\prime \in \setX \quad \exists \, y \in \setY \textnormal{ s.t.\ } \channel y {x_i, s_i^\star} \, W \bigl( y \bigl| x^\prime, s^\prime (i) \bigr) > 0. \label{eq:suffCondCapWithoutFBZero2}
\end{IEEEeqnarray}
Let $\vecs^\prime \in \setS^n$ be the state sequence whose $i$-th component $s^\prime_i$ is $s^\prime (i), \,\, i \in [1:n]$, and let $\vecx^\prime = f (1,\vecs^\prime)$. By \eqref{eq:suffCondCapWithoutFBZero2}
\begin{IEEEeqnarray}{l}
\exists \, \vecy \in \setY^n \textnormal{ s.t.\ } \prod^n_{i = 1} \Bigl( \channel {y_i} {x_i, s_i^\star} \, \channel {y_i} {x^\prime_i,s^\prime_i} \Bigr) > 0.
\end{IEEEeqnarray}
This makes it impossible for the decoder to determine with certainty whether the transmitted bit is $0$ or $1$ even if it is told that the state sequence is $\vecs^\star$ or $\vecs^\prime$. This concludes the proof of the first part of the lemma.\\

It remains to prove that if for some $\kappa \in \bigl[ 2 : |\setY| \bigr]$ and $\lambda \in \bigl[ 2 : \kappa \, |\setX| \bigr]$ there exist channel inputs $\bigl\{ x (s,k) \bigr\}_{(s,k) \in \setS \times [ 1 : \kappa ]}$ and pairwise-disjoint output-sets $\{ \setY_\ell \}_{\ell \in [ 1 : \lambda ]}$ for which \eqref{bl:suffCondCapWithoutFBPositive} holds, then the no-feedback zero-error capacity of the SD-DMC $\channel y {x,s}$ with acausal SI is positive. The proof is similar to that of Remark~\ref{re:chUsesPerBit}. To make up for the missing feedback, we shall choose the inputs so that the encoder---while incognizant of $Y_i$---will know which of the subsets $\{ \setY_\ell \}_{\ell \in [ 1 : \lambda ]}$ contains $Y_i$. The decoder will, of course, know that too.

If there is only one state $s^\star$, i.e., $\setS = \{ s^\star \}$, then upon defining $x \triangleq x (s^\star,1)$ we obtain from \eqref{eq:suffCondCapWithoutFBPositive1} the existence of some $\ell \in [1:\lambda]$ for which
\begin{subequations}\label{bl:leSuffZeroWithoutFBOneState}
\begin{equation}
\channel {\setY_\ell}{x,s^\star} = 1.
\end{equation}
It then follows from \eqref{eq:suffCondCapWithoutFBPositive2} that there exists some $k^\prime \in [1:\kappa]$ with corresponding $x^\prime = x (s^\star,k^\prime)$ for which
\begin{equation}
\channel {\setY_\ell}{x^\prime,s^\star} = 0.
\end{equation}
\end{subequations}
From \eqref{bl:leSuffZeroWithoutFBOneState} we obtain that
\begin{equation}
\channel y {x,s^\star} \, \channel y {x^\prime,s^\star}  = 0, \,\, \forall \, y \in \setY,
\end{equation}
and by sending $x$ or $x^\prime$ we can transmit a bit error-free. We hence consider now $|\setS| \geq 2$.

To transmit a single bit $m \in \{ 0,1 \}$, we use two phases of $n_1$ and $n_2$ channel uses, where
\begin{equation}
n_{\textnormal{bit}} = n_1 + n_2. \label{eq:nbitWithoutFB}
\end{equation}
The goal of Phase~1 is to produce a random subset $\bm {\mathcal S}_{n_1} \subseteq \setS^{n_2}$ with the following three properties: 1) both encoder and decoder know $\bm {\mathcal S}_{n_1}$ before Phase~2 begins; 2) with probability one $\bm {\mathcal S}_{n_1}$ contains the Phase-2 state-sequence $S^{n_1 + n_2}_{n_1 + 1}$; and 3) the cardinality of $\bm {\mathcal S}_{n_1}$ is upper-bounded by
\begin{IEEEeqnarray}{l}
| \bm {\mathcal S}_{n_1} | \leq \biggl( \frac{\kappa - 1}{\kappa} \biggr)^{n_1} |\setS|^{n_2} + \kappa.
\end{IEEEeqnarray}
To that end we partition the set $\bm {\mathcal S}_0 = \setS^{n_2}$ into $\kappa$ different subsets whose size is between $\bigl \lfloor |\bm {\mathcal S}_0| / \kappa \bigr \rfloor$ and $\bigl \lceil |\bm {\mathcal S}_0| / \kappa \bigr \rceil$. We index the $\kappa$ subsets by the set $[ 1 : \kappa ]$ and reveal the result to the encoder and decoder. If, thanks to its acausal SI, the encoder knows that the Time-1 state $S_1$ is $s$ and that $S^{n_1 + n_2}_{n_1 + 1}$ is in the subset of $\bm {\mathcal S}_0$ indexed by $k$, then at Time~1 it transmits $x (s,k)$. By \eqref{eq:suffCondCapWithoutFBPositive1} there exists some $\ell^\star \in [1:\lambda]$ with corresponding subset $\setY_{\ell^\star}$ such that, with probability one, $Y_1$ is in $\setY_{\ell^\star}$. And since the subsets $\{ \setY_\ell \}_{\ell \in [1:\lambda]}$ are pairwise disjoint, the probability of $Y_1$ being in another subset is zero. The decoder can thus compute $\ell^\star$ from $Y_1$ by checking which subset contains $Y_1$. The encoder knows $\ell^\star$, because it knows the pair $(s,k)$. Based on $\setY_{\ell^\star}$ the encoder and decoder can determine all $k^\prime \in [ 1 : \kappa ]$ for which
\begin{IEEEeqnarray}{l}
\bigchannel {\setY_{\ell^\star}}{x (s^\prime,k^\prime), s^\prime} = 0, \,\, \forall \, s^\prime \in \setS. \label{eq:condSubsetOutWithoutFB}
\end{IEEEeqnarray}
(By \eqref{eq:suffCondCapWithoutFBPositive2} at least one such $k^\prime$ exists.) Because $Y_1 \in \setY_{\ell^\star}$ and by \eqref{eq:condSubsetOutWithoutFB}, the Phase-2 state-sequence cannot be contained in a subset of $\bm {\mathcal S}_0$ indexed by such a $k^\prime$, and hence it is in the $\bm {\mathcal S}_0$-complement of these subsets, which we denote $\bm {\mathcal S}_1$. Note that: 1) both encoder and decoder know $\bm {\mathcal S}_1$ after Channel-Use 1; 2) $\bm {\mathcal S}_1$ contains $S^{n_1 + n_2}_{n_1 + 1}$; and 3) the cardinality of $\bm {\mathcal S}_1$ is upper-bounded by
\begin{IEEEeqnarray}{l}
| \bm {\mathcal S}_1 | \leq | \bm {\mathcal S}_0 | - \biggl\lfloor \frac{| \bm {\mathcal S}_0 |}{\kappa} \biggr\rfloor \leq \frac{\kappa - 1}{\kappa} |\bm {\mathcal S}_0| + 1.
\end{IEEEeqnarray}

Phase~1 continues in the same fashion, and hence we obtain that, for every $i \in [ 1 : n_1 ]$, the first~$i$ channel uses produce a random subset $\bm {\mathcal S}_i$ of $\setS^{n_2}$ satisfying that: 1) both encoder and decoder know $\bm {\mathcal S}_i$ after Channel-Use i; 2) $\bm {\mathcal S}_i$ contains $S^{n_1 + n_2}_{n_1 + 1}$; and 3) the cardinality of $\bm {\mathcal S}_i$ is upper-bounded by
\begin{IEEEeqnarray}{l}
| \bm {\mathcal S}_i | \leq | \bm {\mathcal S}_{i-1} | - \biggl\lfloor \frac{| \bm {\mathcal S}_{i-1} |}{\kappa} \biggr\rfloor \leq \frac{\kappa - 1}{\kappa} |\bm {\mathcal S}_{i-1}| + 1.
\end{IEEEeqnarray}
As in the proof of Remark~\ref{re:chUsesPerBit}, this implies that Phase~1 produces a random subset $\bm {\mathcal S}_{n_1}$ of $\setS^{n_2}$ with the desired three properties.

Phase~2 in the proof of Remark~\ref{re:chUsesPerBit} does not use the feedback link, and hence we can use it also in the current setting without feedback. Consequently, we can argue essentially as in the proof of Remark~\ref{re:chUsesPerBit} but with \eqref{bl:explChoiceN1N2} replaced by
\begin{subequations}\label{bl:explChoiceN1N2WithoutFB}
\begin{IEEEeqnarray}{rCl}
n_1 & = & \biggl\lceil \frac{ 2 \, \kappa \log |\setS| - \log \kappa}{\log \kappa - \log ( \kappa - 1 )} \biggr\rceil, \\
n_2 & = & 2 \, \kappa, \\
n_{\textnormal{bit}} & = & \biggl\lceil \frac{ 2 \, \kappa \log |\setS| - \log \kappa}{\log \kappa - \log ( \kappa - 1 )} \biggr\rceil + 2 \, \kappa
\end{IEEEeqnarray}
\end{subequations}
that $n_{\textnormal{bit}}$ channel uses suffice for the error-free transmission of a single bit. This concludes the proof, because $\kappa$ is at most $|\setY|$ and hence it follows from \eqref{bl:explChoiceN1N2WithoutFB} that $n_{\textnormal{bit}}$ satisfies the upper bound
\begin{equation}
n_{\textnormal{bit}} \leq \biggl\lceil \frac{ 2 \, |\setY| \log |\setS| - \log |\setY|}{\log |\setY| - \log \bigl( |\setY| - 1 \bigr)} \biggr\rceil + 2 \, |\setY|.
\end{equation}
\end{proof}

Theorem~\ref{th:zeroWithoutFBPosWithFB} follows from Theorem~\ref{th:positive}, Lemma~\ref{le:suffZeroWithoutFB}, and the following example:

\begin{example}\label{ex:zeroWithoutPosWith}
Suppose $\setX = \{ 0, 1 \}$ and $\setS = \setY = \{ 1, 2, 3, 4, 5 \}$. For every $x \in \setX$ and $s \in \setS$ define $\setY_{x,s}$ according to Table~\ref{tb:exZeroWithoutPosWith}, and let $\channel y {x,s}$ be such that
\begin{IEEEeqnarray}{l}
\bigl\{ y \in \setY \colon \channel {y}{x,s} > 0 \bigr\} = \setY_{x,s}, \,\, \forall \, (x, s) \in \setX \times \setS.
\end{IEEEeqnarray}
Then, the SD-DMC $\channel y {x,s}$ satisfies both \eqref{eq:positive} and \eqref{eq:suffCondCapWithoutFBZero}.
\end{example}

\begin{table}[h]
\begin{center}
  \renewcommand{\arraystretch}{1.2}
  \begin{tabular}{ c c | c c c c c }
    \multicolumn{2}{c|}{\multirow{2}{*}{$\setY_{x,s}$}} & \multicolumn{5}{c}{$s$} \\
  	&& 1 & 2 & 3 & 4 & 5 \\
    \hline    
    \multirow{2}{*}{$x$} & 0 & \{ 2,3 \} & \{ 1,5 \} & \{ 1,2 \} & \{ 2,3 \} & \{ 1,2 \} \\
    & 1 & \{ 4,5 \} & \{ 3,4 \} & \{ 4,5 \} & \{ 1,5 \} &\{ 3,4 \}
  \end{tabular}
  \caption{Nonzero transitions of the SD-DMC in Example~\ref{ex:zeroWithoutPosWith}.}
  \label{tb:exZeroWithoutPosWith}
\end{center}
\end{table}

\begin{remark}\label{re:lemmaSuffZeroWithoutFB}
Lemma~\ref{le:suffZeroWithoutFB} does not fully characterize the SD-DMCs whose capacity is positive in the absence of feedback. For example the SD-DMC of Example~\ref{ex:zeroWithoutPosWith} but with state alphabet $\setS = \{ 1, 2, 4 \}$ satisfies neither the conditions of the lemma. However, when $\channel y {x,s}$ is $\{ 0,1 \}$-valued (cf.\ Example~\ref{ex:YFunXS}), Lemma~\ref{le:suffZeroWithoutFB} implies that the capacity is positive iff
\begin{IEEEeqnarray}{l}
\bigl| \bigl\{ y \in \setY \colon \, \exists \, x \in \setX \textnormal{ s.t.\ } \channel y {x,s} > 0 \bigr\} \bigr| \geq 2, \,\, \forall s \in \setS.
\end{IEEEeqnarray}
(To see this, choose the sets $\{ \setY_\ell \}$ in Lemma~\ref{le:suffZeroWithoutFB} to be the singletons containing the outputs $y \in \setY$ for which $\channel y {x,s} > 0$ holds for some $(x,s) \in \setX \times \setS$.)
\end{remark}

\section{Summary} \label{sec:conclusion}

We now know the zero-error feedback capacity of the state-dependent
channel in all three cases: when the state is revealed to the encoder
strictly-causally, causally, or acausally. In each case the capacity
result comprises two parts: a characterization of the channels for
which the capacity is positive, and a formula for the capacity when it
is.
\begin{itemize}
\item Revealing the state to the encoder stictly-causally does not
  increase capacity (Remark~\ref{re:strCaus}), and the problem reduces
  to the state-less channel, which was solved by Shannon
  \cite{shannon56}, with Ahlswede \cite{ahlswede73} later providing an
  alternative form and an alternative blocks-based coding scheme.
\item When the state is revealed to the encoder causally, the SI is utilized
  optimally by using Shannon strategies, and the zero-error feedback
  capacity is thus that of the state-less channel into which the
  state-dependent channel is transformed when the encoder uses Shannon
  strategies (Theorems~\ref{th:causPositive} and
  \ref{th:causCapacity}). 
\item For the case where the state is revealed to the encoder
  acausally, our positivity characterization
  (Theorem~\ref{th:positive}) is reminiscent of Shannon's, and our
  formula (Theorem~\ref{th:capacity}) is reminiscent of Ahlswede's.
\end{itemize}

The acausal case exhibits phenomena that are not observed in the
strictly-causal and causal cases: The zero-error feedback capacity can
be positive even if in the absence of feedback the zero-error capacity
is zero (Theorem~\ref{th:zeroWithoutFBPosWithFB}), and the error-free
transmission of a single bit may require more than one channel use
(Corollary~\ref{co:moreThanOneChUse}).

Our coding scheme for the acausal case builds on Ahlswede's
blocks-based scheme \cite{ahlswede73} and to a lesser degree on Shannon's
sequential approach \cite{shannon56}. In contrast to Shannon's
sequential scheme, in Ahlswede's scheme the encoder codes over blocks,
and it can therefore take advantage of the acausal SI in a more
natural way. Ahlswede's scheme also seems to be more natural in the
state-less case in the presence of input constraints: his expression
remains valid provided we replace the maximization over the input
distribution with a constrained maximization
(Corollary~\ref{co:ccInputsStateless}). This is not the case for
Shannon's expression (Remark~\ref{re:ccShannonSmallerAhlswede}).

For the acausal case we also established the zero-error feedback
capacity for a scenario where---in addition to the message---also the
state sequence must be recovered
(Theorem~\ref{th:stateAmplification}); for a scenario with an
average-cost constraint on the channel inputs (Theorem~\ref{th:ccInputs}); and
for a scenario with an average-cost constraint on prespecified
$l$-blocks of consecutive channel states
(Theorem~\ref{th:ccStates2Capacity}).

A recurring theme in our coding schemes is that, as of the beginning
of the transmission, the encoder attempts to convey not only the
message but also the state sequence governing the last block, a state
sequence of which it is cognizant because the entire state sequence is
revealed to it acausally. Once the ambiguity about the
last-block's state sequence and the message has been sufficiently
reduced, the last block is used to resolve it, or rather to decode the
message.
%

Another recurring theme in our coding schemes is that---to reduce the
decoder's ambiguity about the message and the last-block's state
sequence---each block uses pairwise disjoint bins that ``completely
cover'' the set of possible state sequences in the sense that all the state
sequences pertaining to the block can be accommodated.

A recurring theme in the converse parts is to select the ``worst
possible'' joint distribution of the message, state sequence, input
sequence, and output sequence. By ``possible'' we mean here that the
distribution is compatible with the encoding mappings and absolutely
continuous w.r.t.\ the distribution that is induced by the uniform
message distribution, the state distribution, the encoding mappings,
and the channel law. By ``worst'' we mean that the distribution
yields---among all ``possible'' distributions---the tightest bound.

A remaining open problem is to characterize the family of channels whose 
zero-error capacity with acausal SI is zero \emph{in the absence of 
feedback.}
%
We provided a sufficient condition (Lemma~\ref{le:suffZeroWithoutFB}),
which we then used to show that some members of this family have positive
zero-error capacity in the presence of feedback
(Theorem~\ref{th:zeroWithoutFBPosWithFB}). We also showed that some
channels outside this family have zero zero-error capacity when the
state is revealed causally (Theorem~\ref{re:zeroCausPosAcaus}). On such 
channels with acausal SI the error-free transmission of a single bit 
requires more than one channel use also in the absence of feedback (Corollary~\ref{co:moreThanOneChUseWithoutFB}). (Recall that in the 
causal case the zero-error capacity---both in the presence and in the 
absence of feedback---is positive iff it is possible to transmit a 
single bit error-free in one channel use.) One way 
to characterize the family might be to upper-bound the maximal number of channel 
uses that could be necessary to transmit a single bit error-free. 
\begin{appendix}
%

\section{A Proof of Remark~\ref{re:detEncWlg}}\label{sec:pfDetEncWlg}

\begin{definition}\label{def:zeroErrorFBCodeStochEnc}
For any finite set $\setM$ and positive integer $n \in \naturals$, an $(n,\setM)$ zero-error feedback code with acausal SI and a stochastic encoder is defined like its deterministic counterpart (Definition~\ref{def:zeroErrorFBCode}) except that the encoding may depend on some chance variable $\Theta$ that is drawn from some finite set $\setT$ according to some PMF $P_{\Theta}$.\footnote{The assumption that $\Theta$ takes values in a finite set is not restrictive, because the channel-input, -state, and -output alphabets are finite (see Remark~\ref{re:stochEncEquivDef} at the end of this section).} The code thus consists of a finite set $\setT$, a PMF $P_\Theta$ on $\setT$, $n$ encoding mappings
\begin{IEEEeqnarray}{l}
f_i \colon \setM \times \setT \times \setS^n \times \setY^{i-1} \rightarrow \setX, \quad i \in [1:n], \label{eq:stochEncoder}
\end{IEEEeqnarray}
and $|\setM|$ disjoint decoding sets
\begin{equation}
\setD_m \subseteq \setY^n, \quad m \in \setM \label{eq:stochEncDecoder}
\end{equation}
such that for every $m \in \setM$ the probability of a decoding error is zero, i.e.,
\begin{IEEEeqnarray}{l}
\distof {Y^n \notin \setD_m | M = m, S^n = \vecs } = 0, \,\, \forall \, m \in \setM, \, \vecs \in \setS^n, \label{eq:stEncPrErrZero}
\end{IEEEeqnarray}
where
\begin{IEEEeqnarray}{l}
\distof {Y^n \notin \setD_m | M = m, S^n = \vecs } \nonumber \\
\quad = \sum_{\theta \in \setT} P_{\Theta} (\theta) \sum_{\vecy \in \setY^n \setminus \setD_m}  \prod^n_{i = 1} W \bigl( y_i \bigl| f_i (m, \theta, \vecs, y^{i-1}), s_i \bigr). 
\end{IEEEeqnarray}
\end{definition}

\begin{proof}[Proof of Remark~\ref{re:detEncWlg}]
Given an $(n, \setM)$ zero-error feedback code with a stochastic encoder \eqref{eq:stochEncoder} and decoding sets \eqref{eq:stochEncDecoder}, we can construct an $(n, \setM)$ zero-error feedback code with a deterministic encoder \eqref{eq:encAcaus} as follows. We fix some element $\theta^\star$ of $\setT$ for which $P_\Theta (\theta^\star) > 0$ and consider the $n$ deterministic encoding mappings
\begin{IEEEeqnarray}{rrCl}
g_i \colon & \setM \times \setS^n \times \setY^{i-1} & \rightarrow & \setX \\
& (m,\vecs,y^{i-1}) & \mapsto & f_i (m,\theta^\star,\vecs,y^{i-1}), \quad i \in [1:n].
\end{IEEEeqnarray}
It then follows from \eqref{eq:stEncPrErrZero} that for every $m \in \setM$ and $\vecs \in \setS^n$
\begin{IEEEeqnarray}{l}
\sum_{\vecy \in \setY^n \setminus \setD_m} \prod^n_{i = 1} W \bigl( y_i \bigl| g_i (m, \vecs, y^{i-1}), s_i \bigr) = 0,
\end{IEEEeqnarray}
so the encoding mappings $\{ g_i \}_{i \in [1:n]}$ and the decoding sets \eqref{eq:stochEncDecoder} constitute an $(n, \setM)$ zero-error feedback code with acausal SI and a deterministic encoder \eqref{eq:encAcaus}.
\end{proof}

To conclude this section, we show that allowing for any (not necessarily discrete) random variable $\Theta$ in Definition~\ref{def:zeroErrorFBCodeStochEnc} does not lead to a more general notion of an $(n, \setM)$ zero-error feedback code with acausal SI and a stochastic encoder. To this end we shall use the following lemma, which is proved, e.g., in \cite{willemsmeulen85}:

\begin{lemma}[Functional Representation Lemma] \label{le:funcRep}
Given two chance variables $X$ and $Y$ of finite support, there exist a chance variable $S$ of finite support $\setS$ that is independent of $X$ and a function $g \colon \setX \times \setS \rightarrow \setY$ such that $Y = g ( X, S )$.
\end{lemma}

\begin{remark}\label{re:stochEncEquivDef}
An $(n, \setM)$ zero-error feedback code with acausal SI and a stochastic encoder can also be viewed as a collection of $n$ conditional PMFs
\begin{equation}
P_{X_i|M,S^n,X^{i-1},Y^{i-1}}, \quad i \in [1:n]  \label{eq:stochEncCondPMFs}
\end{equation}
and $|\setM|$ disjoint decoding sets \eqref{eq:stochEncDecoder} for which \eqref{eq:stEncPrErrZero} holds, where
\begin{IEEEeqnarray}{l}
\distof {Y^n \notin \setD_m | M = m, S^n = \vecs } \nonumber \\
\quad = \sum_{\vecy \in \setY^n \setminus \setD_m} \sum_{\vecx \in \setX^n} \prod^n_{i = 1} P (x_i|m,\vecs,x^{i-1},y^{i-1}) \, W ( y_i | x_i, s_i ). \label{eq:stEncPrErrorCondPMFs}
\end{IEEEeqnarray}
Indeed, for every (not necessarily discrete) random variable $\Theta$ of support $\setT$, encoding mappings \eqref{eq:stochEncoder}, and decoding sets \eqref{eq:stochEncDecoder}, there exist $n$ conditional PMFs \eqref{eq:stochEncCondPMFs} for which
\begin{IEEEeqnarray}{l}
 \sum_{\vecy \in \setY^n \setminus \setD_m} \sum_{\vecx \in \setX^n} \prod^n_{i = 1} P (x_i|m,\vecs,x^{i-1},y^{i-1}) \, W ( y_i | x_i, s_i ) \nonumber \\
 \quad = \Exop_\Theta \!\! \left[ \sum_{\vecy \in \setY^n \setminus \setD_m} \prod^n_{i = 1} W \bigl( y_i \bigl| f_i (m, \Theta, \vecs, y^{i-1}), s_i \bigr) \right] \!\!, \,\, \forall \, m \in \setM, \, \vecs \in \setS^n. \label{eq:condPMFsThetaEquiv}
\end{IEEEeqnarray}
Conversely, for every collection of conditional PMFs \eqref{eq:stochEncCondPMFs} and decoding sets \eqref{eq:stochEncDecoder}, there exist a random variable $\Theta$ of support $\setT$ and encoding mappings \eqref{eq:stochEncoder} for which \eqref{eq:condPMFsThetaEquiv} holds. Since the channel-input, -state, and -output alphabets are finite, repeated application of the Functional Representation lemma, Lemma~\ref{le:funcRep}, yields, moreover, that we can choose the support $\setT$ of $\Theta$ finite as in Definition~\ref{def:zeroErrorFBCodeStochEnc}.
\end{remark}

\section{A Proof of Remarks~\ref{re:posSuffButNotNec1} and \ref{re:posSuffButNotNec}}\label{sec:pfPosSuffButNotNec}

\begin{proof}
We begin with Remark~\ref{re:posSuffButNotNec}. We first show that Condition~\eqref{eq:condCapacityPos} implies that the RHS of \eqref{eq:capacity} is positive. To this end assume that \eqref{eq:condCapacityPos} holds, pick $\setU = \setY$, and let $U$ be independent of $\setS$ and uniform over $\setU$, so
\begin{equation}
P_{U|S} = P_U = \unif (\setU). \label{eq:condCapacityPositiveChoiceU}
\end{equation}
Fix some conditional PMF $P_{X|U,S}$ that satisfies
\begin{IEEEeqnarray}{l}
\biggl( \Bigl( \channel u {x,s} > 0 \Bigr) \implies \Bigl( P_{X|U,S} (x|u,s) = 0 \Bigr) \biggr), \,\, \forall \, (u,x,s) \in \setU \times \setX \times \setS. \label{eq:condCapacityPositiveChoiceX}
\end{IEEEeqnarray}
(Such a $P_{X|U,S}$ exists, because \eqref{eq:condCapacityPos} says that for every pair $(u,s) \in \setU \times \setS$ there exists some $\tilde x = \tilde x (u,s) \in \setX$ for which $\channel u {\tilde x,s}$ is zero, and we can thus choose $P_{X|U,S}$ to assign $\tilde x (u,s)$ probability one.) For $P_{U,X|S} = P_U \times P_{X|U,S}$, for every PMF $P_S$ on $\setS$, and for every conditional PMF $P_{Y|U,X,S}$ satisfying
\begin{equation}
P_{Y|U = u,X,S} \in \mathscr P (W), \,\, \forall \, u \in \setU, \label{eq:condCapacityPositiveChoiceTransLaw}
\end{equation}
we obtain w.r.t.\ the joint PMF $P_S \times P_{U,X|S} \times P_{Y|U,X,S}$
\begin{IEEEeqnarray}{rCl}
I (U;Y) - I (U;S) & \stackrel{(a)}= & I (U;Y) \\
& \stackrel{(b)}= & \log |\setY| - H (U|Y) \\
& \stackrel{(c)}\geq & \log |\setY| - \log ( |\setY| - 1) \\
& > & 0,
\end{IEEEeqnarray}
where $(a)$ holds because $U$ is independent of $S$; $(b)$ holds because $U$ is uniform over its support $\setY$; and $(c)$ holds because \eqref{eq:condCapacityPositiveChoiceX} and \eqref{eq:condCapacityPositiveChoiceTransLaw} imply that $(P_S \times P_{X|U,S} \times P_{Y|U,X,S})$-almost-surely $U \neq Y$, and because the uniform distribution maximizes entropy. From this we conclude that Condition~\eqref{eq:condCapacityPos} is sufficient for the RHS of \eqref{eq:capacity} to be positive.\\


We next turn to proving that if the RHS of \eqref{eq:capacity} is positive, then \eqref{eq:condCapacityPos} holds. We prove the contrapositive: we show that if for some $( s^\star, y^\star ) \in \setS \times \setY$
\begin{equation}
\channel {y^\star} {x,s^\star} > 0, \,\, \forall \, x \in \setX, \label{eq:necSuffCondCapFormPos}
\end{equation}
then the RHS of \eqref{eq:capacity} must be zero. Suppose $s^\star$ and $y^\star$ are as above, introduce the PMF on $\setS$
\begin{equation}
P_S (s) = \begin{cases} 1 & \textnormal{if } s = s^\star, \\ 0 & \textnormal{otherwise}, \end{cases} \label{eq:condCapPosNecChoiceS}
\end{equation}
and choose $P_{Y|U,X,S} = P_{Y|X,S}$, where
\begin{IEEEeqnarray}{l}
P_{Y|X,S} (y | x, s) = \begin{cases} 1 & \textnormal{if } s = s^\star, \, y = y^\star, \\ 0 & \textnormal{if } s = s^\star, \, y \neq y^\star, \\ \channel y {x,s} & \textnormal{otherwise}. \end{cases} \label{eq:condCapPosNecChoiceTransLaw}
\end{IEEEeqnarray}
Note that the conditional PMF $P_{Y|U,X,S}$ satisfies $P_{Y|U = u,X,S} \in \mathscr P (W), \,\, \forall \, u \in \setU$, because \eqref{eq:necSuffCondCapFormPos} and \eqref{eq:condCapPosNecChoiceTransLaw} imply that $P_{Y|X,S} \in \mathscr P (W)$. For every conditional PMF $P_{U,X|S}$, \eqref{eq:condCapPosNecChoiceS} and \eqref{eq:condCapPosNecChoiceTransLaw} imply that $(P_{S} \times P_{U,X|S} \times P_{Y|U,X,S})$-almost-surely $Y = y^\star$, and hence we obtain w.r.t.\ the joint PMF $P_{S} \times P_{U,X|S} \times P_{Y|U,X,S}$
\begin{equation}
I (U;Y) - I (U;S) \leq 0.
\end{equation}
Since this holds for every conditional PMF $P_{U,X|S}$, we conclude that
\begin{IEEEeqnarray}{l}
\min_{P_S} \max_{P_{U,X|S}} \min_{\substack{P_{Y|U,X,S} \\ P_{Y|U = u,X,S} \in \mathscr P (W), \,\, \forall \, u \in \setU}} I (U;Y) - I (U;S) = 0,
\end{IEEEeqnarray}
where the mutual informations are computed w.r.t.\ the joint PMF $P_S \times P_{U,X|S} \times P_{Y|U,X,S}$.\\

Having established Remark~\ref{re:posSuffButNotNec}, we next prove Remark~\ref{re:posSuffButNotNec1} by providing an example for which Theorem~\ref{th:positive} implies that $C_{\textnormal{f},0} = 0$, and yet \eqref{eq:condCapacityPos} holds. Such an example is the SD-DMC $\channel y {x,s}$ for which $\setX = \setY = \{ 0,1,2 \}$ and
\begin{equation}
\channel y {x,s} = \begin{cases} \frac{1}{2} &\textnormal{if } y \neq x \oplus_3 2, \\ 0 &\textnormal{otherwise}. \end{cases}
\end{equation}
\end{proof}

\section{Analysis of Example~\ref{ex:YFunXS} where $\channel y {x,s}$ is $\{ 0,1 \}$-valued}\label{sec:anExYFunXS}

In this appendix we assume that $\channel y {x,s}$ is $\{ 0,1 \}$-valued, and we derive \eqref{eq:capYFunXS} from Theorems~\ref{th:positive} and \ref{th:capacity}.

We first show that Theorem~\ref{th:positive} implies that $C_{\textnormal {f},0}$ is positive iff the RHS of \eqref{eq:capYFunXS} is positive. The latter is positive iff
\begin{IEEEeqnarray}{l} 
\bigl| \bigl\{ y \in \setY \colon \exists \, x \in \setX \textnormal{ s.t.\ } \channel y {x,s} > 0 \bigr\} \bigr| \geq 2, \,\, \forall \, s \in \setS, \label{eq:positivityYFunXS}
\end{IEEEeqnarray}
i.e., iff for every state there exists a pair of inputs that the deterministic channel maps to different outputs. By Theorem~\ref{th:positive} $C_{\textnormal {f},0}$ is positive iff \eqref{eq:positive} holds, and we thus have to show that
\begin{IEEEeqnarray}{l}
\eqref{eq:positive} \iff \eqref{eq:positivityYFunXS}. \label{eq:exYFunXSPositiveEquiv}
\end{IEEEeqnarray}
The assumption that $\channel y {x,s}$ is $\{ 0,1 \}$-valued implies that for every pair of states $s, \, s^\prime \in \setS$ (not necessarily distinct) and every pair of inputs $x, \, x^\prime \in \setX$
\begin{IEEEeqnarray}{l}
\channel y {x,s} \, \channel y {x^\prime,s^\prime} = \begin{cases} 1 &\textnormal{if } \channel y {x,s} = \channel y {x^\prime,s^\prime} = 1, \\ 0 &\textnormal{otherwise}. \end{cases} \label{eq:exYFunXSWYXS}
\end{IEEEeqnarray}
Using this we prove \eqref{eq:exYFunXSPositiveEquiv}, beginning with
\begin{IEEEeqnarray}{l}
\eqref{eq:positive} \implies \eqref{eq:positivityYFunXS}. \label{eq:exYFunXSPositiveImplies}
\end{IEEEeqnarray}
If we let $s^\prime = s$, then \eqref{eq:positive} and \eqref{eq:exYFunXSWYXS} imply that for every state $s \in \setS$ there exists a pair of inputs $x, \, x^\prime \in \setX$ that the channel maps to different outputs $y, \, y^\prime \in \setY$, so
\begin{IEEEeqnarray}{l}
y \neq y^\prime \quad \textnormal{and} \quad \channel y {x,s} = \channel {y^\prime} {x^\prime,s} = 1,
\end{IEEEeqnarray}
and hence
\begin{IEEEeqnarray}{l} 
\bigl| \bigl\{ y \in \setY \colon \exists \, x \in \setX \textnormal{ s.t.\ } \channel y {x,s} > 0 \bigr\} \bigr| \geq 2.
\end{IEEEeqnarray}
This proves \eqref{eq:exYFunXSPositiveImplies}. It remains to show that
\begin{IEEEeqnarray}{l}
\eqref{eq:positive} \impliedby \eqref{eq:positivityYFunXS}. \label{eq:exYFunXSPositiveImplied}
\end{IEEEeqnarray}
From \eqref{eq:positivityYFunXS} and \eqref{eq:exYFunXSWYXS} it follows that for every state $s \in \setS$ there exists a pair of inputs $x, \, x^\prime \in \setX$ that the deterministic channel maps to different outputs $y, \, y^\prime \in \setY$, so
\begin{IEEEeqnarray}{l}
y \neq y^\prime \quad \textnormal{and} \quad \channel y {x,s} = \channel {y^\prime} {x^\prime,s} = 1.
\end{IEEEeqnarray}
This implies that for every pair of states $s, \, s^\prime \in \setS$ (not necessarily distinct) there exists a pair of inputs $x, \, x^\prime \in \setX$ that the deterministic channel maps to different outputs $y, \, y^\prime \in \setY$, so
\begin{IEEEeqnarray}{l}
y \neq y^\prime \quad \textnormal{and} \quad \channel y {x,s} = \channel {y^\prime} {x^\prime,s^\prime} = 1,
\end{IEEEeqnarray}
and hence we conclude that \eqref{eq:exYFunXSPositiveImplied} holds.\\

It remains to show that when $C_{\textnormal {f}, 0}$ is positive, then the RHS of \eqref{eq:capacity} coincides with the RHS of \eqref{eq:capYFunXS}. We first show that
\begin{IEEEeqnarray}{rCl}
C_{\textnormal{f},0} & = & \min_{P_S} \max_{P_{U,X|S}} I (U;Y) - I (U;S), \label{eq:yFunXSCapTransLawW}
\end{IEEEeqnarray}
where the mutual informations are computed w.r.t.\ the joint PMF $P_S \times P_{U,X|S} \times W$. Note that for every $u \in \setU$ the condition that $P_{Y|U = u,X,S} \in \mathscr P (W)$ is satisfied iff for every pair $(x,s) \in \setX \times \setS$ the outputs that have probability zero w.r.t.\ $\channel \cdot {x,s}$ have probability zero w.r.t.\ $P_{Y|U,X,S} ( \cdot | u,x,s )$. By the assumption that $\channel y {x,s}$ is $\{ 0,1 \}$-valued, this holds iff
\begin{IEEEeqnarray}{l}
P_{Y|U=u,X,S} = W, \,\, \forall \, u \in \setU, \label{eq:yFunXSConstTransLawEquiv}
\end{IEEEeqnarray}
and therefore \eqref{eq:yFunXSCapTransLawW} follows from Theorem~\ref{th:capacity}.

With \eqref{eq:yFunXSCapTransLawW} at hand, we are now ready to show that the RHS of \eqref{eq:capacity} is upper-bounded by the RHS of \eqref{eq:capYFunXS}: w.r.t.\ the joint PMF $P_S \times P_{U,X|S} \times W$
\begin{IEEEeqnarray}{rCl}
C_{\textnormal{f},0} &\stackrel{(a)}= & \min_{P_S} \max_{P_{U,X|S}} I (U;Y) - I (U;S) \\
&\stackrel{(b)}\leq & \min_{P_S} \max_{P_{U,X|S}} I (U;Y,S) - I (U;S) \\
&\stackrel{(c)}= & \min_{P_S} \max_{P_{U,X|S}} I (U;Y|S) \\
&\stackrel{(d)}\leq & \min_{P_S} \max_{P_{U,X|S}} H (Y|S) \\
&\stackrel{(e)}\leq & \min_{s \in \setS} \log \bigl| \bigl\{ y \in \setY \colon \exists \, x \in \setX \textnormal{ s.t.\ } \channel y {x,s} > 0 \bigr\} \bigr|,
\end{IEEEeqnarray}
where $(a)$ holds by \eqref{eq:yFunXSCapTransLawW}; $(b)$ holds because conditioning cannot increase entropy; $(c)$ follows from the chain rule; $(d)$ holds because conditional entropy is nonnegative; and $(e)$ holds because the uniform distribution maximizes entropy, and because we can choose $P_S$ to assign probability one to some $s \in \setS$ that minimizes $$\log \bigl| \bigl\{ y \in \setY \colon \exists \, x \in \setX \textnormal{ s.t.\ } \channel y {x,s} > 0 \bigr\} \bigr|.$$

Having shown that the RHS of \eqref{eq:capacity} is upper-bounded by the RHS of \eqref{eq:capYFunXS}, we now conclude by showing that the reverse also holds, i.e., that the RHS of \eqref{eq:capacity} is lower-bounded by the RHS of \eqref{eq:capYFunXS}. Take $\setU = \setY$, and for every $s \in \setS$ choose $P_{U|S} (\cdot | s)$ to be the uniform distribution on the set $$\bigl\{ y \in \setY \colon \exists \, x \in \setX \textnormal{ s.t.\ } \channel {y}{x,s} > 0 \bigr\}.$$ By the assumption that $\channel y {x,s}$ is $\{0,1\}$-valued, this choice of $P_{U|S}$ guarantees that for every pair $(u,s) \in \setU \times \setS$ for which $P_{U|S} (u|s) > 0$ there exists some $x = x (u,s) \in \setX$ for which $\channel u {x,s} = 1$. Now choose $P_{X|U,S}$ to assign $x (u,s)$ probability one. For $P_{U,X|S} = P_{U|S} \times P_{X|U,S}$ and for every PMF $P_S$ on $\setS$, we obtain $(P_S \times P_{U|S} \times P_{X|U,S} \times W)$-almost-surely $U = Y$ and w.r.t.\ the joint PMF $P_S \times P_{U|S} \times P_{X|U,S} \times W$
\begin{IEEEeqnarray}{rCl}
I (U;Y) - I (U;S) & \stackrel{(a)}= & H (U|S) \\
& \stackrel{(b)}= & \sum_{s \in \setS} P_S (s) \log \bigl| \bigl\{ y \in \setY \colon \exists \, x \in \setX \textnormal{ s.t.\ } \channel y {x,s} > 0 \bigr\} \bigr| \\
& \stackrel{(c)}\geq & \min_{s \in \setS} \log \bigl| \bigl\{ y \in \setY \colon \exists \, x \in \setX \textnormal{ s.t.\ } \channel y {x,s} > 0 \bigr\} \bigr|, \label{eq:yFunXsLBCap}
\end{IEEEeqnarray}
where $(a)$ holds because $(P_S \times P_{U|S} \times P_{X|U,S} \times W)$-almost-surely $U = Y$; $(b)$ holds because $P_{U|S} (\cdot | s)$ is for every $s \in \setS$ the uniform distribution on $$\bigl\{ y \in \setY \colon \exists \, x \in \setX \textnormal{ s.t.\ } \channel {y}{x,s} > 0 \bigr\};$$ and $(c)$ holds because the minimum of $$\log \bigl| \bigl\{ y \in \setY \colon \exists \, x \in \setX \textnormal{ s.t.\ } \channel y {x,s} > 0 \bigr\} \bigr|$$ over $s \in \setS$ cannot be larger than its weighted average over $s \in \setS$ with weights $P_S (s), \,\, s \in \setS$. From \eqref{eq:yFunXSCapTransLawW} and \eqref{eq:yFunXsLBCap} we conclude that the RHS of \eqref{eq:capacity} is lower-bounded by the RHS of \eqref{eq:capYFunXS}.

\section{A Cardinality Bound on $\setU$}\label{sec:leCardU}

\begin{lemma}\label{le:cardU}
Given a channel $\channel y {x,s}$ and a PMF $P_S$ on $\setS$, consider
\begin{IEEEeqnarray}{l}
\max_{P_{U,X|S}} \min_{ \substack{ P_{Y|U,X,S} \colon \\ P_{Y|U = u,X,S} \in \mathscr P (W), \,\, \forall \, u \in \setU }} I (U;Y) - I (U;S), \label{eq:leCardUCap}
\end{IEEEeqnarray}
where the maximization is over all chance variables $U$ of finite support, and the mutual informations are computed w.r.t.\ the joint PMF $P_S \times P_{U,X|S} \times P_{Y|U,X,S}$. Restricting $X$ to be a function of $U$ and $S$, i.e., $P_{U,X|S}$ to have the form
\begin{equation}
P_{U,X|S} (u,x|s) = P_{U|S} (u|s) \, \ind {x = g (u,s)}, \label{eq:leCardUXFuncUS}
\end{equation}
does not change \eqref{eq:leCardUCap}. Nor does requiring that $U$ take values in a set $\setU$  whose cardinality $| \setU |$ satisfies
\begin{equation}
|\setU| \leq |\setX|^{|\setS|}. \label{eq:cardULe}
\end{equation}
\end{lemma}

\begin{proof}
We first show that restricting $X$ to be a function of $U$ and $S$ does not change \eqref{eq:leCardUCap}. By the Functional Representation lemma (Lemma~\ref{le:funcRep}), for every conditional PMF $P_{U,X|S}$, there exists a chance variable $V$ of finite support $\setV$ and a function $h \colon \setU \times \setV \times \setS \rightarrow \setX$ such that
\begin{equation}
P_{U,X|S} (u,x|s) = \sum_{v \in \setV} P_{U|S} (u|s) \, P_V (v) \, \ind {x = h (u,v,s)}.
\end{equation}
Consequently, \eqref{eq:leCardUCap} is equal to
\begin{IEEEeqnarray}{l}
\max_{ P_V, h (\cdot), P_{U|S} } \min_{ \substack{ P_{Y|U,X,S} \colon \\ P_{Y|U = u,X,S} \in \mathscr P (W), \,\, \forall \, u \in \setU }} I (U;Y) - I (U;S), \label{eq:leCardUCapEquivV}
\end{IEEEeqnarray}
where the maximization is over all chance variables $V$ of finite support $\setV$, functions $h \colon \setU \times \setV \times \setS \rightarrow \setX$, and conditional PMFs over a finite set $\setU$; and where the mutual informations are computed w.r.t.\ the joint PMF $P_S \times P_V \times P_{U,X|V,S} \times P_{Y|U,X,S}$, where $P_{U,X|V,S}$ is
\begin{IEEEeqnarray}{l}
P_{U,X|V,S} (u,x|v,s) = P_{U|S} (u|s) \, \ind {x = h (u,v,s)}. \label{eq:leCardUCondPMFUXGivVS}
\end{IEEEeqnarray}
Fix some PMF $P_V$ on $\setV$, some function $h \colon \setU \times \setV \times \setS \rightarrow \setX$, and some conditional PMF $P_{U|S}$, and let $(U,V,X,S) \sim P_S \times P_V \times P_{U,X|V,S}$, where $P_{U,X|V,S}$ is given in \eqref{eq:leCardUCondPMFUXGivVS}.
Let $\tilde P_{Y|U,V,X,S}$ be some conditional PMF satisfying
\begin{equation}
\tilde P_{Y|U = u, V = v,X,S} \in \mathscr P (W), \,\, \forall \, (u, v ) \in \setU \times \setV, \label{eq:leCardUCondPMFYGivVUSatReq}
\end{equation}
and note that this implies that
\begin{equation}
\tilde P_{Y|U=u,X,S} \in \mathscr P (W), \,\, \forall \, u \in \setU, \label{eq:leCardUCondPMFYGivOnlyUSatReq}
\end{equation}
where
\begin{IEEEeqnarray}{l}
\tilde P_{Y|U,X,S} (y|u,x,s) = \sum_{v \in \setV} \frac{P_{U,V,X,S} (u,v,x,s)}{\sum_{v^\prime \in \setV} P_{U,V,X,S} (u,v^\prime,x,s)} \, \tilde P_{Y|U,V,X,S} (y|u,v,x,s). \label{eq:leCardUDefPYGivUXS}
\end{IEEEeqnarray}
W.r.t.\ the joint PMF $P_{U,V,X,S} \times \tilde P_{Y|U,V,X,S}$,
\begin{IEEEeqnarray}{rCl}
I (U,V;Y) - I (U,V;S) & \stackrel{(a)}= & I (U;Y) - I (U;S) + I (V;Y|U) \\
& \stackrel{(b)}\geq & I (U;Y) - I (U;S), \label{eq:leCardUMutWithVBigMutWithout}
\end{IEEEeqnarray}
where $(a)$ follows from the chain rule and the independence of $V$ and $(U,S)$ under $P_{U,V,X,S}$ \eqref{eq:leCardUCondPMFUXGivVS}; and $(b)$ holds because mutual information is nonnegative. Since $P_{U,X,S} \times \tilde P_{Y|U,X,S}$ is obtained from $P_{U,V,X,S} \times \tilde P_{Y|U,V,X,S}$ by integrating $V$ out \eqref{eq:leCardUDefPYGivUXS},
\begin{IEEEeqnarray}{ll}
I (U;Y) - I (U;S) \quad &\textnormal{w.r.t.} \quad P_{U,V,X,S} \times \tilde P_{Y|U,V,X,S} \nonumber \\
\quad = I (U;Y) - I (U;S) \quad &\textnormal{w.r.t.} \quad P_{U,X,S} \times \tilde P_{Y|U,X,S} \\
\quad = I (U;Y) - I (U;S) \quad &\textnormal{w.r.t.} \quad P_{U,V,X,S} \times \tilde P_{Y|U,X,S}.
\end{IEEEeqnarray}
This and \eqref{eq:leCardUMutWithVBigMutWithout} imply that
\begin{IEEEeqnarray}{ll}
I (U,V;Y) - I (U,V;S) \quad &\textnormal{w.r.t.} \quad P_{U,V,X,S} \times \tilde P_{Y|U,V,X,S} \nonumber \\
\quad \geq I (U;Y) - I (U;S) \quad &\textnormal{w.r.t.} \quad P_{U,V,X,S} \times \tilde P_{Y|U,X,S}. \label{eq:leCardUMutWithVBigMutWithoutPMFs}
\end{IEEEeqnarray}
Since \eqref{eq:leCardUCondPMFYGivVUSatReq} implies \eqref{eq:leCardUCondPMFYGivOnlyUSatReq}, it follows from \eqref{eq:leCardUMutWithVBigMutWithoutPMFs} that
\begin{IEEEeqnarray}{ll}
\max_{ P_V, h (\cdot), P_{U|S} } \min_{ \substack{ \tilde P_{Y|U,V,X,S} \colon \\ \tilde P_{Y|U = u,V = v,X,S} \in \mathscr P (W), \,\, \forall \, (u,v) \in \setU \times \setV }} I (U,V;Y) - I (U,V;S) \nonumber \\
\quad \geq \max_{ P_V, h (\cdot), P_{U|S} } \min_{ \substack{ \tilde P_{Y|U,X,S} \colon \\ \tilde P_{Y|U = u,X,S} \in \mathscr P (W), \,\, \forall \, u \in \setU }} I (U;Y) - I (U;S), \label{eq:leCardUMutWithVBigMutWithoutPMFs2}
\end{IEEEeqnarray}
where the mutual informations are computed w.r.t.\ $P_S \times P_V \times P_{U,X|V,S} \times \tilde P_{Y|U,V,X,S}$ in the first line and w.r.t.\ $P_S \times P_V \times P_{U,X|V,S} \times \tilde P_{Y|U,X,S}$ in the second line, where $P_{U,X|V,S}$ is given in \eqref{eq:leCardUCondPMFUXGivVS}. The RHS of \eqref{eq:leCardUMutWithVBigMutWithoutPMFs2} is \eqref{eq:leCardUCapEquivV}, which, as we have noted, is equal to \eqref{eq:leCardUCap}. Consequently, the LHS of \eqref{eq:leCardUMutWithVBigMutWithoutPMFs2} upper-bounds \eqref{eq:leCardUCap}. But the LHS of \eqref{eq:leCardUMutWithVBigMutWithoutPMFs2} corresponds to choosing the auxiliary chance variable $\tilde U = (U,V)$, with the result that $X$ is a deterministic function of $( \tilde U, S )$.\\

It remains to show that restricting the cardinality of $\setU$ to \eqref{eq:cardULe} does not change \eqref{eq:leCardUCap} when the maximization in \eqref{eq:leCardUCap} is over all conditional PMFs $P_{U,X|S}$ of the form \eqref{eq:leCardUXFuncUS}. To this end we show that \eqref{eq:leCardUCap} does not change when we require that for every distinct $u_1, \, u_2 \in \setU$ the mappings $g (u_1, \cdot)$ and $g (u_2, \cdot)$ differ. Since there are $|\setX|^{|\setS|}$ different mappings with domain $\setS$ and co-domain $\setX$, this implies that restricting the cardinality of $\setU$ to \eqref{eq:cardULe} does not change \eqref{eq:leCardUCap}.

Fix some finite set $\setU$ and some conditional PMF $P_{U,X|S}$ of the form \eqref{eq:leCardUXFuncUS}, and let $(U,X,S) \sim P_S \times P_{U,X|S}$. Suppose that there exist distinct $u_1, \, u_2 \in \setU$ for which
\begin{equation}
g (u_1,s) = g (u_2,s), \,\, \forall \, s \in \setS. \label{eq:uPrimeUDoublePrimeSameFunc}
\end{equation}
Define the chance variable
\begin{IEEEeqnarray}{l}
T = \begin{cases} U &\textnormal{if } U \neq u_2, \\ u_1 &\textnormal{otherwise} \end{cases} \label{eq:definUPrime}
\end{IEEEeqnarray}
of support $\setT = \setU \setminus \{ u_2 \}$, and denote by $P_{U,T,X,S}$ the joint PMF of $(U,T,X,S)$. By \eqref{eq:uPrimeUDoublePrimeSameFunc}
\begin{IEEEeqnarray}{l}
P_{X|T,S} \bigl( x \bigl| t,s \bigr) = \ind {x = g (t,s)}, \label{eq:xFuncUPrimeS}
\end{IEEEeqnarray}
where
\begin{IEEEeqnarray}{l}
P_{X|T,S} (x|t,s) = \frac{\sum_{u \in \setU} P_{U,T,X,S} (u,t,x,s)}{\sum_{(u^\prime, x^\prime) \in \setU \times \setX} P_{U,T,X,S} (u^\prime, t, x^\prime, s)}.
\end{IEEEeqnarray}
We will show that replacing $U$ with $T$ does not decrease our payoff, i.e., that
\begin{IEEEeqnarray}{l}
\min_{ \substack{ \tilde P_{Y|U,X,S} \colon \\ \tilde P_{Y|U = u,X,S} \in \mathscr P (W), \,\, \forall \, u \in \setU } } I (U;Y) - I (U;S) \nonumber \\
\quad \leq \min_{ \substack{ \tilde P_{Y|T,X,S} \colon \\ \tilde P_{Y|T = t,X,S} \in \mathscr P (W), \,\, \forall \, t \in \setT } } I (T;Y) - I (T;S), \label{eq:cardU4}
\end{IEEEeqnarray}
where the mutual informations are computed w.r.t.\ $P_{U,X,S} \times \tilde P_{Y|U,X,S}$ in the first line and w.r.t.\ $P_{T,X,S} \times \tilde P_{Y|T,X,S}$ in the second line. By repeating this process we can repeatedly reduce the cardinality of the support set of the auxiliary chance variable until $u_1 \neq u_2$ implies that $g (u_1,\cdot)$ and $g (u_2,\cdot)$ differ.

Let $\tilde P_{Y|T,X,S}$ be some conditional PMF satisfying
\begin{equation}
\tilde P_{Y|T = t,X,S} \in \mathscr P (W), \,\, \forall \, t \in \setT, \label{eq:leCardUCondPMFYGivUPrimeSatReq}
\end{equation}
and define the conditional PMF
\begin{IEEEeqnarray}{l}
\tilde P_{Y|U,X,S} (y|u,x,s) = \begin{cases} \tilde P_{Y|T,X,S} (y|u,x,s) &\textnormal{if } u \neq u_2, \\ \tilde P_{Y|T,X,S} (y|u_1,x,s) &\textnormal{otherwise}, \end{cases} \label{eq:definPYGivenUXSFromUPrime}
\end{IEEEeqnarray}
so $\tilde P_{Y|U,X,S} (y|u,x,s) = \tilde P_{Y|T,X,S} (y|t,x,s)$ when $u = t$ or when $u = u_2$ and $t = u_1$. From this and \eqref{eq:definUPrime}, which implies that $P_{U,T,X,S} (u,t,x,s)$ is positive only when $u = t$ or when $u = u_2$ and $t = u_1$, it follows that
\begin{IEEEeqnarray}{l}
P_{U,T,X,S} \times \tilde P_{Y|U,X,S} = P_{U,T,X,S} \times \tilde P_{Y|T,X,S}. \label{eq:leCardUCondUAndCondUPrimeResultInSamePMF}
\end{IEEEeqnarray}

From \eqref{eq:leCardUCondPMFYGivUPrimeSatReq} and the definition of $\tilde P_{Y|U,X,S}$ \eqref{eq:definPYGivenUXSFromUPrime} we see that
\begin{equation}
\tilde P_{Y|U = u,X,S} \in \mathscr P (W), \,\, \forall \, u \in \setU. \label{eq:leCardUCondPMFYGivUSatReq}
\end{equation}
W.r.t.\ the joint PMF $P_{U,T,X,S} \times \tilde P_{Y|U,X,S}$ (which equals $P_{U,T,X,S} \times \tilde P_{Y|T,X,S}$ by \eqref{eq:leCardUCondUAndCondUPrimeResultInSamePMF})
\begin{IEEEeqnarray}{l}
I (U;Y) - I (U;S) \nonumber \\
\quad \stackrel{(a)}= I (T,U;Y) - I (T,U;S) \\
\quad \stackrel{(b)}= I (T;Y) - I (T;S) + I (U;Y|T) - I (U;S|T) \\
\quad \stackrel{(c)}= I (T;Y) - I (T;S) + H (U|T,S) - H (U|T,Y), \label{eq:cardU1}
\end{IEEEeqnarray}
where $(a)$ holds because under $P_{U,T,X,S}$ $T$ is determined by $U$ \eqref{eq:definUPrime}; $(b)$ follows from the chain rule; and $(c)$ holds by definition of mutual information. W.r.t.\ the joint PMF $P_{U,T,X,S} \times \tilde P_{Y|U,X,S} = P_{U,T,X,S} \times \tilde P_{Y|T,X,S}$ the term $H (U|T,S) - H (U|T,Y)$ is not positive, because
\begin{IEEEeqnarray}{l}
H (U|T,S) - H (U|T,Y) \nonumber \\
\quad \stackrel{(a)}= H (U|T,X,S) - H (U|T,Y) \\
\quad \stackrel{(b)}\leq I (Y;U|T,X,S) \\
\quad \stackrel{(c)}= 0, \label{eq:cardU2}
\end{IEEEeqnarray}
where $(a)$ holds because under $P_{U,T,X,S}$ $X$ is determined by $(T,S)$ \eqref{eq:xFuncUPrimeS}; $(b)$ holds because conditioning cannot increase entropy and by definition of mutual information; and $(c)$ holds because under $P_{U,T,X,S} \times \tilde P_{Y|U,X,S} = P_{U,T,X,S} \times \tilde P_{Y|T,X,S}$ $U$ and $Y$ are conditionally independent given $(T,X,S)$. From \eqref{eq:cardU1} and \eqref{eq:cardU2} we obtain
\begin{IEEEeqnarray}{ll}
I (U;Y) - I (U;S) \quad &\textnormal{w.r.t.} \quad P_{U,T,X,S} \times \tilde P_{Y|U,X,S} \nonumber \\
\quad \leq I (T;Y) - I (T;S) \quad &\textnormal{w.r.t.} \quad P_{U,T,X,S} \times \tilde P_{Y|T,X,S},
\end{IEEEeqnarray}
which is equivalent to
\begin{IEEEeqnarray}{ll}
I (U;Y) - I (U;S) \quad &\textnormal{w.r.t.} \quad P_{U,X,S} \times \tilde P_{Y|U,X,S} \nonumber \\
\quad \leq I (T;Y) - I (T;S) \quad &\textnormal{w.r.t.} \quad P_{T,X,S} \times \tilde P_{Y|T,X,S}. \label{eq:cardU3}
\end{IEEEeqnarray}
Since \eqref{eq:leCardUCondPMFYGivUPrimeSatReq} and \eqref{eq:definPYGivenUXSFromUPrime} imply \eqref{eq:leCardUCondPMFYGivUSatReq}, we obtain from \eqref{eq:cardU3} that \eqref{eq:cardU4} holds, i.e., that replacing $U$ with $T$ does not decrease our payoff.

We can repeat the above process until we are left with a chance variable $\bar U$ of finite support $\bar \setU \subseteq \setU$ that satisfies that for every distinct $\bar u_1, \, \bar u_2 \in \bar \setU$ the mappings $g (\bar u_1, \cdot)$ and $g (\bar u_2, \cdot)$ differ, and, by \eqref{eq:cardU4}, that 
\begin{IEEEeqnarray}{l}
\min_{ \substack{ \tilde P_{Y|U,X,S} \colon \\ \tilde P_{Y|U = u,X,S} \in \mathscr P (W), \,\, \forall \, u \in \setU } } I (U;Y) - I (U;S) \nonumber \\
\quad \leq \min_{ \substack{ \tilde P_{Y|\bar U,X,S} \colon \\ \tilde P_{Y|\bar U = \bar u,X,S} \in \mathscr P (W), \,\, \forall \, \bar u \in \bar \setU } } I (\bar U;Y) - I (\bar U;S). \label{eq:cardU5}
\end{IEEEeqnarray}
From \eqref{eq:cardU5} we obtain the claim that \eqref{eq:leCardUCap}---with the maximization being over all conditional PMFs $P_{U,X|S}$ of the form \eqref{eq:leCardUXFuncUS}---does not change when we require that for every distinct $u_1, \, u_2 \in \setU$ the mappings $g (u_1, \cdot)$ and $g (u_2, \cdot)$ differ.
\end{proof}

\section{A Proof of Lemma~\ref{le:covering}}\label{sec:leCovering}

\begin{proof}
Recall that $\Theta = \bigl\lceil 2^{k (H (U|S) - \epsilon)} \bigr\rceil$. If $\epsilon \geq H (U|S)$, then $\Theta = 1$. The only size-1 partition of $\setT^{ (k) }_{P_U}$ is $\setT^{ (k) }_{P_U}$ itself, and, because $\setT^{(k)}_{P_{U,S}}$ is not empty (since $P_{U,S}$ is a $k$-type), this partition satisfies \eqref{eq:leCoveringCond}, i.e.,
\begin{IEEEeqnarray}{l}
\forall \, \vecs \in \setT^{(k)}_{P_S} \quad \exists \, \vecu \in \setT^{(k)}_{P_U} \textnormal{ s.t.\ } (\vecu, \vecs) \in \setT^{(k)}_{P_{U,S}}.
\end{IEEEeqnarray}

Consider now the more interesting case where $\epsilon < H (U|S)$. We will show that if $k$ exceeds some $\eta_0 \bigl( |\setU|, |\setS|, \epsilon \bigr)$ (to be specified later), then the desired partition $\{ \setB_\ell \}_{\ell \in [ 1 : \Theta ]}$ of $\setT^{ (k) }_{P_U}$ exists. We shall do so using the probabilistic method. Fix $k \in \naturals$ and a $k$-type $P_{U,S}$ with corresponding conditional entropy $H (U|S)$. Generate a random partition $\{ \bm {\mathcal B}_\ell \}_{\ell \in [ 1 : \Theta ]}$ of $\setT^{ (k) }_{P_U}$, where $\{ \bm {\mathcal B}_\ell \}$ is short for $\{ \bm {\mathcal B}_\ell \}_{\ell \in [ 1 : \Theta ]}$, by placing each $k$-tuple $\vecu \in \setT_{P_U}^{ (k) }$ in a uniformly-drawn bin. We show that the probability that $\{ \bm {\mathcal B}_\ell \}$ violates \eqref{eq:leCoveringCond} is smaller than one whenever $k \geq \eta_0 \bigl( |\setU|, |\setS|, \epsilon \bigr)$. From this it will follow that the desired partition exists.

To upper-bound the probability that $\{ \bm {\mathcal B}_\ell \}$ violates \eqref{eq:leCoveringCond}, we first upper-bound $$\Bigdistof {\nexists \, \vecu \in  \bm {\mathcal B}_\ell \textnormal{ s.t.\ } (\vecu,\vecs) \in \setT^{(k)}_{P_{U,S}}}$$ for any fixed pair $( \vecs, \ell ) \in \setT^{ (k) }_{P_S} \times [ 1 : \Theta ]$:
\begin{IEEEeqnarray}{l}
\Bigdistof { \nexists \, \vecu \in  \bm {\mathcal B}_\ell \textnormal{ s.t.\ } (\vecu,\vecs) \in \setT^{(k)}_{P_{U,S}} } \nonumber \\
\quad = \Bigdistof { \bm {\mathcal B}_\ell \cap \setT_{P_{U|S}}^{(k) } (\vecs) = \emptyset } \\
\quad \stackrel{(a)}= \bigl( 1 - \Theta^{-1} \bigr)^{\bigl| \setT_{P_{U|S}}^{(k)} (\vecs) \bigr|} \\
\quad \stackrel{(b)}\leq \Bigl( 1 - 2^{-k (H (U|S) - \epsilon) + 1} \Bigr)^{\bigl| \setT_{P_{U|S}}^{(k)} (\vecs) \bigr|} \\
\quad \stackrel{(c)}\leq \exp \Bigl\{ - 2^{k \epsilon - \log (1 + k) |\setU| \, |\setS| + 1} \Bigr\}, \label{eq:leCoveringCondNotFixedM}
\end{IEEEeqnarray}
where $(a)$ holds because each $k$-tuple $\vecu \in \setT_{P_{U|S}}^{(k)} (\vecs)$ is placed in $\setB_\ell$ with probability $\Theta^{-1}$; $(b)$ holds because
\begin{IEEEeqnarray}{l}
\Theta = \Bigl\lceil 2^{k (H (U|S) - \epsilon)} \Bigr\rceil \leq 2^{k (H (U|S) - \epsilon) + 1},
\end{IEEEeqnarray}
where the last inequality holds by assumption that $\epsilon < H (U|S)$; and $(c)$ holds because $1 - \xi \leq e^{-\xi}, \,\, \xi \in \reals$, and because $\bigl|\setT_{P_{U|S}}^{(k)} (\vecs)\bigr| \geq (1 + k)^{- |\setU| \, |\setS|} 2^{k H (U|S)}$ \cite[Lemma~2.5]{csiszarkoerner11}. Having obtained \eqref{eq:leCoveringCondNotFixedM} for every fixed $( \vecs, \ell ) \in \setT^{ (k) }_{P_S} \times [ 1 : \Theta ]$, we use the Union-of-Events bound to upper-bound the probability that $\{ \bm {\mathcal B}_\ell \}$ violates \eqref{eq:leCoveringCond}:
\begin{IEEEeqnarray}{l}
\Bigdistof { \exists \, ( \vecs, \ell ) \in \setT^{ (k) }_{P_S} \times [ 1 : \Theta ] \textnormal{ s.t.\ } \setB_\ell \cap \setT_{P_{U|S}}^{(k)} = \emptyset } \nonumber \\
\quad \stackrel{(a)}\leq \Bigl| \setT_{P_S}^{(k)} \Bigr| \, \Theta \exp \Bigl\{ - 2^{k \epsilon - \log (1 + k) |\setU| \, |\setS| + 1} \Bigr\} \\
\quad \stackrel{(b)}\leq \exp \Bigl\{ - 2^{ k \epsilon - \log (1 + k) |\setU| \, |\setS| + 1 } + k ( \ln |\setS| + \ln |\setU| - \epsilon \ln 2 ) \Bigr\}, \label{eq:leCoveringCondFixedM}
\end{IEEEeqnarray}
where $(a)$ follows from the Union-of-Events bound and \eqref{eq:leCoveringCondNotFixedM}; and $(b)$ holds because $\bigl|\setT_{P_S}^{(k)}\bigr| \leq |\setS|^k$ and $\Theta \leq |\setU|^k$. The exponent on the RHS of \eqref{eq:leCoveringCondFixedM}, $$- 2^{ k \epsilon - \log (1 + k) |\setU| \, |\setS| } + k ( \ln |\setS| + \ln |\setU| - \epsilon \ln 2 ),$$ depends only on $k$, $|\setU|$, $|\setS|$, and $\epsilon$, and it tends to $- \infty$ as $k$ tends to infinity. Consequently, there exists some $\eta_0 \bigl( |\setU|, |\setS|, \epsilon \bigr)$ that guarantees that the exponent is negative whenever $k \geq \eta_0 \bigl( |\setU|, |\setS|, \epsilon \bigr)$. For such values of $k$ the RHS of \eqref{eq:leCoveringCondFixedM} is smaller than one, and the desired partition exists.
\end{proof}

\section{A Proof of Theorem~\ref{th:causPositive}}\label{sec:pfCausPositive}

The proof consists of a direct and a converse part. We first establish the direct part.

\begin{proof}[Direct Part]
If there exists a partition $\setY_0, \, \setY_1$ of $\setY$ satisfying \eqref{eq:causPositive}, then the encoder can transmit a bit $m \in \{ 0,1 \}$ error-free in one channel use: If $m = 0$ and the Time-1 channel-state is $s \in \setS$, then it sends some $x \in \setX$ for which $\channel {\setY_0}{x,s} = 1$, and if $m = 1$ and the Time-1 channel-state is $s \in \setS$, then it sends some $x^\prime \in \setX$ for which $\channel {\setY_1}{x^\prime,s} = 1$. This allows the decoder to recover the transmitted bit error-free by declaring ``$m = 0$'' if the Time-$1$ channel-output is in $\setY_0$ and ``$m = 1$'' if the Time-$1$ channel-output is in $\setY_1$.
\end{proof}

We next prove the converse part of Theorem~\ref{th:causPositive}.

\begin{proof}[Converse Part]
To prove that \eqref{eq:causPositive} is necessary for $C^{\textnormal{caus}}_{\textnormal{f},0}$ to be positive, we will show that if no partition $\setY_0, \, \setY_1$ of $\setY$ satisfies \eqref{eq:causPositive}, then it is impossible to transmit a bit error-free. Assume then that no such partition exists, and let the bit take values in the set $\setM = \{ 0,1 \}$. Fix a blocklength~$n$ and $n$ encoding mappings $$f_i \colon \setM \times \setS^i \times \setY^{i-1} \rightarrow \setX, \quad i \in [ 1:n ].$$ To show that the mappings do not achieve error-free transmission, we will exhibit a pair of state sequences $\vecs, \, \tilde \vecs \in \setS^n$ and an output sequence $\vecy \in \setY^n$ that for every $i \in [1:n]$ satisfy
\begin{IEEEeqnarray}{l}
W \bigl( y_i \bigl| f_i (0,s^i,y^{i-1}), s_i \bigr) \, W \bigl( y_i \bigl| f_i (1,\tilde s^i,y^{i-1}), \tilde s_i \bigr) > 0. \label{eq:causBadOutputs}
\end{IEEEeqnarray}
This will rule out error-free transmission, because if the state sequence is either $\vecs$ or $\tilde \vecs$, then the decoder, not knowing which, cannot recover the bit.

Our construction of $\vecs, \, \tilde \vecs \in \setS^n$ and $\vecy \in \setY^n$ is inductive, i.e., we first exhibit Time-$1$ components $s_1, \, \tilde s_1 \in \setS$ and $y_1 \in \setY$ that satisfy \eqref{eq:causBadOutputs} for $i = 1$, and we then repeatedly increment $i$ by one (until it reaches $n$) and exhibit Time-$i$ components $s_i, \, \tilde s_i \in \setS$ and $y_i \in \setY$ that---together with the previously constructed $\{ s_j, \tilde s_j \}_{j \in [1:i-1]}$ and $\{ y_j \}_{j \in [1:i-1]}$---satisfy \eqref{eq:causBadOutputs}.

We start by exhibiting Time-$1$ components $s_1, \, \tilde s_1 \in \setS$ and $y_1 \in \setY$ that satisfy \eqref{eq:causBadOutputs} for $i = 1$. To this end we show that
\begin{IEEEeqnarray}{l}
\exists \, s, \, \tilde s \in \setS, \, y \in \setY \textnormal{ s.t.\ } \bigchannel {y}{f_1 (0,s), s} \, \bigchannel {y}{f_1 (1,\tilde s), \tilde s} > 0. \label{eq:causPositiveDoesNotHoldCond}
\end{IEEEeqnarray}
Our proof of \eqref{eq:causPositiveDoesNotHoldCond} is by contradiction. To reach a contradiction, suppose that \eqref{eq:causPositiveDoesNotHoldCond} does not hold, so
\begin{IEEEeqnarray}{l}
\Bigl( \bigchannel {y}{f_1 (0,s), s} \, \bigchannel {y}{f_1 (1,\tilde s), \tilde s} = 0, \,\, \forall \, y \in \setY \Bigr), \,\, \forall \, s, \, \tilde s \in \setS. \label{eq:causMICont}
\end{IEEEeqnarray}
Define the set
\begin{IEEEeqnarray}{l}
\setY_0 = \Bigr\{ y \in \setY \colon \exists \, s \in \setS \textnormal{ s.t.\ } \bigchannel {y}{f_1 (0,s), s} > 0 \Bigl\}
\end{IEEEeqnarray}
and its $\setY$-complement $\setY_1 = \setY \setminus \setY_0$. By the definition of the set $\setY_0$
\begin{IEEEeqnarray}{l}
\bigchannel {\setY_0}{f_1 (0,s), s} = 1, \,\, \forall \, s \in \setS,
\end{IEEEeqnarray}
and by \eqref{eq:causMICont}
\begin{IEEEeqnarray}{l}
\bigchannel {\setY_0}{f_1 (1,\tilde s), \tilde s} = 0, \,\, \forall \, \tilde s \in \setS,
\end{IEEEeqnarray}
so
\begin{IEEEeqnarray}{l}
\bigchannel {\setY_1}{f_1 (1,s), s} = 1 - \bigchannel {\setY_0}{f_1 (1,s), s} = 1, \,\, \forall \, s \in \setS.
\end{IEEEeqnarray}
This contradicts our assumption that no partition $\setY_0, \, \setY_1$ of $\setY$ satisfies \eqref{eq:causPositive} and thus establishes \eqref{eq:causPositiveDoesNotHoldCond}. If $s$, $\tilde s$, and $y$ are as promised in \eqref{eq:causPositiveDoesNotHoldCond}, then we choose $s_1 = s$, $\tilde s_1 = \tilde s$, and $y_1 = y$ with the result that \eqref{eq:causBadOutputs} holds for $i = 1$.

For the inductive step, suppose $\ell \in [2:n]$, and that we have already constructed $\{ s_i, \tilde s_i \}_{i \in [1:\ell-1]}$ and $\{ y_i \}_{i \in [1:\ell-1]}$ for which \eqref{eq:causBadOutputs} holds for every $i \in [1:\ell - 1]$. We construct Time-$\ell$ components $s_\ell, \, \tilde s_\ell \in \setS$ and $y_\ell \in \setY$ that---together with the previously constructed $\{ s_i, \tilde s_i \}_{i \in [1:\ell-1]}$ and $\{ y_i \}_{i \in [1:\ell-1]}$---satisfy \eqref{eq:causBadOutputs} when we substitute $\ell$ for $i$ in \eqref{eq:causBadOutputs}, i.e., we show that
\begin{IEEEeqnarray}{l}
\exists \, s_\ell, \, \tilde s_\ell \in \setS, \, y_\ell \in \setY \textnormal{ s.t.\ } \bigchannel {y_\ell}{f_\ell (0,s^\ell,y^{\ell - 1}), s_\ell} \, \bigchannel {y_\ell}{f_\ell (1,\tilde s^\ell,y^{\ell - 1}), \tilde s_\ell} > 0. \label{eq:causPositiveDoesNotHoldCondEll}
\end{IEEEeqnarray}
Our proof of \eqref{eq:causPositiveDoesNotHoldCondEll} is by contradiction. To reach a contradiction, suppose that \eqref{eq:causPositiveDoesNotHoldCondEll} does not hold, so
\begin{IEEEeqnarray}{l}
\!\!\!\!\!\!\!\! \Bigl( \bigchannel {y_\ell}{f_\ell (0,s^\ell,y^{\ell - 1}), s_\ell} \, \bigchannel {y_\ell}{f_\ell (1,\tilde s^\ell,y^{\ell - 1}), \tilde s_\ell} = 0, \,\, \forall \, y_\ell \in \setY \Bigr), \,\, \forall \, s_\ell, \, \tilde s_\ell \in \setS. \label{eq:causMIContEll}
\end{IEEEeqnarray}
Define the set
\begin{IEEEeqnarray}{l}
\setY_0 = \Bigr\{ y_\ell \in \setY \colon \exists \, s_\ell \in \setS \textnormal{ s.t.\ } \bigchannel {y_\ell}{f_\ell (0,s^\ell,y^{\ell - 1}), s_\ell} > 0 \Bigl\}
\end{IEEEeqnarray}
and its $\setY$-complement $\setY_1 = \setY \setminus \setY_0$. By the definition of the set $\setY_0$
\begin{IEEEeqnarray}{l}
\bigchannel {\setY_0}{f_\ell (0,s^\ell,y^{\ell - 1}), s_\ell} = 1, \,\, \forall \, s_\ell \in \setS,
\end{IEEEeqnarray}
and by \eqref{eq:causMIContEll}
\begin{IEEEeqnarray}{l}
\bigchannel {\setY_0}{f_\ell (1,\tilde s^\ell,y^{\ell - 1}), \tilde s_\ell} = 0, \,\, \forall \, \tilde s_\ell \in \setS,
\end{IEEEeqnarray}
so
\begin{IEEEeqnarray}{l}
\bigchannel {\setY_1}{f_\ell (1,\tilde s^\ell,y^{\ell - 1}), \tilde s_\ell} = 1 - \bigchannel {\setY_0}{f_\ell (1,\tilde s^\ell,y^{\ell - 1}), \tilde s_\ell} = 1, \,\, \forall \, \tilde s_\ell \in \setS.
\end{IEEEeqnarray}
This contradicts our assumption that no partition $\setY_0, \, \setY_1$ of $\setY$ satisfies \eqref{eq:causPositive}.

Since the construction goes through for every $\ell \in [1:n]$, when $\ell$ reaches $n$ we have constructed a pair of state sequences $\vecs, \, \tilde \vecs \in \setS^n$ and an output sequence $\vecy \in \setY^n$ that for every $i \in [1:n]$ satisfy \eqref{eq:causBadOutputs}.
\end{proof}

\section{A Proof of Theorem~\ref{th:causCapacity}}\label{sec:pfCausCapacity}

Suppose $\channel y {x,s}$ satisfies the condition in Theorem~\ref{th:causPositive} for $C^{\textnormal{caus}}_{\textnormal{f},0}$ to be positive. In this case the RHS of \eqref{eq:causCapacityAhlswedeFrom} and the RHS of \eqref{eq:causCapacity} are equal, because the latter is the zero-error feedback capacity of the (state-less) DMC $W^\prime (y|u)$ \eqref{eq:capacityShannon} and thus---using Ahlswede's alternative form \eqref{eq:capacityAhlswede}---can be alternatively expressed as \eqref{eq:causCapacityAhlswedeFrom}. It thus suffices to prove \eqref{eq:causCapacity}, i.e.,
\begin{IEEEeqnarray}{l}
C^{\textnormal{caus}}_{\textnormal{f},0} = \max_{P_U} \min_y - \log \sum_{u \colon W^\prime (y|u) > 0} P_U (u). \label{eq:capacityCausShannon}
\end{IEEEeqnarray}

The proof consists of a direct and a converse part. We first establish the direct part.

\begin{proof}[Direct Part]
That the RHS of \eqref{eq:capacityCausShannon} is achievable follows from Shannon's results on the zero-error capacity \cite[Theorem~7]{shannon56} and on channels with states \cite{shannon58}. Indeed, the encoder can convert the channel to a state-less channel whose inputs are Shannon strategies \cite{shannon58}. That is, it can perform the encoding over the set $\setU$, where $\bigl\{ g (u,\cdot) \colon u \in \setU \bigr\}$ equals $\setX^\setS$, and transmit at Time~$i$ the channel input $g \bigl( u_i (m), S_i \bigr)$, where $u_i (m)$ is the $i$-th component of the codeword~$\vecu (m)$ corresponding to the message~$m$ to be transmitted (see Figure~\ref{fig:shannonStrategy} and \cite[Remark~7.6]{gamalkim11}). In doing so, the encoder transforms the SD-DMC $\channel y {x,s}$ with causal SI and feedback into the state-less DMC $$W^\prime (y|u) = \sum_{s \in \setS} Q_S (s) \, \bigchannel {y}{g (u,s),s}$$ with feedback. Because the zero-error feedback capacity of the DMC $W^\prime (y|u)$ is equal to the RHS of \eqref{eq:capacityCausShannon} (see \cite[Theorem~7]{shannon56} or \eqref{eq:capacityShannon}), the RHS of \eqref{eq:capacityCausShannon} is achievable.
\end{proof}

We next establish the converse part.

\begin{proof}[Converse Part]
To establish that $C^{\textnormal{caus}}_{\textnormal{f},0}$ cannot be larger than the RHS of \eqref{eq:capacityCausShannon}, we adapt Shannon's converse of \cite[Theorem~7]{shannon56} to the present setting. Let
\begin{IEEEeqnarray}{l}
\xi = \max_{P_U} \min_y - \log \sum_{u \colon W^\prime (y|u) > 0} P_U (u), \label{eq:defXiConvCausSI}
\end{IEEEeqnarray}
and fix a finite set $\setM$, a blocklength~$n$, and $n$ encoding mappings $$f_i \colon \setM \times \setS^i \times \setY^{i-1} \rightarrow \setX, \quad i \in [ 1:n ].$$

We will exhibit an output sequence $\vecy \in \setY^n$ for which the corresponding post-$n$ survivor-set
\begin{IEEEeqnarray}{l}
\setM_n = \Biggl\{ m \in \setM \colon \exists \, \vecs \in \setS^n \textnormal{ s.t.\ } \prod^n_{i = 1} W \bigl( y_i \bigl| f_i (m,s^i,y^{i-1}),s_i \bigr) > 0 \Biggr\} \label{eq:capacityCausBadOutputs}
\end{IEEEeqnarray}
is of size at least
\begin{equation}
|\setM_n| \geq 2^{-n \xi} \, |\setM|. \label{eq:capacityCausMessageSet}
\end{equation}
From \eqref{eq:capacityCausBadOutputs} and \eqref{eq:capacityCausMessageSet} it will then follow that the probability of a decoding error can only be zero if $|\setM| \leq 2^{n \xi}$, because otherwise $|\setM_n| \geq 2$ and none of the messages in $\setM_n$ can be ruled out by the decoder.

To conclude the proof, we show by mathematical induction over $i \in [0:n]$ that for every $i \in [0:n]$
\begin{subequations}\label{bl:capacityCausBadOutputsPosti}
\begin{IEEEeqnarray}{l}
\exists \, \vecy \in \setY^i \textnormal{ s.t.\ } \Bigl( \bigl| \setM_i (\vecy) \bigr| \geq 2^{- i \xi} \, |\setM| \Bigr), \label{eq:capacityCausBadOutputsSizePosti}
\end{IEEEeqnarray}
where $\setM_i (\vecy)$ is the post-$i$ survivor-set corresponding to $\vecy$, so
\begin{IEEEeqnarray}{l}
\setM_i (\vecy) = \Biggl\{ m \in \setM \colon \exists \, \vecs \in \setS^i \textnormal{ s.t.\ } \prod^i_{j = 1} W \bigl( y_j \bigl| f_j (m,s^j,y^{j-1}),s_j \bigr) > 0 \Biggr\}. \label{eq:capacityCausBadOutputsPosti}
\end{IEEEeqnarray}
\end{subequations}
In \eqref{eq:capacityCausBadOutputsPosti} we use the convention that the empty product is $1$, so $\setM_0 (\emptyset) = \setM$ for $i = 0$.

\begin{enumerate}
\item Basis $i = 0$: Because $\setM_0 (\emptyset) = \setM$, \eqref{bl:capacityCausBadOutputsPosti} holds for $i = 0$.
\item Inductive Step: Fix $\ell \in [1:n]$, and assume that \eqref{bl:capacityCausBadOutputsPosti} holds for $i = \ell - 1$, i.e., that there exists some $y^{\ell - 1} \in \setY^{\ell - 1}$ for which
\begin{equation}
\bigl| \setM_{\ell - 1} (y^{\ell - 1}) \bigr| \geq 2^{- (\ell - 1) \xi} \, |\setM|. \label{eq:capacityCausBadOutputsIndHyp}
\end{equation}
Suppose $y^{\ell-1}$ is as above. By the definition of the set $\setM_{\ell - 1} (y^{\ell - 1})$ \eqref{eq:capacityCausBadOutputsPosti} there exists a collection $\bigl\{ s^{\ell - 1} (m) \bigr\}_{m \in \setM_{\ell - 1} (y^{\ell - 1})}$ of $(\ell - 1)$-tuples from $\setS^{\ell - 1}$ for which
\begin{IEEEeqnarray}{l}
\prod^{\ell - 1}_{i = 1} W \Bigl( y_i \Bigl| f_i (m, s^i (m), y^{i - 1}), s_i (m) \Bigr) > 0, \,\, \forall \, m \in \setM_{\ell - 1}  (y^{\ell - 1}). \label{eq:capacityCausBadOutputsiCondEllMin1}
\end{IEEEeqnarray}
To prove that \eqref{eq:capacityCausBadOutputsSizePosti} holds for $i = \ell$, we show that
\begin{IEEEeqnarray}{rl}
\exists \, y \in \setY \textnormal{ s.t.\ } & \biggl( \Bigl| \Bigl\{ m \in \setM_{\ell - 1} (y^{\ell - 1}) \colon \exists \, s_\ell (m) \in \setS \textnormal{ s.t.\ } \nonumber \\
& \qquad W \bigl( y \bigl| f_\ell \bigl( m, s^\ell (m), y^{\ell - 1} \bigr), s_\ell (m) \bigr) > 0 \Bigr\} \Bigr| \nonumber \\
& \qquad \geq 2^{- \xi} \, \bigl| \setM_{\ell - 1} (y^{\ell - 1}) \bigr| \biggr). \label{eq:capacityCausBadOutputsiEqEllPf}
\end{IEEEeqnarray}
Setting $y_\ell$ to be the $y \in \setY$ promised in \eqref{eq:capacityCausBadOutputsiEqEllPf} will prove \eqref{bl:capacityCausBadOutputsPosti} for $i = \ell$.

Because $\bigl\{ g (u,\cdot) \colon u \in \setU \bigr\}$ equals $\setX^\setS$, for every $m \in \setM_{\ell - 1} (y^{\ell - 1})$ there exists a $u \in \setU$, call it $u_\ell (m)$, satisfying that
\begin{IEEEeqnarray}{l}
f_\ell \bigl( m, s^\ell (m), y^{\ell - 1} \bigr) = g \bigl( u, s_\ell (m) \bigr), \,\, \forall \, s_\ell (m) \in \setS.
\end{IEEEeqnarray}
This and \eqref{eq:causCapWPrimeExistsSW} imply that \eqref{eq:capacityCausBadOutputsiEqEllPf} is equivalent to
\begin{IEEEeqnarray}{rl}
\exists \, y \in \setY \textnormal{ s.t.\ } & \biggl( \Bigl| \Bigl\{ m \in \setM_{\ell - 1} (y^{\ell - 1}) \colon W^\prime \bigl( y \bigl| u_\ell (m) \bigr) > 0 \Bigr\} \Bigr| \nonumber \\
&\qquad \geq 2^{- \xi} \, \bigl| \setM_{\ell - 1} (y^{\ell - 1}) \bigr| \biggr). \label{eq:capacityCausBadOutputsiEqEllPf1}
\end{IEEEeqnarray}
It thus suffices to establish \eqref{eq:capacityCausBadOutputsiEqEllPf1}. The proof is essentially the converse of \cite[Theorem~7]{shannon56}. For every $u \in \setU$ denote by $F_u$ the fraction of all the messages $m \in \setM_{\ell - 1} (y^{\ell - 1})$ for which $u_\ell (m)$ equals $u$, so
\begin{IEEEeqnarray}{l}
F_u \triangleq \frac{\bigl| \bigl\{ m \in \setM_{\ell - 1} (y^{\ell - 1}) \colon u_\ell (m) = u \bigr\} \bigr|}{\bigl| \setM_{\ell - 1} (y^{\ell - 1}) \bigr|}, \quad u \in \setU.
\end{IEEEeqnarray}
The construction of the collection $\{ F_u \}_{u \in \setU}$ guarantees that for every $y \in \setY$
\begin{IEEEeqnarray}{l}
\Bigl| \Bigl\{ m \in \setM_{\ell - 1} (y^{\ell - 1}) \colon W^\prime \bigl( y \bigl| u_\ell (m) \bigr) > 0 \Bigr\} \Bigr| = \!\!\!\! \sum_{u \colon W^\prime (y|u) > 0} \!\!\!\! F_u \, \bigl| \setM_{\ell - 1} (y^{\ell - 1}) \bigr|. \label{eq:capacityCausConvSizeYSurvs}
\end{IEEEeqnarray}
Moreover, the collection $\{ F_u \}_{u \in \setU}$ is like a PMF on $\setU$, i.e.,
\begin{subequations}\label{eq:capacityCausFuLikeProbAss}
\begin{IEEEeqnarray}{l}
F_u \geq 0, \,\, \forall \, u \in \setU, \\
\sum_{u \in \setU} F_u = 1.
\end{IEEEeqnarray}
\end{subequations}
Choose $y$ as one that---among all elements of $\setY$---maximizes
\begin{equation}
\sum_{u \colon W^\prime (y|u) > 0} F_u. \label{eq:choiceYMaxFu}
\end{equation}
For this choice of $y$ we obtain the lower bound
\begin{IEEEeqnarray}{l}
\Bigl| \Bigl\{ m \in \setM_{\ell - 1} (y^{\ell - 1}) \colon W^\prime \bigl( y \bigl| u_\ell (m) \bigr) > 0 \Bigr\} \Bigr| \nonumber \\
\quad \stackrel{(a)}= \sum_{u \colon W^\prime (y|u) > 0} F_u \, \bigl|\setM_{\ell - 1} (y^{\ell - 1})\bigr| \\
\quad \stackrel{(b)}= \max_{y \in \setY} \sum_{u \colon W^\prime (y|u) > 0} F_u \, \bigl|\setM_{\ell - 1} (y^{\ell - 1})\bigr| \\
\quad \stackrel{(c)}\geq \min_{P_U} \max_{y \in \setY} \sum_{u \colon W^\prime (y|u) > 0} P_U (u) \, \bigl|\setM_{\ell - 1} (y^{\ell - 1})\bigr| \\
\quad \stackrel{(d)}\geq 2^{-\xi} \, \bigl|\setM_{\ell - 1} (y^{\ell - 1})\bigr|,
\end{IEEEeqnarray}
where $(a)$ holds by \eqref{eq:capacityCausConvSizeYSurvs}; $(b)$ holds because $y$ maximizes \eqref{eq:choiceYMaxFu} and consequently also \eqref{eq:capacityCausConvSizeYSurvs} among all elements of $\setY$; $(c)$ holds by \eqref{eq:capacityCausFuLikeProbAss}; and $(d)$ holds by \eqref{eq:defXiConvCausSI}. This proves \eqref{eq:capacityCausBadOutputsiEqEllPf1} and consequently also \eqref{eq:capacityCausBadOutputsiEqEllPf}. If $y$ is as promised in \eqref{eq:capacityCausBadOutputsiEqEllPf} and we choose $y_\ell$ to be $y$, then it follows from \eqref{eq:capacityCausBadOutputsIndHyp} and \eqref{eq:capacityCausBadOutputsiCondEllMin1} that for $i = \ell$ the post-$\ell$ survivor-set $\setM_\ell (y^\ell)$ of \eqref{eq:capacityCausBadOutputsPosti} is of size at least $2^{- \ell \xi} \, |\setM|$, and hence that \eqref{bl:capacityCausBadOutputsPosti} holds for $i = \ell$.
\end{enumerate}

\end{proof}

\section{A Proof of Remarks~\ref{re:causPosSuffButNotNec1} and \ref{re:causPosSuffButNotNec}}\label{sec:pfCausPosSuffButNotNec}

\begin{proof}
We begin with Remark~\ref{re:causPosSuffButNotNec}. We first show that Condition~\eqref{eq:condCausCapFormulaPos} implies that the RHS of \eqref{eq:causCapacityAhlswedeFrom} is positive. To this end assume that \eqref{eq:condCausCapFormulaPos} holds. Recall that $\bigl\{ g (u,\cdot) \colon u \in \setU \bigr\}$ equals $\setX^\setS$. This, combined with \eqref{eq:condCausCapFormulaPos}, implies that for every $y \in \setY$ there must exist a $u \in \setU$, call it $u_y$, that satisfies
\begin{IEEEeqnarray}{l}
W \bigl( y \bigl| g (u_y, s), s \bigr) = 0, \,\, \forall \, s \in \setS. \label{eq:prResCausDefUY}
\end{IEEEeqnarray}
The mapping $y \mapsto u_y$ need not be one-to-one, but it follows from \eqref{eq:prResCausDefUY} that the cardinality of its range must exceed one. Let $U$ be uniform over the set $\{ u_y \colon y \in \setY \}$, so
\begin{IEEEeqnarray}{l}
P_U = \unif \bigl( \{ u_y \colon y \in \setY \} \bigr).
\end{IEEEeqnarray}
For every $P_{Y|U} \in \mathscr P (W^\prime)$ we obtain w.r.t.\ $P_U \times P_{Y|U}$
\begin{IEEEeqnarray}{rCl}
I (U;Y) & \stackrel{(a)}= & \log | \{ u_y \colon y \in \setY \} | - H (U|Y) \\
& \stackrel{(b)}\geq & \log | \{ u_y \colon y \in \setY \} | - \log \bigl( | \{ u_y \colon y \in \setY \} | - 1 \bigr) \\
& \stackrel{(c)}\geq & \log |\setY| - \log \bigl( |\setY| - 1 \bigr),
\end{IEEEeqnarray}
where $(a)$ holds because $U$ is uniform over $\{ u_y \colon y \in \setY \}$; $(b)$ holds because $U \neq u_Y$ and because the uniform distribution maximizes entropy; and $(c)$ holds because $|\setY| \geq 2$ (which follows from \eqref{eq:condCausCapFormulaPos}), and because the function $$\xi \mapsto \frac{\xi}{\xi-1}, \quad \xi > 1$$ is strictly monotonically decreasing in $\xi$. This implies that the RHS of \eqref{eq:causCapacityAhlswedeFrom} is positive:
\begin{IEEEeqnarray}{l}
\max_{P_U} \min_{P_{Y|U} \in \mathscr P (W^\prime)} I (U;Y) \nonumber \\
\quad \geq \log |\setY| - \log ( |\setY| - 1) \\
\quad > 0,
\end{IEEEeqnarray}
where the mutual information is computed w.r.t.\ the joint PMF $P_U \times P_{Y|U}$.\\

We next turn to proving that if the RHS of \eqref{eq:causCapacityAhlswedeFrom} is positive, then \eqref{eq:condCausCapFormulaPos} holds. We prove the contrapositive: we show that if for some $( s^\star, y^\star ) \in \setS \times \setY$
\begin{equation}
\channel {y^\star} {x,s^\star} > 0, \,\, \forall \, x \in \setX, \label{eq:causNecSuffCondCapFormPos}
\end{equation}
then the RHS of \eqref{eq:causCapacityAhlswedeFrom} must be zero. Suppose $s^\star$ and $y^\star$ are as above, and introduce the conditional PMF
\begin{equation}
P_{Y|U} (y|u) = \begin{cases} 1 &\textnormal{if } y = y^\star, \\ 0 &\textnormal{otherwise}. \end{cases} \label{eq:causNecSuffCondCapFormPosPYU}
\end{equation}
Note that $P_{Y|U} \in \mathscr P (W^\prime)$, because \eqref{eq:causNecSuffCondCapFormPos} implies that
\begin{IEEEeqnarray}{l}
W \bigl( y^\star \bigl| g (u,s^\star),s^\star \bigr) > 0, \,\, \forall \, u \in \setU.
\end{IEEEeqnarray}
For every PMF $P_U$ on $\setU$ \eqref{eq:causNecSuffCondCapFormPosPYU} implies that $(P_U \times P_{Y|U})$-almost-surely $Y = y^\star$, and hence we obtain w.r.t.\ the joint PMF $P_U \times P_{Y|U}$
\begin{equation}
I (U;Y) = 0.
\end{equation}
Because this holds for every PMF $P_U$ on $\setU$, we conclude that the RHS of \eqref{eq:causCapacityAhlswedeFrom} is zero:
\begin{IEEEeqnarray}{l}
\max_{P_U} \min_{P_{Y|U} \in \mathscr P (W^\prime)} I (U;Y) = 0,
\end{IEEEeqnarray}
where the mutual information is computed w.r.t.\ the joint PMF $P_U \times P_{Y|U}$.\\

Having established Remark~\ref{re:causPosSuffButNotNec}, we next prove Remark~\ref{re:causPosSuffButNotNec1} by providing an example for which Theorem~\ref{th:causPositive} implies that $C^{\textnormal{caus}}_{\textnormal{f},0} = 0$, and yet \eqref{eq:condCausCapFormulaPos} holds. Such an example is the SD-DMC $\channel y {x,s}$ for which $\setX = \setY = \{ 0,1,2 \}$ and
\begin{equation}
\channel y {x,s} = \begin{cases} \frac{1}{2} &\textnormal{if } y \neq x \oplus_3 2, \\ 0 &\textnormal{otherwise}. \end{cases}
\end{equation}
\end{proof}

\section{A Proof of Theorem~\ref{th:stateAmplification}}\label{sec:pfThStateAmplification}

The proof consists of a direct and a converse part. We first establish the direct part.

\begin{proof}[Direct Part]
We assume that \eqref{eq:positive} holds and show that the RHS of \eqref{eq:capacitySA} is achievable. If the RHS of \eqref{eq:capacitySA} is zero, then there is nothing to prove, so we assume that it is positive. The proof builds on the proofs of Remark~\ref{re:chUsesPerBit} and the direct part of Theorem~\ref{th:capacity}, adapting both to the case where---in addition to the message---the encoder wants to convey to the receiver error-free also the state sequence. We partition the blocklength-$n$ transmission into $B + 1$ blocks, with each of the first $B$ blocks being of length~$k$, and with Block~$(B + 1)$ being of length $k^\prime$. The choice we shall later make for $k^\prime$ will be such that the last block be of negligible length compared to $B k$ and therefore not affect the code's asymptotic rate.

Before the transmission begins, the encoder is revealed the realization $\vecs \triangleq S^n$ of the state sequence, from which it can compute the realization $\vecs^{(b)} \triangleq s^{b k}_{(b-1) k + 1}$ of the Block-$b$ state-sequence for every $b \in [1:B]$ and the realization $\vecs^{(B + 1)} \triangleq s^{B k + k^\prime}_{B k + 1}$ of the Block-$(B + 1)$ state-sequence. In the first $B$ blocks our scheme draws on the scheme we used in the direct part of Theorem~\ref{th:capacity} but with the following two modifications: 1)~to guarantee that the decoder can recover the Block-$(B + 1)$ state-sequence $\vecs^{(B + 1)}$, the encoder transmits the pair $\bigl( m,\vecs^{(B + 1)} \bigr) \in \setM \times \setS^{k^\prime}$ comprising the message to be sent and the Block-$(B + 1)$ state-sequence; and 2)~to guarantee that the decoder can recover the state sequences $\bigl\{ \vecs^{(b)} \bigr\}_{b \in [1:B]}$ during the first $B$ blocks, we choose the auxiliary chance variable $U$ to comprise the channel state $S$ and consequently to be $(X,S)$ (because we can w.l.g.\ restrict $X$ to be a function of $U$ and $S$). The last block draws on Phase~2 of the scheme we used to prove Remark~\ref{re:chUsesPerBit}. We next describe the proposed coding scheme in detail, beginning with the first $B$ blocks and ending with the last block.\\

For every $b \in [1:B]$ we adapt the Block~$b$ transmission of the scheme we used in the direct part of Theorem~\ref{th:capacity} as follows. Assume for now that the decoder---while incognizant of $\vecs^{(1)}, \ldots, \vecs^{(B)}$---knows the empirical types $P_{\vecs^{(1)}}, \ldots, P_{\vecs^{(B)}}$: Block~$(B + 1)$ will ensure that the scheme works even though the decoder is incognizant of these types. Let $\bm { \mathcal I}_0 = \setM \times \setS^{k^\prime}$ be the set of all possible pairs of message $m^\prime \in \setM$ and Block-$(B + 1)$ state-sequence $\vecs^\prime \in \setS^{k^\prime}$, and for every $b \in [1:B]$ let $\bm { \mathcal I}_b \subset \setM \times \setS^{b k} \times \setS^{k^\prime}$ be the (random) set comprising all the triples of message $m^\prime \in \setM$, state sequence $\hat \vecs \in \setS^{bk}$ pertaining to the first $b$ blocks, and Block-$(B + 1)$ state-sequence $\vecs^\prime \in \setS^{k^\prime}$ that have a positive posterior probability given the channel outputs $Y^{bk}$ and the empirical types $\bigl\{ P_{\vecs^{(b^\prime)}} \bigr\}_{b^\prime \in [1:b]}$ during the first $b$ blocks. Choose some $k$-type $P_{X,S}^{(b)}$ whose $\setS$-marginal $P_S^{(b)}$ equals $P_{\vecs^{(b)}}$. In the following, unless otherwise specified, all entropies and mutual informations are computed w.r.t.\ the joint PMF $P_{X,S}^{(b)}$.

For any $k$-length state-sequence $\vecs^\prime \in \setT^{(k)}_{P^{(b)}_S}$, let $L^{(k)}_{P^{(b)}_{X|S}}$ denote the size of the $P^{(b)}_{X|S} (\vecs^\prime)$-shell $\setT^{(k)}_{P^{(b)}_{X|S}} (\vecs^\prime)$, i.e.,
\begin{equation}
L^{(k)}_{P^{(b)}_{X|S}} = \biggl|\setT^{(k)}_{P_{X|S}^{(b)}} (\vecs^\prime) \biggr|, \quad \vecs^\prime \in \setT^{(k)}_{P^{(b)}_S}. \label{eq:defSizeSetsSA}
\end{equation}
This size does not depend on $\vecs^\prime \in \setT^{(k)}_{P^{(b)}_S}$, and by \cite[Lemma~2.5]{csiszarkoerner11}
\begin{equation}
L^{(k)}_{P^{(b)}_{X|S}} \geq (1 + k)^{- |\setX| \, |\setS|} \, 2^{k H (X|S)}. \label{eq:sizeSetsSA}
\end{equation}
We partition $\bm { \mathcal I}_{b-1}$ into $L^{(k)}_{P^{(b)}_{X|S}}$ subsets whose size is between $$\biggl\lfloor |\bm { \mathcal I}_{b-1}| / L^{(k)}_{P^{(b)}_{X|S}} \biggr\rfloor \quad \textnormal{and} \quad \biggl\lceil |\bm { \mathcal I}_{b-1}| / L^{(k)}_{P^{(b)}_{X|S}} \biggr\rceil;$$ and we associate with each set a different bin from the bins $$\setB_\ell \subseteq \setT^{(k)}_{P^{(b)}_{X,S}}, \quad \ell \in \biggl[ 1 : L^{(k)}_{P^{(b)}_{X|S}} \biggr],$$ where the bins $\{ \setB_\ell \}$ are pairwise disjoint subsets of $\setT^{(k)}_{P^{(b)}_{X,S}}$
\begin{subequations}
\begin{IEEEeqnarray}{l}
\setB_\ell \cap \setB_{\ell^\prime} = \emptyset, \,\, \Biggl( \forall \, \ell, \, \ell^\prime \in \biggl[ 1 : L^{(k)}_{P^{(b)}_{X|S}} \biggr], \, \ell^\prime \neq \ell \Biggr), \label{eq:binsDisjointSA}
\end{IEEEeqnarray}
and where each bin ``covers'' $\setT^{(k)}_{P^{(b)}_S}$ exactly in the sense that
\begin{IEEEeqnarray}{l}
\forall \, (\vecs, \ell) \in \setT^{(k)}_{P^{(b)}_S} \times \biggl[ 1 : L^{(k)}_{P^{(b)}_{X|S}} \biggr] \quad \exists ! \, \vecx \in \setT^{(k)}_{P^{(b)}_X} \textnormal{ s.t.\ } (\vecx,\vecs) \in \setB_\ell. \label{eq:coveringCondSA}
\end{IEEEeqnarray}
\end{subequations}
(Unlike the direct part of Theorem~\ref{th:capacity}, here we need not invoke Lemma~\ref{le:covering} to guarantee the existence of such bins. Indeed, that such bins exist follows from the definition of $L^{(k)}_{P^{(b)}_{X|S}}$ \eqref{eq:defSizeSetsSA}: for every $\vecs \in \setT^{(k)}_{P^{(b)}_S}$ there exist $L^{(k)}_{P^{(b)}_{X|S}}$ different $\vecx$ for which $(\vecx, \vecs) \in \setT^{(k)}_{P^{(b)}_{X,S}}$ \eqref{eq:defSizeSetsSA}, and hence we can choose some collection $\{ \setB_\ell \}$ satisfying that for every $\vecs \in \setT^{(k)}_{P^{(b)}_S}$ each of the $L^{(k)}_{P^{(b)}_{X|S}}$ pairs in $\setT^{(k)}_{P^{(b)}_{X,S}}$ whose second component is $\vecs$ is contained in a different bin from the $L^{(k)}_{P^{(b)}_{X|S}}$ bins $\{ \setB_\ell \}$.) To transmit the triple~$\bigl( m,s^{bk},\vecs^{(B + 1)} \bigr)$, the encoder picks from the bin that is associated with the subset of $\bm { \mathcal I}_{b-1}$ containing $\bigl( m,s^{(b-1) k},\vecs^{(B + 1)} \bigr)$ the pair $(\vecx^\prime,\vecs^\prime)$ satisfying $\vecs^\prime = \vecs^{(b)}$ \eqref{eq:coveringCondSA} and chooses as the Block-$b$ channel-inputs $\vecx^{(b)}$ the $k$-tuple $\vecx^\prime$.

Based on the Block-$b$ outputs $\vecy^{(b)} \triangleq Y^{b k}_{(b-1) k + 1}$ and the empirical type $P_{\vecs^{(b)}}$, the encoder and decoder compute $\bm { \mathcal I}_b$ as follows. First, they identify all the pairs $(\tilde \vecx, \tilde \vecs) \in \setT^{ (k) }_{ P^{(b)}_{X,S} }$ that could have produced the observed Block-$b$ outputs $\vecy^{(b)}$. For each such pair $(\tilde \vecx, \tilde \vecs)$ they identify the unique bin that contains it, and they include in $\bm { \mathcal I}_b$ all the triples $(m^\prime,\hat \vecs,  \vecs^\prime) \in \setM \times \setS^{bk} \times \setS^{k^\prime}$ satisfying that $\hat s^{bk}_{(b-1)k + 1} = \tilde \vecs$ and that $(m^\prime,\hat s^{(b-1)k}, \vecs^\prime)$ is an element of the subset of $\bm { \mathcal I}_{b-1}$ with which this bin is associated.

Using arguments similar to those in the direct part of Theorem~\ref{th:capacity}, we next show that
\begin{subequations}\label{bl:noTriplesIbGivenOutputsSA}
\begin{IEEEeqnarray}{l}
| \bm { \mathcal I}_b | \leq \biggl( \max_{P_{Y|X,S} \in \mathscr P (W)} 2^{ - k ( I (X,S;Y) - H (S) - \beta_k ) } \biggr) |\bm { \mathcal I}_{b-1}|, \label{eq:noTriplesIbGivenOutputsSA}
\end{IEEEeqnarray}
whenever
\begin{IEEEeqnarray}{l}
|\bm { \mathcal I}_{b-1}| \geq L^{(k)}_{P^{(b)}_{X|S}}, \label{eq:noTriplesIbGivenOutputsSACondIBMin1}
\end{IEEEeqnarray}
and
\begin{IEEEeqnarray}{l}
| \bm { \mathcal I}_b | \leq 2^{ k ( H (X,S) + \beta_k ) } \label{eq:noTriplesIbGivenOutputsSAAlt}
\end{IEEEeqnarray}
\end{subequations}
otherwise, where the mutual information is computed w.r.t.\ the joint PMF $P_{X,S}^{(b)} \times P_{Y|X,S}$, and where $\beta_k$ is given by
\begin{IEEEeqnarray}{l}
\beta_k = \frac{\log (1 + k ) \, |\setX| \, |\setS| \, (1 + |\setY|) + 1}{k} \label{eq:defBetakSA}
\end{IEEEeqnarray}
and hence converges to zero as $k$ tends to infinity. To this end note that, with probability one, the empirical type of the tuple $\bigl( \vecx^{(b)}, \vecs^{(b)}, \vecy^{(b)} \bigr)$ satisfies
\begin{subequations}\label{bl:constTypeXSYSA}
\begin{IEEEeqnarray}{l}
P_{\vecx^{(b)},\vecs^{(b)}} = P^{(b)}_{X,S}, \\
\Bigl( \channel y {x,s} = 0 \Bigr) \implies \Bigl( P_{\vecx^{(b)},\vecs^{(b)},\vecy^{(b)}} (x,s,y) = 0 \Bigr).
\end{IEEEeqnarray}
\end{subequations}
This allows us to upper-bound the number of pairs in $\setT^{ (k) }_{ P^{(b)}_{X,S} }$ that could have produced the observed Block-$b$ outputs $\vecy^{(b)}$: For every fixed $k$-type $P_{X,S,Y}$ on $\setX \times \setS \times \setY$ the number of pairs $(\tilde \vecx, \tilde \vecs) \in \setT^{ (k) }_{ P^{(b)}_{X,S} }$ that satisfy $\bigl( \tilde \vecx, \tilde \vecs, \vecy^{(b)} \bigr) \in \setT^{(k)}_{P_{X,S,Y}}$ cannot exceed $2^{k H (X,S|Y)}$, where the conditional entropy is computed w.r.t.\ the joint PMF $P_{X,S,Y}$ \cite[Lemma~2.5]{csiszarkoerner11}. This, combined with \eqref{bl:constTypeXSYSA} and the fact that the number of $k$-types on $\setX \times \setS \times \setY$ cannot exceed $(1 + k)^{|\setX| \, |\setS| \, |\setY|}$, implies that the number of pairs in $\setT^{(k)}_{P^{(b)}_{X,S}}$ that could have produced the observed Block-$b$ outputs $\vecy^{(b)}$ is upper-bounded by
\begin{IEEEeqnarray}{l}
2^{\log (1 + k) \, |\setX| \, |\setS| \, |\setY| } \max_{P_{Y|X,S} \in \mathscr P (W)} 2^{k H (X,S|Y)}, \label{eq:noSeqXSSA}
\end{IEEEeqnarray}
where the conditional entropy is computed w.r.t.\ the joint PMF $P^{(b)}_{X,S} \times P_{Y|X,S}$. Since the bins are pairwise disjoint \eqref{eq:binsDisjointSA}, no pair is contained in more than one bin. Every bin is associated with a subset of $\bm { \mathcal I}_{b-1}$ whose size is at most $\Bigl\lceil |\bm { \mathcal I}_{b-1}| / L^{(k)}_{P^{(b)}_{X|S}} \Bigr\rceil$; and by \eqref{eq:sizeSetsSA}
\begin{IEEEeqnarray}{l}
\biggl\lceil |\bm { \mathcal I}_{b-1}| / L^{(k)}_{P^{(b)}_{X|S}} \biggr\rceil \leq 2^{- k H (X|S) + \log (1 + k) \, |\setX| \, |\setS| + 1} \, |\bm { \mathcal I}_{b-1}|, \label{eq:sizeSetIbUBSA}
\end{IEEEeqnarray}
whenever \eqref{eq:noTriplesIbGivenOutputsSACondIBMin1} holds, and
\begin{IEEEeqnarray}{l}
\biggl\lceil |\bm { \mathcal I}_{b-1}| / L^{(k)}_{P^{(b)}_{X|S}} \biggr\rceil = 1 \label{eq:sizeSetIbUBSAAlt}
\end{IEEEeqnarray}
otherwise. From \eqref{eq:noSeqXSSA}--\eqref{eq:sizeSetIbUBSAAlt}, the fact that
\begin{IEEEeqnarray}{rCl}
H (X|S) - H (X,S|Y) & = & I (X,S;Y) - H (S),
\end{IEEEeqnarray}
and the inequality
\begin{IEEEeqnarray}{rCl}
H (X,S|Y) & \leq & H (X,S),
\end{IEEEeqnarray}
which holds because conditioning cannot increase entropy, we obtain \eqref{bl:noTriplesIbGivenOutputsSA}.

We next use \eqref{bl:noTriplesIbGivenOutputsSA} to show that---for some choice of the $k$-type $P_{X,S}^{(b)}$ and some $\gamma_k = \gamma_k \bigl( |\setX|, |\setS|, |\setY| \bigr)$, which converges to zero as $k$ tends to infinity---we can guarantee that
\begin{subequations}\label{bl:noTriplesIbGivenOutputsSA2}
\begin{IEEEeqnarray}{l}
| \bm { \mathcal I}_b | \leq \biggl( \max_{P_S} \min_{P_{X|S}} \max_{P_{Y|X,S} \in \mathscr P (W)} 2^{ - k ( I (X,S;Y) - H (S) - \gamma_k ) } \biggr) |\bm { \mathcal I}_{b-1}|, \label{eq:noTriplesIbGivenOutputsSA2}
\end{IEEEeqnarray}
whenever
\begin{IEEEeqnarray}{l}
|\bm { \mathcal I}_{b-1}| \geq 2^{k \log |\setX|}, \label{eq:noTriplesIbGivenOutputsSACondIBMin12}
\end{IEEEeqnarray}
and
\begin{IEEEeqnarray}{l}
| \bm { \mathcal I}_b | \leq 2^{ k (\log |\setX| + \log |\setS| + \gamma_k ) } \label{eq:noTriplesIbGivenOutputsSAAlt2}
\end{IEEEeqnarray}
\end{subequations}
otherwise, where the mutual information and the entropy are computed w.r.t.\ the joint PMF $P_S \times P_{X|S} \times P_{Y|X,S}$. To this end we will first infer from \eqref{bl:noTriplesIbGivenOutputsSA} that
\begin{subequations}\label{bl:noTriplesIbGivenOutputsSA3}
\begin{IEEEeqnarray}{l}
| \bm { \mathcal I}_b | \leq \biggl( \max_{P_{Y|X,S} \in \mathscr P (W)} 2^{ - k ( I (X,S;Y) - H (S) - \beta_k ) } \biggr) |\bm { \mathcal I}_{b-1}|, \label{eq:noTriplesIbGivenOutputsSA3}
\end{IEEEeqnarray}
whenever
\begin{IEEEeqnarray}{l}
|\bm { \mathcal I}_{b-1}| \geq 2^{k \log |\setX|}, \label{eq:noTriplesIbGivenOutputsSACondIBMin13}
\end{IEEEeqnarray}
and
\begin{IEEEeqnarray}{l}
| \bm { \mathcal I}_b | \leq 2^{ k (\log |\setX| + \log |\setS| + \beta_k ) } \label{eq:noTriplesIbGivenOutputsSAAlt3}
\end{IEEEeqnarray}
\end{subequations}
otherwise, where the mutual information is computed w.r.t.\ the joint PMF $P_{X,S}^{(b)} \times P_{Y|X,S}$, and where $\beta_k$ is defined in \eqref{eq:defBetakSA}. The following three observations show that $\eqref{bl:noTriplesIbGivenOutputsSA} \implies \eqref{bl:noTriplesIbGivenOutputsSA3}:$
\begin{enumerate}
\item[1)] By \eqref{eq:defSizeSetsSA} $L^{(k)}_{P_{X|S}^{(b)}} \leq |\setX|^k$, so whenever Condition~\eqref{eq:noTriplesIbGivenOutputsSACondIBMin13} holds so does \eqref{eq:noTriplesIbGivenOutputsSACondIBMin1}. Consequently, we obtain from \eqref{bl:noTriplesIbGivenOutputsSA} that whenever Condition~\eqref{eq:noTriplesIbGivenOutputsSACondIBMin13} holds the inequality \eqref{eq:noTriplesIbGivenOutputsSA3} holds.
\item[2)] Since
\begin{IEEEeqnarray}{l}
\log |\setX| - \bigl( I (X,S;Y) - H (S) - \beta_k \bigr) \leq \log |\setX| + \log |\setS| + \beta_k,
\end{IEEEeqnarray}
it follows from \eqref{bl:noTriplesIbGivenOutputsSA} that the inequality \eqref{eq:noTriplesIbGivenOutputsSAAlt3} holds whenever
\begin{IEEEeqnarray}{l}
L^{(k)}_{P_{X|S}^{(b)}} \leq |\bm { \mathcal I}_{b-1}| < 2^{k \log |\setX|}.
\end{IEEEeqnarray}
\item[3)] Since $H (X,S) \leq \log |\setX| + \log |\setS|$, it follows from \eqref{bl:noTriplesIbGivenOutputsSA} that the inequality \eqref{eq:noTriplesIbGivenOutputsSAAlt3} holds whenever
\begin{IEEEeqnarray}{l}
|\bm { \mathcal I}_{b-1}| < L^{(k)}_{P_{X|S}^{(b)}}.
\end{IEEEeqnarray}
\end{enumerate}
Having established \eqref{bl:noTriplesIbGivenOutputsSA3}, we are now ready to prove \eqref{bl:noTriplesIbGivenOutputsSA2}. Since we can choose any $k$-type $P_{X,S}^{(b)}$ whose $\setS$-marginal $P_S^{(b)}$ is $P_{\vecs^{(b)}}$, we can choose $P_{X,S}^{(b)} = P_{\vecs^{(b)}} \times P^{(b)}_{X|S}$, where $P^{(b)}_{X|S}$ is the conditional $k$-type that---among all conditional $k$-types---maximizes
\begin{IEEEeqnarray}{l}
\min_{P_{Y|X,S} \in \mathscr P (W)} I (X,S;Y) - H (S),
\end{IEEEeqnarray}
where the mutual information and the entropy are computed w.r.t.\ the joint PMF $P_{X,S}^{(b)} \times P_{Y|X,S}$. Every conditional PMF can be approximated in the total variation distance by a conditional $k$-type when $k$ is sufficiently large; and, because entropy and mutual information are continuous in this distance \cite[Lemma~2.7]{csiszarkoerner11}, it follows that---for the above choice of the conditional $k$-type and some $\gamma_k = \gamma_k \bigl( |\setX|, |\setS|, |\setY| \bigr)$, which converges to zero as $k$ tends to infinity---\eqref{bl:noTriplesIbGivenOutputsSA3} implies \eqref{bl:noTriplesIbGivenOutputsSA2}.

Since we assume that the RHS of \eqref{eq:capacitySA} is positive, we can choose $B$ and $k$ sufficiently large so that
\begin{IEEEeqnarray}{l}
\!\! \left( \max_{P_S} \min_{P_{X|S}} \max_{P_{Y|X,S} \in \mathscr P (W)} 2^{ - B k ( I (X,S;Y) - H(S) - \gamma_k ) } \right) |\setM| \, |\setS|^{k^\prime} \nonumber \\
\quad \leq 2^{k ( \log |\setX| + \log |\setS| + \gamma_k )}; \label{eq:choiceBKForLastBlockSA}
\end{IEEEeqnarray}
and by \eqref{bl:noTriplesIbGivenOutputsSA2} this guarantees that, with probability one,
\begin{IEEEeqnarray}{l}
| \bm { \mathcal I}_B | \leq 2^{k (\log |\setX| + \log |\setS| + \gamma_k)}. \label{eq:UBIBSA}
\end{IEEEeqnarray}

We now deal with Block~$(B + 1)$. Because the decoder is incognizant of the empirical types $\bigl\{ P_{\vecs^{(b)}} \bigr\}_{b \in [1:B]}$, it cannot compute the post-Block-$B$ ambiguity-set $\bm { \mathcal I}_B$ comprising the pairs of message and length-$n$ state-sequence of positive posterior probability given the channel outputs $\bigl\{ \vecy^{(b)} \bigr\}_{b \in [1:B]}$ and the $k$-types $\bigl\{ P_{\vecs^{(b)}} \bigr\}_{b \in [1:B]}$. The uncertainty that needs to be addressed is about the message, the length-$n$ state-sequence, as well as the $B$ empirical types of $\vecs^{(1)}, \ldots, \vecs^{(B)}$. Let $\bm { \mathcal J}_B \subseteq \setM \times \setS^n$ denote the union of the post-Block-$B$ ambiguity-sets corresponding to all the different $B$-tuples of $k$-types on $\setS$, i.e., $\bm { \mathcal J}_B$ is the set of pairs of messages and state sequences that have a positive posterior probability given only the outputs $\bigl\{ \vecy^{(b)} \bigr\}_{b \in [1:B]}$ (and not the $k$-types $\bigl\{ P_{\vecs^{(b)}} \bigr\}_{b \in [1:B]}$). Because the post-Block-$B$ ambiguity-set corresponding to any given $B$-tuple of $k$-types on $\setS$ satisfies \eqref{eq:UBIBSA}, and because there are at most $(1 + k)^{B \, |\setS|}$ $B$-tuples of $k$-types on $\setS$,
\begin{IEEEeqnarray}{l}
|\bm { \mathcal J}_B| \leq 2^{k (\log |\setX| + \log |\setS| + \gamma_k) + B \log (1 + k) \, |\setS|}. \label{eq:UBJBSA}
\end{IEEEeqnarray}
In Block~$(B + 1)$ we resolve the set $\bm { \mathcal J}_B$. This will guarantee that the decoder can recover the transmitted message~$m$ and the length-$n$ state-sequence $\vecs$ error-free.

Block~$(B + 1)$ is similar to Phase~2 of the scheme we used to prove Remark~\ref{re:chUsesPerBit}: the encoder allocates to every pair $(m^\prime,\vecs^\prime) \in \bm { \mathcal J}_B$ a length-$k^\prime$ codeword $\vecx (m^\prime,\vecs^\prime)$, where the codewords are chosen so that
\begin{IEEEeqnarray}{l}
\Bigl( \forall \, (m^\prime,\vecs^\prime), \, (m^{\prime\prime},\vecs^{\prime\prime}) \in \bm { \mathcal J}_B \textnormal{ s.t.\ } m^\prime \neq m^{\prime\prime} \Bigr) \quad \exists \, i \in [1:k^\prime] \textnormal{ s.t.\ } \nonumber \\*[-0.6\normalbaselineskip]
  \label{eq:condLastBlockSA}
\\*[-0.6\normalbaselineskip]
\qquad \Bigl( \bigchannel y {x_i (m^\prime,\vecs^\prime), s_{B k + i}^\prime} \, \bigchannel y {x_i (m^{\prime\prime}, \vecs^{\prime\prime}), s_{B k + i}^{\prime\prime}} = 0 , \,\, \forall \, y \in \setY \Bigr). \nonumber
\end{IEEEeqnarray}
(We shall shortly use a random coding argument to show that this can be done.) To convey the message~$m$ and the state sequence $\vecs$, the encoder transmits in Block~$(B + 1)$ the codeword $\vecx ( m,\vecs )$. Condition~\eqref{eq:condLastBlockSA} implies that, upon observing the Block-$(B + 1)$ outputs $\vecy^{(B + 1)} \triangleq Y^{Bk + k^\prime}_{Bk + 1}$, the decoder, who knows $\bm { \mathcal J}_B$ and the codewords $\bigl\{ \vecx (m^\prime,\vecs^\prime) \bigr\}$, can determine the transmitted message~$m$ and the state sequence $\vecs$ error-free, because
\begin{IEEEeqnarray}{l}
\prod^{k^\prime}_{i = 1} W \Bigl( y_i^{(B + 1)} \Bigl| x_i \bigl(m, \vecs \bigr), s_{B k + i} \Bigr) > 0,
\end{IEEEeqnarray}
whereas \eqref{eq:condLastBlockSA} implies for every other pair $(m^\prime,\vecs^\prime) \in \bm { \mathcal J}_B$
\begin{IEEEeqnarray}{l}
\prod^{k^\prime}_{i = 1} W \Bigl( y_i^{(B + 1)} \Bigl| x_i (m^\prime, \vecs^\prime), s^\prime_{B k + i} \Bigr) = 0.
\end{IEEEeqnarray}
The decoder can thus calculate $\prod_i W \bigl( y_i^{(B + 1)} \bigl| x_i (\tilde m, \tilde \vecs), \tilde s_{B k + i} \bigr)$ for each $( \tilde m, \tilde \vecs ) \in \bm { \mathcal J}_B$ and produce the pair $(\tilde m, \tilde \vecs)$ for which this product is positive.

We next show that, for some choice of $k^\prime$, there exist codewords $\bigl\{ \vecx (m^\prime,\vecs^\prime) \bigr\}$ satisfying \eqref{eq:condLastBlockSA}. To this end we use a random coding argument. Draw the length-$k^\prime$ codewords $\bigl\{ \rndvecX (m^\prime, \vecs^\prime) \bigr\}$ independently, each uniformly over $\setX^{k^\prime}$. From \eqref{eq:positive} it then follows that for any fixed distinct $(m^\prime,\vecs^\prime), \, (m^{\prime\prime},\vecs^{\prime\prime}) \in \bm { \mathcal J}_B$
\begin{IEEEeqnarray}{l}
\Bigdistof {\forall \, i \in [1 : k^\prime] \,\, \exists \, y \in \setY \textnormal{ s.t.\ } W \bigl( y \, \bigl| X_i (m^\prime, \vecs^\prime), s^\prime_{B k + i} \bigr) \, W \bigl( y \, \bigl| X_i (m^{\prime\prime}, \vecs^{\prime\prime}), s^{\prime\prime}_{B k + i} \bigr) > 0} \nonumber \\
\quad \leq \biggl( 1 - \frac{1}{|\setX|^2} \biggr)^{\!\! k^{\prime}} \\
\quad = 2^{- k^{\prime} ( 2 \log |\setX| - \log ( |\setX|^2 - 1 ) ) }.
\end{IEEEeqnarray}
This, the Union-of-Events bound, and \eqref{eq:UBJBSA} imply that the probability that the randomly drawn length-$k^\prime$ codewords do not satisfy \eqref{eq:condLastBlockSA} is upper-bounded by
\begin{IEEEeqnarray}{l}
| \bm { \mathcal J}_B |^2 \, 2^{- k^{\prime} ( 2 \log |\setX| - \log ( |\setX|^2 - 1 ) ) } \nonumber \\
\quad \leq 2^{- k^{\prime} ( 2 \log |\setX| - \log ( |\setX|^2 - 1 ) ) + 2 ( k ( \log |\setX| + \log |\setS| + \gamma_k ) + B \log (1 + k) \, |\setS| )},
\end{IEEEeqnarray}
which is smaller than one whenever
\begin{IEEEeqnarray}{l}
k^{\prime} > \frac{k \bigl( \log |\setX| + \log |\setS| + \gamma_k \bigr) + B \log (1 + k) \, |\setS|}{\log |\setX| - \frac{1}{2} \log \bigl( |\setX|^2 - 1 \bigr)}. \label{eq:kPrimeLBSA}
\end{IEEEeqnarray}
Consequently, if we choose some $k^\prime$ that satisfies \eqref{eq:kPrimeLBSA}, then there exist length-$k^\prime$ codewords $\bigl\{ \vecx (m^\prime,\vecs^\prime) \bigr\}$ satisfying \eqref{eq:condLastBlockSA}.\\

We are now ready to join the dots and conclude that the coding scheme asymptotically achieves any rate smaller than the RHS of \eqref{eq:capacitySA}. More precisely, we will show that, for every rate $R$ smaller than the RHS of \eqref{eq:capacitySA} and every sufficiently-large blocklength~$n$, our coding scheme can convey $n R$ bits and the length-$n$ state-sequence error-free in $n$ channel uses.

It follows from \eqref{eq:choiceBKForLastBlockSA} and \eqref{eq:kPrimeLBSA} that if the positive integers $n, \, B, \, k, \, k^\prime$ are such that \eqref{eq:kPrimeLBSA} holds,
\begin{equation}
n = B k + k^\prime, \label{eq:chUsesBlocks1ThroughBPl1SA}
\end{equation}
and
\begin{IEEEeqnarray}{l}
n R + k^\prime \log |\setS| \leq B k \biggl( \min_{P_S} \max_{P_{X|S}} \min_{P_{Y|X,S} \in \mathscr P (W)} I (X,S;Y) - H(S) - \gamma_k \biggr), \label{eq:choiceBKForLastBlockSA2}
\end{IEEEeqnarray}
then our coding scheme can convey $n R$ bits and the length-$n$ state-sequence error-free in $n$ channel uses. It thus remains to exhibit positive integers $B, \, k, \, k^\prime$ such that for every sufficiently-large blocklength~$n$ \eqref{eq:kPrimeLBSA}--\eqref{eq:choiceBKForLastBlockSA2} hold. As we argue next, when $n$ is sufficiently large we can choose
\begin{subequations}\label{bl:defBkSA}
\begin{IEEEeqnarray}{rCl}
B & = & \bigl\lfloor \sqrt n \bigr\rfloor - \Biggl( \biggl\lfloor \frac{\log |\setX| + \log |\setS| + \gamma_k + \log (1 + \sqrt n) \, |\setS|}{\log |\setX| - \frac{1}{2} \log \bigl( |\setX|^2 - 1 \bigr)} \biggr\rfloor + 1 \Biggr), \\
k & = & \bigl\lfloor \sqrt n \bigr\rfloor, \label{eq:defBkkSA} \\
k^\prime & = & n - B k.
\end{IEEEeqnarray}
\end{subequations}
Note that, whenever $n$ is sufficiently large, $B$, $k$, and $k^\prime$ are positive, and \eqref{eq:kPrimeLBSA} and \eqref{eq:chUsesBlocks1ThroughBPl1SA} are satisfied. To see that \eqref{eq:choiceBKForLastBlockSA2} holds whenever $n$ is sufficiently large, we first observe from \eqref{eq:defBkkSA} that $k$ tends to infinity as $n$ tends to infinity. Because $\gamma_k = \gamma_k \bigl( |\setX|, |\setS|, |\setY| \bigr)$ converges to zero as $k$ tends to infinity, this implies that $\gamma_k$ converges to zero as $n$ tends to infinity. We next observe that \eqref{bl:defBkSA} implies that $B k / n$ converges to one as $n$ tends to infinity and consequently that $k^\prime / n = 1 - B k / n$ converges to zero as $n$ tends to infinity. This, combined with the facts that $\gamma_k$ converges to zero as $n$ tends to infinity and that $R$ is smaller than the RHS of \eqref{eq:capacitySA}, implies that \eqref{eq:choiceBKForLastBlockSA2} holds whenever $n$ is sufficiently large.
\end{proof}

We next establish the converse part of Theorem~\ref{th:stateAmplification}.

\begin{proof}[Converse Part]
That \eqref{eq:positive} is a necessary condition for $C^{\textnormal{m} + \textnormal{s}}_{\textnormal{f},0}$ to be positive follows from Theorem~\ref{th:positive}, because $C^{\textnormal{m} + \textnormal{s}}_{\textnormal{f},0}$ is upper-bounded by $C_{\textnormal{f},0}$. We next show that---irrespective of whether or not \eqref{eq:positive} holds---$C^{\textnormal{m} + \textnormal{s}}_{\textnormal{f},0}$ is upper-bounded by the RHS of \eqref{eq:capacitySA}. The proof is similar to the converse of Theorem~\ref{th:capacity}. Fix a finite set $\setM$, a blocklength~$n$, and an $(n,\setM)$ zero-error state-conveying code with $n$ encoding mappings
\begin{IEEEeqnarray}{l}
f_i \colon \setM \times \setS^n \times \setY^{i-1} \rightarrow \setX, \quad i \in [1:n] \label{eq:convStateAmpEnc}
\end{IEEEeqnarray}
and $|\setM| \, |\setS|^n$ disjoint decoding sets $\setD_{m, \vecs} \subseteq \setY^n, \,\, (m, \vecs) \in \setM \times \setS^n$. We will show that the rate $\frac{1}{n} \log |\setM|$ of the code is upper-bounded by the RHS of \eqref{eq:capacitySA}.

Draw $M$ uniformly over $\setM$, and denote its distribution $P_M$. Since the code is a zero-error state-conveying code,
\begin{equation}
\distof {Y^n \in \setD_{M, S^n}} = 1, \label{eq:convZeroDecErrSA}
\end{equation}
where $\dist$ is the distribution \eqref{eq:capacityConvDist} of $(M,S^n,X^n,Y^n)$ induced by $P_M$, the state distribution $Q$, the encoding mappings \eqref{eq:convStateAmpEnc}, and the channel law $\channel y {x,s}$. Similarly as in the converse of Theorem~\ref{th:capacity}, fix any PMF $\tilde P_S$ on $\setS$ and any collection of $n$ conditional PMFs $\bigl\{ \tilde P_{Y_i|X_i,S_i} \bigr\}_{i \in [1:n]}$ that satisfy
\begin{equation}
\tilde P_{Y_i|X_i,S_i} \in \mathscr P (W). \label{eq:convAbsContTrProbSA}
\end{equation}
These PMFs induce the PMF on $\setM \times \setS^n \times \setX^n \times \setY^n$
\begin{IEEEeqnarray}{l}
\tilde P_{M, S^n, X^n, Y^n} = P_M \times \tilde P_S^n \times \prod^n_{i = 1} \bigl( P_{X_i|M,S^n,Y^{i-1}} \times \tilde P_{Y_i|X_i,S_i} \bigr). \label{eq:convAbsContPMFSA}
\end{IEEEeqnarray}
It follows from \eqref{eq:assQAssignsPosProbs} and \eqref{eq:convAbsContTrProbSA} that $\tilde P_{M, S^n, X^n, Y^n} \ll \dist$ and consequently that \eqref{eq:convZeroDecErrSA} implies
\begin{IEEEeqnarray}{l}
\tilde P_{M, S^n, X^n, Y^n} [Y^n \in \setD_{M,S^n}] = 1. \label{eq:convZeroDecErrImpSA}
\end{IEEEeqnarray}

We upper-bound $\frac{1}{n} \log |\setM|$ by carrying out the following calculation under $\tilde P_{M, S^n, X^n, Y^n}$ of \eqref{eq:convZeroDecErrImpSA}:
\begin{IEEEeqnarray}{l}
\frac{1}{n} \log |\setM| \nonumber \\
\quad \stackrel{(a)}= \frac{1}{n} \Bigl[ H (M) + H ( S^n ) - H ( S^n ) \Bigr] \\
\quad \stackrel{(b)}= \frac{1}{n} \Bigl[ I ( S^n, M ; Y^n ) - H ( S^n ) \Bigr] \\
\quad \stackrel{(c)}= \frac{1}{n} \sum^n_{i = 1} \Bigl[ I ( S^n, M ; Y_i | Y^{i-1} ) - H ( S_i | S^{i-1} ) \Bigr] \\
\quad \stackrel{(d)}\leq \frac{1}{n} \sum^n_{i = 1} \Bigl[ I ( S^n, M, Y^{i-1} ; Y_i ) - H ( S_i ) \Bigr] \\
\quad \stackrel{(e)}\leq \frac{1}{n} \sum^n_{i = 1} \Bigl[ I ( X_i, S_i ; Y_i ) - H ( S_i ) \Bigr], \label{eq:convFixedDistsSA}
\end{IEEEeqnarray}
where $(a)$ holds because $M$ is uniform over $\setM$ under $\tilde P_{M, S^n, X^n, Y^n}$; $(b)$ holds by \eqref{eq:convZeroDecErrImpSA} and because $M$ is independent of $S^n$ under $\tilde P_{M, S^n, X^n, Y^n}$; $(c)$ follows from the chain rule; $(d)$ holds because conditioning cannot increase entropy and by the independence of $S_i$ and $S^{i-1}$ under $\tilde P_{M, S^n, X^n, Y^n}$; and $(e)$ holds because under $\tilde P_{M, S^n, X^n, Y^n}$ $( S^n, M, Y^{i-1} )$, $( X_i, S_i )$, and $Y_i$ form a Markov chain in that order.

We will conclude the proof by exhibiting a PMF $\tilde P_S$ and a collection of conditional PMFs $\bigl\{ \tilde P_{Y_i|X_i,S_i} \bigr\}_{i \in [1:n]}$ satisfying \eqref{eq:convAbsContTrProbSA} for which each summand on the RHS of \eqref{eq:convFixedDistsSA} is upper-bounded by the RHS of \eqref{eq:capacitySA}.

We begin with the choice of $\bigl\{ \tilde P_{Y_i|X_i,S_i} \bigr\}_{i \in [1:n]}$. We first choose $\tilde P_{Y_i|X_i,S_i}$ for $i = 1$, and we then repeatedly increment $i$ by one until it reaches $n$. Key to our choice is the observation, which will be justified shortly, that $\tilde P_{X_i,S_i}$ is determined by $\tilde P_S$ and $\bigl\{ \tilde P_{Y_j|X_j,S_j} \bigr\}_{j \in [i-1]}$. Our choice of $\tilde P_{Y_i|X_i,S_i}$ can thus depend not only on our choice of $\tilde P_S$ and our previous choices of $\bigl\{ \tilde P_{Y_j|X_j,S_j} \bigr\}_{j \in [1:i-1]}$ but also on $\tilde P_{X_i,S_i}$. This will allow us to choose $\tilde P_{Y_i|X_i,S_i}$ as one that---among all conditional PMFs satisfying \eqref{eq:convAbsContTrProbSA}---minimizes
\begin{IEEEeqnarray}{l}
I ( X_i, S_i ; Y_i ) - H ( S_i ),
\end{IEEEeqnarray}
where the mutual information and the entropy are computed w.r.t.\ the joint PMF $\tilde P_{X_i,S_i} \times \tilde P_{Y_i|X_i,S_i}$. Since \eqref{eq:convAbsContPMFSA} implies that
\begin{equation}
\tilde P_{S_i} = \tilde P_S, \quad i \in [1:n], \label{eq:confPMFPSiSA}
\end{equation}
we will then find that, for our choice of $\bigl\{ \tilde P_{Y_i|X_i,S_i} \bigr\}$,
\begin{IEEEeqnarray}{l}
I ( X_i, S_i ; Y_i ) - H ( S_i ) \nonumber \\
\quad \leq \max_{\tilde P_{X_i|S_i}} \min_{\tilde P_{Y_i|X_i,S_i} \in \mathscr P (W)} I ( X_i, S_i ; Y_i ) - H ( S_i ), \quad i \in [1:n], \label{eq:convFixedDistsSA2}
\end{IEEEeqnarray}
where the mutual information and the entropy are computed w.r.t.\ the joint PMF $\tilde P_{S_i} \times \tilde P_{X_i|S_i} \times \tilde P_{Y_i|X_i,S_i}$. The chosen conditional PMFs $\bigl\{ \tilde P_{Y_i|X_i,S_i} \bigr\}_{i \in [1:n]}$ satisfy \eqref{eq:convAbsContTrProbSA}, and hence \eqref{eq:convFixedDistsSA}, \eqref{eq:confPMFPSiSA}, and \eqref{eq:convFixedDistsSA2} will imply that
\begin{IEEEeqnarray}{l}
\frac{1}{n} \log |\setM| \nonumber \\
\quad \leq \max_{\tilde P_{X|S}} \min_{\tilde P_{Y|X,S} \in \mathscr P (W)} I ( X, S ; Y ) - H ( S ), \label{eq:convFixedDistsSA3}
\end{IEEEeqnarray}
where the mutual information and the entropy in the $i$-th summand are computed w.r.t.\ the joint PMF $\tilde P_S \times \tilde P_{X|S} \times \tilde P_{Y|X,S}$.

We now prove that indeed $\tilde P_{X_i,S_i}$ is determined by $\tilde P_S$ and $\bigl\{ \tilde P_{Y_j|X_j,S_j} \bigr\}_{j \in [1:i-1]}$. In fact, we will show that the latter two determine $\tilde P_{M,S^n,X^i,Y^{i-1}}$. The latter determines $\tilde P_{X_i,S_i}$, because the tuple $(X_i,S_i)$ is determined by $(M,S^n,X^i,Y^{i-1})$.

We use mathematical induction, but first we note that the PMF $\tilde P_{M,S^n,X^n,Y^n}$ is constructed inductively: by \eqref{eq:convAbsContPMFSA}
\begin{equation}
\tilde P_{M, S^n, X_1} = P_M \times \tilde P_S^n \times P_{X_1|M,S^n}, \label{eq:convDistMSnX1SA}
\end{equation}
and, for every $\ell \in [2:n]$, $\tilde P_{M, S^n, X^\ell, Y^{\ell-1}}$ is constructed from $\tilde P_{M, S^n, X^{\ell-1}, Y^{\ell-2}}$ by
\begin{IEEEeqnarray}{l}
\tilde P_{M, S^n, X^\ell, Y^{\ell-1}} = \tilde P_{M, S^n, X^{\ell - 1}, Y^{\ell - 2}} \times \tilde P_{Y_{\ell-1}|X_{\ell-1},S_{\ell-1}} \times P_{X_\ell|M,S^n,Y^{\ell-1}}. \label{eq:convAbsContPMFXYFromiToiPl1SA}
\end{IEEEeqnarray}
In describing the proof we shall make the dependence on $P_M$, our choice of $\tilde P_S$, and $\bigl\{ P_{X_j|M,S^n,Y^{j-1}} \bigr\}_{j \in [1:n]}$, whose components are determined by the encoding mappings \eqref{eq:convStateAmpEnc} via \eqref{eq:convCondPMFXDetEncMapp}, implicit.
\begin{enumerate}
\item Basis $\ell = 1$: It follows from \eqref{eq:convDistMSnX1SA} that $\tilde P_{M,S^n,X_1}$ is determined.
\item Inductive Step: Fix $\ell \in [2:i]$, and suppose that $\tilde P_{M,S^n,X^{\ell-1},Y^{\ell-2}}$ is determined by $\bigl\{ \tilde P_{Y_j|X_j,S_j} \bigr\}_{j \in [1:\ell - 2]}$. This implies that $\tilde P_{M,S^n,X^{\ell-1},Y^{\ell-2}}$ and $\tilde P_{Y_{\ell-1}|X_{\ell-1},S_{\ell-1}}$ are determined by $\bigl\{ \tilde P_{Y_j|X_j,S_j} \bigr\}_{j \in [1:\ell - 1]}$. Consequently, it follows from \eqref{eq:convAbsContPMFXYFromiToiPl1SA} that $\tilde P_{M,S^n,X^\ell,Y^{\ell-1}}$ is determined by $\bigl\{ \tilde P_{Y_j|X_j,S_j} \bigr\}_{j \in [1:\ell - 1]}$.
\end{enumerate}
This proves that, for every $i \in [1:n]$, $\tilde P_{M,S^n,X^i,Y^{i-1}}$ and consequently also $\tilde P_{X_i,S_i}$ are determined by $\tilde P_S$ and $\bigl\{ \tilde P_{Y_j|X_j,S_j} \bigr\}_{j \in [1:i-1]}$, and hence \eqref{eq:convFixedDistsSA3} holds.

Having established \eqref{eq:convFixedDistsSA3}, we are now ready to conclude the proof. Since we can choose any PMF $\tilde P_S$ on $\setS$, we can choose one that---among all PMFs on $\setS$---yields the tightest bound, i.e., minimizes
\begin{IEEEeqnarray}{l}
\max_{\tilde P_{X|S}} \min_{\tilde P_{Y|X,S} \in \mathscr P (W)} I ( X, S ; Y ) - H ( S ),
\end{IEEEeqnarray}
where the mutual information and the entropy are computed w.r.t.\ the joint PMF $P_S \times P_{X|S} \times P_{Y|X,S}$. For this choice of $\tilde P_S$ \eqref{eq:convFixedDistsSA3} implies that
\begin{IEEEeqnarray}{l}
\frac{1}{n} \log |\setM| \leq \min_{\tilde P_S} \max_{\tilde P_{X|S}} \min_{\tilde P_{Y|X,S} \in \mathscr P (W)} I ( X, S ; Y ) - H ( S ), \label{eq:convRateUBSA}
\end{IEEEeqnarray}
where the mutual information and the entropy are computed w.r.t.\ the joint PMF $P_S \times P_{X|S} \times P_{Y|X,S}$. If the RHS of \eqref{eq:convRateUBSA} is negative, then---irrespective of $|\setM| \geq 1$---\eqref{eq:convRateUBSA} is a contradiction and consequently \eqref{eq:convZeroDecErrImpSA} cannot hold. This implies that---even if $| \setM |$ is one---the state sequence cannot be conveyed error-free. Since we say that $C^{\textnormal{m} + \textnormal{s}}_{\textnormal{f},0} = 0$ if the state sequence cannot be conveyed error-free, \eqref{eq:convRateUBSA} implies that $C^{\textnormal{m} + \textnormal{s}}_{\textnormal{f},0}$ is upper-bounded by the RHS of \eqref{eq:capacitySA}.
\end{proof}

\section{A Proof of Theorem~\ref{th:ccInputs}}\label{sec:pfThCcInputs}

We already showed in Section~\ref{sec:constrainedInputs} using \eqref{eq:ccInpustConstFact} that, whenever $\Gamma > \Gamma_{\textnormal{min}}$, \eqref{eq:positive} is a necessary and sufficient condition for $C_{\textnormal{f},0} (\Gamma)$ to be positive, and we hence prove that if $C_{\textnormal{f},0} (\Gamma)$ is positive, then it is equal to \eqref{eq:ccCapacity}.

To that end we first show that restricting $X$ to be a function of $U$ and $S$, i.e., $P_{U,X|S}$ to have the form \eqref{eq:thCapCardUXFuncUS}, does not change the RHS of \eqref{eq:ccCapacity}, nor does restricting the cardinality of $\setU$ to \eqref{eq:cardU}:

\begin{lemma}\label{le:cardUCCInputs}
Given a channel $\channel y {x,s}$ and a PMF $P_S$ on $\setS$, consider
\begin{IEEEeqnarray}{l}
\max_{\substack{P_{U,X|S} \colon \\ \Ex {}{\gamma (X)} \leq \Gamma}} \min_{ \substack{ P_{Y|U,X,S} \colon \\ P_{Y|U = u,X,S} \in \mathscr P (W), \,\, \forall \, u \in \setU }} I (U;Y) - I (U;S), \label{eq:leCardUCapCCInputs}
\end{IEEEeqnarray}
where the maximization is over all chance variables $U$ of finite support, the expectation is computed w.r.t.\ the joint PMF $P_S \times P_{U,X|S}$, and the mutual informations are computed w.r.t.\ the joint PMF $P_S \times P_{U,X|S} \times P_{Y|U,X,S}$. Restricting $X$ to be a function of $U$ and $S$, i.e., $P_{U,X|S}$ to have the form
\begin{equation}
P_{U,X|S} (u,x|s) = P_{U|S} (u|s) \, \ind {x = g (u,s)}, \label{eq:leCardUXFuncUSCCInputs}
\end{equation}
does not change \eqref{eq:leCardUCapCCInputs}. Nor does requiring that $U$ take values in a set $\setU$  whose cardinality $| \setU |$ satisfies
\begin{equation}
|\setU| \leq |\setX|^{|\setS|}. \label{eq:cardULeCCInputs}
\end{equation}
\end{lemma}

\begin{proof}
The proof is essentially that of Lemma~\ref{le:cardU} in Appendix~\ref{sec:leCardU}. We first show that restricting $X$ to be a function of $U$ and $S$ does not change \eqref{eq:leCardUCapCCInputs}. In the proof of Lemma~\ref{le:cardU} it is shown that \eqref{eq:leCardUCap} is equal to \eqref{eq:leCardUCapEquivV}, and the same line of argument implies here that \eqref{eq:leCardUCapCCInputs} is equal to
\begin{IEEEeqnarray}{l}
\max_{ \substack{ P_V, h (\cdot), P_{U|S} \colon \\ \Ex {}{\gamma (X)} \leq \Gamma }} \min_{ \substack{ P_{Y|U,X,S} \colon \\ P_{Y|U = u,X,S} \in \mathscr P (W), \,\, \forall \, u \in \setU }} I (U;Y) - I (U;S), \label{eq:leCardUCapEquivVCCInputs}
\end{IEEEeqnarray}
where the maximization is over all chance variables $V$ of finite support $\setV$, functions $h \colon \setU \times \setV \times \setS \rightarrow \setX$, and conditional PMFs over a finite set $\setU$ for which
\begin{equation}
\bigEx {}{\gamma (X) \leq \Gamma}, \label{eq:pfCCInputsCC}
\end{equation}
where the expectation is computed w.r.t.\ the joint PMF $P_S \times P_V \times P_{U,X|V,S}$ and $P_{U,X|V,S}$ is defined in \eqref{eq:leCardUCondPMFUXGivVS}; and where the mutual informations are computed w.r.t.\ the joint PMF $P_S \times P_V \times P_{U,X|V,S} \times P_{Y|U,X,S}$. Unlike the proof of Lemma~\ref{le:cardU}, where we fix any PMF $P_V$ on $\setV$, any function $h \colon \setU \times \setV \times \setS \rightarrow \setX$, and any conditional PMF $P_{U|S}$, here we fix any $P_V$, any $h \colon \setU \times \setV \times \setS \rightarrow \setX$, and any $P_{U|S}$ for which \eqref{eq:pfCCInputsCC} holds w.r.t.\ $P_S \times P_V \times P_{U,X|V,S}$, where $P_{U,X|V,S}$ is defined in \eqref{eq:leCardUCondPMFUXGivVS}. The line of argument leading to \eqref{eq:leCardUMutWithVBigMutWithoutPMFs2} in the proof of Lemma~\ref{le:cardU} then implies that restricting $X$ to be a function of $U$ and $S$ does no change \eqref{eq:leCardUCapCCInputs}. To show that restricting the cardinality of $\setU$ to \eqref{eq:cardULeCCInputs} does not change \eqref{eq:leCardUCapCCInputs}, we fix any conditional PMF $P_{U,X|S}$ of the form \eqref{eq:leCardUXFuncUSCCInputs} for which \eqref{eq:pfCCInputsCC} holds w.r.t.\ $P_S \times P_{U,X|S}$. The line of argument leading to \eqref{eq:cardU5} in the proof of Lemma~\ref{le:cardU} then implies that restricting the cardinality of $\setU$ to \eqref{eq:cardULeCCInputs} does not change \eqref{eq:leCardUCapCCInputs}.
\end{proof}

\begin{proof}[Direct Part of Theorem~\ref{th:ccInputs}]

From Lemma~\ref{le:cardUCCInputs} it follows that it suffices to establish the direct part of Theorem~\ref{th:ccInputs} for the case where the cardinality of $\setU$ is restricted to \eqref{eq:cardU}. The direct part is essentially that of Theorem~\ref{th:capacity} but with the following two modifications: 1) During the first $B$ blocks we choose $k$-types $\bigl\{ P_{U,X,S}^{(b)} \bigr\}_{b \in [1:B]}$ w.r.t.\ which
\begin{equation}
\bigEx {}{\gamma (X)} \leq \Gamma. \label{eq:firstBBlocksTypeConstCCInputs}
\end{equation}
This will guarantee that
\begin{equation}
\frac{1}{B k} \sum^{B k}_{i = 1} \gamma (X_i) \leq \Gamma. \label{eq:firstBBlocksSatisfyCCInputs}
\end{equation}
2) We pad Block~$B + 1$ with as many symbols from the set $\setX^\prime$ as are needed to guarantee that
\begin{equation}
\frac{1}{k^\prime} \sum^{B k + k^\prime}_{B k + 1} \gamma (X_i) \leq \Gamma, \label{eq:lastBlocksSatisfiesCCInputs}
\end{equation}
where $k^\prime$ denotes the length of Block~$B + 1$.

By \eqref{eq:firstBBlocksSatisfyCCInputs} and \eqref{eq:lastBlocksSatisfiesCCInputs} the channel inputs' average cost satisfies the cost constraint \eqref{eq:ccInputs}. Padding Block~$(B + 1)$ to guarantee \eqref{eq:lastBlocksSatisfiesCCInputs} increases its length by a factor of at most
\begin{equation}
\tau \triangleq \biggl\lceil \frac{\gamma_{\textnormal {max}} - \gamma_{\textnormal {min}}}{\Gamma - \gamma_{\textnormal {min}}} \biggr\rceil. \label{eq:lastBlocksIncCCInputs}
\end{equation}
Consequently, also with the padding, the last block does not affect the rate of the code.

To show that the coding scheme asymptotically achieves any rate smaller than the RHS of \eqref{eq:ccCapacity}, we can argue essentially as in the proof of the direct part of Theorem~\ref{th:capacity}. We will show that, for every rate $R$ smaller than the RHS of \eqref{eq:ccCapacity} and every sufficiently-large blocklength~$n$, our coding scheme can convey $n R$ bits error-free in $n$ channel uses. It follows from \eqref{bl:choiceBKEpsForLastBlock}, \eqref{eq:chUsesLastBlock}, \eqref{eq:firstBBlocksTypeConstCCInputs}, and \eqref{eq:lastBlocksIncCCInputs} that if the positive integers $n, \, B, \, k$ and $\epsilon > 0$ are such that \eqref{eq:choiceBKEpsForLastBlock2K} holds and
\begin{IEEEeqnarray}{l}
n R \leq B k \left( \min_{P_S} \max_{\substack{P_{U,X|S} \colon \\ \Ex {}{\gamma (X)} \leq \Gamma}} \min_{ \substack{P_{Y|U,X,S} \colon \\ P_{Y|U = u, X,S} \in \mathscr P (W), \,\, \forall \, u \in \setU}} \!\!\!\!\! I (U;Y) - I (S;Y) - \delta (\epsilon, k ) \right) \!\!, \label{eq:choiceBKEpsForLastBlock2CCInputs}
\end{IEEEeqnarray}
then our coding scheme can convey $n R$ bits error-free in
\begin{IEEEeqnarray}{l}
B k + \tau \bigl\lceil k \log |\setU| + B \log (1 + k) \, |\setS| \bigr\rceil n_{\textnormal{bit}} \label{eq:chUsesBlocks1ThroughBPl1CCInputs}
\end{IEEEeqnarray}
channel uses. It thus remains to exhibit positive integers $B, \, k$ and some $\epsilon > 0$ such that, for every sufficiently-large blocklength~$n$, \eqref{eq:choiceBKEpsForLastBlock2K} and \eqref{eq:choiceBKEpsForLastBlock2CCInputs} hold and
\begin{IEEEeqnarray}{l}
B k + \tau \bigl\lceil k \log |\setU| + B \log (1 + k) \, |\setS| \bigr\rceil n_{\textnormal{bit}} \leq n. \label{eq:chUsesBlocks1ThroughBPl1SmNCCInputs}
\end{IEEEeqnarray}
As we argue next, when $n$ is sufficiently large we can choose
\begin{subequations}\label{bl:defBkCCInputs}
\begin{IEEEeqnarray}{rCl}
B & = & \lfloor \sqrt n \rfloor - \tau \bigl\lceil \log |\setU| + \log (1 + \sqrt n) \, |\setS| \bigr\rceil n_{\textnormal{bit}}, \\
k & = & \lfloor \sqrt n \rfloor, \label{eq:defBkkCCInputs}
\end{IEEEeqnarray}
\end{subequations}
and we can choose any $\epsilon > 0$ for which
\begin{IEEEeqnarray}{l}
R + \epsilon < \min_{P_S} \max_{\substack{P_{U,X|S} \colon \\ \Ex {}{\gamma (X)} \leq \Gamma}} \min_{ \substack{P_{Y|U,X,S} \colon \\ P_{Y|U = u, X,S} \in \mathscr P (W), \,\, \forall \, u \in \setU}} I (U;Y) - I (S;Y). \label{eq:defEpsCCInputs}
\end{IEEEeqnarray}
Note that, whenever $n$ is sufficiently large, $B$ is positive and \eqref{eq:chUsesBlocks1ThroughBPl1SmNCCInputs} is satisfied. To see that also \eqref{eq:choiceBKEpsForLastBlock2K} and \eqref{eq:choiceBKEpsForLastBlock2CCInputs} hold whenever $n$ is sufficiently large, we first observe from \eqref{eq:defBkkCCInputs} that $k$ tends to infinity as $n$ tends to infinity. This implies that \eqref{eq:choiceBKEpsForLastBlock2K} holds whenever $n$ is sufficiently large, and that $\delta (\epsilon, k)$ (which is defined in \eqref{eq:defDeltaOfEpsK}, where $\gamma_k = \gamma_k \bigl( |\setU|, |\setX|, |\setS|, |\setY| \bigr)$ converges to zero as $k$ tends to infinity) converges to $\epsilon$ as $n$ tends to infinity. We next observe that \eqref{bl:defBkCCInputs} implies that $B k / n$ converges to one as $n$ tends to infinity. This, combined with the fact that $\delta (\epsilon, k)$ converges to $\epsilon$ as $n$ tends to infinity and with \eqref{eq:defEpsCCInputs}, implies that \eqref{eq:choiceBKEpsForLastBlock2CCInputs} holds whenever $n$ is sufficiently large.
\end{proof}

\begin{proof}[Converse Part of Theorem~\ref{th:ccInputs}]
From Lemma~\ref{le:cardUCCInputs} it follows that it suffices to establish the converse part of Theorem~\ref{th:ccInputs} for the case where $\setU$ is any finite set. The converse is similar to that of Theorem~\ref{th:capacity}. Fix a finite set $\setM$, a blocklength~$n$, and an $(n,\setM)$ zero-error code with $n$ encoding mappings
\begin{IEEEeqnarray}{l}
f_i \colon \setM \times \setS^n \times \setY^{i-1} \rightarrow \setX, \quad i \in [1:n] \label{eq:convCCInputsEnc}
\end{IEEEeqnarray}
and $|\setM|$ disjoint decoding sets $\setD_m \subseteq \setY^n, \,\, m \in \setM$, where the code is chosen so that, with probability one, the channel inputs $X^n$ satisfy the cost constraint \eqref{eq:ccInputs}. We will show that, for some chance variable $U$ of finite support $\setU$, the rate $\frac{1}{n} \log |\setM|$ of the code is upper-bounded by the RHS of \eqref{eq:ccCapacity}.

Draw $M$ uniformly over $\setM$, and denote its distribution $P_M$. Since the code is a zero-error code, and since, with probability one, the channel inputs $X^n$ satisfy the cost constraint \eqref{eq:ccInputs}, the following two hold:
\begin{subequations}\label{bl:ccInputsConvZeroDecErrAndCost}
\begin{IEEEeqnarray}{rCl}
\distof {Y^n \in \setD_M} & = & 1, \label{eq:ccInputsConvZeroDecErr} \\
\bigdistof {\gamma^{(n)} (X^n) \leq \Gamma} & = & 1, \label{eq:ccInputsConv}
\end{IEEEeqnarray}
\end{subequations}
where $\dist$ is the distribution \eqref{eq:capacityConvDist} of $(M,S^n,X^n,Y^n)$ induced by $P_M$, the state distribution $Q$, the encoding mappings \eqref{eq:convCCInputsEnc}, and the channel law $\channel y {x,s}$. As in the converse of Theorem~\ref{th:capacity}, fix any PMF $\tilde P_S$ on $\setS$ and any collection of $n$ conditional PMFs $\bigl\{ \tilde P_{Y_i|M,Y^{i-1},S^n_{i + 1},X_i,S_i} \bigr\}_{i \in [1:n]}$ satisfying \eqref{eq:convAbsContTrProb}. These PMFs induce the PMF $\tilde P_{M,S^n,X^n,Y^n}$ of \eqref{eq:convAbsContPMF} on $\setM \times \setS^n \times \setX^n \times \setY^n$. Since this PMF satisfies $\tilde P_{M,S^n,X^n,Y^n} \ll \dist$, \eqref{bl:ccInputsConvZeroDecErrAndCost} implies
\begin{subequations}\label{bl:ccInputsConvZeroDecErrAndCostImp}
\begin{IEEEeqnarray}{rCl}
\tilde P_{M, S^n, X^n, Y^n} [Y^n \in \setD_M] & = & 1, \label{eq:ccInputsConvZeroDecErrImp} \\
\tilde P_{M, S^n, X^n, Y^n} \bigl[ \gamma^{(n)} (X^n) \leq \Gamma \bigr] & = & 1. \label{eq:ccInputsConvImp}
\end{IEEEeqnarray}
\end{subequations}
Note that the latter \eqref{eq:ccInputsConvImp} implies that
\begin{IEEEeqnarray}{rCl}
\frac{1}{n} \sum_{i = 1}^n \bigEx {}{\gamma (X_i)} & \leq & \Gamma, \label{eq:ccInputsConvImp2}
\end{IEEEeqnarray}
where the expectation in the $i$-th summand is computed w.r.t.\ the PMF $\tilde P_{X_i}$ induced by $\tilde P_{M,S^n,X^n,Y^n}$.

The line of argument leading to \eqref{eq:convFixedDists2} in the converse of Theorem~\ref{th:capacity} implies that every choice of $\tilde P_S$ and $\bigl\{ \tilde P_{Y_i|M,Y^{i-1},S^n_{i + 1},X_i,S_i} \bigr\}_{i \in [1:n]}$ gives rise to an upper bound
\begin{IEEEeqnarray}{l}
\frac{1}{n} \log |\setM| \leq \frac{1}{n} \sum^n_{i = 1} \Bigl[ I (U_i;Y_i) - I (U_i;S_i) \Bigr], \label{eq:convFixedDists2CCInputs}
\end{IEEEeqnarray}
where the chance variables $\{ U_i \}_{i \in [1:n]}$ are defined in \eqref{eq:convDefRVUi}, and the mutual informations in the $i$-th summand are computed w.r.t.\ the joint PMF $\tilde P_{U_i,X_i,S_i,Y_i}$ induced by $\tilde P_{M,S^n,X^n,Y^n}$. We next exhibit a PMF $\tilde P_S$ and a collection of conditional PMFs $\bigl\{ \tilde P_{Y_i|M,Y^{i-1},S^n_{i + 1},X_i,S_i} \bigr\}_{i \in [1:n]}$ satisfying \eqref{eq:convAbsContTrProb} for which each summand on the RHS of \eqref{eq:convFixedDists2CCInputs} is upper-bounded by the RHS of \eqref{eq:ccCapacity}.

We begin with the choice of $\bigl\{ \tilde P_{Y_i|M,Y^{i-1},S^n_{i + 1},X_i,S_i} \bigr\}_{i \in [1:n]}$. As in the converse of Theorem~\ref{th:capacity}, choosing a collection of conditional PMFs $\bigl\{ \tilde P_{Y_i|M,Y^{i-1},S^n_{i + 1},X_i,S_i} \bigr\}_{i \in [1:n]}$ that satisfy \eqref{eq:convAbsContTrProb} is tantamount to choosing a collection of conditional PMFs $\bigl\{ \tilde P_{Y_i|U_i,X_i,S_i} \bigr\}_{i \in [1:n]}$ that satisfy \eqref{eq:convDistYCondUXSAbsContTrProb}. We shall choose the latter collection, and we shall do so as in the converse of Theorem~\ref{th:capacity}. Consequently, our choice of $\tilde P_{Y_i|U_i,X_i,S_i}$ can depend not only on our choice of $\tilde P_S$ and our previous choices of $\bigl\{ \tilde P_{Y_j|U_j,X_j,S_j} \bigr\}_{j \in [1:i-1]}$ but also on $\tilde P_{U_i,X_i,S_i}$, and hence we can choose $\tilde P_{Y_i|U_i,X_i,S_i}$ as one that---among all conditional PMFs satisfying \eqref{eq:convDistYCondUXSAbsContTrProb}---minimizes
\begin{equation}
I ( U_i ; Y_i ) - I ( U_i ; S_i ),
\end{equation}
where the mutual informations are computed w.r.t.\ the joint PMF $\tilde P_{U_i,X_i,S_i} \times \tilde P_{Y_i|U_i,X_i,S_i}$. Because \eqref{eq:convAbsContPMF} implies that
\begin{IEEEeqnarray}{l}
\tilde P_{S_i} = \tilde P_{S}, \quad i \in [1:n] \label{eq:convDistSiCCInputs}
\end{IEEEeqnarray}
and by \eqref{eq:ccInputsConvImp2}, which holds because the chosen conditional PMFs $\bigl\{ \tilde P_{Y_i|U_i,X_i,S_i} \bigr\}_{i \in [1:n]}$ satisfy \eqref{eq:convDistYCondUXSAbsContTrProb}, we find that, for our choice of $\bigl\{ \tilde P_{Y_i|U_i,X_i,S_i} \bigr\}_{i \in [1:n]}$,
\begin{IEEEeqnarray}{l}
\frac{1}{n} \sum^n_{i = 1} \Bigl[ I ( U_i ; Y_i ) - I ( U_i ; S_i ) \Bigr] \nonumber \\
\quad \leq \max_{\substack{\{ \tilde P_{U_i,X_i|S_i} \}_{i \in [1:n]} \colon \\ \frac{1}{n} \sum_{i = 1}^n \Ex {}{\gamma (X_i)} \leq \Gamma}} \min_{\substack{ \{ \tilde P_{Y_i|U_i,X_i,S_i} \}_{i \in [1:n]} \colon \\ \tilde P_{Y_i|U_i = u_i,X_i,S_i} \in \mathscr P (W), \,\, \forall \, u_i \in \setU_i }} \frac{1}{n} \sum^n_{i = 1}  \Bigl[ I ( U_i ; Y_i ) - I ( U_i ; S_i ) \Bigr], \label{eq:convChoiceCondPMFsGuaranteeCCInputs}
\end{IEEEeqnarray}
where the expectation in the $i$-th summand is computed w.r.t.\ the joint PMF $\tilde P_{S_i} \times \tilde P_{U_i,X_i|S_i}$, and the mutual informations in the $i$-th summand are computed w.r.t.\ the joint PMF $\tilde P_{S_i} \times \tilde P_{U_i,X_i|S_i} \times \tilde P_{Y_i|U_i,X_i,S_i}$. The chosen conditional PMFs $\bigl\{ \tilde P_{Y_i|U_i,X_i,S_i} \bigr\}_{i \in [1:n]}$ satisfy \eqref{eq:convDistYCondUXSAbsContTrProb}, and hence \eqref{eq:convFixedDists2CCInputs} and \eqref{eq:convChoiceCondPMFsGuaranteeCCInputs} imply that
\begin{IEEEeqnarray}{l}
\frac{1}{n} \log |\setM| \nonumber \\
\quad \leq \max_{\substack{\{ \tilde P_{U_i,X_i|S_i} \}_{i \in [1:n]} \colon \\ \frac{1}{n} \sum_{i = 1}^n \Ex {}{\gamma (X_i)} \leq \Gamma}} \min_{\substack{ \{ \tilde P_{Y_i|U_i,X_i,S_i} \}_{i \in [1:n]} \colon \\ \tilde P_{Y_i|U_i = u_i,X_i,S_i} \in \mathscr P (W), \,\, \forall \, u_i \in \setU_i }} \frac{1}{n} \sum^n_{i = 1} \Bigl[ I ( U_i ; Y_i ) - I ( U_i ; S_i ) \Bigr], \label{eq:convFixedDists3CCInputs}
\end{IEEEeqnarray}
where the expectation in the $i$-th summand is computed w.r.t.\ the joint PMF $\tilde P_{S_i} \times \tilde P_{U_i,X_i|S_i}$, and the mutual informations in the $i$-th summand are computed w.r.t.\ the joint PMF $\tilde P_{S_i} \times \tilde P_{U_i,X_i|S_i} \times \tilde P_{Y_i|U_i,X_i,S_i}$.

By the definition of $U_i$ \eqref{eq:convDefRVUi} the cardinality of the support $\setU_i$ of $U_i$ satisfies \eqref{eq:convSuppRVUi}. Consequently, \eqref{eq:convDistSiCCInputs} and \eqref{eq:convFixedDists3CCInputs} imply that
\begin{IEEEeqnarray}{l}
\frac{1}{n} \log |\setM| \nonumber \\
\quad \leq \max_{\substack{ \tilde P_{U,X|V,S} \colon \\ \Ex {}{\gamma (X)} \leq \Gamma}} \min_{\substack{ \tilde P_{Y|V,U,X,S} \colon \\ \tilde P_{Y|(V,U) = (i,u),X,S} \in \mathscr P (W), \,\, \forall \, (i,u) \in [1:n] \times \setU }} I ( U ; Y | V ) - I ( U ; S |V ) \\
\quad \leq \max_{\substack{ \tilde P_{U,X|V,S} \colon \\ \Ex {}{\gamma (X)} \leq \Gamma}} \min_{\substack{ \tilde P_{Y|V,U,X,S} \colon \\ \tilde P_{Y|(V,U) = (i,u),X,S} \in \mathscr P (W), \,\, \forall \, (i,u) \in [1:n] \times \setU }} I ( V, U ; Y ) - I ( V, U ; S ), \label{eq:convFixedDists4CCInputs}
\end{IEEEeqnarray}
where $V$ is a time-sharing random-variable that is drawn uniformly over $[1:n]$ and $U$ an auxiliary chance variable taking values in a finite set $\setU$; where the mutual informations are computed w.r.t.\ the joint PMF $\tilde P_S \times P_V \times \tilde P_{U,X|V,S} \times \tilde P_{Y|V,U,X,S}$; and where the second inequality holds because conditioning cannot increase entropy, and because $S$ and $V$ are independent under $\tilde P_S \times P_V \times \tilde P_{U,X|V,S} \times \tilde P_{Y|V,U,X,S}$. By defining the auxiliary chance variable $\tilde U = (U,V)$, we obtain from \eqref{eq:convFixedDists4CCInputs} that every choice of $\tilde P_S$ gives rise to an upper bound
\begin{IEEEeqnarray}{l}
\frac{1}{n} \log |\setM| \nonumber \\
\quad \leq \max_{\substack{ \tilde P_{\tilde U,X|S} \colon \\ \Ex {}{\gamma (X)} \leq \Gamma}} \min_{\substack{ \tilde P_{Y|\tilde U,X,S} \colon \\ \tilde P_{Y|\tilde U = \tilde u,X,S} \in \mathscr P (W), \,\, \forall \, \tilde u \in \tilde \setU }} I ( \tilde U ; Y ) - I ( \tilde U ; S ), \label{eq:convFixedDists5CCInputs}
\end{IEEEeqnarray}
where $\tilde U$ is an auxiliary chance variable taking values in a finite set $\tilde \setU$, and the mutual informations are computed w.r.t.\ the joint PMF $\tilde P_S \times \tilde P_{\tilde U,X|S} \times \tilde P_{Y|\tilde U,X,S}$.

Having established \eqref{eq:convFixedDists5CCInputs}, we are now ready to conclude the proof of the converse. Since we can choose any PMF $\tilde P_S$ on $\setS$, we can choose one that---among all PMFs on $\setS$---yields the tightest bound, i.e., minimizes
\begin{IEEEeqnarray}{l}
\max_{\substack{ \tilde P_{\tilde U,X|S} \colon \\ \Ex {}{\gamma (X)} \leq \Gamma}} \min_{\substack{ \tilde P_{Y|\tilde U,X,S} \colon \\ \tilde P_{Y|\tilde U = \tilde u,X,S} \in \mathscr P (W), \,\, \forall \, \tilde u \in \tilde \setU }} I ( \tilde U ; Y ) - I ( \tilde U ; S ),
\end{IEEEeqnarray}
where $\tilde U$ is an auxiliary chance variable taking values in a finite set $\tilde \setU$, and the mutual informations are computed w.r.t.\ the joint PMF $\tilde P_S \times \tilde P_{\tilde U,X|S} \times \tilde P_{Y|\tilde U,X,S}$. For this choice of $\tilde P_S$ \eqref{eq:convFixedDists5CCInputs} implies that
\begin{IEEEeqnarray}{l}
\frac{1}{n} \log |\setM| \nonumber \\
\quad \leq \min_{\tilde P_S} \max_{\substack{ \tilde P_{\tilde U,X|S} \colon \\ \Ex {}{\gamma (X)} \leq \Gamma}} \min_{\substack{ \tilde P_{Y|\tilde U,X,S} \colon \\ \tilde P_{Y|\tilde U = \tilde u,X,S} \in \mathscr P (W), \,\, \forall \, \tilde u \in \tilde \setU }} I ( \tilde U ; Y ) - I ( \tilde U ; S ), \label{eq:convFixedDists6CCInputs}
\end{IEEEeqnarray}
where $\tilde U$ is an auxiliary chance variable taking values in a finite set $\tilde \setU$, and the mutual informations are computed w.r.t.\ the joint PMF $\tilde P_S \times \tilde P_{\tilde U,X|S} \times \tilde P_{Y|\tilde U,X,S}$.
\end{proof}

\section{A Proof of Remark~\ref{re:ccShannonSmallerAhlswede}}\label{sec:pfReCcShannonSmallerAhlswede}

\begin{proof}
Fix some PMF $P_X$ on $\setX$, and define the function
\begin{IEEEeqnarray}{rCl}
\rho \colon \mathscr P (W) & \rightarrow & \reals^+_0 \nonumber \\*[-0.5\normalbaselineskip]
  \label{eq:ccShannonSmallerAhlswedeDefR}
\\*[-0.5\normalbaselineskip]
V & \mapsto & I (P_X,V).\nonumber
\end{IEEEeqnarray}
To prove \eqref{eq:ccShannonSmallerAhlswede}, we will show that every $V$ that minimizes $\rho (\cdot)$ satisfies
\begin{IEEEeqnarray}{l}
\rho (V) \geq \min_{ y \in \setY } - \log \sum_{x \in \setX \colon \channel y x > 0 } P_X (x). \label{eq:ccShannonSmallerAhlswedePfIneq}
\end{IEEEeqnarray}
From this we will then obtain \eqref{eq:ccShannonSmallerAhlswede} by maximizing both sides over all choices of $P_X$ for which $\bigEx {}{\gamma (X)} \leq \Gamma$. This will conclude the proof of Remark~\ref{re:ccShannonSmallerAhlswede}, because Example~\ref{ex:ccShannonStrictlySmallerAhlswede} demonstrates that Inequality~\eqref{eq:ccShannonSmallerAhlswede} can be strict.

To show that every minimizer of $\rho (\cdot)$ satisfies \eqref{eq:ccShannonSmallerAhlswedePfIneq}, we first establish that $V  \in \mathscr P (W)$ minimizes $\rho (\cdot)$ only if
\begin{IEEEeqnarray}{l}
\!\!\!\!\!\!\!\! \frac{V (y|x)}{( P_X V ) (y)} = \frac{V (y^\prime|x)}{( P_X V ) (y^\prime)}, \,\, \Bigl( \forall \, (x,y,y^\prime) \in \setX \times \setY \times \setY \textnormal{ s.t.\ } V (y|x) \, V (y^\prime|x) > 0 \Bigr). \label{eq:partialVYXRMinimizerNec}
\end{IEEEeqnarray}
We prove the contrapositive: we show that if for some $V \in \mathscr P (W)$
\begin{IEEEeqnarray}{l}
\exists \, (x,y,y^\prime) \in \setX \times \setY \times \setY \textnormal{ s.t.\ } \nonumber \\*[-0.5\normalbaselineskip]
  \label{eq:partialVYXRMinimizerNecCP}
\\*[-0.5\normalbaselineskip]
\qquad V (y|x) \, V (y^\prime|x) > 0 \quad \textnormal{and} \quad \frac{V (y|x)}{( P_X V ) (y)} < \frac{V (y^\prime|x)}{( P_X V ) (y^\prime)}, \nonumber
\end{IEEEeqnarray}
then $V$ cannot be a minimizer of $\rho (\cdot)$. Our proof is by contradiction. To reach a contradiction, suppose that $V \in \mathscr P (W)$ minimizes $\rho (\cdot)$ and \eqref{eq:partialVYXRMinimizerNecCP} holds. Since
\begin{equation}
I (P_X,V) = \sum_{( x, y) \in \setX \times \setY} P_X (x) \, V (y|x) \log \frac{V (y|x)}{( P_X V ) (y)}, \label{eq:ccShannonSmallerAhlswedeRMuti}
\end{equation}
it follows that for every $(x,y) \in \setX \times \setY$
\begin{IEEEeqnarray}{l}
\frac{\partial \rho}{\partial V (y|x)} = P_X (x) \log \frac{V (y|x)}{( P_X V ) (y)}. \label{eq:partialVYXR}
\end{IEEEeqnarray}
This and \eqref{eq:partialVYXRMinimizerNecCP} imply that for all sufficiently-small $\delta$, and a fortiori for some $\delta$ satisfying
\begin{equation}
0 < \delta \leq \bigl( 1 - V (y|x) \bigr) \wedge V (y^\prime|x), \label{eq:ccShannonSmallerAhlswedeDefDelta}
\end{equation}
$\rho (\cdot)$ decreases when we replace $V (y|x)$ by $V (y|x) + \delta$ and $V (y^\prime|x)$ by $V (y^\prime|x) - \delta$. This contradicts our assumption that $V$ minimizes $\rho (\cdot)$, because replacing $V (y|x)$ by $V (y|x) + \delta$ and $V (y^\prime|x)$ by $V (y^\prime|x) - \delta$ yields some transition law $V^\prime$ in $\mathscr P (W)$. (The transition law $V^\prime$ is in $\mathscr P (W)$, because $V \in \mathscr P (W)$ and by \eqref{eq:ccShannonSmallerAhlswedeDefDelta}.) This contradiction proves that \eqref{eq:partialVYXRMinimizerNec} is a necessary condition for $V \in \mathscr P (W)$ to minimize $\rho (\cdot)$.

Having proved the necessity of \eqref{eq:partialVYXRMinimizerNec}, we are now ready to establish \eqref{eq:ccShannonSmallerAhlswedePfIneq}. To that end let
\begin{equation}
\gamma = \max_{y \in \setY} \sum_{x \in \setX \colon \channel y x > 0 } P_X (x), \label{eq:ccShannonSmallerAhlswedeDefH}
\end{equation}
and fix some transition law $V \in \mathscr P (W)$ that minimizes $\rho (\cdot)$ and for which \eqref{eq:partialVYXRMinimizerNec} hence holds. By \eqref{eq:partialVYXRMinimizerNec} there exist $\{ \alpha_x \}_{x \in \setX}$ satisfying that, whenever $V (y|x) > 0$,
\begin{IEEEeqnarray}{l}
\alpha_x = \frac{V (y|x)}{( P_X V ) (y)}. \label{eq:ccShannonSmallerAhlswedeDefAlphaX}
\end{IEEEeqnarray}
Consequently,
\begin{IEEEeqnarray}{rCl}
\rho (V) & = & \sum_{x \in \setX} P_X (x) \sum_{y \in \setY \colon V (y|x) > 0} (P_X V) (y) \, \alpha_x \log \alpha_x \\
& = & \sum_{y \in \setY} (P_X V) (y) \sum_{x \in \setX \colon V (y|x) > 0} P_X (x) \, \alpha_x \log \alpha_x \\
& \stackrel{(a)}\geq & \sum_{y \in \setY} (P_X V) (y) \Biggl( \sum_{x^{\prime \prime} \in \setX \colon V (y|x^{\prime \prime}) > 0} P_X (x^{\prime \prime}) \Biggr) \nonumber \\
& & \times \frac{\sum_{x \in \setX \colon V (y|x) > 0} P_X (x) \, \alpha_x}{\sum_{x^{\prime \prime} \in \setX \colon V (y|x^{\prime \prime}) > 0} P_X (x^{\prime \prime})} \log \frac{\sum_{x \in \setX \colon V (y|x) > 0} P_X (x) \, \alpha_x}{\sum_{x^{\prime \prime} \in \setX \colon V (y|x^{\prime \prime}) > 0} P_X (x^{\prime \prime})} \\
& \stackrel{(b)}= & \sum_{y \in \setY} (P_X V) (y) \log \frac{1}{\sum_{x^{\prime \prime} \in \setX \colon V (y|x^{\prime \prime}) > 0} P_X (x^{\prime \prime})} \\
& \stackrel{(c)}\geq & - \sum_{y \in \setY} (P_X V) (y) \log \gamma \\
& = & - \log \gamma \\
& \stackrel{(d)} = & \min_{ y \in \setY } - \log \sum_{x \in \setX \colon \channel y x > 0 } P_X (x), \label{eq:ccShannnonSmallerAhlswedeFinalIneq}
\end{IEEEeqnarray}
where $(a)$ holds because the function $$\xi \mapsto \xi \log \xi, \quad \xi \in \reals^+$$ is convex; $(b)$ holds because \eqref{eq:ccShannonSmallerAhlswedeDefAlphaX} holds whenever $V (y|x) > 0$; $(c)$ holds because $V \in \mathscr P (W)$ and \eqref{eq:ccShannonSmallerAhlswedeDefH} combine to imply that
\begin{IEEEeqnarray}{l}
\sum_{x \in \setX \colon V (y|x) > 0} P_X (x) \leq \sum_{x \in \setX \colon \channel y x > 0} P_X (x) \leq \gamma, \quad y \in \setY;
\end{IEEEeqnarray}
and $(d)$ holds by \eqref{eq:ccShannonSmallerAhlswedeDefH}. Inequality~\eqref{eq:ccShannnonSmallerAhlswedeFinalIneq} concludes the proof of \eqref{eq:ccShannonSmallerAhlswedePfIneq}.
\end{proof}

\section{Analysis of Example~\ref{ex:stuckAtOne}}\label{sec:exStuckAtOne}

For the SD-DMC $\channel y {x,s}$ of Example~\ref{ex:stuckAtOne} we show that, subject to the cost constraint \eqref{eq:ccStates1} with $\Gamma > 0$ satisfying \eqref{eq:exStuckAtOneLambda}, the zero-error capacity with acausal SI is positive. Given some blocklength~$n$, some message set $\setM$, and some encoding mapping
\begin{equation}
f \colon \setM \times \setS^n \rightarrow \setX^n, \label{eq:encExStuckAtOne}
\end{equation}
let $\vecy (m, \vecs)$ denote the output sequence that is produced when the transmitter uses the encoding mapping \eqref{eq:encExStuckAtOne} to convey Message~$m$ and the channel-state sequence is $\vecs$. Let $\setS^n (\Lambda)$ denote the set of $n$-length state-sequences of highest allowed cost
\begin{IEEEeqnarray}{l}
\setS^n (\Lambda) = \bigl\{ \vecs \in \setS^n \colon n \lambda^{(n)} (\vecs) = \left\lfloor \Lambda n \right\rfloor \bigr\}.
\end{IEEEeqnarray}

We begin with the following two observations: 1) From Table~\ref{tb:exStuckAtOne} we see that, if
\begin{equation}
\bigl\{ \tilde \vecy (m, \vecs) \bigr\}_{(m,\vecs) \in \setM \times \setS^n} \subseteq \setY^n \label{eq:exStuckAtOneNTuplesAllS}
\end{equation}
is such that for every $\vecs \in \setS^n$
\begin{IEEEeqnarray}{l}
\biggl( \Bigl( s_i = 1 \Bigr) \implies \Bigl( \tilde y_i (m,\vecs) = 1 \Bigr) \biggr), \,\, \forall \, (m,i) \in \setM \times [1:n], \label{eq:exStuckAtOneSiOneYiOne}
\end{IEEEeqnarray}
then there exists an encoding mapping $f$ of the form \eqref{eq:encExStuckAtOne} for which
\begin{equation}
\vecy (m, \vecs) = \tilde \vecy (m, \vecs), \,\, \forall \, (m,\vecs) \in \setM \times \setS^n.
\end{equation}
2) From the definition of $\setS^n (\Lambda)$ it follows that, if $\vecs \in \setS^n$ is such that $n \lambda^{(n)} (\vecs) < \left\lfloor \Lambda n \right\rfloor$, then there exists some $\vecs^\prime \in \setS^n (\Lambda)$ satisfying
\begin{IEEEeqnarray}{l}
\biggl( \Bigl( s_i = 1 \Bigr) \implies \Bigl( s_i^\prime = 1 \Bigr) \biggr), \,\, \forall \, i \in [1:n]. \label{eq:exStuckAtOneSiPrimeOneSiOne}
\end{IEEEeqnarray}
For such $\vecs^\prime$, any binary $n$-tuple $\tilde \vecy (m, \vecs^\prime)$ satisfying
\begin{equation}
\biggl( ( s_i^\prime = 1 ) \implies \Bigl( \tilde y_i (m,\vecs^\prime) = 1 \Bigr) \biggr), \,\, \forall \, i \in [1:n] \label{eq:exStuckAtOneSiOneYiOneTildeSY}
\end{equation}
also satisfies
\begin{IEEEeqnarray}{l}
\biggl( \Bigl( s_i = 1 \Bigr) \implies \Bigl( \tilde y_i (m,\vecs^\prime) = 1 \Bigr) \biggr), \,\, \forall \, i \in [1:n]. \label{eq:exStuckAtOneSiOneYiOneTildeY}
\end{IEEEeqnarray}
These two observations imply that to every collection
\begin{equation}
\bigl\{ \tilde \vecy (m, \vecs) \bigr\}_{(m, \vecs) \in \setM \times \setS^n (\Lambda)} \subseteq \setY^n \label{eq:exStuckAtOneNTuples}
\end{equation}
that satisfies \eqref{eq:exStuckAtOneSiOneYiOne} for every $\vecs \in \setS^n (\Lambda)$ there corresponds an encoding mapping of the form \eqref{eq:encExStuckAtOne} for which: 1) for every $\vecs \in \setS^n (\Lambda)$
\begin{subequations}\label{bl:exStuckAtOneYSi}
\begin{equation}
\vecy (m, \vecs) = \tilde \vecy (m, \vecs), \,\, \forall \, m \in \setM; \label{eq:exStuckAtOneYSiLargeCost}
\end{equation}
and 2) for every $\vecs \in \setS^n$ for which $n \lambda^{(n)} (\vecs) < \left\lfloor \Lambda n \right\rfloor$
\begin{IEEEeqnarray}{l}
\exists \, \vecs^\prime \in \setS^n (\Lambda) \textnormal{ s.t.\ } \Bigl( \vecy (m, \vecs) = \tilde \vecy (m , \vecs^\prime), \,\, \forall \, m \in \setM \Bigr). \label{eq:exStuckAtOneYSiPrimeSi}
\end{IEEEeqnarray}
\end{subequations}
The state sequence $\vecs \in \setS^n$ satisfies the cost constraint \eqref{eq:ccStates1} if $n \lambda^{(n)} (\vecs) \leq \left\lfloor \Lambda n \right\rfloor$. Consequently, if the collection in \eqref{eq:exStuckAtOneNTuples}---in addition to satisfying \eqref{eq:exStuckAtOneSiOneYiOne} for every $\vecs \in \setS^n (\Lambda)$---also satisfies that
\begin{IEEEeqnarray}{l}
\biggl( \Bigl( (m,\vecs) \neq (m^\prime,\vecs^\prime) \Bigr) \implies \Bigl( \tilde \vecy (m, \vecs) \neq \tilde \vecy (m^\prime, \vecs^\prime) \Bigr) \biggr), \nonumber \\
\quad \,\, \forall \, (m,\vecs), \, (m^\prime, \vecs^\prime) \in \setM \times \setS^n (\Lambda),
\end{IEEEeqnarray}
then we obtain from \eqref{bl:exStuckAtOneYSi} that the encoding mapping $f$ corresponding to the collection and the decoding sets
\begin{IEEEeqnarray}{l}
\setD_m = \bigcup_{\vecs \in \setS^n (\Lambda)} \bigl\{ \tilde \vecy (m, \vecs) \bigr\}, \quad m \in \setM
\end{IEEEeqnarray}
constitute an $(n, \setM)$ zero-error code for our channel under the cost constraint \eqref{eq:ccStates1}.

To show that under the cost constraint \eqref{eq:ccStates1} the zero-error capacity with acausal SI is positive, it thus suffices to exhibit some positive rate $R > 0$ for which for every sufficiently-large $n$ there exists some finite set $\setM$ of cardinality $|\setM| \geq 2^{n R}$ and some collection of $|\setM| \, |\setS^n (\Lambda)|$ distinct binary $n$-tuples \eqref{eq:exStuckAtOneNTuples} that satisfies \eqref{eq:exStuckAtOneSiOneYiOne} for every $\vecs \in \setS^n (\Lambda)$.

To that end we first note that the cardinality of $\setS^n (\Lambda)$ is upper-bounded by
\begin{IEEEeqnarray}{l} 
\bigl| \setS^n (\Lambda) \bigr| = {n \choose \left\lfloor \Lambda n \right\rfloor} \leq 2^{n h_{\textnormal b} (\Lambda)}, \label{eq:exStuckAtOneNoStateSequences}
\end{IEEEeqnarray}
where we used the inequality $h_{\textnormal b} \bigl( \left\lfloor \Lambda n \right\rfloor / n \bigr) \leq h_{\textnormal b} (\Lambda)$ (which holds because $\Lambda < 1/2$). We also note that for every state sequence $\vecs \in \setS^n (\Lambda)$ there exist $2^{\lceil n (1 - \Lambda) \rceil}$ binary $n$-tuples $\tilde \vecy$ that satisfy
\begin{IEEEeqnarray}{l}
\biggl( \Bigl( s_i = 1 \Bigr) \implies \Bigl( \tilde y_i = 1 \Bigr) \biggr), \,\, \forall \, i \in [1:n]. \label{eq:exStuckAtOneSiOneYiOneNormalY}
\end{IEEEeqnarray}
We now construct a collection of $|\setM| \, |\setS^n (\Lambda)|$ distinct binary $n$-tuples \eqref{eq:exStuckAtOneNTuples} that satisfies \eqref{eq:exStuckAtOneSiOneYiOne} for every $\vecs \in \setS^n (\Lambda)$ as follows. We sequentially allocate to each pair $(m,\vecs) \in \setM \times \setS^n (\Lambda)$ some $\tilde \vecy (m,\vecs)$ from the binary $n$-tuples $\tilde \vecy$ that satisfy \eqref{eq:exStuckAtOneSiOneYiOneNormalY} and that have not yet been allocated to some other pair $(m^\prime,\vecs^\prime)$. There are $| \setM | \, \bigl| \setS^n (\Lambda) \bigr| - 1$ such other pairs $(m^\prime,\vecs^\prime)$ to which we may or may not have allocated some $\tilde \vecy (m^\prime,\vecs^\prime)$ yet, and there are at least $2^{\lceil n (1 - \Lambda) \rceil}$ binary $n$-tuples $\tilde \vecy$ that satisfy \eqref{eq:exStuckAtOneSiOneYiOneNormalY}. Consequently, at most $| \setM | \, \bigl| \setS^n (\Lambda) \bigr| - 1$ binary $n$-tuples could have already been allocated, and if
\begin{equation}
| \setM | \, \bigl| \setS^n (\Lambda) \bigr| - 1 < 2^{\lceil n (1 - \Lambda) \rceil}, \label{eq:exStuckAtOneOutputSequenceExists}
\end{equation}
then there is at least one binary $n$-tuples $\tilde \vecy$ that satisfies \eqref{eq:exStuckAtOneSiOneYiOneNormalY} and that has not been allocated yet. Hence, if \eqref{eq:exStuckAtOneOutputSequenceExists} holds, then our construction produces a collection of $|\setM| \, |\setS^n (\Lambda)|$ distinct binary $n$-tuples \eqref{eq:exStuckAtOneNTuples} that satisfies \eqref{eq:exStuckAtOneSiOneYiOne} for every $\vecs \in \setS^n (\Lambda)$. From \eqref{eq:exStuckAtOneNoStateSequences} we obtain that \eqref{eq:exStuckAtOneOutputSequenceExists} holds whenever
\begin{equation}
| \setM | \leq 2^{n (1 - \Lambda - h_{\textnormal b} (\Lambda))},
\end{equation}
and hence every positive rate $R > 0$ satisfying
\begin{equation}
R \leq 1 - \Lambda - h_{\textnormal b} (\Lambda)
\end{equation}
is achievable. This, combined with \eqref{eq:exStuckAtOneLambda}, implies that under the cost constraint \eqref{eq:ccStates1} the zero-error capacity with acausal SI is positive.

\section{A Proof of Theorem~\ref{th:ccStates2Capacity}}\label{sec:pfThCcStates2Capacity}

Lemma~\ref{le:cardU} in Appendix~\ref{sec:leCardU} implies that restricting $X$ to be a function of $U$ and $S$, i.e., $P_{U,X|S}$ to have the form \eqref{eq:thCapCardUXFuncUS}, does not change the RHS of \eqref{eq:capacityCCStates}, nor does restricting the cardinality of $\setU$ to \eqref{eq:cardU}. To prove Theorem~\ref{th:ccStates2Capacity} it thus suffices to establish a direct part for the case where $\setU$ is restricted to \eqref{eq:cardU} and a converse part for the case where $\setU$ is any finite set. We first establish the direct part.

\begin{proof}[Direct Part]
We assume that \eqref{eq:positiveCCStates} holds and show that the RHS of \eqref{eq:capacityCCStates} is achievable. The necessity of \eqref{eq:positiveCCStates} is part of the converse. If the RHS of \eqref{eq:capacityCCStates} is zero, then there is nothing to prove, so we assume that it is positive. To prove that the RHS of \eqref{eq:capacityCCStates} is achievable, we shall show that for every $l \in \naturals$ the RHS of \eqref{eq:capacityCCStates} is a lower bound for $C_{\textnormal{f},0}^{(2)} (\Lambda,l)$. To that end fix any $l \in \naturals$. The proof builds on the proofs of Remark~\ref{re:chUsesPerBit} and the direct part of Theorem~\ref{th:capacity}, adapting both to the state constraint \eqref{eq:ccStates2}. We partition the blocklength-$n$ transmission into $B + 2$ blocks, with each of the first $B$ blocks being of length~$k$, where $k$ is a multiple of $l$; with Block~$(B + 1)$ being of length $k^\prime$; and with Block~$(B + 2)$ being of length~$n - B k - k^\prime$. The only purpose of Block~$(B + 2)$ is to allow $B k + k^\prime$ to be smaller than $n$: in this block the encoder can thus transmit arbitrary inputs with the decoder ignoring the corresponding outputs. The choice we shall later make for $k$ and $k^\prime$ will be such that the last two blocks be of negligible length compared to $B k$ and therefore not affect the code's asymptotic rate.

Before the transmission begins, the encoder is revealed the realization $\vecs^{(b)} \triangleq S^{b k}_{(b-1) k + 1}$ of the Block-$b$ state-sequence for every $b \in [1:B]$ and the realization $\vecs^{(B + 1)} \triangleq S^{B k + k^\prime}_{B k + 1}$ of the Block-$(B + 1)$ state-sequence. In the first $B$ blocks our scheme draws on the scheme we used in the direct part of Theorem~\ref{th:capacity}. But instead of reducing the set of messages of positive posterior probability given the channel outputs, in the present setting we consider pairs of messages and possible Block-$(B + 1)$ state-sequences, and each of the blocks~1 through $B$ reduces the set of such pairs that have a positive posterior probability given the channel outputs. For every $b \in [1:B]$ we thus adapt the Block~$b$ transmission as follows. Because $k$ is a multiple of $l$, the cost constraint \eqref{eq:ccStates2} implies that in the first $B$ blocks
\begin{subequations}\label{bl:ccStatesConstStatesBlocks}
\begin{IEEEeqnarray}{l}
\sum_{s \in \setS} P_{\vecs^{(b)}} (s) \, \lambda (S) \leq \Lambda, \,\, \forall \, b \in [1:B], \label{eq:ccStatesConstStatesBlockb}
\end{IEEEeqnarray}
and in Block~$(B + 1)$
\begin{IEEEeqnarray}{l}
\frac{1}{l} \sum^{j l}_{i = (j - 1) l + 1} \lambda \Bigl( s^{(B + 1)}_i \Bigr) \leq \Lambda, \,\, \Bigl( \forall \, j \in \naturals \textnormal{ s.t.\ } j l \leq k^\prime \Bigr). \label{eq:ccStatesConstStatesBlockBPl1}
\end{IEEEeqnarray}
\end{subequations}
Assume for now that the decoder---while incognizant of $\vecs^{(1)}, \ldots, \vecs^{(B)}$---knows the empirical types $P_{\vecs^{(1)}}, \ldots, P_{\vecs^{(B)}}$: Block~$(B + 1)$ will ensure that the scheme works even though the decoder is incognizant of these types. Let $\bm { \mathcal I}_0 \subseteq \setM \times \setS^{k^\prime}$ be the set of all possible pairs of message $m^\prime \in \setM$ and Block-$(B + 1)$ state-sequence $\vecs^\prime \in \setS^{k^\prime}$ satisfying \eqref{eq:ccStatesConstStatesBlockBPl1}, i.e.,
\begin{equation}
\frac{1}{l} \sum^{j l}_{i = (j - 1) l + 1} \lambda ( s^\prime_i ) \leq \Lambda, \,\, \Bigl( \forall \, j \in \naturals \textnormal{ s.t.\ } j l \leq k^\prime \Bigr), \label{eq:ccStatesConstStatesSetJB}
\end{equation}
and let $\bm { \mathcal I}_b$ be the post-Block-$b$ ambiguity-set, i.e., the (random) subset of $\bm { \mathcal I}_{b-1}$ comprising the elements in $\bm { \mathcal I}_{b-1}$ of positive posterior probability given the Block-$b$ outputs $\vecy^{(b)} \triangleq Y^{bk}_{(b-1)k + 1}$ and the empirical type $P_{\vecs^{(b)}}$. Choose some $k$-type $P_{U,X,S}^{(b)}$ whose $\setS$-marginal $P_S^{(b)}$ equals $P_{\vecs^{(b)}}$, fix some $\epsilon > 0$, and define $\Theta$ as in \eqref{eq:defTheta}. In the following, unless otherwise specified, all entropies and mutual informations are computed w.r.t.\ the joint PMF $P_{U,X,S}^{(b)}$. Unlike the scheme we used in the direct part of Theorem~\ref{th:capacity}, where it was the survivor set $\bm { \mathcal M}_{b-1}$ that was partitioned into $\Theta$ subsets, here it is the ambiguity set $\bm { \mathcal I}_{b-1}$ that is partitioned into $\Theta$ subsets. The arguments leading to \eqref{bl:noMessMbGivenOutputs2} in the direct part of Theorem~\ref{th:capacity} then imply that we can find a positive integer $\eta_0 = \eta_0 \bigl( |\setX|, |\setS|, \epsilon \bigr)$ that guarantees that, for every $k \geq \eta_0$,
\begin{subequations}\label{bl:ccStatesNoMessMbGivenOutputs}
\begin{IEEEeqnarray}{l}
| \bm { \mathcal I}_b | \leq \left( \max_{\substack{P_{Y|U,X,S} \colon \\ P_{Y|U = u, X,S} \in \mathscr P (W), \,\, \forall \, u \in \setU}} 2^{ - k ( I (U;Y) - I (U;S) - ( \epsilon + \beta_k ) ) } \right) \! |\bm { \mathcal I}_{b-1}|, \label{eq:ccStatesNoMessMbGivenOutputs}
\end{IEEEeqnarray}
whenever
\begin{IEEEeqnarray}{l}
|\bm { \mathcal I}_{b-1}| \geq 2^{k \log |\setU|}, \label{eq:ccStatesNoMessMbCondMbMin1}
\end{IEEEeqnarray}
\end{subequations}
where the mutual informations are computed w.r.t.\ the joint PMF $P^{(b)}_{U,X,S} \times P_{Y|U,X,S}$, and where $\beta_k$ is defined in \eqref{eq:gammaKPfCapacity} and hence converges to zero as $k$ tends to infinity.

Since we can choose any $k$-type $P_{U,X,S}^{(b)}$ whose $\setS$-marginal $P_S^{(b)}$ is $P_{\vecs^{(b)}}$, we can choose $P_{U,X,S}^{(b)} = P_{\vecs^{(b)}} \times P_{U,X|S}^{(b)}$, where $P_{U,X|S}^{(b)}$ is the conditional $k$-type that---among all conditional $k$-types---maximizes
\begin{IEEEeqnarray}{l}
\min_{\substack{P_{Y|U,X,S} \colon \\ P_{Y|U = u, X,S} \in \mathscr P (W), \,\, \forall \, u \in \setU}} I (U;Y) - I (U;S),
\end{IEEEeqnarray}
where the mutual informations are computed w.r.t.\ the joint PMF $P_{\vecs^{(b)}} \times P_{U,X|S}^{(b)} \times P_{Y|U,X,S}$. Every conditional PMF can be approximated in the total variation distance by a conditional $k$-type when $k$ is sufficiently large; and, because entropy and mutual information are continuous in this distance \cite[Lemma~2.7]{csiszarkoerner11}, it follows that---for the above choice of the conditional $k$-type and some $\gamma_k = \gamma_k (|\setU|,|\setX|,|\setS|,|\setY|)$, which converges to zero as $k$ tends to infinity---\eqref{eq:ccStatesConstStatesBlockb} and \eqref{bl:ccStatesNoMessMbGivenOutputs} imply that when $|\bm { \mathcal I}_{b-1}| \geq 2^{k \log |\setU|}$
\begin{IEEEeqnarray}{l}
| \bm { \mathcal I}_b | \leq \left( \max_{\substack{P_S \colon \\ \Ex {}{\lambda(S)} \leq \Lambda}} \min_{P_{U,X|S}} \max_{\substack{P_{Y|U,X,S} \colon \\ P_{Y|U = u, X,S} \in \mathscr P (W), \,\, \forall \, u \in \setU}} \!\!\!\!\! 2^{ - k ( I (U;Y) - I (U;S) - \epsilon + \gamma_k ) } \right) \! |\bm { \mathcal I}_{b-1}|, \label{eq:ccStatesNoMessMbGivenOutputs2}
\end{IEEEeqnarray}
where the mutual informations are computed w.r.t.\ the joint PMF $P_S \times P_{U,X|S} \times P_{Y|U,X,S}$. Because our scheme works for any $\epsilon > 0$, it follows that for every $\epsilon > 0$ and positive integer $k \geq \eta_0 (|\setX|,|\setS|,\epsilon)$ each of Blocks~1 through~$B$ is guaranteed to reduce the ambiguity set by a factor of at least
\begin{IEEEeqnarray}{l}
\max_{\substack{P_S \colon \\ \Ex {}{\lambda(S)} \leq \Lambda}} \min_{P_{U,X|S}} \max_{\substack{P_{Y|U,X,S} \colon \\ P_{Y|U = u, X,S} \in \mathscr P (W), \,\, \forall \, u \in \setU}} \!\!\!\!\! 2^{ - k ( I (U;Y) - I (U;S) - \delta (\epsilon,k) ) }, \label{eq:ccStatesNoMessMbGivenOutputs3}
\end{IEEEeqnarray}
until $|\bm { \mathcal I}_B|$ is smaller than $2^{k \log |\setU|}$. Here the mutual informations are computed w.r.t.\ the joint PMF $P_S \times P_{U,X|S} \times P_{Y|U,X,S}$, and $\delta (\epsilon, k)$ is defined in \eqref{eq:defDeltaOfEpsK} and hence converges to zero as $\epsilon$ tends to zero and $k$ to infinity.

Since we assume that the RHS of \eqref{eq:capacityCCStates} is positive; and, because $\delta (\epsilon, k)$ converges to zero as $\epsilon \downarrow 0$ and $k \rightarrow \infty$, it follows that we can choose $\epsilon$ sufficiently small and $B$ and $k$ sufficiently large so that
\begin{subequations}\label{bl:choiceBKEpsForLastBlockCCStates}
\begin{IEEEeqnarray}{l}
k \geq \eta_0 \bigl( |\setX|, |\setS|, \epsilon \bigr) \label{eq:choiceBKEpsForLastBlockKCCStates}
\end{IEEEeqnarray}
and
\begin{IEEEeqnarray}{l}
\!\! \left( \max_{\substack{P_S \colon \\ \Ex {}{\lambda (S)} \leq \Lambda}} \min_{P_{U,X|S}} \max_{ \substack{ P_{Y|U,X,S} \colon \\ P_{Y|U = u, X,S} \in \mathscr P (W) , \,\, \forall \, u \in \setU}} 2^{ - B k ( I (U;Y) - I (U;S) - \delta (\epsilon, k) ) } \right) \! |\setM| \, |\setS|^{k^\prime} \nonumber \\
\quad \leq 2^{k \log |\setU|}. \label{eq:choiceBKEpsForLastBlockCCStates}
\end{IEEEeqnarray}
\end{subequations}
This guarantees that
\begin{equation}
|\bm { \mathcal I}_B| \leq 2^{k \log |\setU|}, \label{eq:setIBSmUPowKCCStates}
\end{equation}
because each block reduces the ambiguity set by the factor in \eqref{eq:ccStatesNoMessMbGivenOutputs3} until $|\bm { \mathcal I}_B|$ is smaller than $2^{k \log |\setU|}$.

We now deal with Block~$(B + 1)$. Because the decoder is incognizant of the empirical types $\{ P_{\vecs^{(b)}} \}_{b \in [1:B]}$, it cannot compute the post-Block-$B$ ambiguity-set $\bm { \mathcal I}_B$. The uncertainty that needs to be addressed is about the message, the Block-$(B + 1)$ state-sequence, as well as the $B$ empirical types of $\vecs^{(1)}, \ldots, \vecs^{(B)}$. Let $\bm { \mathcal J}_B \subseteq \setM \times \setS^{k^\prime}$ denote the union of the post-Block-$B$ ambiguity-sets corresponding to all the different $B$-tuples of $k$-types on $\setS$, i.e., $\bm { \mathcal J}_B$ is the set of pairs of messages and possible Block-$(B + 1)$ state-sequences that have a positive posterior probability given only the outputs $\bigl\{ \vecy^{(b)} \bigr\}_{b \in [1:B]}$ (and not the $k$-types $\bigl\{ P_{\vecs^{(b)}} \bigr\}_{b \in [1:B]}$). Because the post-Block-$B$ ambiguity-set corresponding to any given $B$-tuple of $k$-types on $\setS$ satisfies \eqref{eq:setIBSmUPowKCCStates}, and because there are at most $(1 + k)^{B \, |\setS|}$ $B$-tuples of $k$-types on $\setS$
\begin{IEEEeqnarray}{l}
|\bm { \mathcal J}_B| \leq 2^{k \log |\setU| + B \log (1 + k) \, |\setS|}. \label{eq:ccStatesSizeSetJB}
\end{IEEEeqnarray}
In Block~$(B + 1)$ we resolve the set $\bm { \mathcal J}_B$. This will guarantee that the decoder can recover the transmitted message~$m$ error-free.

Block~$(B + 1)$ is similar to Phase~2 of the scheme we used to prove Remark~\ref{re:chUsesPerBit}: the encoder allocates to every pair $(m^\prime,\vecs^\prime) \in \bm { \mathcal J}_B$ a length-$k^\prime$ codeword $\vecx (m^\prime,\vecs^\prime)$, where the codewords are chosen so that
\begin{IEEEeqnarray}{l}
\Bigl( \forall \, (m^\prime,\vecs^\prime), \, (m^{\prime\prime},\vecs^{\prime\prime}) \in \bm { \mathcal J}_B \textnormal{ s.t.\ } m^\prime \neq m^{\prime\prime} \Bigr) \quad \exists \, i \in [1:k^\prime] \textnormal { s.t.\ } \nonumber \\*[-0.6\normalbaselineskip]
  \label{eq:condLastBlockCCStates}
\\*[-0.6\normalbaselineskip]
\qquad \Bigl( \bigchannel y {x_i (m^\prime,\vecs^\prime), s_i^\prime} \, \bigchannel y {x_i (m^{\prime\prime}, \vecs^{\prime\prime}), s_i^{\prime\prime}} = 0 , \,\, \forall \, y \in \setY \Bigr). \nonumber
\end{IEEEeqnarray}
(We shall shortly use a random coding argument to show that this can be done.) To convey the message~$m$, the encoder transmits in Block~$(B + 1)$ the codeword $\vecx \bigl( m,\vecs^{(B + 1)} \bigr)$. Condition~\eqref{eq:condLastBlockCCStates} implies that, upon observing the Block-$(B + 1)$ outputs $\vecy^{(B + 1)} \triangleq Y^{Bk + k^\prime}_{Bk + 1}$, the decoder, who knows $\bm { \mathcal J}_B$ and the codewords $\bigl\{ \vecx (m^\prime,\vecs^\prime) \bigr\}$, can determine the transmitted message~$m$ error-free, because, for the true realization $\vecs^{(B + 1)}$ of the Block-$(B + 1)$ state-sequence,
\begin{IEEEeqnarray}{l}
\prod^{k^\prime}_{i = 1} W \Bigl( y_i^{(B + 1)} \Bigl| x_i \bigl(m, \vecs^{(B + 1)} \bigr), s^{(B + 1)}_i \Bigr) > 0,
\end{IEEEeqnarray}
whereas \eqref{eq:condLastBlockCCStates} implies for $m^\prime \neq m$
\begin{IEEEeqnarray}{l}
\prod^{k^\prime}_{i = 1} W \Bigl( y_i^{(B + 1)} \Bigl| x_i (m^\prime, \tilde \vecs), \tilde s_i \Bigr) = 0, \,\, \Bigl( \forall \, \tilde \vecs \textnormal{ s.t.\ } (m^\prime,\tilde \vecs) \in \bm { \mathcal J}_B \Bigr).
\end{IEEEeqnarray}
The decoder can thus calculate $\prod_i W \bigl( y_i^{(B + 1)} \bigl| x_i (\tilde m, \tilde \vecs), \tilde s_i \bigr)$ for each $( \tilde m, \tilde \vecs ) \in \bm { \mathcal J}_B$ and produce the message~$\tilde m$ for which this product is positive for some $\tilde \vecs$ for which $(\tilde m, \tilde \vecs) \in \bm { \mathcal J}_B$.

We next show that, for some choice of $k^\prime$, there exist codewords $\bigl\{ \vecx (m^\prime,\vecs^\prime) \bigr\}$ satisfying \eqref{eq:condLastBlockCCStates}. To this end we use a random coding argument. Draw the length-$k^\prime$ codewords $\bigl\{ \rndvecX (m^\prime, \vecs^\prime) \bigr\}$ independently, each uniformly over $\setX^{k^\prime}$, and let
\begin{subequations}\label{bl:defKPrimePrimeAndMoreCCStates}
\begin{IEEEeqnarray}{rCl}
q & = & \Bigl\lfloor \frac{k^\prime}{l} \Bigr\rfloor, \label{eq:defQCCStates} \\
\lambda^\star & = & \min_{s, \, s^\prime \in \setS \colon \lambda (s) + \lambda (s^\prime) > 2 \Lambda} \label{eq:defLambdaStarCCStates} \frac{\lambda (s) + \lambda (s^\prime)}{2}, \\
\alpha & = & \frac{\lambda^\star - \Lambda}{\lambda^\star - \lambda_{\textnormal{min}}}, \label{eq:defAlphaCCStates} \\
k^{\prime\prime} & = & \left\lceil \alpha q l \right\rceil. \label{eq:defKPrimePrimeCCStates}
\end{IEEEeqnarray}
\end{subequations}
From the cost constraint \eqref{eq:ccStatesConstStatesSetJB} and the definition of $q$ \eqref{eq:defQCCStates} it follows that every pair of (not necessarily distinct) state sequences $\vecs^\prime, \, \vecs^{\prime\prime} \in \setS^{k^\prime}$ for which
\begin{IEEEeqnarray}{l}
\exists \, m^\prime, \, m^{\prime\prime} \in \setM \textnormal{ s.t.\ } (m^\prime, \vecs^\prime), \, (m^{\prime\prime}, \vecs^{\prime\prime}) \in \bm { \mathcal J}_B \label{eq:stateSequencesInGameCCStates}
\end{IEEEeqnarray}
satisfies
\begin{IEEEeqnarray}{l}
\frac{1}{q l} \sum^{q l}_{i = 1} \frac{\lambda (s^{\prime}_i) + \lambda (s^{\prime\prime}_i)}{2} \leq \Lambda. \label{eq:ccStatesOnQl}
\end{IEEEeqnarray}
As we argue next, \eqref{eq:ccStatesOnQl} can hold only if for all such $\vecs^\prime, \, \vecs^{\prime\prime}$ there exist at least $k^{\prime\prime}$ distinct epochs $\setL (\vecs^\prime, \vecs^{\prime\prime}) \subset [1:q l]$ for which
\begin{equation}
\frac{\lambda (s^{\prime}_\ell) + \lambda (s^{\prime\prime}_\ell)}{2} \leq \Lambda, \,\, \forall \, \ell \in \setL (\vecs^\prime, \vecs^{\prime\prime}). \label{eq:ccStatesSatisfiedAtJ}
\end{equation}
To simplify the typography, we shall refer to $\setL (\vecs^\prime, \vecs^{\prime\prime})$ as $\setL$. The claim can then be stated equivalently as
\begin{IEEEeqnarray}{l}
\exists \, \setL \subset [1:ql] \textnormal{ s.t.\ } \Bigl( |\setL| = k^{\prime\prime} \Bigr) \wedge \Biggl( \frac{\lambda \bigl( s^{\prime}_\ell \bigr) + \lambda \bigl( s^{\prime\prime}_\ell \bigr)}{2} \leq \Lambda, \,\, \forall \, \ell \in \setL \Biggr). \label{eq:ccStatesSatisfiedAtJMath}
\end{IEEEeqnarray}
To prove \eqref{eq:ccStatesSatisfiedAtJMath}, note that, by the definition of $\lambda^\star$ \eqref{eq:defLambdaStarCCStates},
\begin{equation}
\Biggl( \biggl( \frac{\lambda (s^{\prime}_i) + \lambda (s^{\prime\prime}_i)}{2} > \Lambda \biggr) \implies \biggl( \frac{\lambda (s^{\prime}_i) + \lambda (s^{\prime\prime}_i)}{2} \geq \lambda^\star \biggr) \Biggr), \,\, \forall \, i \in [1:k^\prime], \label{eq:ccStatesSatisfiedAtJPf1}
\end{equation}
and, because $\lambda (s) \geq \lambda_{\textnormal {min}}, \,\, s \in \setS$,
\begin{equation}
\frac{\lambda (s^{\prime}_i) + \lambda (s^{\prime\prime}_i)}{2} \geq \lambda_{\textnormal {min}}, \,\, \forall \, i \in [1:k^\prime]. \label{eq:ccStatesSatisfiedAtJPf2}
\end{equation}
The definitions of $\alpha$ \eqref{eq:defAlphaCCStates} and $k^{\prime\prime}$ \eqref{eq:defKPrimePrimeCCStates} combine with \eqref{eq:ccStatesSatisfiedAtJPf1} and \eqref{eq:ccStatesSatisfiedAtJPf2} to prove our claim that \eqref{eq:ccStatesSatisfiedAtJMath} holds for all $\vecs^\prime, \, \vecs^{\prime\prime} \in \setS^{k^\prime}$ satisfying \eqref{eq:stateSequencesInGameCCStates}. An immediate consequence of \eqref{eq:ccStatesSatisfiedAtJMath} and the assumption \eqref{eq:positiveCCStates} is that for all $\vecs^\prime, \, \vecs^{\prime\prime}$ satisfying \eqref{eq:stateSequencesInGameCCStates}.
\begin{IEEEeqnarray}{l}
\exists \, \setL \subset [1:ql] \textnormal{ s.t.\ } \Bigl( |\setL| = k^{\prime\prime} \Bigr) \wedge \biggl( \forall \, \ell \in \setL \quad \exists \, x^{\prime}, x^{\prime\prime} \in \setX \textnormal{ s.t.\ } \nonumber \\*[-0.55\normalbaselineskip]
  \label{eq:pfCCStatesCondKPrimePrime}
\\*[-0.55\normalbaselineskip]
\qquad \Bigl( \bigchannel y {x^{\prime},s^{\prime}_\ell} \, \bigchannel y {x^{\prime\prime},s^{\prime\prime}_\ell} = 0, \,\, \forall \, y \in \setY \Bigr) \biggr). \nonumber
\end{IEEEeqnarray}

Having established \eqref{eq:pfCCStatesCondKPrimePrime}, we are now ready to show that---for some choice of $k^\prime$---the probability that the random codewords $\bigl\{ \rndvecX (m^\prime,\vecs^\prime) \bigr\}$ satisfy \eqref{eq:condLastBlockCCStates} is positive. For every distinct $(m^\prime,\vecs^\prime), \, (m^{\prime\prime},\vecs^{\prime\prime}) \in \bm { \mathcal J}_B$
\begin{IEEEeqnarray}{l}
\Bigdistof {\forall \, i \in [1 : k^\prime] \,\, \exists \, y \in \setY \textnormal{ s.t.\ } W \bigl( y \bigl| X_i (m^\prime, \vecs^\prime), s^\prime_i \bigr) \, W \bigl( y \bigl| X_i (m^{\prime\prime}, \vecs^{\prime\prime}), s^{\prime\prime}_i \bigr) > 0} \nonumber \\
\quad \leq \biggl( 1 - \frac{1}{|\setX|^2} \biggr)^{\!\! k^{\prime\prime}} \\
\quad = 2^{- k^{\prime\prime} ( 2 \log |\setX| - \log ( |\setX|^2 - 1 ) ) },
\end{IEEEeqnarray}
where we used \eqref{eq:pfCCStatesCondKPrimePrime} and that $\rndvecX (m^\prime, \vecs^\prime)$ and $\rndvecX (m^{\prime\prime}, \vecs^{\prime\prime})$ are independent and uniform over $\setX^{k^\prime}$. This, the Union-of-Events bound, and \eqref{eq:ccStatesSizeSetJB} imply that the probability that the randomly drawn length-$k^\prime$ codewords do not satisfy \eqref{eq:condLastBlockCCStates} is upper-bounded by
\begin{IEEEeqnarray}{l}
| \bm { \mathcal J}_B |^2 \, 2^{- k^{\prime\prime} ( 2 \log |\setX| - \log ( |\setX|^2 - 1 ) ) } \nonumber \\
\quad \leq 2^{- k^{\prime\prime} ( 2 \log |\setX| - \log ( |\setX|^2 - 1 ) ) + 2 ( k \log |\setU| + B \log (1 + k) \, |\setS| )},
\end{IEEEeqnarray}
which is smaller than one whenever
\begin{IEEEeqnarray}{l}
k^{\prime\prime} > \frac{k \log |\setU| + B \log (1 + k) \, |\setS|}{\log |\setX| - \frac{1}{2} \log \bigl( |\setX|^2 - 1 \bigr)}. \label{eq:kPrimePrimeLBCCStates}
\end{IEEEeqnarray}
Consequently, \eqref{eq:defQCCStates} and \eqref{eq:defKPrimePrimeCCStates} imply that, if we choose
\begin{IEEEeqnarray}{l}
k^\prime = \Biggl( \biggl\lfloor \frac{k \log |\setU| + B \log (1 + k) \, |\setS|}{\alpha l \bigl( \log |\setX| - \frac{1}{2} \log ( |\setX|^2 - 1 ) \bigr)} \biggr\rfloor + 1 \Biggr) l, \label{eq:kPrimeLBCCStates}
\end{IEEEeqnarray}
then there exist length-$k^\prime$ codewords $\bigl\{ \vecx (m^\prime,\vecs^\prime) \bigr\}$ satisfying \eqref{eq:condLastBlockCCStates}.\\

We are now ready to join the dots and conclude that the coding scheme asymptotically achieves any rate smaller than the RHS of \eqref{eq:capacityCCStates}. More precisely, we will show that, for every rate $R$ smaller than the RHS of \eqref{eq:capacityCCStates} and every sufficiently-large blocklength~$n$, our coding scheme can convey $n R$ bits error-free in $n$ channel uses. It follows from \eqref{bl:choiceBKEpsForLastBlockCCStates} and \eqref{eq:kPrimeLBCCStates} that if the positive integers $n, \, B, \, k$ and $\epsilon > 0$ are such that $k$ is a multiple of $l$,
\begin{subequations}\label{bl:choiceBKEpsForLastBlockCCStates2}
\begin{IEEEeqnarray}{l}
k \geq \eta_0 \bigl( |\setX|, |\setS|, \epsilon \bigr), \label{eq:choiceBKEpsForLastBlockKCCStates2}
\end{IEEEeqnarray}
and
\begin{IEEEeqnarray}{l}
\!\!\!\!\!\!\!\! n R + \Biggl( \biggl\lfloor \frac{k \log |\setU| + B \log (1 + k) \, |\setS|}{\alpha l \bigl( \log |\setX| - \frac{1}{2} \log ( |\setX|^2 - 1 ) \bigr)} \biggr\rfloor + 1 \Biggr) l \log |\setS| \nonumber \\
\!\!\!\!\!\!\!\! \quad \leq B k \left( \min_{\substack{P_S \colon \\ \Ex {}{\lambda (S)} \leq \Lambda}} \max_{P_{U,X|S}} \min_{ \substack{P_{Y|U,X,S} \colon \\ P_{Y|U = u, X,S} \in \mathscr P (W), \,\, \forall \, u \in \setU}} \!\!\!\!\! I (U;Y) - I (S;Y) - \delta (\epsilon, k ) \right) \!\!, \label{eq:choiceBKEpsForLastBlockCCStates2}
\end{IEEEeqnarray}
\end{subequations}
then the first $B + 1$ blocks of our coding scheme can convey $n R$ bits error-free in
\begin{IEEEeqnarray}{l}
B k + \Biggl( \biggl\lfloor \frac{k \log |\setU| + B \log (1 + k) \, |\setS|}{\alpha l \bigl( \log |\setX| - \frac{1}{2} \log ( |\setX|^2 - 1 ) \bigr)} \biggr\rfloor + 1 \Biggr) l \label{eq:chUsesBlocks1ThroughBPl1CCStates}
\end{IEEEeqnarray}
channel uses. It thus remains to exhibit positive integers $B, \, k$, where $k$ is a multiple of $l$, and some $\epsilon > 0$ such that for every sufficiently-large blocklength~$n$ \eqref{bl:choiceBKEpsForLastBlockCCStates2} holds and
\begin{IEEEeqnarray}{l}
B k + \Biggl( \biggl\lfloor \frac{k \log |\setU| + B \log (1 + k) \, |\setS|}{\alpha l \bigl( \log |\setX| - \frac{1}{2} \log ( |\setX|^2 - 1 ) \bigr)} \biggr\rfloor + 1 \Biggr) l \leq n. \label{eq:chUsesBlocks1ThroughBPl1SmNCCStates}
\end{IEEEeqnarray}
(When the inequality in \eqref{eq:chUsesBlocks1ThroughBPl1SmNCCStates} is strict, then Block~$(B + 2)$ deals with all the superfluous epochs: recall that in this block the encoder can transmit arbitrary inputs with the decoder ignoring the corresponding outputs.) As we argue next, when $n$ is sufficiently large we can choose
\begin{subequations}\label{bl:defBkCCStates}
\begin{IEEEeqnarray}{rCl}
B & = & \lfloor \sqrt n \rfloor - \Biggl( \biggl\lfloor \frac{\log |\setU| + \log (1 + \sqrt n) \, |\setS|}{\alpha l \bigl( \log |\setX| - \frac{1}{2} \log ( |\setX|^2 - 1 ) \bigr)} \biggr\rfloor + 1 \Biggr) l, \\
k & = & \biggl\lfloor \frac{\sqrt n}{l} \biggr\rfloor l, \label{eq:defBkkCCStates}
\end{IEEEeqnarray}
\end{subequations}
and we can choose $\epsilon > 0$ for which
\begin{IEEEeqnarray}{l}
R + \epsilon < \min_{\substack{P_S \colon \\ \Ex {}{\lambda (S)} \leq \Lambda}} \max_{P_{U,X|S}} \min_{ \substack{P_{Y|U,X,S} \colon \\ P_{Y|U = u, X,S} \in \mathscr P (W), \,\, \forall \, u \in \setU}} I (U;Y) - I (S;Y). \label{eq:defEpsCCStates}
\end{IEEEeqnarray}
Note that, whenever $n$ is sufficiently large, $B$ and $k$ are positive, $k$ is a multiple of $l$, and \eqref{eq:chUsesBlocks1ThroughBPl1SmNCCStates} is satisfied. To see that also \eqref{bl:choiceBKEpsForLastBlockCCStates2} holds whenever $n$ is sufficiently large, we first observe from \eqref{eq:defBkkCCStates} that $k$ tends to infinity as $n$ tends to infinity. This implies that \eqref{eq:choiceBKEpsForLastBlockKCCStates2} holds whenever $n$ is sufficiently large, and that $\delta (\epsilon, k)$ (which is defined in \eqref{eq:defDeltaOfEpsK}, where $\gamma_k = \gamma_k \bigl( |\setU|, |\setX|, |\setS|, |\setY| \bigr)$ converges to zero as $k$ tends to infinity) converges to $\epsilon$ as $n$ tends to infinity. We next observe that \eqref{bl:defBkCCStates} implies that $B k / n$ converges to one as $n$ tends to infinity, and that
\begin{equation}
\frac{1}{n} \Biggl( \biggl\lfloor \frac{k \log |\setU| + B \log (1 + k) \, |\setS|}{\alpha l \bigl( \log |\setX| - \frac{1}{2} \log ( |\setX|^2 - 1 ) \bigr)} \biggr\rfloor + 1 \Biggr) l \log |\setS|
\end{equation}
converges to zero as $n$ tends to infinity. This, combined with the fact that $\delta (\epsilon, k)$ converges to $\epsilon$ as $n$ tends to infinity and with \eqref{eq:defEpsCCStates}, implies that \eqref{eq:choiceBKEpsForLastBlockCCStates2} holds whenever $n$ is sufficiently large.
\end{proof}

We next prove the converse part of Theorem~\ref{th:ccStates2Capacity}.

\begin{proof}[Converse Part]
We first show that \eqref{eq:positiveCCStates} is necessary for $C_{\textnormal{f},0}^{(2)} (\Lambda)$ to be positive. To this end suppose that \eqref{eq:positiveCCStates} does not hold, i.e., that there exists a pair of states $s, \, s^\prime \in \setS$ satisfying
\begin{equation}
\frac{\lambda (s) + \lambda (s^\prime)}{2} \leq \Lambda \label{eq:positiveCCStates1}
\end{equation}
for which 
\begin{IEEEeqnarray}{l}
\forall \, x, \, x^\prime \in \setX \quad \exists \, y \in \setY \textnormal{ s.t.\ } \channel y {x,s} \, \channel y {x^\prime,s^\prime} > 0.
\end{IEEEeqnarray}
We will show that in this case it is impossible to transmit a single bit error-free whenever $l$ is even. This will imply that $C_{\textnormal{f},0}^{(2)} (\Lambda,l)$ is zero whenever $l$ is even and consequently that $C_{\textnormal{f},0}^{(2)} (\Lambda)$ is zero, because, by definition,
\begin{IEEEeqnarray}{l}
C_{\textnormal{f},0}^{(2)} (\Lambda) = \liminf_{l \rightarrow \infty} C_{\textnormal{f},0}^{(2)} (\Lambda,l).
\end{IEEEeqnarray}

Fix some even $l$ and $s, \, s^\prime$ as above. The proof is similar to that of the converse of Theorem~\ref{th:positive}. Let the bit take values in the set $\setM = \{ 0,1 \}$, and fix a blocklength~$n$ and $n$ encoding mappings $$f_i \colon \setM \times \setS^n \times \setY^{i-1} \rightarrow \setX, \quad i \in [1:n].$$ Denote by $\hat \vecs, \, \check \vecs \in \setS^n$ the state sequences that at odd times are $s$ and $s^\prime$, respectively, and at even times $s^\prime$ and $s$, respectively:
\begin{subequations}
\begin{IEEEeqnarray}{rClCrrClC}
( \hat s_{2i - 1}, \check s_{2i - 1} ) & = & ( s, s^\prime ), \quad &i \in \bigl[ 1 : \left\lceil n / 2 \right\rceil \bigr], \\
( \hat s_{2i}, \check s_{2i} ) & = & ( s^\prime, s ), \quad &i \in \bigl[ 1 : \left\lfloor n / 2 \right\rfloor \bigr].
\end{IEEEeqnarray}
\end{subequations}
Note that, by \eqref{eq:positiveCCStates1} and because $l$ is even, they meet the cost constraint \eqref{eq:ccStates2}. The line of argument leading to \eqref{eq:badOutputs} in the converse of Theorem~\ref{th:positive} implies that there exists an output sequence $\vecy \in \setY^n$ for which
\begin{IEEEeqnarray}{l}
W \bigl( y_i \bigl| f_i (0,\hat \vecs,y^{i-1}), \hat s_i \bigr) \, W \bigl( y_i \bigl| f_i (1,\check \vecs,y^{i-1}), \check s_i \bigr) > 0, \,\, \forall \, i \in [1:n].
\end{IEEEeqnarray}
This rules out error-free transmission, because if the state sequence is either $\hat \vecs$ or $\check \vecs$, then the decoder, not knowing which, cannot recover the bit.\\

We next show that---irrespective of whether or not \eqref{eq:positiveCCStates} holds---$C_{\textnormal{f},0}^{(2)} (\Lambda)$ is upper-bounded by the RHS of \eqref{eq:capacityCCStates}. The proof is similar to the converse of Theorem~\ref{th:capacity}, but in the current setting we cannot fix some PMF $\tilde P_S$ on $\setS$ and assume that the state sequence $S^n$ is drawn IID $\tilde P_S$, because this might violate the cost constraint \eqref{eq:ccStates2}. In fact, \eqref{eq:ccStates2} need not hold even if $\bigEx {}{\lambda (S)} \leq \Lambda$ under $\tilde P_S$.

Fix any $l \in \naturals$, and assume that $n = J l$ for some $J \in \naturals$. We can make this assumption w.l.g., because $$\lim_{n \rightarrow \infty} \frac{ \lfloor n / l \rfloor l}{n} = 1.$$ To satisfy \eqref{eq:ccStates2}, we fix some $l$-type $\tilde P_S$ on $\setS$ w.r.t.\ which
\begin{equation}
\bigEx {}{\lambda (S)} \leq \Lambda, \label{eq:ccStatesAdmissiblePS}
\end{equation}
and we set $\tilde P_{S^n}$ to be the uniform distribution over $\bigl( \setT^{(l)}_{\tilde P_S} \bigr)^J$. Let the PMF $\tilde P_{M,S^n,X^n,Y^n}$ be as in \eqref{eq:convAbsContPMF} but with $\tilde P_S^n$ replaced by $\tilde P_{S^n}$. We can now upper-bound $\frac{1}{n} \log |\setM|$ essentially along the line of argument leading to \eqref{eq:convFixedDists} in the converse of Theorem~\ref{th:capacity}. The main difference is that under $\tilde P_{M,S^n,X^n,Y^n}$ of the current setting $S_i$ and $S^n_{i + 1}$ need not be independent and consequently $$\frac{1}{n} \sum^n_{i = 1} I (S^n_{i + 1};S_i)$$ need not be zero. However, it does tend to zero as $l$ tends to infinity, because
\begin{IEEEeqnarray}{l}
\frac{1}{n} \sum^n_{i = 1} I (S^n_{i + 1};S_i) \nonumber \\
\quad \stackrel{(a)}= \frac{1}{n} \sum^n_{i = 1} \Bigl[ H (S_i) - H (S_i|S^n_{i + 1}) \Bigr] \\
\quad \stackrel{(b)}= H (\tilde P_S) - \frac{1}{n} H (S^n) \\
\quad \stackrel{(c)}= H (\tilde P_S) - \frac{1}{n} \log \, \bigl| \setT^{(l)}_{\tilde P_S} \bigr|^J \\
\quad \stackrel{(d)}\leq H (\tilde P_S) - \frac{J}{n} \Bigl( l H (\tilde P_S) - \log ( 1 + l ) \, |\setS| \Bigr) \\
\quad \stackrel{(e)}= \frac{\log (1 + l) \, |\setS|}{l} \label{eq:mutiFutureCurrentStatesCCStatesLastBefore} \\
\quad \rightarrow 0 \, (l \rightarrow \infty), \label{eq:mutiFutureCurrentStatesCCStates}
\end{IEEEeqnarray}
where $(a)$ holds by the definition of mutual information; $(b)$ follows from the chain rule and the fact that $S_i \sim \tilde P_S$ under $\tilde P_{M,S^n,X^n,Y^n}$; $(c)$ holds because $S^n$ is uniform over $\bigl( \setT^{(l)}_{\tilde P_S} \bigr)^J$ under $\tilde P_{M,S^n,X^n,Y^n}$; $(d)$ follows from the inequality $\bigl| \setT^{(l)}_{\tilde P_S} \bigr| \geq (1 + l)^{-|\setS|} \, 2^{l H (\tilde P_S)}$, where the entropy is computed w.r.t.\ $\tilde P_S$ \cite[Lemma~2.3]{csiszarkoerner11}; and $(e)$ holds because $n = J l$.

Having established \eqref{eq:mutiFutureCurrentStatesCCStates}, we are now ready to conclude the proof. The arguments leading to \eqref{eq:convFixedDists4} in the converse of Theorem~\ref{th:capacity} and \eqref{eq:mutiFutureCurrentStatesCCStatesLastBefore} imply that
\begin{IEEEeqnarray}{rCl}
\frac{1}{n} \log |\setM| & \leq & \frac{\log (1 + l) \, |\setS|}{l} \nonumber \\
& & + \max_{\tilde P_{U,X|S}} \min_{\substack{\tilde P_{Y|U,X,S} \colon \\ \tilde P_{Y|U=u,X,S} \in \mathscr P (W), \,\, \forall \, u \in \setU}} I (U;Y) - I (U;S), \label{eq:convFixedDists4CCStates}
\end{IEEEeqnarray}
where $U$ is an auxiliary chance variable taking values in a finite set $\setU$, and the mutual informations are computed w.r.t.\ the joint PMF $\tilde P_S \times \tilde P_{U,X|S} \times \tilde P_{Y|U,X,S}$. Moreover, it is enough to consider the second term on the RHS of \eqref{eq:convFixedDists4CCStates}, because the first converges to zero as $l$ tends to infinity \eqref{eq:mutiFutureCurrentStatesCCStates} and, by definition, $$C_{\textnormal{f},0}^{(2)} (\Lambda) = \liminf_{l \rightarrow \infty} C_{\textnormal{f},0}^{(2)} (\Lambda,l).$$ To conclude that $C_{\textnormal{f},0}^{(2)} (\Lambda)$ is upper-bounded by the RHS of \eqref{eq:capacityCCStates}, we would have liked to choose some PMF $\tilde P_S$ that---among all PMFs on $\setS$ w.r.t.\ which \eqref{eq:ccStatesAdmissiblePS} holds--- yields the tightest bound, i.e., minimizes
\begin{equation}
\max_{\tilde P_{U,X|S}} \min_{\substack{\tilde P_{Y|U,X,S} \colon \\ \tilde P_{Y|U=u,X,S} \in \mathscr P (W), \,\, \forall \, u \in \setU}} I (U;Y) - I (U;S). \label{eq:ccStatesToMinimizeChoicePS}
\end{equation}
But this is not possible, because $\tilde P_S$ must be an $l$-type. We can, however, choose $\tilde P_S$ as one that---among all $l$-types on $\setS$ w.r.t.\ which \eqref{eq:ccStatesAdmissiblePS} holds---minimizes \eqref{eq:ccStatesToMinimizeChoicePS}. For this choice \eqref{eq:convFixedDists4CCStates} implies that
\begin{IEEEeqnarray}{rCl}
C_{\textnormal{f},0}^{(2)} (\Lambda) & \leq & \liminf_{l \rightarrow \infty} \min_{\substack{\tilde P_S \in \Gamma^{(l)} \colon \\ \Ex {}{\lambda (s)} \leq \Lambda}} \max_{\tilde P_{U,X|S}} \min_{\substack{\tilde P_{Y|U,X,S} \colon \\ \tilde P_{Y|U=u,X,S} \in \mathscr P (W), \,\, \forall \, u \in \setU}} I (U;Y) - I (U;S), \label{eq:convFixedDists5CCStates}
\end{IEEEeqnarray}
where $\Gamma^{(l)}$ denotes the set of $l$-types on $\setS$. To conclude, note that the RHS of \eqref{eq:convFixedDists5CCStates} is equal to that of \eqref{eq:capacityCCStates}: every PMF $\tilde P_S$ on $\setS$ w.r.t.\ which \eqref{eq:ccStatesAdmissiblePS} holds can be approximated in the total variation distance by an $l$-type on $\setS$ w.r.t.\ which \eqref{eq:ccStatesAdmissiblePS} holds when $l$ is sufficiently large; and (conditional) entropy is continuous in this distance \cite[Lemma~2.7]{csiszarkoerner11}.
\end{proof}

%
\end{appendix}

%
\lhead[\fancyplain{\scshape Appendix}
{\scshape Appendix}]
{\fancyplain{\scshape \leftmark}
  {\scshape \leftmark}}
\rhead[\fancyplain{\scshape \leftmark}
{\scshape \leftmark}]
{\fancyplain{\scshape Appendix}
  {\scshape Appendix}}
%
%
%
%
%

%
%
%
%
\end{document}
%